\documentclass[12pt,titlepage,twoside,a4paper,BCOR0.25cm]{scrartcl}
\usepackage[title]{appendix}

\usepackage{amsmath, amsthm}
\usepackage{dsfont}
\usepackage{mathrsfs}
\usepackage{ amssymb}
\usepackage{amsfonts}
\usepackage{accents}
\usepackage{eucal} 

\usepackage{array}

\usepackage{scrpage2}
\usepackage{graphicx,color}

\usepackage{hyperref}

\clearscrheadfoot
\ohead{\headmark}
\ofoot[\pagemark]{\pagemark}
\automark[subsection]{section}
\hypersetup{bookmarks=true,
colorlinks=true,
linkcolor=blue, citecolor=blue, menucolor=blue,
 pdfauthor={Olaf Milbredt}, pdftitle={The Cauchy Problem for
 Membranes}, pdfsubject={Geometric Analysis},
pdfkeywords={Nonlinear Wave Equations, Membranes, Cauchy Problem}}

\newcounter{vol2}
\newcommand {\rr}{{\mathds R}}

\newcommand{\Lie}{{\mathcal{L}}} 
\newcommand {\tsum}{{\textstyle{\sum}}}
\newcommand {\tint}{{\textstyle{\int}}}
\newcommand {\rrangle}{{\rangle\kern - 2.5pt \rangle}}
\newcommand {\llangle}{{\langle\kern - 2.5pt \langle}}

\newcommand{\restr}[1]{#1\bigr|_{t=0}}
\newcommand{\restrb}[1]{#1\bigr|_{\bar{t}=0}}
\newcommand{\abs}[1]{\lvert#1\rvert}
\newcommand{\babs}[1]{\bigl\lvert#1\bigr\rvert}
\newcommand{\mabs}[1]{\langle#1\rangle}
\newcommand{\norm}[1]{\lVert#1\rVert}
\newcommand{\bnorm}[1]{\bigl\lVert#1\bigr\rVert}

\newcommand {\init}[1]{{\mathring{#1}}}

\newcommand {\si}[1]{{\underaccent{\bar}{#1}}}
\newcommand{\II}{{\init{I\kern - 2.0pt I}}{}}
\newcommand{\IIfull}{{I\kern - 2.0pt I}{}}
\newcommand{\slice}{\mathscr{S}}
\newcommand{\D}{\mathbf{D}}

\newcommand{\R}{\mathbf{R}}
\newcommand{\Rr}{\mathcal{R}}
\newcommand{\G}{\mathbf{\Gamma}}
\newcommand{\BF}{\mathbf{f}}
\newcommand{\del}{\mathcal{D}} 
\newcommand{\hg}{\widehat{\mathbf{\Gamma}}{}}

\newcommand{\hT}{\widehat{T}}
\newcommand{\indg}{\widehat{\Gamma}}
\newcommand{\hG}{\widehat{\Gamma}}
\newcommand{\ca}{C_{\mathrm{arg}}}
\newcommand{\lip}{\mathrm{Lip}}
\newcommand{\hn}{\widehat{\nabla}}
\newcommand{\dt}{\tfrac{d}{dt}}
\newcommand{\dtb}{\tfrac{d}{d\bar{t}}}

\newcommand{\ig}{\init{g}{}}
\newcommand{\iG}{\init{\Gamma}{}}
\newcommand{\inab}{\init{\nabla}{}}
\newcommand{\ibarg}{\init{\bar{g}}{}}
\newcommand{\hbarG}{\widehat{\bar{\Gamma}}{}}

\newcommand{\2}{1/2}

\DeclareMathOperator{\Riem}{Rm}
\DeclareMathOperator{\RRiem}{\mathbf{Rm}}
\DeclareMathOperator{\Rriem}{\mathcal{R}}

\DeclareMathOperator{\tr}{tr}

\DeclareMathOperator{\im}{im}
\DeclareMathOperator{\dom}{domain}
\DeclareMathOperator{\diag}{diag}

\DeclareMathOperator{\id}{id}

\DeclareMathOperator{\Lip}{Lip}

\theoremstyle{plain}
    \newtheorem{thm}{Theorem}[section]
    \newtheorem{lem}[thm]{Lemma}
    \newtheorem{prop}[thm]{Proposition}
    \newtheorem{cor}[thm]{Corollary}
    
    \newtheorem{assum}[thm]{Assumptions}

\theoremstyle{definition}
    \newtheorem{defn}[thm]{Definition}
\theoremstyle{remark}
    \newtheorem{rem}[thm]{Remark}
    
    \newtheorem{case}{Case}

\numberwithin{equation}{section}
\numberwithin{equation}{subsection}

\begin{document}
\setcounter{page}{-6}

\begin{titlepage}
  \subject{Dissertation}
 \title{The Cauchy Problem for Membranes}
 \publishers{ dem Fachbereich Mathematik und Informatik\\
 der Freien Universit\"at Berlin zur Erlangung des Grades eines Doktor Rerum 
 Naturalium (Dr.\,rer.\,nat.)\\
 angefertigt am Max-Planck-Institut f\"ur Gravitationsphysik\\
 (Albert-Einstein-Institut)
  \\
 \begin{center}
 \begin{tabular}{cc}
    \begin{minipage}[c]{.32\linewidth}
     {\par
           }
    \end{minipage}
    \begin{minipage}[c]{.32\linewidth}
     { \par
           }
    \end{minipage} 
 \end{tabular}
 \end{center}
 1. Gutachter: Prof.\ Dr.\ Gerhard Huisken \\
 2. Gutachter: Dr.\ Alan Rendall \\
 Disputation: 14. Juli 2008
 }
 \author{vorgelegt von \\[0.5cm]
 Olaf Milbredt 
 }
 \date{}
 \end{titlepage}

\maketitle
\thispagestyle{empty} 
\begin{center}
{\Large Zusammenfassung in deutscher Sprache}
\end{center}
Diese Dissertation untersucht das Anfangswertproblem f\"ur Membranen.
Eine Membran ist eine raumartige Untermannigfaltigkeit $\Sigma_0$
in einer umgebenden
Mannigfaltigkeit ausgestattet mit einer Lorentz-Metrik.
Die Bewegungsgleichung f\"ur eine Membran ist gegeben durch die 
 folgende Bedingung f\"ur das Weltvolumen $\Sigma$ der Membran, das
eine zeitartige Untermannigfaltigkeit darstellt, welches durch Bewegung
aus der Membran entsteht.
Die Bewegung soll so erfolgen, dass das induzierte
Volumen des Weltvolumen extremal wird. Die Euler-Lagrange Gleichung f\"ur das
Volumenfunktional ergibt sich zu
\begin{gather*}
  H(\Sigma) \equiv 0,
\end{gather*}
 wobei $H(\Sigma)$ als der mittlere Kr\"ummungsvektor von $\Sigma$ definiert 
ist. In der String-Theorie, ein Kanditat f\"ur die Vereinheitlichung der
Allgemeinenen Relativit\"atstheorie und der Quantenfeldtheorie, werden
solche Objekte $p$-Branen genannt. Diese sind h\"oher-dimensionale 
Analoga zu strings mit den gleichen dynamischen Eigenschaften.
\\
Wir betrachten folgendes Anfangswertproblem. 
\begin{center}
  \parbox{0.95\textwidth}{
    Sei eine raumartige Untermannigfaltigkeit $\Sigma_0$  gegeben. 
    Sei $\nu$ ein zeitartiges zukunfts\-gerichtetes
    Einheitsvektorfeld auf $\Sigma_0$.
    
    \textbf{Existenz:}  Finde offene zeitartige
    L\"osung  $\Sigma$ des Problems 
    \begin{gather*}
    H(\Sigma) \equiv 0,~ \Sigma_0 \subset \Sigma
    \text{ und }\nu \text{ ist tangential zu }\Sigma.
    \end{gather*}
      
    \textbf{Eindeutigkeit:} Seien $\Sigma_1$ und $\Sigma_2$ zwei L\"osungen
    des Problems. Zeige, dass eine offene Umgebung 
    $\Sigma_0 \subset U$ von $\Sigma_0$ existiert mit
    \begin{gather*}
      U \cap \Sigma_1 = U \cap \Sigma_2.
    \end{gather*}
    }
\end{center}
Die Membrangleichung f\"uhrt zu einem quasilinearen System hyperbolischer
Differential\-gleichungen. Um jedoch Existenz und Eindeutigkeit f\"ur solche
Systeme benutzen zu k\"onnen, wird die Diffeomorphismeninvarianz der
Gleichung durch die Wahl einer geeigneten Eichung gebrochen.
In dieser Arbeit wird eine Eichung benutzt, die auch bei der L\"osung
der Einstein-Gleichung der Allgemeinen Relativit\"atstheorie zum Einsatz kommt.

Es wird eine positive Antwort auf das Problem der Existenz und Eindeutigkeit
gegeben.  Das Ergebnis gilt
f\"ur jede Dimension der umgebenden Mannigfaltigkeit
und jede Codimension von $\Sigma_0$.
Ebenso sind nicht-kompakte Anfangswerte zul\"assig und Anfangswerte, die
beliebig nahe am Lichtkegel sind.
Im Falle einer global hyperbolischen Umgebungsmannigfaltigkeit und
uniformen Schranken an Ableitungen der zweiten Fundamentalform
und Ableitungen des Anfangswertes $\nu$ sowie quantitative Kontrolle
des Abstandes der Anfangswerte zum Lichtkegel wird eine Existenzzeit
der L\"osung $\Sigma$ in geometrischen Gr\"o\ss en gegeben.

\newpage
\thispagestyle{empty} 
 \mbox{}
 \vfill
\mbox{}\\
\begin{center}
\large{Danksagungen}
\end{center}
An dieser Stelle nehme ich die Gelegenheit wahr, allen zu danken, die
an der Fertigstellung dieser Arbeit beteiligt waren.
An erster Stelle danke ich meinem Betreuer \textsc{Prof.\,Dr.\,Gerhard Huisken}
f\"ur seine Zeit, seine Ermutigung und vor allem f\"ur sein Verst\"andnis. 
Ohne seine Hilfe w\"are diese Arbeit niemals Realit\"at geworden.
Insbesondere bin ich sehr dankbar f\"ur die M\"oglichkeit hier am 
Albert-Einstein-Institut studiert haben zu k\"onnen.
Des Weiteren m\"ochte ich Steve White, Paul T.\ Allen und Carla Cederbaum
f\"ur Korrekturen an Teilen der Arbeit danken.
Den Angeh\"origen der Arbeitsgruppe ``Geometrische Analysis und Gravitation''
m\"ochte ich f\"ur die gute Atmosph\"are danken, insbesondere 
Bernhard List und Tilman Vogel,
die mir bei mathematischen Problemen und  bei der technischen Umsetzung eine 
gro\ss e Hilfe waren.
Einen besonders gro\ss en Dank an Sameh Keliny f\"ur die st\"andige
Ermutigung. %
Einen gro\ss en Dank auch an Melanie Henniger f\"ur die Unterst\"utzung bei
praktischen Problemen.
Am Schlu\ss\ danke ich meiner Familie: Christel und Klaus-Dieter Milbredt
und meinem Bruder Jens Milbredt f\"ur viele Gespr\"ache und Zuspruch.

\newpage
\pagenumbering{roman}

\tableofcontents

\setcounter{section}{-1}
\section{Introduction}
\subsection{Geometric differential equations}
Nonlinear geometric partial differential equations play a fundamental
role in mathematical physics. The requirement that the physical
laws be formulated in a coordinate invariant way leads 
naturally to geometric equations
concerning curvature. Recently, developments in the field of nonlinear 
equations
make it possible to justify models of theoretical physics.

This can be done by showing existence and uniqueness of solutions
under circumstances where the physics suggests that existence and uniqueness
should hold. The main result 
of this thesis is existence and uniqueness for an equation originating in
higher-dimensional extensions of String Theory.

Examples of solutions of a geometric equation 
are minimal surfaces and surfaces of prescribed
mean curvature. Additionally, physical situations involving surface tension
such as capillary surfaces and soap surfaces lead to curvature equations.
In General Relativity the fundamental equation consists of the vanishing of
the Ricci curvature, which leads to a quasilinear hyperbolic equation.
Solving the initial value problem is a central issue in General Relativity.
In contrast to solutions of linear equations,
solutions of nonlinear equations can have singular behaviour. 
To cope with this behaviour new analytical methods need
to be applied. In this work we consider existence and uniqueness of
solutions to the initial value problem for a
hyperbolic minimal surface, whose structure is similar to the initial value
problem of General Relativity.

\subsection{Unified Theories}
In theoretical physics the search for a unified description of the universe is 
a central issue. Two successful theories exist at this time representing 
different viewpoints. Quantum Field Theory deals with very small length scales. 
Forces between entities are transferred by particles --- a 
central 
part of this theory. The other viewpoint, General Relativity, acts at great 
length scales and is the theoretical groundwork for the movement of planets 
and stars. 
The difference becomes more obvious when considering the role of space and
time in both theories. Whereas in Quantum Field Theory they are assumed to
be fixed, in General Relativity both quantities are glued together in
a manifold called \emph{spacetime}. This manifold carries a Lorentzian 
metric, i.\,e.\ the metric has
one negative 
eigenvalue and the others are positive. 
The metric of the spacetime has to satisfy
the Einstein equations relating the Ricci curvature to the stress-energy 
tensor. The latter tensor models information about the distribution of matter 
and
energy within the universe. Some attempts exist to 
overcome the inconsistency in the description of nature. One of them is String 
Theory. 
In String Theory one-dimensional objects called \emph{strings}
are  studied in place of zero-dimensional particles 
of %
Quantum Field Theory. Considered 
as classical objects, i.\,e.\ without 
connection to
Quantum Field Theory, strings are submanifolds of a Lorentzian manifold.
Recently, even higher dimensional objects called \emph{$p$-branes} 
have been introduced
as part of an extension of String Theory ($p$ = 1). 
Branes also arise as Dirichlet 
boundary values for strings if one is not dealing with free strings
and are called \emph{D-branes} in this case.
The next section describes  the equation of motion for 
$p$-branes.

\subsection{Membranes}
\label{sec:intro_membranes}
To phrase the problem we consider %
 two notions from
Lorentzian geometry. A submanifold of a Lorentzian ambient manifold is
called \emph{spacelike} if the induced metric is Riemannian and \emph{timelike}
if the induced metric is Lorentzian.

A spacelike submanifold $\Sigma_0$ of a Lorentzian ambient manifold is 
called
a \emph{membrane} in this thesis.
During a time evolution a membrane sweeps out a timelike
submanifold $\Sigma$ called %
\emph{world volume}. 
If $\dim \Sigma_0 = p$, then $\dim \Sigma = p + 1$;
in physics literature the term membrane is reserved for two-dimensional objects
($p = 2$),
and the spacelike submanifold $\Sigma_0$ is called a $p$-brane
for higher-dimensional objects.
If $p = 1$, then $\Sigma_0$ is called string and
$\Sigma$ is called the worldsheet of the
string.

The equation of motion of a membrane is determined by 
the condition that the world volume of the membrane is a critical
point for the volume functional induced by the ambient manifold.
The Euler-Lagrange equation of the volume functional is the vanishing of the 
mean curvature of the 
world volume $\Sigma$. A discussion of this fact in the context of  
Riemannian manifolds can be found in 
\cite{Gal:2004}.
The central system of partial differential equations considered in this
thesis is therefore
\begin{gather}
  \label{eq:H0}
  H(\Sigma) \equiv 0, 
\end{gather}
where $H(\Sigma)$ denotes the mean curvature vector of the submanifold $\Sigma$.
Equation \eqref{eq:H0} will be called the \emph{membrane equation} in
the sequel.

Surfaces in Euclidean space satisfying this equation
are called minimal surfaces. 
Many aspects of minimal surfaces are well understood, a treatment in Euclidean
space was done by J. Douglas in \cite{Do:1939} and in general Riemannian 
manifolds by C.\,B. Morrey in \cite{Mo:1948}.

A powerful tool for %
investigating such geometric equations is the
theory of partial differential equations. 
Here, a distinction arises between minimal submanifolds in Riemannian geometry
and timelike
world volumes in Lorentzian geometry.
Whereas the equation describing minimal submanifolds leads to elliptic 
equations, the
Lorentzian metric induced on the world volume causes the resulting equations
for membranes
to be hyperbolic. Due to the dependency of the 
induced volume on the induced metric of the submanifold the derived equations
in both cases are nonlinear.

Static examples of membranes can be obtained from minimal surfaces,
which are solutions in Minkowski space if a timelike real line is attached.
Exact solutions corresponding to pulsating and rotating objects
have been studied by H. Nicolai and J. Hoppe in \cite{NiHo:1987}.

A natural question for 
hyperbolic equations is the \emph{initial value problem} (IVP) or
\emph{Cauchy problem}. Solving such a Cauchy  problem  under appropriate 
assumptions only involving geometric
quantities
will be the central issue of this work.
The Cauchy problem for closed strings, i.\,e.\ membranes
diffeomorphic to the circle line, has been studied 
in Minkowski space and in globally hyperbolic ambient manifolds.
T. Deck \cite{Deck:1994} studied  geometric evolution problems
for a string using the
work of C.\,H. Gu on the motion of strings (cf.\ \cite{Gu:1983}). 
Solutions are obtained as timelike immersion in a neighborhood of the initial
data.
In  \cite{Muller:2007}, 
O.\ M\"uller also solved 
the Cauchy problem for a closed
string adapting another work of C.\,H. Gu on the existence of harmonic maps 
from
two-dimensional Minkowski space to a Riemannian manifold 
(cf.\ \cite{Gu:1980}). In these works the question of global existence 
was investigated for globally hyperbolic ambient manifolds. 
The notion of global existence used 
involves the construction of a map 
which does not necessarily represent a regular submanifold.

The existence
of global-in-time solutions and stability are further questions. 
T. Deck showed stability for a subclass of globally hyperbolic target manifolds
and O. M\"uller showed global existence 
for general globally hyperbolic targets.
In case the world volume in Minkowski space is represented as a graph,
global existence for small
initial data was shown in \cite{Lind:2004}. 
This work uses techniques of D. 
Christodoulou (cf.\ \cite{Cr:1986})
and S. Klainerman (cf.\ \cite{Kl:1986}). These authors 
showed the stability of Minkowski space
satisfying the Einstein equations
of General Relativity (cf.\ \cite{CHRKLA:1993}). 
This result was generalized to arbitrary codimensions in \cite{AA:2006}.

Using techniques previously applied to the Einstein equation
by D. 
Christodoulou %
and S. Klainerman, %
S. Brendle (cf.\ \cite{Brendle:2002}) investigated 
the stability of a
hyperplane satisfying the membrane equation \eqref{eq:H0} in Minkowski space.
In contrast to the work of H. Lindblad, S. Brendle requires less regularity
of the initial data, comparable to that in \cite{CHRKLA:1993}.

\subsection{Main results}
We now introduce a formulation of an initial value problem
for the membrane equation \eqref{eq:H0}. The goal is to set up a problem
independent of any choice of coordinates or parametrizations of the
geometric quantities.
Initial values for hyperbolic equations are composed of an initial position
and an initial velocity. Here, one needs to
find a geometric quantity which can fill these parts.

In the following we state the main problem of this thesis in non-technical
terms.
\begin{center}
  \parbox{0.93\textwidth}{
    \begin{center}
      \textbf{Main Problem}
    \end{center}
    Let $(N^{n+1}, h)$ be an $(n+1)$-dimensional Lorentzian manifold called 
    \emph{ambient manifold}
  and let $\Sigma_0$ be an $m$-dimensional spacelike
  regularly immersed submanifold of $N$ called \emph{initial submanifold}. 
  Assume $\nu$ to be a unit timelike future-directed vector field normal to
  $\Sigma_0$  which will be called 
  \emph{initial direction}.
\vspace{0.5cm}
}

\parbox{0.93\textwidth}{
   \textbf{Existence} \\
  Find an open $(m + 1)$-dimensional regularly immersed submanifold $\Sigma$
  carrying a Lorentzian metric and satisfying
  \begin{gather}
    \label{eq:geom_problem}
    H(\Sigma) \equiv 0,~ \Sigma_0 \subset \Sigma,
    \text{ and }\nu \text{ is tangential to }\Sigma.
  \end{gather}
}
\parbox{0.93\textwidth}{
  \textbf{Uniqueness} \\
  Show that for  $\Sigma_1$ and $\Sigma_2$ solving the IVP 
  \eqref{eq:geom_problem}, 
  there exists an open set $V$ with
  $\Sigma_0 \subset V \subset N$ such that 
  \begin{gather}
    \label{eq:geom_uni}
    V \cap \Sigma_1 = V \cap \Sigma_2.
  \end{gather} 
}
\end{center}  
The main result of this thesis is an affirmative answer to the problem above
under suitable conditions.
Since equation \eqref{eq:H0}, as a geometric equation,
 is invariant under 
diffeomorphisms of the world volume $\Sigma$, the membrane equation
is degenerate hyperbolic; this means that technically, the principal symbol 
of the 
differential operator has zero eigenvalues.
As for the Einstein equations a choice of  gauge is needed. From
H. Friedrich and A. Rendall's treatise of the Cauchy problem for the Einstein
equations in 
\cite{FriRe:2000} we make use of
the so-called harmonic map gauge, a generalized
form of the harmonic (or wave-) coordinates used by Y. Choquet-Bruhat
to solve the Cauchy problem for the Einstein equations (cf.\ \cite{ChBr:1952}). 
This choice leads to
a reduced equation which is quasilinear hyperbolic of second order; thus
we need to investigate an existence theory for such equations.

By adapting the results of T. Kato (cf.\ \cite{Kato:1975}) for symmetric
hyperbolic systems of first order to second-order equations, we obtain 
an existence theory with emphasis on 
a lower
bound on the time of existence for a solution. 
The corresponding Theorem
\ref{qlin_ex} shows existence of classical
solutions  of quasilinear wave equations
using Sobolev space theory. 
To make use  of the Sobolev embedding theorem, boundedness
of higher-order derivatives is needed.

In Theorem \ref{thm:ex_uni_geom} we consider general 
initial velocities
with normal component pointing in direction $\nu$ and obtain parametrized 
immersions solving the Cauchy problem 
\eqref{eq:geom_problem} in Minkowski space. 
The conditions imposed on the initial data are bounds for
derivatives of the second fundamental form %
of the
initial submanifold and
the initial velocity up to order $s$, resp. $s + 1$,
where $s > \tfrac{m}{2} + 1$
denotes an integer. These tensors were measured with the Euclidean metric on 
Minkowski space.
The value $s$, and therefore the dimension of the 
initial submanifold, enters the conditions through 
the use of Sobolev spaces in the
existence theory of T. Kato (cf.\ \cite{Kato:1975}).
To bound the initial values away from the light cone
we introduce a bound for the infimum of angles between unit timelike directions
normal to the initial submanifold and the timelike direction of Minkowski
space. Furthermore, a bound is imposed for the angle between the initial 
velocity
and the timelike direction. These conditions give us a measure of how far
the initial values are away from the light cone. A lower bound for the
time parameter of the differential equation
is obtained by uniformity of these conditions
and in view of the bounds on the smoothness of the initial data.

By introducing a graph representation of the initial submanifold similar to 
graph representations for 
hypersurfaces of Euclidean space, but adapted to Minkowski space and valid
for any codimension, we fulfill the conditions needed by Theorem 
\ref{thm:ex_uni_atlas} for the problem when localized in space.

A generalization to a globally hyperbolic 
Lorentzian ambient manifold is given in Theorem
\ref{thm:ex_uni_geom_hyp}. 
To be able to adopt the conditions imposed on the initial values in Minkowski
space,  
a substitute for the
timelike direction is obtained in the following manner. The Lorentzian ambient
manifold is assumed to admit a time function in a neighborhood of the initial
submanifold, i.\,e.\ a function whose gradient is timelike everywhere. 
The unit normal to the level sets of this time function in
that neighborhood is used to formulate the conditions on the 
angles between the lightcone and the timelike and spacelike directions of
the initial data, replacing
the timelike direction of Minkowski space.

The ambient manifold has to satisfy 
certain conditions. The foliation with levelsets of the time function
gives a natural Riemannian metric by flipping the sign of the unit normal
to the slices. The conditions imposed on the ambient manifold 
are boundedness of derivatives of the
curvature up to order $s + 1$, control over the 
norm of the gradient of the time function measured with the Lorentzian
metric,  and boundedness of derivatives of the gradient
of the time function up to order $s + 2$. 
The conditions on the initial values are the same as in Minkowski space
if one uses %
the flipped Riemannian metric to measure the norm of 
the second fundamental
form and the initial velocity. 

Solutions obtained by the above procedure  are non-geometric objects. 
As a next step towards the solution to problem \eqref{eq:geom_problem}
we show in Theorem \ref{thm:ex_uni_lapse_hyp} that we can solve the initial
value problem  for 
immersions if we allow the initial
velocity to have a tangential part called \emph{shift} and a factor ---
called \emph{lapse} --- scaling the initial direction.
The uniqueness part of this theorem
shows that the construction of geometric solutions to \eqref{eq:geom_problem}
in the preceding cases are
independent of the choice of an immersion of the initial submanifold,
the initial  
lapse and shift.
This is done by constructing a reparametrization of two solutions
in order to change one gauge into the other.
The condition for such a change leads 
 to the problem of solving a harmonic 
map equation from a Lorentzian manifold to another Lorentzian manifold.
This equation is solved using the results for quasilinear hyperbolic
equations previously developed for solving the membrane equation.

The existence part of the main problem is shown in Theorem \ref{thm:main_ex}.
Uniform conditions on ambient manifold, initial submanifold, and initial 
direction
lead to a lower bound on the time of existence of a solution in the sense 
that one obtains a lower bound on the length of timelike curves emanating
from the initial submanifold.
 Local conditions 
lead to a 
solution in a neighborhood of the initial submanifold as is shown
in Corollary \ref{cor:main_ex}.
Uniqueness of solutions to the Cauchy problem \eqref{eq:geom_problem}
is established in Theorem \ref{thm:main_uni}.
The strategy is to compare a solution to the initial value problem
\eqref{eq:geom_problem}
with a solution constructed in Corollary \ref{cor:main_ex}.
A parametrization of the given solution 
meeting the gauge condition used for the construction of
a solution is obtained by solving a harmonic map equation similar
to that in the above case. 
No conditions on the dimensions are necessary, nor is  compactness.  Therefore
known examples for minimal surfaces, such as the plane and the catenoid, 
are permitted as 
initial values.

A natural question for scale-invariant equations such as the membrane 
equation
is the behaviour of a solution under rescaling. 
The results presented here are scale-invariant, in the sense that
we show that the time of existence
 equals the scale multiplied by a scale-invariant constant.

Solutions to the geometric Cauchy-problem \eqref{eq:geom_problem} were 
obtained for various
notions of the initial submanifold to be regularly immersed,
locally embedded,
and with locally finite intersections, in such a way that the solution has
the same structure as the initial submanifold. Furthermore, uniqueness results
in the sense of the main problem are established for these types
of submanifolds.

The goal of this work was to provide short-time solutions 
to the Cauchy problem \eqref{eq:geom_problem}.
Let us conclude this discussion by proposing directions for future research.
An extension result and the analysis of the long-time behaviour of solutions
of the membrane equation would be interesting to study.
Another related goal would be the understanding of the structure
of singularities occurring at finite time.
From the viewpoint of the theory of partial differential equations
it would also be desirable to weaken the condition needed on the 
differentiability of the initial data.

\subsection{Thesis overview}
We will now give a short overview of this work.
\emph{Chapter \ref{sec:prel}} provides notations to be used, various
formulas, and the choice of the gauge to be used to reduce the membrane
equation. The existence theory for  hyperbolic equations obtained 
by the reduction process is the subject of \emph{Chapter \ref{sec:hyp}}.
We obtain a lower bound on the time of existence of a solution and 
spatially-local uniqueness of solutions. More technical and non-enlightening
lemmata are proven
in \emph{Chapter \ref{sec:proofs}} of the appendix.
To provide an easier access to solutions
of the membrane equation, \emph{Chapter \ref{sec:ex_mink}} is devoted to the 
membrane
equation in the case where the ambient manifold is Minkowski space.
By examining the Minkowski case one gains an understanding of
what conditions should be imposed when generalizing to an arbitrary ambient
manifold, which is done in \emph{Chapter \ref{sec:mem_ex_hyp}}.
Up to this point all solutions were obtained as 
immersions; \emph{Chapter \ref{indpar}} is then devoted to showing 
that each solution is independent of the choice of %
immersion of the initial submanifold, and of lapse and shift of
the initial velocity.
Finally, \emph{Chapter \ref{main}} gives proofs of the main 
existence
and uniqueness results providing an estimate on the time of existence.

\newpage
\thispagestyle{empty}
 \mbox{}
 \newpage
\setcounter{page}{1}
\pagenumbering{arabic}

\section{Membrane equation}
\label{sec:prel}
\subsection{Notations}
\label{sec:notation}
Let $N^{n+1}$ be an $(n+1)$-dimensional  smooth manifold endowed with
a Lorentzian metric $h$. 
The Levi-Civita connection w.r.t.\ $h$ and the corresponding Christoffel
symbols will be denoted by $\D$ and
$\G$, respectively. The $(1,3)$ or the $(0,4)$-version 
of the curvature of $h$ will be denoted by $\RRiem$ or $\R$.

Coordinates on $N$ will carry two sets of indices as follows.
Capital Latin letters as $A,B,C,\dots$ 
will run from $0$ to $n$ and small underlined Latin letters 
as $\si{a},\si{b},\si{c},\dots$ will run from
$1$ to $n$. Our convention for the signature of a Lorentzian metric will be 
$(\,{-}\,{+}\,\cdots\,{+}\,)$, so that the $0$th component denotes the timelike
direction.
Partial derivatives in coordinates will be abbreviated by 
$\partial_A = \tfrac{\partial}{\partial y^A}$ and we abbreviate covariant
derivatives  by $\D_A = \D_{\partial_A}$. A covariant derivative of order
$\ell$ will be denoted by $\D^{\ell}$.
We will use the notation of
contracted Christoffel symbols defined by $\G^A = h^{BC} \G_{BC}^A$.
If $E$ denotes a Riemannian metric on $N$, then we use the coordinate 
independent norm on
tensors $T \in {\cal T}_k^{\ell}(N)$ induced by $E$ as follows
\begin{gather*}
  \abs{T}^2_E = E^{A_1 C_1} \cdots E^{A_k C_k} \, E_{B_1 D_1} \cdots E_{B_{\ell} D_{\ell}}\,
    T_{A_1, \dots, A_k}^{B_1, \dots, B_{\ell}} \, T_{C_1, \dots, C_k}^{D_1, \dots, D_{\ell}}.
\end{gather*}
Coordinate derivatives of order $\ell$ will be denoted by $\partial^{\ell}$
and the vector of all coordinate derivatives of order $\ell$
by $D^{\ell}$. For a function $f$ it then follows 
$D^{\ell} f = (\partial^{\alpha} f)_{\abs{\alpha} = \ell}$, where 
$\alpha$ %
designates a multi-index.
The notation for the Euclidean norm induced by a coordinate system of a 
coordinate dependent
term  $T_{A_1, \dots, A_k}^{B_1, \dots, B_{\ell}}$,
not necessarily the components of a tensor, 
is given by
\begin{gather*}
  \abs{(T_{A_1, \dots, A_k}^{B_1, \dots, B_{\ell}})}^2 = 
  \abs{T}^2_e = \tsum_{A_1, \dots, A_k,B_1, \dots, B_{\ell}} 
  \abs{T_{A_1, \dots, A_k}^{B_1, \dots, B_{\ell}}}^2.
\end{gather*}
To designate coordinate indices on timelike and spacelike submanifolds
of $N$ we use the following convention.
Let $\Sigma_0$ be a spacelike submanifold of dimension $m$. Coordinates on
$\Sigma_0$ will carry small Latin indices as $i,j,k,\dots$ running run
from $1$ to $m$.
An adapted coordinate system $\partial_A$ on $N$ will carry %
small Latin indices as $i,j,k,\dots$ for directions tangent to $\Sigma_0$,
Greek indices $\alpha,\beta,\dots$ having the range $\{0, m+1,\dots, n\}$
denoting normal directions to $\Sigma_0$. %
Small Latin letters as $a,b,c, \dots$ will range from $m+1$ to $n$ denoting
spacelike normal directions to $\Sigma_0$.

Let $\Sigma$ be a timelike $(m+1)$-dimensional submanifold of $N$.
Local coordinates on $\Sigma$ will carry 
Greek indices as $\mu,\nu,\lambda,\dots$ and 
will run from $0$ to $m$. The indices $a,b,c, \dots$ described above will be 
used for directions normal to $\Sigma$ being part of an
adapted coordinate system.

In the following we first fix some notations concerning geometric quantities
of timelike submanifolds. The next section will then discuss the
notation for a spacelike submanifold.
The metric and connection induced on $\Sigma$ are given by
\begin{gather*}
  g := h\bigr|_{\Sigma} \quad\text{and}\quad
  \nabla_X Y := (\D_X Y)^{\top} \text{ for vector fields } X, Y \text{ tangent
    to }
  \Sigma.
\end{gather*}
The Christoffel symbols of $\nabla$ will be denoted by $\Gamma$.
The curvature will be denoted by $\Riem$ or $R$.
To describe the structure of the submanifold we need a notion describing 
the interplay of the curvature of $N$ and the submanifold.
Such a notion is given by the exterior curvature of $\Sigma$ in $N$.
In the following definition we use the notation $\mathcal{T}(\Sigma)$
for smooth sections of the tangent bundle $T\Sigma$ of $\Sigma$ and
$\mathcal{N}(\Sigma)$ for smooth sections of the normal bundle $N\Sigma$
of $\Sigma$.
\begin{defn}[2nd fundamental form]
  Let $X, Y$ be two vector fields on $\Sigma$ extended to vector fields
  on $N$. Define a mapping
  \begin{gather*}
    \IIfull : \mathcal{T}(\Sigma) \times \mathcal{T}(\Sigma) 
    \rightarrow \mathcal{N}(\Sigma)
    \quad\text{by} \quad
    \IIfull(X,Y) := (\D_X Y)^{\bot}.
  \end{gather*}
  The tensor $\IIfull$ is called \emph{second fundamental form}.
  It is symmetric and independent of the extensions of the vector fields.
\end{defn}
The relation between the curvature $\RRiem$ of $N$ and $\Riem$ of $\Sigma$
is given by the \emph{Gau\ss-equation}
\begin{gather*}
  g\bigl(R(X,Y)Z, W \bigr)  = g\bigl( \R(X,Y) Z, W\bigr)
  + g\bigl( \IIfull(X,W), \IIfull(Y,Z)\bigr) 
  - g\bigl(\IIfull(X,Z), \IIfull(Y, W) \bigr)
\end{gather*}
and the Codazzi-equation which needs a further definition.
Set
\begin{multline*}
  (\nabla^{\bot}_X \IIfull)(Y,Z)  :=
  \nabla^{\bot}_X\bigl(\IIfull(Y,Z)
  \bigr) - \IIfull(\nabla_X Y, Z) - \IIfull(Y, \nabla_X Z)
  \\
  \text{with } \nabla^{\bot}_X W
   = ( \D W )^{\bot}.
\end{multline*}
The \emph{Codazzi-equation} is then given by
\begin{gather*}
  \bigl(\R(X,Y) Z\bigr)^{\bot}  =
  (\nabla^{\bot}_X\IIfull)(Y,Z) - (\nabla^{\bot}_Y \IIfull)(X,Z)
\end{gather*}
We are now able to define the main notion of this thesis.
\begin{defn}[Mean curvature]
The \emph{mean curvature} $H \in \mathcal{N}(\Sigma)$ 
of a submanifold $\Sigma$ is given by
\begin{gather}
  \label{eq:mean_def}
  H = %
      \tr_g \IIfull
\end{gather}
\end{defn}
\begin{rem}
  The usual definition of the mean curvature involves a factor 
  dependent on the 
  dimension. It was omitted here since we are only interested in the 
  homogeneous equation.
\end{rem}
The next section describes the geometric equations for an immersion
of a spacelike submanifold which will play the role of an initial value for the
membrane equation.
\subsection{Geometry of immersions}
\label{sec:immersion}
We begin with a definition of the kind of manifolds we will use as initial 
data.
\begin{defn}
  \label{defn:immersed_smf}
  Let $M^m$ be an $m$-dimensional manifold and $\varphi: M \rightarrow N$
  be a mapping. The image of $\varphi$, denoted by $\Sigma_0:= \im \varphi$,
  is called a \emph{regularly immersed submanifold of $N$} if
  the differential of $\varphi$ is non-singular
  at every point of $M$ .
\end{defn}
Let $\Sigma_0$ be a $m$-dimensional spacelike regularly immersed submanifold of 
$N$.
We denote geometric quantities referring to $\Sigma_0$ by symbols analog 
to those corresponding 
to $\Sigma$ with 
a ``$\,\mathring{~}\,$'' to emphasize the role of initial value. Therefore
the induced metric, connection, and the corresponding Christoffel
symbols are denoted by $\ig, ~ \inab, ~ \iG$, respectively.

Let $\varphi: M \rightarrow N$ denote the local diffeomorphism with image
$\Sigma_0$ occurring
in definition \eqref{defn:immersed_smf}.
Local coordinates on $M$ will carry small Latin indices such as $i,j,k,\dots$
with range $1$ to $m$.
By abuse of language the metric, connection and Christoffel symbols
induced on $M$  by $\varphi$ will also be denoted by $\ig, ~ \inab$ and $ \iG$,
respectively.
The next considerations are derived from \cite{ER:1993} and
\cite{CB:1999}.

We first give a definition of a vector field along $\varphi$ which will
be used to compute the geometric quantities in this setting
in dependence on the immersion $\varphi$ and the ambient space $N$.
\begin{defn}
  A mapping $\xi: M \rightarrow TN$ is said to be a 
  \emph{vector field along $\varphi$}
  if $\xi(p) \in T_{\varphi(p)}N$ for all $p \in M$. 
  The bundle of vector fields along $\varphi$ 
  is denoted by 
  $\varphi^{\ast}TN$. 
  The connection $\D$ on $N$ induces a covariant derivative on such
  vector fields.
\end{defn}
\begin{rem}
  For a smooth 
  vector field
  $X \in \mathcal{T}(M)$ 
  the mapping $p \mapsto d\varphi(X(p))$ is a vector field 
  along $\varphi$.
\end{rem}
The ambient space $N$ induces a connection on $\varphi^{\ast}TN$
which will be described in the next definition.
\begin{defn}
  Let $\xi$ be a vector field along $\varphi$ and $v \in T_p M$. 
  Choose coordinates $y^A$
  in a neighborhood of $\varphi(p) \in N$ giving the representation
  $\xi(q) = \lambda^A(q) \partial_A\bigl(\varphi(q)\bigr)$.
  Then the covariant derivative $\hn_v \xi$ is given by
  \begin{align}
    \label{eq:hn}
    \hn_v \xi & 
    =  \Bigl(\partial_{v} \lambda^{A} + \lambda^{B}(p) d\varphi^{C}
    (v) \G_{BC}^{A}\bigl(\varphi(p)\bigr)\Bigr) 
    \partial_{A}\bigl(\varphi(p)\bigr).
  \end{align}
\end{defn}
\begin{rem}
  \label{rem:cov_properties}
  We summarize some properties of the covariant derivative on 
  $\varphi^{\ast}TN$:
  \begin{enumerate}
   \setcounter{enumi}{-1}
 \item $\hn$ is a connection on $\varphi^{\ast}TN$.
  \item{ It is compatible with the metric on $N$, namely
      $ \partial_{v }h(\xi, \eta) =
      h(\hn_{v} \xi, \eta) + h(\xi, \hn_{v}
      \eta)
      $ and}
  \item{\label{en:cov_dev_N}
      it is a restriction of the covariant derivative on N to 
      $\varphi^{\ast}TN$. If $\xi$ is of the
      form $\xi = \zeta \circ \varphi$ and $\zeta$ a vector field on $N$,
      then it follows that $\hn_v \xi = \D_{d\varphi(v)} \zeta$.}
  \end{enumerate}
\end{rem}
To state local computations, we introduce 
coordinates $x^j$ on $M$ and use the notations
$\bar{\xi}$ and $\Phi$ for the representations of the vector field
$\xi$ and $\varphi$ w.r.t.\ the coordinates $x$ and $y$, respectively, yielding
that
\begin{align*}
  \bar{\xi} = dy_{\varphi \circ x^{-1}}(\xi) \quad\text{and}\quad
  \Phi = y \circ \varphi \circ x^{-1}.
\end{align*}
Using this notation we can state the derivative w.r.t.\ a coordinate basis 
vector field
$\partial_j$ on $M$ and the derivative of the vector field 
$d\varphi(\partial_k)$.
We have that %
\begin{subequations}
  \begin{align}
    \label{eq:cov_coord}
  && \hn_j \xi & = \bigl(\partial_j
 \bar{\xi}^{A} + \xi^{B}
  \partial_j \Phi^{C}
  \G_{BC}^{A}(\Phi)\bigr)
  \partial_{A}(\Phi)
  \\
  \label{eq:dphi_cov_der}
  \text{and} && \hn_k \bigl(d\varphi(\partial_j )\bigr)
  & =  \bigl( \partial_k \partial_j \Phi^A
   + \partial_k \Phi^B \partial_j \Phi^C \G_{BC}^A(\Phi)\bigr)
   \partial_{A}(\Phi).
\end{align}
\end{subequations}
The Christoffel symbols w.r.t.\ the connection $\hn$ will be denoted
by $\indg$. They are given by
\begin{gather}
  \label{eq:ind_chr}
  \indg_{jB}^A = \partial_j \Phi^{C}
  \G_{BC}^{A}(\Phi).
\end{gather}

For vector fields $X$ and $Y$ on $M$ we define the
\emph{second fundamental form}  of
$\varphi$ in this setting by
\begin{gather}
  \label{eq:2ndform_mean}
   \II(X,Y) = \bigl( \hn_X d\varphi(Y)\bigr)^{\bot}, 
\end{gather}
where the normal part is taken w.r.t.\ $\Sigma_0$, the image of $\varphi$.
According to the Koszul formula, the following decomposition holds
for $X,Y \in {\cal T}(M)$
$$  \hn_X d\varphi(Y) = \bigl( \hn_X d\varphi(Y) 
   \bigr)^{\top} + \bigl( \hn_X d\varphi(Y)\bigr)^{\bot}
   = d\varphi(\inab_X Y) + \II(X,Y).
$$

For further reference we now state the representation of the second
fundamental form and the induced Christoffel symbols $\iG_{ij}^k$
on $M$. They are given by %
\begin{align}
  \label{eq:2ndfform}
  &&\II_{ij}^A & = \partial_{i} \partial_{j}
    \Phi^{A} + \partial_i \Phi^{B}
    \partial_j \Phi^{C} \G_{BC}^{A}(\Phi)  - \iG_{ij}^k
    \partial_k \Phi^A  \\
  \label{eq:ChrSymb}
  \text{and} &&\iG_{ij}^k & 
  = \ig^{k\ell} \partial_i\partial_j \Phi^A h_{AB}(\Phi)
    \partial_{\ell} \Phi^B + \ig^{k\ell}\partial_i \Phi^B \partial_j
    \Phi^C \G_{BC}^A(\Phi) h_{AD}(\Phi)
    \partial_{\ell} \Phi^D. &&
\end{align}

We introduce the notation of a norm to be applied to a 
covariant derivative of a vector field along $\varphi$ as follows.
Set
\begin{gather}
  \label{eq:gE_norm}
  \abs{\hn^{\ell} \xi}_{\ig,E}^2 = 
  \ig^{i_1 j_1} \cdots \ig^{i_{\ell} j_{\ell} }E_{AB} \hn_{i_1, \dots, i_{\ell} }\xi^A
  \hn_{j_1, \dots, j_{\ell}} \xi^B
\end{gather}
with the abbreviation $\hn_{i_1, \dots, i_{\ell} } = \hn_{i_1} \cdots \hn_{i_{\ell}}$,
where $E$ denotes a Riemannian metric on the ambient manifold $N$.
In this situation we introduce a further notation of %
coordinate dependent
norm. Set $  \abs{\hn^{\ell} \xi}_{e,e}^2$ equal to the RHS of the above
term replacing $\ig$ and $E$ by the Euclidean metric w.r.t.\ coordinates on
$M$ and $N$, respectively. 

To simplify computations we introduce the term $S \ast T$ which
denotes a linear combination of contractions of the tensors $S$ and $T$, 
where we suppress
factors of the metric $h$ of the manifold $N$.
 We use this notation if the precise form
is irrelevant for the arguments used. 
For example, with this notation, equation \eqref{eq:cov_coord} 
reads
\begin{gather*}
  \hn \xi = \partial \xi + \hg \ast \xi,
\end{gather*}
where $\hg$ denote the Christoffel symbols  defined in equation
\eqref{eq:ind_chr}.

In the following section we will describe how we intend to access the 
geometric problem via PDE-techniques.
\subsection{Reduction}
\label{sec:reduction}
Wishing to use standard PDE techniques it is convenient to introduce 
parametrizations.
Let $F:U \subset \rr \times M \rightarrow (N,h)$ denote an immersion
with induced metric $g:= F^{\ast} h$. The Christoffel symbols
w.r.t.\ $g$ will be denoted by $\Gamma$. The image of $F$ corresponds to 
a solution $\Sigma$ of the geometric IVP \eqref{eq:geom_problem}. 
Again, by abuse of 
language
we use the same notation for the metric on $\Sigma$ and the induced metric
on $U \subset \rr\times M$.
Let $x^{\mu}$
be coordinates on $ \rr\times M$ and $y^A$ coordinates on $N$.
Then the 
considerations from %
the previous section carry over to this situation.
We denote the covariant derivative on
the pullback bundle $F^{\ast}TN$ by $\hn^F$.
It then follows from the formula
\eqref{eq:2ndfform} for the second fundamental form
that the membrane equation \eqref{eq:H0} reads 
\begin{subequations}
  \begin{align}
    \label{eq:mem_geom}
  && g^{\mu\nu}\hn^F_{\mu} \partial_{\nu} F - g^{\mu\nu} \bigl(\hn^F_{\mu}
  \partial_{\nu} F \bigr)^{\top} & = 0 %
  \\ 
  \label{eq:membrane}
  \Longleftrightarrow &&
  g^{\mu\nu} \partial_{\mu} \partial_{\nu} F^A  + g^{\mu\nu} \partial_{\mu} F^B
  \partial_{\nu} F^C \G_{BC}^A(F) - 
  \Gamma^{\lambda} \partial_{\lambda} F^A & = 0.
\end{align}
\end{subequations}
From the PDE viewpoint the latter equation seems to be quasilinear
hyperbolic %
of second order. The hyperbolicity comes from the 
signature of the metric induced by $F$ on $\rr \times M$.

Based on equation \eqref{eq:membrane} we will consider the following
existence and uniqueness problems. \\
Let $\Sigma_0$ be a regularly immersed $m$-dimensional 
  submanifold  
  with immersion
  $\varphi: M^m \rightarrow N$. Assume $\nu$ to be a unit timelike 
  future-directed vector field normal to $\Sigma_0$.
  Let $\alpha$ be a function on
  $M$ called \emph{initial lapse} and
  let $\beta$  be a vector field on $M$ called \emph{initial shift}. 
  \begin{center}
  \parbox{0.93\textwidth}{
     \textbf{Existence} \\
  Find a neighborhood $U$ of $\{0\} \times M$ in $\rr \times M$ and
  an 
  immersion $F:U \subset \rr\times M \rightarrow N$
  solving the IVP 
  \begin{gather}
    \label{eq:param_ivp}
    H(\im F)\equiv 0,~\restr{F} = \varphi,~\restr{\dt F} 
    = \alpha \,\nu\circ \varphi + d\varphi(\beta)
  \end{gather}
  such that $\dt F$ is timelike future-directed.
  The parameter $t$ denotes the first component of $\rr \times M$.\\
  \textbf{Uniqueness} \\
  Show that for an immersion $\bar{F}$ solving the IVP \eqref{eq:param_ivp}
  for another choice $\bar{\varphi}$ of %
  immersion of $\Sigma_0$
  and another choice of initial lapse $\bar{\alpha}$ and initial shift 
  $\bar{\beta}$
  there exists a local diffeomorphism $\Psi$ %
  such that
  $F\circ \Psi^{-1} = \bar{F}$.   
}
\end{center}  
Before standard existence and uniqueness results for hyperbolic equations apply
one issue needs to be solved.
The invariance of the membrane equation under diffeomorphisms of a solution
$\Sigma$
leads to
a degenerate equation. 
This can be seen by considering the contracted induced Christoffel
symbols according to formula \eqref{eq:ChrSymb}.
It follows that
\begin{gather}
  \label{eq:red_chr}
     g^{\mu\nu} \Gamma_{\mu\nu}^{\lambda}  = g^{\kappa\lambda}\bigl(
    g^{\mu\nu} \partial_{\mu} \partial_{\nu} F^A + g^{\mu\nu} \partial_{\mu} F^B
    \partial_{\nu} F^C \G_{BC}^A(F) \bigr) h_{AD}(F) \partial_{\kappa} F^D. 
\end{gather}
The first term cancels the tangential part of the
leading term of the membrane equation \eqref{eq:membrane}. 
Therefore a gauge is needed
to remove the freedom of performing diffeomorphisms.
Our choice is the harmonic map gauge taken from \cite{FriRe:2000}.
In this article this gauge is used in order to reduce
the Einstein equations which also lead
to a degenerate hyperbolic system of second order. 
\begin{defn}
  A solution $F: U \subset \rr \times M \rightarrow N$ of the membrane 
  equation \eqref{eq:membrane}
  is in \emph{harmonic map gauge},
  if the following condition is satisfied:
  \begin{align}
    \label{eq:harm_cond}
    \id: \bigl(U, g = F^{\ast} h\bigr) 
    \rightarrow \bigl(U, \hat{g}\bigr)
    \quad\text{is a harmonic map,}
  \end{align}
  where $\hat{g}$ is a fixed background metric on $U$.
  This condition will be called the \emph{harmonic map gauge condition}.
\end{defn}
\begin{rem}
  In coordinates, this condition reads 
\begin{gather}
  \label{eq:harm_cond_coord}
  g^{\mu\nu}(\Gamma_{\mu\nu}^{\lambda} - 
\hat{\Gamma}_{\mu\nu}^{\lambda}) = 0.
\end{gather}
As a difference of two connections 
the term is a tensor and therefore the condition is invariant under change of
coordinates.
\end{rem}
Inserting  the harmonic map gauge condition \eqref{eq:harm_cond_coord}
into equation \eqref{eq:membrane} we derive the \emph{reduced membrane
equation}
\begin{equation}
  \label{eq:mem_red}
  g^{\mu\nu} \partial_{\mu} \partial_{\nu} F^A + g^{\mu\nu} \partial_{\mu} F^B
  \partial_{\nu} F^C \G_{BC}^A(F) - 
  g^{\mu\nu}\hat{\Gamma}_{\mu\nu}^{\lambda} 
  \partial_{\lambda} F^A = 0.
\end{equation}
The reduced equation is hyperbolic since the dependency on the second-order 
derivatives
of $F$ in the Christoffel symbols was replaced by fixed functions. 
Let us now prove the equivalence of the equations.
\begin{lem}
  \label{lem:equiv}
  The membrane equation \eqref{eq:membrane} together with condition 
  \eqref{eq:harm_cond_coord} is equivalent to equation \eqref{eq:mem_red}.
\end{lem}
\begin{proof}
  Straightforwardly, inserting condition \eqref{eq:harm_cond_coord} into the 
  membrane  equation 
  \eqref{eq:membrane}
  gives us the reduced membrane equation \eqref{eq:mem_red}.
  
  Suppose the reduced membrane equation
  \eqref{eq:mem_red} holds. To show that the mean curvature
  vanishes we need to compute the contracted Christoffel symbols w.r.t.
  the induced metric. By inserting equation \eqref{eq:mem_red} into
  equation \eqref{eq:red_chr} we get that
  \begin{align*}
    g^{\mu\nu} \Gamma_{\mu\nu}^{\lambda} 
    & =  g^{\kappa\lambda} g^{\mu\nu} \hat{\Gamma}_{\mu\nu}^{\delta} \partial_{\delta}
    F^A h_{AD} \partial_{\kappa} F^D = g^{\mu\nu} \hat{\Gamma}_{\mu\nu}^{\lambda}
  \end{align*}
  Using this identity together with equation
  \eqref{eq:mem_red} gives us the desired result.
\end{proof}

\begin{rem}
  This result can be derived in a somewhat more abstract fashion
  if we use equation \eqref{eq:mem_geom} instead.
  The reduced membrane equation \eqref{eq:mem_red} can then be read
  as $g^{\mu\nu}\hn^F_{\mu} \partial_{\nu} F = g^{\mu\nu} 
  \hat{\Gamma}_{\mu\nu}^{\lambda} \partial_{\lambda} F$. 
  The RHS of the latter equation is tangential, therefore the
  normal part of the LHS %
  has to vanish and the tangential part has to
  coincide with the RHS which yields us the desired identity
  \eqref{eq:harm_cond_coord}.
\end{rem}

\begin{rem}
  An examination of other gauges for the Einstein equations 
  (cf. \cite{FriRe:2000}) shows that these are also possible here, 
  since the contracted Christoffel symbols can be replaced by arbitrary fixed
  functions.

  A similar coordinate dependent choice of gauge is the use of the so-called
  harmonic (or wave-) coordinates used by Y. Choquet-Bruhat 
  (cf. \cite{ChBr:1952})
  to show well-posedness of the Einstein equations.
  Their gauge condition is the vanishing of the contracted Christoffel symbols.
\end{rem}

\begin{rem}
  This reduction process for the membrane equation is complete in the
  sense that no further equation is needed. This contrasts the situation
  for the Einstein equations where the propagation of the gauge condition
  is shown by using the Bianchi identity. In view of R.\,S. Hamilton's result
  for 
  geometric parabolic flows (cf. \cite{Hamilton:1982})
  the membrane equation seems more closely analogous to mean curvature flow
  than to the Einstein equations. In the reduction process for both equations,
  the projector to 
  the tangent
  space plays a fundamental role.
\end{rem}

\paragraph{Background metric}
Throughout this text we use a \emph{special background metric} 
$\hat{g}$.
If the initial values of $F$ are denoted by $\restr{F} = \varphi$
and $\restr{\dt F} = \chi$, then
we define
\begin{align}
  \label{eq:back_metric}
  \hat{g}  & := 
  - \lambda^2 dt^2 + \ig_{ij}(\beta^i dt + dx^i)(\beta^j dt
  + dx^j)
  \\
 \text{ with }\beta  = \beta^i \partial_i,~
\beta^i & = \ig^{ij} h\bigl(\chi, d\varphi(\partial_j)\bigr) \text{ and }
- \lambda^2 = h(\chi, \chi) - \ig(\beta,\beta).
\nonumber
\end{align}

\section{Existence and Uniqueness for Hyperbolic Equations
of Second Order}
\label{sec:hyp}
Due to the signature of the metric, the reduced membrane equation 
\eqref{eq:mem_red} is a hyperbolic equation of second order.
Nonlinear second-order equations whose coefficients and RHS depend on 
derivatives of the solution only up to first order are called quasilinear.
In this section we develop an existence theory for such equations relying
on Sobolev spaces. More precisely we search for solutions
$u: [0,T'] \times \rr^m \rightarrow \rr^N$ of the initial value problem
\begin{equation}
\label{second}
g^{\mu\nu}(u, Du,\partial_t u) \partial_{\mu} \partial_{\nu} u = f(u, Du,
\partial_t u), ~\restr{u} = u_0, ~\restr{\partial_t u} = u_1.
\end{equation}
The equation is supposed to be hyperbolic in the sense that 
$g^{00} \le - \lambda < 0$ and $g^{ij} \ge \mu \delta^{ij}$ for constants
$\lambda$ and $\mu$.
We will discuss the existence of a solution only within a small time interval.
The solution will be sought within $L^2$-Sobolev spaces $H^s$ for 
$s$ large enough.

The strategy will mainly follow the work of T. Kato for symmetric hyperbolic
systems (cf. \cite{Kato:1975}).
By using ideas from \cite{FiMa:1972}, \cite{Taylor:1996} and \cite{StSh:1998}
we adapt it to second
order equations obtaining a lower bound on the existence time.
Following \cite{Kato:1976} and \cite{FiMa:1972} we consider asymptotic
initial value problems. A solution to these problems satisfy the equation
\eqref{second} and its difference to a linear function lies in a Sobolev 
space.

Throughout this section $c$ will denote various constants depending on
the dimensions involved.
\subsection{Uniformly local Sobolev spaces}
This section fixes some notations used for the existence theory.
Let $H^s$ denote the usual $L^2$-Sobolev space  with norm 
$\norm{\,.\,}_s$ given by
\begin{gather*}
  \norm{u}_s^2 = \sum_{\abs{\alpha} \le s} \norm{\partial^{\alpha} u}_{L^2}^2.
\end{gather*}
Further, let $C^s$ designate the space of $s$-times continuously differentiable
functions with norm $\norm{\,.\,}_{C^s}$ defined by
\begin{gather}
  \label{eq:norm_Cs}
  \norm{u}_{C^s}^2 = \sum_{\abs{\alpha} \le s} \norm{\partial^{\alpha} u}_{\infty}^2
\end{gather}
and let $C^s_b$ denote functions in $C^s$ which have bounded
derivatives up to order $s$.
In the following definition we introduce spaces of functions which are only 
locally contained in a Sobolev space. These will be used to formulate
conditions for the coefficients of the equation.
\begin{defn}[Uniformly local $L^p$ and Sobolev Spaces]
  \begin{enumerate}
  \item  Suppose $1 \le p < \infty$ and $V$ is a finite dimensional vector 
    space. 
    Let $L_{\mathrm{ul}}^p(\rr^m, V)$ denote the space of all functions $u$ on 
    $\rr^m$ with
    values in $V$ such that 
    $\sup_{x\in\rr^m}\norm{\varphi_x u}_{L^p} < \infty$ where $\varphi \in 
    C_c^{\infty}(\rr^m)$
    and $\varphi_x(y) = \varphi(y - x)$.
  \item   For an integer $s \ge 0$ we denote by $H^s_{\mathrm{ul}}(\rr^m, V)$ the 
    space
    of all functions $u \in L^2_{\mathrm{ul}}$ whose distributional 
    derivatives
    $\partial^{\alpha} u$ of order $\abs{\alpha} \le s$ are in $L^2_{\mathrm{ul}}$. \\
    The space $H^s_{\mathrm{ul}}(\rr^m, V)$ will be endowed with
    the norm $\norm{u}_{s,\mathrm{ul}} 
    = \sup_{\abs{ \alpha} \le s} \norm{\partial^{\alpha} u}_{L^2_{\mathrm{ul}}}$.
  \end{enumerate}
\end{defn}
The next lemma lists some properties of the uniformly local Sobolev spaces
concerning embeddings which are similar to properties of the usual Sobolev
spaces.
\begin{lem}
  \label{lem:prop_ul}
  \begin{enumerate}
  \item The space $H^s_{\mathrm{ul}}(\rr^m, V)$ is a Banach space and 
    $H^s \subset 
        H^s_{\mathrm{ul}}$. Let $\varphi\in C_c^{\infty}(\rr^m)$. Then we have the
        equivalent norm $\sup_{x\in\rr^m} \norm{\varphi_x u}_{s}$. 
  \item Some properties of the usual Sobolev spaces apply to the uniformly
        local case: If $s < \tfrac{m}{2}$, then $H^s_{\mathrm{ul}}  
        \subset L^p_{\mathrm{ul}} $ for $p$ with $1/2 - s/m \le 1/p \le 1/2$.
        If $s = \tfrac{m}{2}$ then $H^s_{\mathrm{ul}}  
        \subset L^p_{\mathrm{ul}} $ for any $p$ with $2 \le p < \infty$ and
        if $s > \tfrac{m}{2}$ then $H^s_{\mathrm{ul}} \subset C_b^k$ for 
        $k = s - \lfloor m/2 \rfloor - 1$, where $\lfloor m/2 \rfloor$
        denotes the integer part of $m/2$.

  \item \label{mult_ul}
        If $r = \min(s,t,s+t - \lfloor m/2 \rfloor -1) > 0$, then 
        $H^s_{\mathrm{ul}} H^t_{\mathrm{ul}} \subset H^r_{\mathrm{ul}}$
        The LHS is either the product of two functions or a linear 
        operator on $V$ applied to a function. \\
        With the same values of $r, s, t$ it follows $H^s_{\mathrm{ul}}
        H^t \subset H^r$
        and $H^s H^t_{\mathrm{ul}} \subset H^r$.
  \end{enumerate}
\end{lem}
As a consequence we state a result about the commutator of a
differential operator and multiplication with a function analog to 
a result for the usual
Sobolev space to be found in \cite{Taylor:1996}.
\begin{cor}
  \label{cor:comm}
  Let $u \in H^s_{\mathrm{ul}}$ and $v \in H^{s-1}$ with $s> \tfrac{m}{2} +1$.
  Then the following estimate holds
  \begin{equation}
    \label{diff_komm}
    \norm{[\partial^{\alpha}, u]v}_{L^2} \le c \norm{u}_{s, \mathrm{ul}} \norm{v}_{s-1}
    \qquad\text{for } \abs{\alpha} \le s.
  \end{equation}
\end{cor}
\begin{proof}
  From \cite{Taylor:1996} we take the expression
  $$  [\partial^{\alpha}, u] v = \partial^{\alpha}(uv) - u \partial^{\alpha} v
  = \tsum_{\abs{\beta} + \abs{\gamma} = s-1} c_{j\beta\gamma} (\partial^{\beta} \partial_j u)
  (\partial^{\gamma} v)
  $$
  and therefore
  $$ \abs{[\partial^{\alpha}, u] v}_{L^2} \le c \tsum_{\abs{\beta} + \abs{\gamma} = s-1}
  \abs{(\partial^{\beta} u_j ) (\partial^{\gamma} v)}_{L^2}
  $$
  with $\partial_j u = u_j$. 
  The conditions $u \in H^s_{\mathrm{ul}}$ and $v \in H^{s-1}$ imply
  $\partial^{\beta} u_j \in H^{s- \abs{\beta} - 1}_{\mathrm{ul}}$
  and $\partial^{\gamma} v \in H^{s - 1 - \abs{\gamma}}$.
  To derive the desired result we want to use part \ref{mult_ul} of lemma
  \ref{lem:prop_ul}. Therefore, we need to check the orders of
  differentiation. We have
  \begin{gather*}
    2s - 2 - \abs{\beta} - \abs{\gamma} = 2(s-1) - (s-1) > \tfrac{m}{2}
  \end{gather*}
  which ends the proof.
\end{proof}
We now turn to invertible elements within $H^s_{\mathrm{ul}}$.
For a vector space $V$ we denote the space of linear bounded operators
$L:V \rightarrow V$ by $B(V)$.
\begin{lem}
  \label{inverse}
  Let $s > \tfrac{m}{2}$ be an integer and $A \in H^s_{\mathrm{ul}}(\rr^m, B(V))$. 
  A necessary
  and sufficient condition for $A$ to be invertible within $H^s_{\mathrm{ul}}$
  is $\bigl(A(x)\bigr)^{-1} \in B(V)$ for all $x$ and $\abs{A^{-1}}_{\infty} 
  \le \delta^{-1}$. Then it follows that
  \begin{equation}
    \label{inv_ul}
    \norm{A^{-1}}_{s,\mathrm{ul}} \le c \delta^{-1}\bigl(1 
    + (\delta^{-1} \norm{A}_{s,
  \mathrm{ul}})^s\bigr)
  \end{equation}
\end{lem}
The next lemma states estimates for the composition of functions.
The proof is an immediate consequence of the proof of 
theorem IV in \cite{Kato:1975}.
\begin{lem}
  \label{lem:comp_ul_est}
  Assume $s> \tfrac{m}{2}$ to be an integer,  
  $u \in H^s_{\mathrm{ul}}(\rr^m, \rr^N)$ and  $F \in C_c^s(\rr^N)$.
  Then it holds that
  \begin{gather*}
    \norm{F(u)}_{s,\mathrm{ul}} \le c \norm{F}_{C^s} (1 + \norm{u}_{s,\mathrm{ul}}^s) 
  \end{gather*}
  with  a constant $c$ depending on $N$ and $m$.
\end{lem}
\subsection{Linear Equation}
The argument for solving the quasilinear equation will use estimates
for the linearized equation. Therefore we establish existence for linear
equations first using the result for linear symmetric hyperbolic systems
of \cite{Kato:1975}. To this end it is necessary to reduce the second-order
equation to a first-order symmetric hyperbolic system. This will be done by
a method taken from \cite{FiMa:1972}, where it is used
to establish a reduction of the Einstein
equations to a first-order symmetric hyperbolic system.

To keep this thesis self-contained we state the existence result
for linear symmetric hyperbolic systems of first order from \cite{Kato:1975}
(cf. p. 189, Theorem I). 
\begin{thm}
  \label{thm:Kato_thmI}
  Let $P$ be a Hilbert-space, $s > \tfrac{m}{2} + 1$ be an integer and let 
  $1 \le s' \le s$.
  Assume the following conditions:
  \begin{subequations}
    \begin{gather}
    \label{eq:Kato_cond1}
    A^{\mu}, B \in C\bigl([0,T], H^0_{\mathrm{ul}}\bigl(\rr^m, B(P)\bigr)\bigr), 
    0 \le \mu \le m 
    \\
    \label{eq:Kato_cond2}
    \norm{A^0(t)}_{s,\mathrm{ul}} \le K_0, ~\norm{A^j(t)}_{s,\mathrm{ul}},~
    \norm{B(t)}_{s,\mathrm{ul}}
    \le K \quad \text{ for } 0 \le t \le T, 1 \le j\le m,
    \\
    \label{eq:Kato_cond3}
    \norm{A^0(t') - A^0(t)}_{s-1, \mathrm{ul}} \le L \abs{t' - t}\quad 0 \le t,t'\le T
    \\
    \label{eq:Kato_cond4}
    A^{\mu}(t,x) \text{ is symmetric for } 0 \le t \le T, x\in \rr^m, 0 \le \mu 
    \le m
    \\
    \label{eq:Kato_cond5}
    A^0(t,x) \ge \eta>0
    \text{ for all } t,x \text{ (uniformly positive definite)}
    \\
    \label{eq:Kato_cond6}
    F \in L^1\bigl([0,T], H^{s'}(\rr^m, P)\bigr) 
    \cap C\bigl([0,T], H^{s'-1}(\rr^m, P)\bigr)
    \\
    v_0 \in H^{s'}(\rr^m, P)
  \end{gather}
  \end{subequations}
  Then the IVP
  \begin{gather*}
     A^0 \partial_t v + A^j \partial_j v + Bv = F, ~ 
  \restr{v} = v_0
  \end{gather*}
has a solution
  \begin{gather*}
    v \in C\bigl([0,T], H^{s'}(\rr^m, P)\bigr) \cap 
    C\bigl([0,T], H^{s'-1}(\rr^m, P)\bigr).
  \end{gather*}
  The solution is unique in a larger class $C\bigl([0,T], H^{1}(\rr^m, P)
  \bigr) \cap 
    C\bigl([0,T], H^{0}(\rr^m, P)\bigr)$.
\end{thm}
The next step is to find suitable operators $A^0, A^j$ and $B$ 
connecting first-order systems with second-order equations.
To this end we consider the matrices
$$
\begin{array}{c@{\text{,}\quad}c}
  C^0 = 
  \begin{pmatrix}
   1 & 0 & 0 \\
   0 & -g^{00} & 0 \\
   0 & 0 & (g^{ij})_{i,j} 
  \end{pmatrix}

  &
  C^j = 
  \begin{pmatrix}
   0 & 0 & 0 \\
   0 & -2g^{j0} & (-g^{ij})_i \\
   0 & (-g^{ij})_i & 0 
  \end{pmatrix}
 \quad\text{for } 1\le j \le m
\end{array}
$$
and 
$$ D =
\begin{pmatrix}
 0 & -1 & 0 \\
 0 & 0 &  0 \\
 0 & 0 & 0
\end{pmatrix}.
$$
Let $V$ be a vector space.
We define a member $C$ of $B(V^{m+2})$ arising from a matrix 
$(c^k_{\ell})_{1 \le k,\ell \le m + 2}$ by 
\begin{gather}
  \label{eq:op_arise}
  (Cv)^k = \sum_{1 \le \ell \le m + 2} c^k_{\ell} v^{\ell} \quad \text{for } v \in V^{m+2}.
\end{gather}

The following fact was established in \cite{FiMa:1972}.
\begin{prop}
  \label{prop:reduction}
  Let $f \in C([0,T] \times \rr^m, \rr^N)$ and $(g^{\mu\nu})_{0\le \mu, \nu \le m}$ be a 
  matrix-valued
  continuous function defined on $[0,T] \times \rr^m$.
  Let $A^0, ~A^j$ and $B$ be the linear operators 
  $ \in B\bigl((\rr^N)^{m+2}\bigr)$
  arising
  from the matrices $C^0, C^j$ and $D$ by definition \eqref{eq:op_arise}. 
  Let $u_0 \in C^2([0,T] \times \rr^m, \rr^N)$ and 
  $u_1 \in C^1([0,T] \times \rr^m, \rr^N)$.

  Then the initial value problem
\begin{gather}
  \label{eq:sym_ivp}
  A^0 \partial_t v + A^j \partial_j v + Bv = F, ~ 
  \restr{v} = \bigl(u_0, u_1, (\partial_j u_0)_j\bigr)
\end{gather}
with the RHS $F = (0,f,0)$ is equivalent to the initial value problem
\begin{gather}
  \label{eq:lin_ivp}
  g^{\mu\nu} \partial_{\mu} \partial_{\nu} u = f,~ \restr{u} = u_0,~
\restr{\partial_t u} = u_1
\end{gather}
within $C^1$-solutions of \eqref{eq:sym_ivp} and $C^2$-solutions of 
\eqref{eq:lin_ivp}.
\end{prop}

Relying on this fact we state existence and uniqueness for 
linear second-order equations.
\begin{thm}
  \label{lin_ex}
  Let $s > \tfrac{m}{2} +1$ be an integer. Assume that the following 
  inequalities hold:
  \begin{subequations}
      \begin{gather}
  \label{cond_lin1}
  g^{\mu\nu} \in C\bigl([0,T], H^0_{\mathrm{ul}}(\rr^m)\bigr) \text{ for } 
  0 \le \mu,\nu \le m, \\
  \label{cond_lin2}
  \norm{(g^{\mu\nu})}_{e,s,\mathrm{ul}}\le K \\
  \label{cond_lin3}
  \norm{g^{00}(t) - g^{00}(t')}_{s-1, \mathrm{ul}} \le L \abs{t - t'}, ~
  \bnorm{\bigl(g^{ij}(t)\bigr) - \bigl(g^{ij}(t')\bigr)}_{e,s-1, \mathrm{ul}} 
  \le L \abs{t - t'} \\
  \label{cond_lin4}
  \bigl(g^{\mu\nu}\bigr) \text{ is symmetric and }
  - g^{00} \ge \eta > 0, ~ \bigl(g^{ij}) \ge \eta I \\
  \label{cond_lin5}
  f \in L^1\bigl([0,T], H^s(\rr^m, \rr^N)\bigr) \cap 
  C\bigl([0,T], H^{s-1}(\rr^m, \rr^N)\bigr)
\end{gather}
  \end{subequations}
If $u_0 \in H^{s+1}(\rr^m, \rr^N)$ and $u_1 \in H^{s}(\rr^m, \rr^N)$, then the IVP 
\eqref{eq:lin_ivp} has a unique solution 
$$
u \in C\bigl([0,T], H^{s+1}(\rr^m, \rr^N)\bigr) \cap 
C^1\bigl([0,T], H^{s}(\rr^m, \rr^N)\bigr).
$$
\end{thm}

\begin{proof}
  Consider the linear operators  $A^0, ~A^j, ~B \in B\bigl((\rr^N)^{m+2}\bigr)$
  arising
  from the matrices $C^0, C^j$ and $D$ by definition \eqref{eq:op_arise}.
  We have to check the conditions of theorem
  \ref{thm:Kato_thmI}. 
  Observe that we have for an operator $C \in B\bigl((\rr^N)^{m+2}\bigr)$ 
  arising from a matrix $(c_{i}^j)$
  by definition \eqref{eq:op_arise}
  \begin{gather*}
    \norm{C}_{s, \mathrm{ul}} \le \norm{(c^j_i)}_{e,s,\mathrm{ul}} \quad
    \text{and}\quad \langle v, Cv \rangle \ge \langle (\lambda^i),
    (c^j_i) (\lambda^j) \rangle,
  \end{gather*}
  where $(\rr^N)^{m+2}$ is endowed with the Euclidean metric.
  Hence, it suffices to consider the matrices $C^{\mu}$ and $D$ 
  to obtain the conditions of theorem
  \ref{thm:Kato_thmI}. 
  
  Constants are in $L^2_{\mathrm{ul}}$, so condition 
  \eqref{eq:Kato_cond1}
  holds for $B$. From assumption \eqref{cond_lin1} we get the rest of
  condition  \eqref{eq:Kato_cond1}.
  To obtain condition \eqref{eq:Kato_cond2} we use the Frobenius-norm of 
  the matrices $C^{\mu}$. 
  With this norm we
  have 
  \begin{gather*}
    \norm{C^0}_{e, s, \mathrm{ul}}^2 \le 1 + \norm{g^{00}}^2_{s,\mathrm{ul}} +
  \norm{(g^{ij})}_{e,s,\mathrm{ul}}^2 \le 1 +  K^2.
  \end{gather*}
  The same 
  consideration for $C^j$
  gives us $\norm{C^j}_{e,s,\mathrm{ul}}^2 \le 4K^2 $ and the matrix
  $D$ is bounded by $\norm{D}_{e,s,
   \mathrm{ul}} \le 1$ and therefore  
 condition \eqref{eq:Kato_cond2} follows.

  The Lipschitz-condition \eqref{eq:Kato_cond3} can be obtained by
  assumption \eqref{cond_lin3} as follows:
  \begin{gather*}
    \norm{C^0(t) - C^0(t')}_{e,s,\mathrm{ul}}^2 \le \norm{g^{00}(t) - g^{00}(t')}_{s,
   \mathrm{ul}}^2 + \norm{(g^{ij}(t)) - (g^{ij}(t'))}_{e,s,\mathrm{ul}}^2 \le
  2 L^2\abs{t - t'}^2.
  \end{gather*}
  Positive definiteness  of the matrix $C^0$ is given by
  $C^0 \ge \min(1,\eta) I$, where $I$ is
  the identity matrix.

  From Theorem \ref{thm:Kato_thmI} we obtain a solution $v = (w, w_0, w_j)$
  of the IVP \eqref{eq:sym_ivp} and from the differentiability of $v$ we 
  derive 
  $w \in C([0,T], H^{s+1}) \cap C^1([0,T], H^s)$.
The Sobolev-embedding yields that Proposition \ref{prop:reduction} is applicable
and therefore $w$ solves the IVP \eqref{eq:lin_ivp}.
The equivalence of the first-order and second-order equations stated
in Proposition \ref{prop:reduction} also shows that a solution $u$ to
the IVP \eqref{eq:lin_ivp} is unique.
\end{proof}

\subsection{Quasilinear Equation}
\label{sec:qlinear}
In this section we will develop an existence theory for the IVP
\eqref{second}. A key ingredient will be Banach's fixed point theorem
which gives an outline of the proof. A complete metric space and a suitable
mapping has to be constructed. The value of the mapping consists of
the solution to the linearized equation and energy estimates will be
established
to show the desired properties. \\
For convenience we postpone the proofs of several more technical statements
to section \ref{sec:proofs} of the appendix.

Let $s > \tfrac{m}{2} + 1$ be an integer 
and $W \subset H^{s+1}(\rr^m, \rr^{N}) \times 
H^{s}(\rr^m, \rr^{N})$ be open.
Suppose 
  \begin{align}
  \label{eq:domain_second}
  \begin{split}
       (g^{\mu\nu}) & : [0,T] \times W 
      \longrightarrow H^{s}_{\mathrm{ul}}\bigl(\rr^m, 
  \rr^{(m+1) \times (m+1)}\bigr)
  \\
  \text{and} \qquad\qquad f & : [0,T] \times
  W \longrightarrow H^{s}(\rr^m, \rr^N)
  \end{split}
\end{align}
to be  nonlinear operators.  
For $v = (v_0, v_1) \in W$ we introduce the abbreviation
$g^{\mu\nu}(t,v) = g^{\mu\nu}(t, v_0, D v_0, v_1)$; $f(t,v)$ is
defined analogously.
The main result of this section is stated in the following theorem.
\begin{thm}
\label{qlin_ex}
Suppose the 
following conditions hold:
\begin{subequations}
  \begin{gather}
 \label{cond_qlin1}
\bnorm{\bigl(g^{\mu\nu}(t,v)\bigr)}_{e,s, \mathrm{ul}} \le K \text{ and } \norm{f(t,v)}_{s} 
\le K_f  \\
\label{cond_qlin2}
\bnorm{\bigl(g^{\mu\nu}(t,v)\bigr) - \bigl(g^{\mu\nu}(t,w)\bigr)}_{e,s-1, \mathrm{ul}} \le \theta E_{s}(v - w)  \\
\label{cond_qlin_add3}
\bnorm{\bigl(g^{\mu\nu}(t,v)\bigr) - \bigl(g^{\mu\nu}(t,w)\bigr)}_{e,0, \mathrm{ul}} \le \theta' E_{1}(v - w)  \\
\label{cond_qlin_t_lip}
\bnorm{\bigl(g^{\mu\nu}(t,v)\bigr) - \bigl(g^{\mu\nu}(t',v)\bigr)}_{e,s-1, \mathrm{ul}} 
\le \nu \abs{t - t'}
\\
\label{cond_qlin_add1}
\norm{f(t,v) - f(t,w)}_{s-1} \le \theta_f' E_s(v - w) \\
\label{cond_qlin_add2}
\norm{f(t,v) - f(t,w)}_{L^2} \le \theta_f E_1(v - w) \\
\label{cond_qlin3}
 t \mapsto f(t, v)
\text{ is continuous w.r.t.\ } H^{s-1} \text{ for all } v \in W \\
\label{cond_qlin4}
 \bigl(g^{\mu\nu}(t,v)\bigr) \text{ is symmetric and }
g^{00}(t, v) \le - \lambda, ~ \bigl(g^{ij}(t,v)\bigr) \ge \mu \delta^{ij}, ~
\lambda, \mu >0 
\end{gather}
\end{subequations}
for $v = (v_0, v_1), w = (w_0, w_1) \in W$.

If $(u_0, u_1) \in W$ then there exists a constant $0 < T'\le T $ and a 
unique solution 
\begin{gather}
  \label{eq:qlin_sol_diff}
  u \in C\bigl([0,T'], H^{s+1}(\rr^m, \rr^{N})\bigr) \cap C^1\bigl([0,T'],
H^s(\rr^m, \rr^{N})\bigr)
\end{gather}
with $(u,\partial_t u) \in W$ to the initial value problem
\eqref{second}.
\end{thm}
The strategy to proof this theorem will be to apply Banach's fixed point 
theorem to the map
\begin{gather}
  \label{eq:def_map}
  \Phi: v \longmapsto \text{solution } u
\text{ to the linearized equation } g^{\mu\nu}(t,v) \partial_{\mu} 
\partial_{\nu} u = f(t,v).
\end{gather}
In the next lemma we derive ingredients for the definition of the metric
space on which $\Phi$ will be defined. \\
Set $c_E = 2 (\tfrac{2}{\mu} + 
\tfrac{1}{\lambda})$ and $\tau_s = \#\{ \beta \in N_0^m : \abs{\beta} \le s\}$.
For notational convenience we set $H^{r} = H^{r}(\rr^m, \rr^N)$ and 
\begin{gather*}
  E_{r+1}(v - w) := \norm{v_0 - w_0}_{r+1} + \norm{
    v_1 - w_1}_{r}\text{ for }v = (v_0, v_1), ~w = (w_0, w_1) 
  \in H^{r+1} \times H^{r}.
\end{gather*}
\begin{lem}
\label{u00}
For arbitrary $ \init{u} = (u_0, u_1) \in W$ there exist $\delta >0, ~0 \le 
\rho \le
\delta/3$ 
and $u_{00} = (y_0, y_1) \in W \cap (H^{s+2} \times
H^{s+1})$ such that 
\begin{gather*}
E_{s+1}(v - u_{00}) \le \delta \quad \Longrightarrow \quad 
(v_0, v_1) \in W \\
E_{s+1}(\init{u} - u_{00}) \le \rho \\
\rho  C^{1/2} \le \delta/3
\end{gather*}
with 
\begin{gather}
  \label{eq:equiv_const}
   \hat{C} := \sup_{v \in W}
   ( \mu + \abs{g^{00}(0,v)}_{\infty}
   + \babs{\bigl(g^{ij}(0,v)\bigr)}_{e,\infty})
   \text{ and }
   C := 4 c_E^{1/2} \tau_{s}^{1/2} \hat{C}. 
\end{gather}
\end{lem}
We are now able to define the metric space.
\begin{lem}
\label{e_est}
Let $L'$ be a constant chosen later and let 
\begin{multline*} Z_{\delta, L'}  :=  \{ u \in L^2([0,T']\times\rr^m, \rr^N): 
  \restr{u} = u_0, ~\restr{\partial_t u} = u_1, ~ 
  \\
  E_{s+1}(u(t) - u_{00}) = \norm{u(t) - y_0}_{s+1} + 
  \norm{\partial_t u(t) - y_1}_{{s}} \le \delta,
  \\
  ~ \norm{\partial_t u(t) - \partial_t
    u(t')}_{{s-1}} \le L' \abs{t - t'}  \}
\end{multline*}
where $u_{00} = (y_0, y_1) $ from Lemma \ref{u00}.
$Z_{\delta, L'}$ equipped with the metric 
\begin{gather*}
  d(u,v) = \sup_t \bigl( 
\norm{u(t) - v(t)}_{H^1} + \norm{\partial_t u(t) - \partial_t v(t)}_{L^2}
\bigr)
\end{gather*}
is complete.
\end{lem}
The next lemma establishes the fact that the mapping $\Phi$ is defined on
$Z_{\delta, L'}$ for all choices of $\delta$ and $L'$.
\begin{lem}
  The IVP for the linearized equation 
  \begin{gather}
    \label{eq:linearized}
    g^{\mu\nu}(t,v)
  \partial_{\mu} 
  \partial_{\nu} u = f(t,v), ~\restr{u} = u_0, ~ \restr{\partial_t
    u} = u_1
  \end{gather}
 for $v \in Z_{\delta , L'}$ has a solution.
\end{lem}

\begin{proof}
  We need to verify that the assumptions of the linear existence theorem
  \ref{lin_ex} are met. \\
  It holds
  \begin{gather}
    \norm{\bigl(g^{\mu\nu}(t,v(t))\bigr) 
      - \bigl(g^{\mu\nu}(t',v(t'))\bigr)}_{e,s-1, \mathrm{ul}}
  \le \nu \abs{t-t'} + \theta E_s\bigl(v(t) - v(t')\bigr) 
  \nonumber
  \\
  \label{eq:Es_lip_t}
  \text{with} \quad
  E_s\bigl(v(t) - v(t')\bigr) 
  \le  \bigl(\delta + \norm{y_1}_s + L'\bigr)\abs{t - t'}
  \end{gather}
  derived from the conditions \eqref{cond_qlin2}, \eqref{cond_qlin_t_lip} and 
  the definition
  of the metric space $Z_{\delta, L'}$.
  Hence, the conditions \eqref{cond_lin1} and \eqref{cond_lin3} follow.
  Condition
  \eqref{cond_lin2} follows directly from condition \eqref{cond_qlin1}. 
  The condition
  \eqref{cond_lin5} for $f$ is satisfied since $f\in L^{\infty}([0,T],H^s)$.
  The continuity of $f$ is 
  shown in a manner analogous to that for the coefficients.
\end{proof}
The next proposition establishes second-order energy estimates for
the linearized equation. To this end, first-order energy estimates taken
from \cite{Kato:1975} were adapted using ideas from \cite{Taylor:1996}
and \cite{StSh:1998}. 
\begin{prop}
  \label{prop:e_estimate}
  If $u$ is a solution to the linearized equation 
  \eqref{eq:linearized}, then we have the estimate
  \begin{equation}
    \label{eq:e_estimate}
    E_{s+1}\bigl(u(t) - u_{00}\bigr) \le 
    e^{c_E 
      \tau_{s}^2 C_2 t}\bigl( C^{1/2} E_{s+1}(\init{u} - u_{00}) + 
    c_E \tau^2_{s} C_1 t\bigr)
  \end{equation}
  where 
  \begin{align*}
    C_1 & = c\bigl( (\mu + K)E_{s+2}(u_{00}) + K 
    \bigl( 1 + \bigl(\tfrac{1}{\lambda} K\bigr)^s \bigr)
    K_f + K^2  \bigl( 1 + \bigl(\tfrac{1}{\lambda} K\bigr)^s \bigr)
    E_{s+1}(u_{00})
    \bigr) \\
    C_2 & = c\bigl( \mu + K + \nu + \theta
    ( \delta + E_{s+1}(u_{00}) + L')+ K^2 
    \bigl( 1 + \bigl(\tfrac{1}{\lambda} K\bigr)^s \bigr)
    \bigr) 
  \end{align*}
  and the constant 
  $C$ is taken from Lemma \ref{u00}.
\end{prop}
In the next lemma it will be shown using the preceding energy estimate 
that the mapping $\Phi$ maps to 
the metric space $Z_{\delta, L'}$.
\begin{lem}
  \label{lem:domain_map}
$\Phi$ maps $ Z_{\delta, L'}$ to $Z_{\delta, L'}$ if $L'$ and $T'$ 
are chosen appropriately.
\end{lem}
\begin{proof}
  Let $u = \Phi(v)$ for $v \in Z_{\delta, L'}$.
  To obtain the Lipschitz condition for $\partial_t u$ we estimate the 
  second-order time derivative of $u$ by using the modified equation
  \eqref{eq:modified}. 
  Using 
  the bound for $(g^{00})^{-1}$
  derived from Lemma \ref{inverse} gives us a bound for $\partial_t^2 u$
dependent on the term $E_{s+1}(u(t) - u_{00})$ which is controlled by
energy estimate \eqref{eq:e_estimate}. \\
In the limiting case
$T' = 0$ the condition $u \in Z_{\delta, L'}$ reduces to 
\begin{subequations}
  \begin{align}
  \label{eq:delta_init}
  && C^{1/2} E_{s+1}(\init{u} - u_{00}) & < \delta \\
  \label{eq:L_init}
   \text{and} && c \, \bigl( 1 + \bigl(\tfrac{1}{\lambda} K\bigr)^s \bigr)
  (2K C^{1/2} E_{s+1}(\init{u} - u_{00}) + 2K E_{s+1}(u_{00}) + K_f) & < L'.
\end{align}
\end{subequations}
The first inequality
\eqref{eq:delta_init} is satisfied since $C^{\2}
E_{s+1}(\init{u} - u_{00}) \le \rho
\le \delta/3$ by construction in Lemma \ref{u00}. 
Condition \eqref{eq:L_init} can be satisfied
choosing $L'$ appropriately large. Due to the continuity of the right member
of \eqref{eq:e_estimate} 
 the desired conditions
hold for sufficiently small $T' > 0$.
\end{proof}
To show that $\Phi$ is a contraction we will show energy estimates
for the difference equation of two solutions to the linearized equation.
\begin{prop}
  \label{prop:constraction}
Assume $u_1,~ u_2 \in Z_{\delta, L'}$ to be  two solutions to the 
linear equation \eqref{eq:linearized}
for $v_1, ~v_2 \in Z_{\delta, L'}$ resp. Then the following estimate holds
\begin{equation}
  \label{eq:contr_estimate}
  d(u_1 , u_2 ) \le c\, c_E T' e^{C_1 T'}  \bigl(\theta_f + (E_{s+1}(u_{00})
  +\delta + L')\theta' 
  \bigr)d(v_1, v_2)
\end{equation}
with $C_1 = c \, c_E\bigl( \mu +  K +  \nu + \theta (\delta 
+ E_{s+1}(u_{00})
+ L')
\bigr)$. 
\end{prop} 

\begin{proof}
  We consider the difference equation
  \begin{align*}
   g^{\mu\nu}(t,v_1) \partial_{\mu} \partial_{\nu} (u_1 - u_2) 
   & = \hat{f}
  \\
  \text{with} \qquad \hat{f} & = f(t,v_1) - f(t,v_2) 
  {}+ \bigl( g^{\mu\nu}(t,v_2)  -  
  g^{\mu\nu}(t,v_1)\bigr) \partial_{\mu} \partial_{\nu} u_2.
 \end{align*}
 Following the proof of Proposition \ref{prop:e_estimate}
 we arrive at an energy estimate for the difference equation similar
 to \eqref{energy1}.
 From the constants occurring we derive that, taking $u_{00} = 0$, only
 an estimate on $\norm{\hat{f}}_{L^{\infty}L^2}$ is needed. 
We begin with the term including second-order derivatives of the solution $u_2$.
They  are bounded by the condition $u_2 \in Z_{\delta, L'}$ and 
the linearized equation solved by $u_2$ which gives us an estimate on the
second time derivatives as in the proof of Lemma \ref{lem:domain_map}.

The assumptions \eqref{cond_qlin_add2} and \eqref{cond_qlin_add3} yield that 
the differences
inheriting the RHS and the coefficients of the original equation
can be estimated by a constant times the metric $d$.
Hence,  the desired result follows.
\end{proof}
We immediately obtain from the preceding energy estimate that the mapping 
$\Phi$ is a contraction. 
\begin{lem}
$\Phi$ is a contraction w.r.t.\ the metric $d$, if $T'$ is chosen small enough.
\end{lem}
We now turn to the proof of the existence theorem for quasilinear
hyperbolic equations.
\begin{proof}[\textbf{Proof of Theorem \ref{qlin_ex}}]
  From Lemma \ref{e_est} we derive that the metric space $Z_{\delta, L'}$
  is complete. The mapping $\Phi$ defined in \eqref{eq:def_map} 
  is defined on $Z_{\delta, L'}$ with values in $Z_{\delta, L'}$ for 
  appropriate choices of the constants $\delta$ and $ L'$ showed in lemma
  \ref{lem:domain_map}.
  It is a continuous map since the solution to the linearized
  equation depends continuously on the initial data.
  We obtain from the preceding lemma that $\Phi$ is a contraction
  and therefore Banach's fixed point theorem yields the existence
  of a unique fixed point solving the quasilinear IVP \eqref{second}. \\
  The proof of the differentiability property of such a solution is postponed 
  to Lemma \ref{conti}.
\end{proof}

\begin{rem}
  \label{rem:gen_lower}
    To obtain a lower bound for the existence time $T'$ of a solution
    observe that it has 
  to satisfy the two inequalities
  \begin{gather}
    \label{eq:time_ineq}
    e^{c_1 T'} (c_2 + c_3 T') \le \delta \qquad\text{and}
    \qquad c_4 T' e^{c_5 T'} \le 
    \zeta \quad\text{for a fixed } \zeta < 1
  \end{gather}
  with the constraint $c_2 \le \delta/3$.
  Set 
  \begin{gather*}
    T' = \min(\tfrac{\delta}{3c_3},
  \tfrac{1}{c_1} \ln(3/2),\tfrac{\zeta}{2c_4}, \tfrac{1}{c_5} \ln 2 ),
  \end{gather*}
  then previous inequalities \eqref{eq:time_ineq} are satisfied. For 
  convenience we record the 
  constants occurring in this inequalities derived from the
  energy estimates \eqref{eq:e_estimate} and \eqref{eq:contr_estimate}
  \begin{align*}
    c_1 & = c\,c_E 
      \tau_{s}^2 \bigl( \mu + K + \nu + \theta
   ( \delta + E_{s+1}(u_{00}) + L')+ K^2 
      \bigl( 1 + \bigl(\tfrac{1}{\lambda} K\bigr)^s \bigr)
      \bigr)\\
    c_3 & =  c\,c_E \tau^2_{s} \bigl( (\mu + K)E_{s+2}(u_{00}) + K 
      \bigl( 1 + \
      \bigl(\tfrac{1}{\lambda} K\bigr)^s \bigr)
      K_f + K^2  \bigl( 1 + \bigl(\tfrac{1}{\lambda} K\bigr)^s \bigr)
      E_{s+1}(u_{00})
      \bigr)\\
    c_4 & = c\, c_E 
      \bigl(\theta_f + (E_{s+1}(u_{00}) +\delta + L')\theta' 
      \bigr)\\
    c_5 & = c \, c_E\bigl( \mu +  K +  \nu + \theta
    ( \delta + E_{s+1}(u_{00}) + L')
      \bigr)
  \end{align*}
  with $c_E = 2 \bigl(\tfrac{2}{\mu} + \tfrac{1}{\lambda}\bigr)$. 
\end{rem}

\begin{rem}
  \label{rem:special_lower}
  If $\init{u} = (u_0, u_1) \in H^{s+2}\times H^{s+1}$ then we can
  choose $u_{00} = \init{u}$ in Lemma \ref{u00}.
  A lower bound for the existence time $T'$ is then given by
  \begin{gather}
    \label{eq:def_ex_time}
    T' = \min\bigl(\tfrac{\delta}{2c_3},
    \tfrac{1}{c_1} \ln(2),\tfrac{\zeta}{2c_4}, \tfrac{1}{c_5} \ln 2 \bigr).
  \end{gather}
\end{rem}
At the end of this section we will turn to higher regularity of solutions.
We refer to \cite{Kato:1976} and \cite{FiMa:1972} for further reference.
\begin{cor}
  \label{cor:sol_diff}
  Consider the Cauchy problem \eqref{second} with coefficients and
  a RHS 
  having the domains described in \eqref{eq:domain_second}.
  Let $\ell_0$ and $r> \tfrac{m}{2} + 1 + \ell_0$ be  integers.
  Suppose the coefficients and the RHS satisfy the assumptions
  of Theorem \ref{qlin_ex} with $s = r$. For $1 \le \ell
  \le \ell_0$ let the derivatives of the coefficients and the RHS of
  order $\ell$
  satisfy the assumptions %
  of Theorem \ref{qlin_ex} with $s = r - \ell$.

  Then the unique solution $u$ to the IVP \eqref{second}
  has the property
  \begin{gather}
     \label{eq:sol_diff}
    u \in C^{2 + \ell}\bigl([0,T'], H^{r -1 - \ell}(\rr^m, \rr^{N})\bigr)\quad
    \text{for } 0 \le \ell \le \ell_0.
  \end{gather}
\end{cor}
\begin{proof}
  The proof will be an induction on $\ell_0$ and we will only show 
  the first step, the general step follows in an analog manner.

   Let $\ell_0 = 1$. Observe that from the existence Theorem \ref{qlin_ex}
   we obtain a solution $u$ satisfying \eqref{eq:qlin_sol_diff} with $s = r$ and
   therefore the second-order derivatives of $u$ have the property
   \begin{gather*}
     D^2 u, ~D \partial_t u \text{ and }\partial_t^2 u \in 
     C\bigl([0,T'], H^{r -1}\bigr)
   \text{ with } r - 1 > \tfrac{m}{2} + \ell_0 = \tfrac{m}{2} + 1.
   \end{gather*}
   By differentiating the equation solved by $u$ we get that
   the resulting equation can be seen as a linear equation for $v = \partial_t
   u$. The second-order derivatives can be seen as RHS for the linear
   equation satisfying the conditions of the linear
   existence Theorem \ref{lin_ex} with $s = r - 1$.
   From the uniqueness result of this theorem we obtain the desired
   differentiability property of $\partial_t u$.
\end{proof}

\begin{rem}
  Due to the Sobolev embedding theorem it follows from property 
  \eqref{eq:sol_diff} that $u$ is in fact a classical solution satisfying
  $u \in C^{2+ \ell}([0,T'] \times \rr^m)$ for 
  $0 \le \ell \le \ell_0$.
\end{rem}

\subsubsection{Asymptotic equations}
\label{sec:asym}
In this paragraph we discuss solutions to the IVP \eqref{second}
which do not tend to $0$ at infinity, but instead tend to a linear function.
In the sequel, considerations follow roughly \cite{Kato:1976}
and \cite{FiMa:1972}
covering asymptotic solutions to the Einstein equations.

To make it precise let $w(t,x)$ be a linear function on 
$\rr^m \times \rr$ with
$w(t,x) = w_0(x) + tw_1$.
Denote the set of functions $v:\rr^m \rightarrow \rr^N$ satisfying  
$v - w_0 \in H^{s+1}$
by $H^{s+1}_{w_0}$ and let $H^s_{w_1}$ be defined analogously.

To simplify computations we use a special norm on uniformly Sobolev spaces. 
Let $\varphi \in C_c^{\infty}(\rr^m)$
with $0 \le \varphi \le 1$ and $\tint \abs{\varphi}^2 
\, dx = 1$. Throughout this section we endow $H^s_{\mathrm{ul}}$ with the norm
given by this test function as in lemma    \ref{lem:prop_ul}.
The next lemma enlightens the role of this norm.
\begin{lem}
  \label{lem:est_ul_asym}
  If $v \in H^s(\rr^m, \rr^N)$ and $v_0 \in \rr^N$ is a constant vector, then
$\norm{v + v_0}_{s, \mathrm{ul}} \le \norm{v}_s + \abs{v_0}$.
\end{lem}
\begin{proof}
  The claim follows from the observation 
  \begin{gather*}
    \int_{\rr^m} \babs{\varphi(y - x) \bigl(v(y) + v_0\bigr)}^2 \, dx 
    \le  \norm{v}^2_{L^2} + \abs{v_0}^2 + 2\norm{v}_{L^2} 
    \norm{\varphi_x v_0}_{L^2}
  \end{gather*}
  where we used the H\"older inequality on the last term.
\end{proof}

We search for a solution $u$ of the equation \eqref{second}
where the coefficients and the RHS are assumed to be defined on
$[0,T] \times \widetilde{W}$ with 
an open subset $\widetilde{W} \subset H^{s+1}_{w_0} 
\times H^s_{w_1}$ .
\begin{thm}
  \label{asym_ex}
  Let
  \begin{gather}
    \label{eq:coeff_asym_def}
    g^{\mu\nu}_{\mathrm{a}}(t,\varphi_0, \varphi_1) : = g^{\mu\nu}\bigl(t,w(t) + \varphi_0,
  Dw_0 + D\varphi_0, w_1 + \varphi_1\bigr)
  \end{gather}
  and $f_{\mathrm{a}}$ be  defined analogously. Assume $g^{\mu\nu}_{\mathrm{a}}$ and $f_{\mathrm{a}}$
  to be defined on $[0,T] \times W \subset \rr \times H^{s+1} 
  \times H^s$ and to satisfy the assumptions of the existence 
  Theorem \ref{qlin_ex}.

  Then there exists a constant $0 < T' \le T$ and a unique solution 
  $u$
  to the asymptotic IVP  
  \eqref{second} with initial values
  $(u_0, u_1) \in \widetilde{W} \subset H^{s+1}_{w_0} \times H^s_{w_1}$.
  Further, it has the property
  \begin{gather}
    \label{eq:der_sol_asym}
    u(t) - w(t) 
  \in C([0,T'], H^{s+1})\text{ and }
  \partial_t u - w_1 \in C([0,T'], H^s) \text{ for } 0 \le t \le T'.
  \end{gather}
\end{thm}

\begin{proof}
  The operators $g^{\mu\nu}_{\mathrm{a}}$ and $f_{\mathrm{a}}$ are supposed to be defined on
  $[0,T] \times W \subset \rr \times H^{s+1} 
  \times H^s$ and satisfy the assumptions of the existence theorem 
  \ref{qlin_ex}.
  It therefore follows the existence of a solution 
  $\psi \in  C([0,T], H^{s+1})\cap  C^1([0,T], H^s)$ to the 
  IVP
  \begin{gather}
    \label{eq:asym_ivp}
    g^{\mu\nu}_{\mathrm{a}}(t, \psi, \partial_t \psi) \partial_{\mu} \partial_{\nu}
    \psi = f_{\mathrm{a}}(t, \psi, \partial_t \psi), ~ \restr{\psi} = u_0 - w_0,
    ~ \restr{\partial_t \psi} = u_1 - w_1.
  \end{gather}
  Let $u(t) = 
  \psi(t) + w(t)$. Then we have $g^{\mu\nu}_{\mathrm{a}}(t,\psi, \partial_t \psi) = 
  g^{\mu\nu}(t,u, Du, \partial_t u)$ and the same consideration holds
  for $f_{\mathrm{a}}(t, \psi,  \partial_t \psi)$.
  Since $\partial_{\mu}
  \partial_{\nu} w(t) \equiv 0$ and $\restr{u} = \restr{\psi } + w_0,
  ~ \restr{\partial_t u} = \restr{\partial_t \psi} + w_1$,
  the function $u$ is a solution to the IVP \eqref{second}.
  The differentiability claim of $u$ follows from the property of $\psi$.
\end{proof}
\begin{rem}
  Corollary \ref{cor:sol_diff} applies to the solution of the asymptotic
  equation \eqref{eq:asym_ivp} since the asymptotic is a linear function.
  the 
\end{rem}

\subsection{Spatially local energy estimates}
In this section we will develop uniqueness locally in space and time. 
To this end
it is necessary to impose local conditions on the coefficients and the
RHS of the quasilinear hyperbolic equation \eqref{second}.
The result will be obtained by means of local energy estimates 
in a similar way to the Sobolev-energy estimates of proposition
\ref{prop:e_estimate}.

Consider the coefficients $g^{\mu\nu}$ and the RHS $f$ of equation
\eqref{second} to be nonlinear operators defined as in \eqref{eq:domain_second}
satisfying the assumptions of Theorem \ref{qlin_ex}.
We first fix some notations concerning the local energies to be used.
Let $v = (v_0, v_1) \in W$ and define the energies 
$e_1(x,v)$ and 
$e_2(x,v)$  by
\begin{subequations}
  \begin{align}
  \label{eq:def_e1}
  & e_1(x,v) = \abs{v_0(x)} +
  \abs{v_1(x)} + \abs{D v_0(x)} 
  \\
  \text{and} \quad &
  e_2(x,v) = \abs{v_0(x)}^2 + \abs{v_1(x)}^2 
  + \abs{D v_0(x)}^2.
\end{align}
\end{subequations}
Further we introduce a local energy adapted to the equation as follows
\begin{gather*}
  e(x,v) = \abs{v_0}^2 + \langle - g^{00} v_1, v_1 \rangle 
  + \langle g^{ij} \partial_i v_0, \partial_j v_0 \rangle.
\end{gather*}
The brackets $\langle\, .\,, \,. \,\rangle$ denote 
the inner product of $\rr^N$.

The assumptions on the coefficients give us equivalence of $e$ to
$e_1$ and $e_2$.
We only state the parts of the equivalence which will be used in the
sequel. It holds that
\begin{gather*}
  e_1^2 \le 4e_2 \le 4\bigl(1 + \max(\mu^{-1}, \lambda^{-1})\bigr) e.
\end{gather*}

Assume $u \in C([0,T'], H^{s+1}) \cap C^1([0,T'], H^s)$ to be a solution to the 
IVP \eqref{second} satisfying $(u, \partial_t) \in W$ for $0 \le t \le T'$. 
Suppose the RHS $f$ satisfies the following additional local assumption
\begin{gather}
  \label{eq:loc_cond_uni_full}
  \abs{f\bigl(t,x,v(x)\bigr)} \le K_f^{\mathrm{loc}} e_1(x,v)
\end{gather}
for a function $v = (v_0, v_1) \in W$.
Let $c_0 > 0$ be a constant to be chosen later and let
$(t_0, x_0) \in [0,T]\times\rr^m$ be arbitrary.
Define the cone $C$ and its section with the $\{t = \text{const}\}$-slices
by %
\begin{gather}
  \label{eq:def_cone}
  C = \{ (x,t)\,|\, \abs{x - x_0} < R(t) \} \quad\text{and}\quad
  C_t = \{ x \,|\, \abs{x - x_0} < R(t) \}
\end{gather}
where $R(t) = - c_0( t - t_0)$. We want to show that the value of the
solution $u$ in the vertex $(t_0, x_0)$ of the cone
only depends on the values of $u$ in the cone $C$. 

To this end we will consider the quantity 
\begin{gather*}
  E(t) = E\bigl(x,u(t,x), \partial_t u(t,x)\bigr)
  := \int_{C_t} e\bigl(x, u(x), \partial_t u(x)\bigr) \, dx
\end{gather*}
the integral of the adapted energy $e$ over a spatial part of the 
previously introduced cone.
To estimate the deviation of $E(t) $ in time we 
consider the $\limsup_t $. Following \cite{Evans:1998} we make use of the
Co-area formula adapted to our case. If $v$ is a continuous real-valued 
function defined on
$\rr^m$, then 
$$ \frac{d}{dt} \biggl( \int_{C_t} v \, dx \biggr) = - c_0 \int_{\partial C_t}
v \, dS. 
$$
Thus
\begin{align}
  \limsup_{\tau} &\tfrac{1}{\tau}
  \bigl( E(t+ \tau)  - E(t) \bigr) = \nonumber
  \\
  &  \int_{C_t} \biggl(2 \langle u, \partial_t u
  \rangle 
  - \langle \limsup_{\tau} \tfrac{1}{\tau}
  \bigl( g^{00}(u(t+ \tau)) - g^{00}(u(t)) \bigr)
  \partial_t u, \partial_t u \rangle \nonumber\\
  &~~~~{}+ \langle 
  \limsup_{\tau} \tfrac{1}{\tau}
  \bigl( g^{ij}(u(t+ \tau)) - g^{ij}(u(t)) \bigr) \partial_i u, \partial_j u
  \rangle \nonumber\\
  \label{eq:lok_eest}
  &~~~~{}- 2 \langle g^{00} \partial_t^2 u, \partial_t u \rangle + 
  2 \langle g^{ij} \partial_i \partial_t u, \partial_j u \rangle
  \biggr)\, dx 
  {}- c_0 \int_{\partial C_t}  e(t)\, dS  
\end{align}
The $\limsup_t$ applied to the coefficients can be estimated  via
\begin{gather*}
  \abs{u(t) - u(t')}_s \le \abs{\partial_t u}_{L^{\infty}H^s} \abs{t - t'}
\end{gather*}
and an analogous result holds  for $\partial_t u$. Therefore 
$$ 
\limsup_{\tau} \tfrac{1}{\tau}
  \babs{ (g^{\mu\nu})(t + \tau,u(t+ \tau)) - (g^{\mu\nu})(t,u(t))} \le 
  c\nu \abs{\tau} + c\theta \bigl(
  \abs{\partial_t u}_{L^{\infty}H^s} + \abs{\partial_t^2 u}_{L^{\infty}H^{s-1}}\bigr).
$$
We calculate via integration by parts
taking care of the boundary terms
\begin{multline*}
  \int_{C_t} \langle g^{ij} \partial_i \partial_j u, \partial_t u \rangle
  + \langle g^{ij} \partial_i \partial_t u, \partial_j u \rangle 
  =  \int_{C_t} \bigl(\langle g^{ij} \partial_i \partial_j u, \partial_t u \rangle
  - \langle g^{ij} \partial_i \partial_j u, \partial_t u \rangle \\
  {}-  \langle \partial_i g^{ij}  \partial_t u, \partial_j u \rangle \bigr)\, dx
  {}+ \int_{\partial C_t} 
  \langle g^{ij} \partial_j u, \partial_t u \rangle \nu_i
  \, dS
\end{multline*}
where $\nu$ denotes the unit outer normal to $\partial C_t$.
Analogously we obtain %
\begin{gather*}
  \int_{C_t} 2 \langle 
  g^{0j} \partial_j \partial_t u, \partial_t u \rangle \, dx
  = - \int_{C_t} \langle \partial_j g^{0j} \partial_t u, \partial_t u \rangle
  \, dx + \int_{\partial C_t}  \langle 
  g^{0j} \partial_t u, \partial_t u \rangle \nu_j \, dS.
\end{gather*}
These identities yield 
\begin{align*}
  \limsup_{\tau}& \tfrac{1}{\tau}
  \bigl( E(t+ \tau) - E(t) \bigr) \le 
  {} - \int_{C_t} \Bigl( \langle \partial_j g^{0j} \partial_t u, 
  \partial_t u \rangle
  +  \langle \partial_i g^{ij}  \partial_t u, \partial_j u \rangle \Bigr)\, dx\\
  &{}+\int_{C_t} \Bigl(
  c\theta 
  \bigl(\abs{\partial_t u}_{L^{\infty}H^s} + \abs{\partial_t^2 u}_{L^{\infty}H^{s-1}} 
  \bigr)\bigl(\abs{\partial_t u}^2 +  \abs{D u}^2 \bigr) + \langle  2 u - f, 
  \partial_t u \rangle 
  \Bigr)\, dx \\
  &{}+ \int_{\partial C_t}   \Bigl(- c_0 \,e(t) + 2
  \langle g^{ij} \partial_j u, \partial_t u \rangle \nu_i
  + 2 \langle 
  g^{0j} \partial_t u, \partial_t u \rangle \nu_j\Bigr)\, dS
\end{align*}
It now follows from the local condition \eqref{eq:loc_cond_uni_full} for the 
RHS and the equivalence of
the energies
$e_2(t)$ and $e(t)$ 
that
\begin{align*}
  &\limsup_{\tau} \tfrac{1}{\tau}
  \bigl( E(t+ \tau) - E(t) \bigr) \le 
  {} - \int_{C_t} \bigl( \langle \partial_j g^{0j} \partial_t u, 
  \partial_t u \rangle
  +  \langle \partial_i g^{ij}  \partial_t u, \partial_j u \rangle \bigr)\, dx\\
  &{}+\int_{C_t} \Bigl(\bigl[ 
  c \theta 
  \bigl(\abs{\partial_t u}_{L^{\infty}H^s} + \abs{\partial_t^2 u}_{L^{\infty}H^{s-1}} 
  \bigr) + 8(1 + K_f^{\mathrm{loc}}) \bigr]
  \bigl(1 + \max(\mu^{-1}, \lambda^{-1}) \bigr)e(t) 
  \Bigr)\, dx \\
  &{}+ \int_{\partial C_t}   \Bigl(- c_0 \,e(t) + 2
  \langle g^{ij} \partial_j u, \partial_t u \rangle \nu_i
  + 2 \langle 
  g^{0j} \partial_t u, \partial_t u \rangle \nu_j\Bigr)\, dS.
\end{align*}
The first integral can be estimated by condition \eqref{cond_qlin1} using the
Sobolev embedding arriving at
$$ \langle \partial_j g^{0j} \partial_t u, 
  \partial_t u \rangle
  +  \langle \partial_i g^{ij}  \partial_t u, \partial_j u \rangle
  \le cK e_1^2. %
$$
We infer from this inequality an estimate for the deviation of the energy
$E(t)$ by term involving the energy itself and boundary terms as follows
\begin{multline*}
  \limsup_{\tau} \tfrac{1}{\tau}
  \bigl( E(t+ \tau) - E(t) \bigr) \le \int_{\partial C_t}   
  \Bigl(- c_0 \,e(t) + 2
  \langle g^{ij} \partial_j u, \partial_t u \rangle \nu_i
  + 2 \langle 
  g^{0j} \partial_t u, \partial_t u \rangle \nu_j\Bigr)\, dS
  \\
  {}+ \bigl( 4cK 
  + c \theta 
  \bigl(\abs{\partial_t u}_{L^{\infty}H^s} + \abs{\partial_t^2 u}_{L^{\infty}H^{s-1}} 
  \bigr) + 8(1 + K_f^{\mathrm{loc}})
  \bigr)\bigl(1 + \max(\mu^{-1}, \lambda^{-1}) \bigr) E(t). 
\end{multline*}
The idea is now to adjust the constant $c_0$ such that the boundary term 
becomes 
non-positive. Denoting components of $\rr^N$ by the index $A$ it
follows for the second boundary term by using the generalized H\"older 
inequality
\begin{gather*}
  \abs{ \nu_i g^{ij} \partial_j u^A} \le \bigl( g^{ij} \partial_i u^A \partial_j
u^A\bigr)^{1/2} \bigl( g^{ij} \nu_i \nu_j\bigr)^{1/2}.
\end{gather*}
The last term can be
estimated by $\abs{(g^{ij})}_e^{1/2}$ since $\nu$ has unit Euclidean
length. Summation over $A$
gives us
$$
\biggl( \sum_A \abs{\nu_i g^{ij} \partial_j u^A}^2\biggr)^{1/2}
\le  \abs{(g^{ij})}_e^{1/2}
\biggl( \sum_A g^{ij} \partial_i u^A \partial_j u^A\biggr)^{1/2}
=  \abs{(g^{ij})}_e^{1/2}  \langle g^{ij} \partial_i u , \partial_j 
u\rangle^{1/2}.
$$
Hence 
$$
2\langle g^{ij} \partial_j u, \partial_t u \rangle \nu_i \le 
2  \abs{(g^{ij})}_e^{1/2} \abs{\partial_t u} 
\langle g^{ij} \partial_i u, \partial_j 
u\rangle^{1/2}.
$$
By a similar device we conclude $2\langle 
  g^{0j} \partial_t u, \partial_t u \rangle \nu_j \le 2 \abs{(g^{0j})}_e
  \abs{\partial_t u}^2
 $ and
this yields for the integrand of the boundary term
\begin{align}
  & - c_0 \,e(t) + 2
  \langle g^{ij} \partial_j u, \partial_t u \rangle \nu_i
  + 2 \langle 
  g^{0j} \partial_t u, \partial_t u \rangle \nu_j 
  \nonumber\\
  & \le
  {}- c_0 \,e(t) + 2  \abs{(g^{0j})}_e
  \abs{\partial_t u}^2 + 2  \abs{(g^{ij})}_e^{1/2}
  \abs{\partial_t u}\langle g^{ij} 
  \partial_i u, \partial_j 
  u\rangle^{1/2} 
  \nonumber\\
  & \le - c_0 \,e(t) + 2  \abs{(g^{0j})}_e\tfrac{1}{\lambda} \langle
  - g^{00} \partial_t u, \partial_t u \rangle + 2  \abs{(g^{ij})}_e^{1/2}
  \bigl( \tfrac{1}{\lambda} \langle
  - g^{00} \partial_t u, \partial_t u \rangle\bigr)^{1/2}\langle g^{ij} 
  \partial_i u, \partial_j 
  u\rangle^{1/2} 
  \nonumber\\
  \label{eq:slope}
  &  \le - c_0 \,e(t) 
  + \bigl(  2 \abs{(g^{0j})}_e + \abs{(g^{ij})}_e\bigr)  
  \tfrac{1}{\lambda} \langle
  - g^{00} \partial_t u, \partial_t u \rangle  + \langle g^{ij} 
  \partial_i u , \partial_j 
  u\rangle,
\end{align}
where we used the condition  $g^{00} \le - \lambda$.
The Sobolev embedding provides us with a local estimate 
$\abs{(g^{\mu\nu})}_e \le
c \abs{(g^{\mu\nu})}_{e,s,\mathrm{ul}} \le c K$. 
By choosing 
\begin{equation}
  \label{eq:def_slope}
  c_0 =   1 + 3cK\tfrac{1}{\lambda},
\end{equation}
the 
desired non-positivity of \eqref{eq:slope} follows. Here, the constant $c$ 
only depends on
 the embedding
$H^{s-1}_{\mathrm{ul}} \hookrightarrow C_b^0$.
Applying 
Gronwall's lemma 
yields the  following result.
\begin{lem}
  \label{lem:eest}
  Let  $c_0$ be defined by  \eqref{eq:def_slope} 
  and $0 \le t_0 < t_1$. Then the following
  estimate holds
  \begin{align}
    \label{lok_eest_final}
    & E(t) \le e^{C_1 (t_1 - t_0)} E(t_0) \quad \text{for } 0 \le t_0 \le t \le t_1
    \\
    & \text{with} \quad C_1 = \bigl(4cK 
  + c \theta 
  \bigl(\abs{\partial_t u}_{L^{\infty}H^s} + \abs{\partial_t^2 u}_{L^{\infty}H^{s-1}} 
  \bigr) + 8(1 + K_f^{\mathrm{loc}})
  \bigr)\bigl(1 + \max(\mu^{-1}, \lambda^{-1}) \bigr). \nonumber
  \end{align}
\end{lem}
A consequence of this estimate is that a solution depends only on its values
within the cone $C$ introduced in \eqref{eq:def_cone}.
\begin{thm}[Domain of dependence]
  \label{dodep}
  Assume the nonlinear operators $g^{\mu\nu}$ and $f$ defined in
  \eqref{eq:domain_second}
  satisfy the assumptions of Theorem \ref{qlin_ex}.
  Suppose the local condition \eqref{eq:loc_cond_uni_full} for the RHS holds
  and set $c_0$ as in \eqref{eq:def_slope}.
  Let the initial data satisfy $u_0 \equiv 0 $ and $u_1 \equiv 0$ 
  on $B_{c_0 t_0}(x_0)$. 

  Then $u \equiv 0$ on the closure of the 
  cone with
  base $B_{c_0 t_0}(x_0)$ and vertex $(t_0, x_0)$.
\end{thm}

\begin{proof}
  From the local energy estimate \eqref{lok_eest_final} we derive 
  $ E(t) = 0$ for $0 \le t \le t_0$ since $E(0) = 0$.
  It follows from the assumptions on the coefficients that $\partial_t u
  \equiv 0$ and $D u \equiv 0$ on the spatial part $C_t$ for all $t$.
  The vanishing of the initial values gives the desired result.
\end{proof}
Theorem \ref{dodep} immediately yields the following local uniqueness
result.
\begin{thm}[Local Uniqueness]
  \label{loc_uni}
  Let the coefficients $g^{\mu\nu}$ and the RHS $f$ of the hyperbolic equation
  \eqref{second} satisfy the assumptions of the existence theorem
  \ref{qlin_ex}. Suppose $g^{\mu\nu}$ and $f$ admit 
  constants $\theta_f^{\mathrm{loc}} $
  and $\theta^{\mathrm{loc}} $ satisfying 
  \begin{equation}
    \label{eq:loc_cond_uni}
    \begin{split}
      \babs{f\bigl(t,x,v(x)\bigr) - f\bigl(t,x,w(x)\bigr)} \le 
      \theta_f^{\mathrm{loc}}  e_1(x, v - w)
      \\
      \babs{\bigl(g^{\mu\nu}(t,x,v(x))\bigr) - \bigl(g^{\mu\nu}(t,x,w(x))\bigr)}
      \le \theta^{\mathrm{loc}} 
      e_1(x, v - w)
    \end{split}
  \end{equation}
  for $v, w \in W$ and $e_1$ defined in \eqref{eq:def_e1}.
  Let $c_0$ be defined as in \eqref{eq:def_slope}.
  Assume $u_1$ and $u_2$ to be two solutions of the hyperbolic equation
  \eqref{second} satisfying
  $\restr{u_1} = \restr{u_2}$ and 
  $\restr{\partial_t u_1} =\restr{ \partial_t u_2}$ on $B_{R}(x_0)$.

  Then $u_1 = u_2$ on the closure of the 
  cone with
  base $B_{R}(x_0)$ and vertex $(R/c_0, x_0)$.
\end{thm}

\begin{proof}
  Consider the difference equation
  \begin{gather}
    \label{eq:diff_eq}
    g^{\mu\nu}(u_1) \partial_{\mu} \partial_{\nu} (u_1 - u_2) %
  = f(u_1) - f(u_2) + \bigl( g^{\mu\nu}(u_2) -  
  g^{\mu\nu}(u_1) \bigr)\partial_{\mu} \partial_{\nu} u_2.
  \end{gather}
  The result follows from applying Theorem \ref{dodep} to this equation.
\end{proof}

\begin{rem}
  \label{rem:dep_c0}
  From the proof of the energy estimate 
  \eqref{lok_eest_final} we derive that the slope $c_0$ of the cone on which 
  uniqueness holds depends on the coefficients of the equation 
  stated in \eqref{second}.
  Therefore, as can be seen from  the difference equation \eqref{eq:diff_eq},
  $c_0$ only depends on the solution $u_1$. By interchanging the role of 
  the two solutions it is possible to achieve dependency only on $u_2$.
\end{rem}

\section{Minkowski space}
\label{sec:ex_mink}
In this section we will construct a solution to the Cauchy problem
\eqref{eq:param_ivp} for an immersion of the initial
submanifold in the case where the ambient manifold equals the 
$(n+1)$-dimensional
Minkowski space $\rr^{n,1}$. Only fixed direction, lapse and shift will be 
treated; they will be combined in a timelike vector field serving as
initial velocity.

The layout of this section is as follows.
In section \ref{par:graph_repr} a graph representation for spacelike
submanifolds in Minkowski space will be developed independently of the 
codimension. A first solution to the Cauchy problem will be given
in section \ref{sec:sol_atlas} for given coordinates.
It will provide a connection between assumptions on the initial values in 
given coordinates
and the existence theory for quasilinear second-order hyperbolic
equations. In section \ref{sec:prop_graph_mink} we will derive a solution
to the Cauchy problem \eqref{eq:param_ivp}.

The standard basis of the Minkowski space will be denoted by $\tau_0,
\dots, \tau_n$, where $\tau_0$ denotes the timelike direction.
The metric on $\rr^{n,1}$ will be denoted by 
$\llangle \,\cdot\,,\,\cdot\,\rrangle$ or 
$\eta$ with components $\eta_{AB} = \diag(-1,1,\dots,1)$ and  we set 
$\varepsilon_A = \eta_{AA}$.
Recall that we use the notation $\mabs{v} 
= \bigl( - \llangle v, v\rrangle\bigr)^{\2}$ for a timelike vector $v$.

\subsection{Special graph representation}
\label{par:graph_repr}
In this section we develop coordinates on the initial submanifold $\Sigma_0$.
This will be done by considering a graph representation similar to 
representations for hypersurfaces of Euclidean space derived by A. Stahl in
\cite{Stahl:1996}. 

Let $\Sigma_0$ be a regularly immersed submanifold of dimension $m$. 
Assume $M$ to be an $m$-dimensional manifold and 
$\varphi: M \rightarrow \rr^{n,1}$ to be an immersion with $\im \varphi = 
\Sigma_0$.
Let $p_0 \in M$ and $z: U \subset M \rightarrow \Omega \subset \rr^m$ be
a chart centered at $p_0$. 
Since $\varphi$ is a spacelike immersion we have that the mapping 
$\varphi \circ z^{-1}$ has rank $m$ and we only need a renumbering of the 
spacelike directions to get 
\begin{gather*}
  \det\bigl(d(\varphi^1 \circ z^{-1}, 
\dots, \varphi^m\circ z^{-1})\bigr)(0)
\neq 0.
\end{gather*}
The inverse function theorem now yields the existence of
a diffeomorphism
\begin{gather*}
  \Psi = (\varphi^1 \circ z^{-1}, \dots, \varphi^m\circ z^{-1})^{-1}: 
  V \subset \rr^m
\rightarrow W \subset \Omega.
\end{gather*}
By setting $x = \Psi^{-1} \circ z$ we get a graph representation
$\varphi \circ x^{-1}(w) = w^j \tau_j + u^{\alpha}(w) \tau_{\alpha}$
 in a neighborhood of $p_0$.
Applying an orthogonal transformation to the spacelike directions of the 
Minkowski space gives us the property $Du^a(0) = 0$ for all $a$, where 
$a$ denotes the spacelike directions of the graph. Together with this
orthogonal transformation we derive 
the representation 
\begin{gather}
  \label{eq:graph_repr}
  \Phi(w) := v \circ  \varphi \circ x^{-1}(w) = 
  w^j \tau_j + u^{\alpha}(w) \tau_{\alpha} 
\end{gather}
where $v$ is a choice of coordinates
for the Minkowski space. 
A representation of the geometry is as follows:
\begin{subequations}
  \begin{align}
    \label{eq:def_tangent}
  \text{tangent vectors} && \partial_j \Phi^A \tau_A
  & = \tau_j + \partial_j u^{\alpha} \tau_{\alpha} 
  \\
  \label{eq:def_normal}
  \text{normal vectors} && N_{\alpha} & = \tau_{\alpha} - \varepsilon_{\alpha} 
  \partial_k u_{\alpha}\delta^{k\ell} \tau_{\ell} 
  \\
  \label{eq:ind_metric}
  \text{induced metric} && \ig_{ij} & = \delta_{ij} + \partial_i u^{\alpha}
    \eta_{\alpha\beta} \partial_{j} u^{\beta} 
    \\
    \label{eq:ind_chr_sym}
  \text{induced Christoffel symbols} && \iG_{ij}^k & = 
  g^{\ell k} \partial_i \partial_j u^{\alpha}
    \eta_{\alpha\beta} \partial_{\ell} u^{\beta}
\end{align}
\end{subequations}
and the second fundamental form: 
\begin{align}
  \label{eq:extr_curv}
  \II_{ij} &= 
  \partial_i \partial_j 
    u^{\alpha} \tau_{\alpha} 
    -  \ig^{\ell k} \partial_i \partial_j u^{\alpha} \eta_{\alpha\beta}
    \partial_{\ell} u^{\beta}
    (\tau_k + \partial_k u^{\delta}
    \tau_{\delta}).
\end{align}
The choice of the normal vectors $N_{\alpha}$ 
is based on the construction
of Gram-Schmidt. Instead of using the tangential vectors $\tau_j + \partial_j
u^{\alpha} \tau_{\alpha}$ we only use the vectors $\tau_j$ to develop
the tangential part of $\tau_{\alpha}$.
We take all vectors $\tau_{\alpha}$ which are normal
to the domain of the graph and subtract the tangential part as follows
\begin{align*}
\llangle \tau_{\alpha}, \tau_{\ell} + \partial_{\ell} u^{\beta} \tau_{\beta}
\rrangle = \varepsilon_{\alpha} \partial_{\ell} u_{\alpha}
\overset{!}{=} \llangle \lambda^{k} \tau_k,   
\tau_{\ell} + \partial_{\ell} u^{\beta} \tau_{\beta}
\rrangle = \lambda^k \delta_{k\ell}. 
\end{align*}
This consideration gives us the normal vectors $N_{\alpha}$ as defined
above.

The following lemma establishes control over the eigenvalues of the induced
metric $\ig_{ij}$ involving the graph functions.
\begin{lem}
  \label{lem:posdef}
  The components $\ig_{ij}$ defined in \eqref{eq:ind_metric} of the induced 
  metric 
  satisfies
  \begin{gather}
    \label{eq:posdef}
    (1 - \abs{Du^0}^2_e) \delta_{ij} \le \ig_{ij} \le (1 + \abs{Du}_e^2)\delta_{ij}.
  \end{gather}
\end{lem}
\begin{proof}
  The second estimate comes from a straight-forward calculation using
  the definition \eqref{eq:ind_metric}.
  
  To show the positive definiteness we examine the last term
  of the definition. It reads
  \begin{gather*}
    - \partial_i u^0 \partial_j u^0 
    + \tsum_a (Du^a \otimes Du^a)_{ij},
  \end{gather*}
  where
  the index $a$ denotes the spacelike directions normal to the submanifold.
  The second term is $\ge 0$, it remains therefore to consider the
  first term. An orthogonal transformation $O = (v_1, \dots, v_m)$ with $v_1 = 
  \tfrac{Du^0}{\abs{Du^0}_e}$ gives us 
  \begin{gather*}
    \bigl(O^{\ast}(\delta_{ij} - \partial_i u^0 \partial_j u^0)O\bigr)_{k\ell}
    = \delta_{k\ell}
    - Du^0[v_k] Du^0[v_{\ell}] = \delta_{k\ell} 
    - \delta_{k\ell 1} \abs{Du^0}^2
  \end{gather*}
  and the  desired estimate follows.
\end{proof}
\begin{rem}
  \label{rem:orig}
  At the origin we have the condition $Du^a(0) = 0$ for all $a$,
  the proof therefore yields that $1 - \abs{Du^0(0)}^2$ is the smallest
  eigenvalue of $\ig_{ij}(0)$.  Since the submanifold is 
  assumed to be spacelike and the graph representation is valid at least in
  a neighborhood of the origin, we derive $\abs{Du^0(0)} < 1$.
\end{rem}
The next lemma will establish a bound for second derivatives
of the graph functions in terms of first derivatives and 
the norm of the second fundamental form of the immersion.
\begin{lem}
  \label{lem:2der}
  The following estimate holds
  \begin{gather}
    \label{eq:2der}
    \abs{D^2 u}_e \le 
    m^{1/2} (n+1 - m)^{1/2}\bigl(
    1 + \abs{Du}^2 \bigr)^2\abs{\II}_{\ig,e}.
  \end{gather}
\end{lem}
\begin{proof}
  Let $n_{\alpha\beta} = \llangle N_{\alpha}, N_{\beta} \rrangle$
  and let $h_{ij}^{\alpha}$ be the coefficients of the second fundamental
  form w.r.t.\ the chosen normal vectors $N_{\alpha}$. From 
  representation \eqref{eq:extr_curv} for the second fundamental form we get
  $h_{ij}^{\alpha} = n^{\alpha\beta} \epsilon_{\beta}
  \partial_i \partial_j
  u_{\beta}$. 
  We further set 
  $e_{\alpha\beta} = e(N_{\alpha}, N_{\beta})$ with 
  $e_{\alpha\beta} \ge C_e \delta_{\alpha\beta}$, where $e$
  denotes the Euclidean metric on $\rr^{n+1}$.
  The value of $C_e$ and the invertability of the matrix $n_{\alpha\beta}$
  will be shown later.
  
  We derive from the expression $\abs{\II}^2_{\ig,e} =  
  \ig^{ij} \ig^{\ell k} h_{i\ell}^{\alpha} h_{jk}^{\beta} e_{\alpha\beta}$ that
  \begin{align}
    \label{eq:2der_est_raw}
    \abs{D^2 u}^2_e 
    & %
    \le \abs{(\ig_{ij})}_e^2 \abs{(n_{\alpha\beta})}_e^2 C_e^{-1} \abs{\II}_{\ig,e}^2
  \end{align}
  From a consideration similar to the proof of Lemma
  \ref{lem:posdef} we get $n_{\alpha\beta} \le (1 + \abs{Du}_e^2) 
  \delta_{\alpha\beta}$.
  This yields for the norms of $\ig_{ij}$ and $n_{\alpha\beta}$
  \begin{gather*}
    \abs{(\ig_{ij})}_e \le m^{\2} (1 + \abs{Du}_e^2)
    \quad\text{and}\quad \abs{(n_{\alpha\beta})}_e
    \le (n+1 - m)^{\2}(1 + \abs{Du}_e^2).
  \end{gather*}

  It remains to estimate the matrices $e_{\alpha\beta}$ and $n_{\alpha\beta}$.
  We compute $n_{00} = -1 + \abs{Du^0}^2$ and $n_{ab} = 
  \delta_{ab} + \partial_k u_a \delta^{k\ell} \partial_{\ell} u_b$.
  From Remark \ref{rem:orig} and Lemma 
  \ref{lem:coeff_inverse} we get that $n_{\alpha\beta}$
  is invertible in a neighborhood of $0$.
  To establish the positive definiteness of $e_{\alpha\beta}$ let
  $(v^{\alpha}) \in \rr^{n+1 - m}$, then it follows that
  \begin{gather*}
    v^{\alpha} e_{\alpha\beta} v^{\beta} \ge 
    \abs{v}^2 + \bigl( \abs{v^0} \abs{Du^0} - \abs{w}\bigr)^2 
  \end{gather*}
  where we used the abbreviation $w = (w_j)$ and $w_j = v^a \partial_j u_a$.
  Therefore, the desired result follows with $C_e = 1$.
\end{proof}
At this stage it is convenient to introduce the following definition.
\begin{defn}
  \label{defn:spacelike}
  A submanifold $\Sigma_0$ of the Minkowski space $\rr^{n,1}$ 
  is called 
  \emph{uniformly spacelike with bounded curvature},
  if there exist constants $\omega_1, C_0$ %
  such
  that  
  \begin{gather*}
    \begin{split}
        \inf\{ - \eta( \gamma , \tau_0 ) : \gamma & \text{ timelike 
          future-directed
        unit normal to }\Sigma_0 \}\le \omega_1
      \\
      &\text{and}\quad\abs{\II}_{\ig,e} \le C_0.
    \end{split}
  \end{gather*}
\end{defn}

\begin{rem}
  By abuse of notation we use $\ig$ to denote the induced metric on $\Sigma_0$
  and on $M$ as well as we use $\II$ to denote the second fundamental form
  of $\Sigma_0$ and of the immersion $\varphi$ as a tensor along $\varphi$.
\end{rem}
From now on we suppose the submanifold $\Sigma_0 $ %
to be uniformly spacelike with bounded curvature.
The following lemma provides a definition of a unit timelike
normal for which 
the assumption on unit timelike normals to the submanifold $\Sigma_0$ 
of Definition \ref{defn:spacelike} can be used.
\begin{lem}
  \label{lem:cond_normal}
  Set $\nu_0 = \mabs{N_0}^{-1} N_0$ with the normal vector $N_0$ defined in
  \eqref{eq:def_normal}. Then it follows that $- \eta(\nu_0(0), \tau_0) \le \omega_1$.
\end{lem}
\begin{proof}
  We consider
  a perturbation $\gamma = \mabs{\nu_0 + \lambda^a
    N_a}^{-1}(\nu_0 + \lambda^a N_a)$ of $\nu_0$ with 
  scalars $\lambda^a$. At the origin we have $Du^a(0) = 0$ for all $a$ and
  therefore $N_a(0) = \tau_a$ for all $a$. This yields
  $ \mabs{\nu_0 + \lambda^a N_a}^2 = 1 - \lambda^a \delta_{ab} \lambda^b$
  since $\nu_0 \bot \tau_a$.
  Hence
  \begin{gather*}
    - \eta(\tau_0, \gamma) = (1 - \lambda^a \delta_{ab} \lambda^b)^{-\2} 
    \bigl(- \eta(\tau_0,\nu_0)\bigr) > - \eta(\tau_0,\nu_0)
  \end{gather*}
  if $(\lambda^a) \neq 0$.
\end{proof}
The Definition \ref{defn:spacelike} yields the estimate
\begin{gather}
  \label{eq:cond_nu_0}
  - \eta\bigl(\tau_0, \nu_0(0)\bigr) = \mabs{N_0(0)}^{-1} 
  = (1 - \abs{Du^0(0)}^2)^{-\2}
  \le \omega_1.
\end{gather}
The following lemma is the key ingredient for developing estimates on the
graph representation of uniformly spacelike submanifolds with bounded curvature.
it follows that the strategy of \cite{Stahl:1996} establishing similar estimates
for Euclidean hypersurfaces.
\begin{lem}
  \label{lem:graph_est}
  Let $\lambda\ge 800$ and
  $\rho_1 >0$ be an arbitrary constant.
  Suppose $\Sigma_0$ has bounded curvature with a bound $C_0$ of the form
  $C_0 = (2\lambda m^{1/2}
  (n+1 - m)^{1/2} \rho_1)^{-1}$.
  Assume $x$ to be the coordinates on $M$ which are part
  of the graph representation \eqref{eq:graph_repr}.
  Furthermore, suppose $y \in \rr^m$ such
  that $\abs{y} = r < \rho_1$ and the set $\{\tau y:
  0 \le \tau \le 1\}$ to be contained in the image of $x$. 
  Define a function $v$ by  $v(z) = 
  \tfrac{1}{\sqrt{1 - \abs{Du^0}^2}}$ then
  the following inequalities hold
  \begin{subequations}
    \begin{align}
      \label{eq:stahl_1}
      \bigl(1 + \abs{Du(y)}_e^2\bigr)^{\2} & < 
      (1 + 2\abs{Du^0(0)}^2)^{\2} + 
      \tfrac{1}{2} \bigl(\tfrac{r}{\lambda\rho_1}\bigr)^2 \le B_1
      \\
      \label{eq:stahl_3}
      \abs{D^2 u}_e & \le \tfrac{1}{2\lambda\rho_1} B_1^4
      \\
      \label{eq:stahl_2}
      v(y) & <  \omega_1 + \bigl(1 + B_1^4 \tfrac{r}{2\lambda\rho_1 }
      \bigr) B_1^4 \tfrac{r}{\lambda\rho_1 } 
    \end{align}
  \end{subequations}
  where $B_1 = 3^{\2} 
    + \tfrac{1}{2\lambda^2}$.
\end{lem}

\begin{proof}
  To show the first inequality we consider the bound 
  \eqref{eq:2der} for the second derivatives of the graph functions $u^{\alpha}$.
  The expression of the bound $C_0$ for the curvature yields
  \begin{gather}
    \label{eq:2nd_der_full}
    \abs{D^2 u}_e \le \tfrac{1}{2\lambda\rho_1}(1 + \abs{Du}^2)^2.
  \end{gather}
  From the mean value theorem along the direction
  $e_y = \tfrac{y}{\abs{y}_e}$ and the H\"older inequality we get
    \begin{align*}
    \abs{\partial_j u^{\alpha}(re_y)} & \le \abs{\partial_j u^0(0)} +
    \bigl( \tint_0^r \tsum_i \abs{\partial_i \partial_j u^{\alpha}}^2
   \, d\tau\, r
   \bigr)^{1/2} 
   \end{align*}
   since $Du^a(0) = 0$.
   Summation over $j$ and $\alpha$ gives us
    \begin{align}
    \label{eq:Du_est}
    \abs{Du(re_y)} & \le \abs{Du^0(0)} 
    + \biggl( \int_0^r \abs{D^2 u}^2_e
  \, d\tau\, r
  \biggr)^{1/2}.
  \end{align}
  This yields the following integral inequality 
   \begin{gather*}
     1 + \abs{Du(y)}^2 \le 1 + 2 \abs{Du^0(0)}^2 
     + \tfrac{r}{2\lambda^2 \rho_1^2} \int_0^r
     (1 + \abs{Du}^2)^4 \, d\tau.
   \end{gather*}
  We solve the corresponding ODE giving us the solution $f(t) = 
  \bigl(t_0 - 3 \tfrac{r}{2\lambda^2 
    \rho_1^2}\,t\bigr)^{-1/3}$.
  We want to use a Taylor expansion for $f^{\2}$; it is therefore necessary
  to estimate the derivative.
  This can be done via the choice of the parameter $\lambda$, so that
  it follows that
  \begin{gather*}
    f'(\xi) \le \tfrac{r}{2\lambda^2 
    \rho_1^2} \bigl(t_0 - 3 \tfrac{r}{2\lambda^2 
    \rho_1^2}\,t\bigr)^{-1/3} \le 2\,\tfrac{r}{2\lambda^2 
    \rho_1^2} \text{ for } 0 \le \xi < r
  \end{gather*}
  The Taylor expansion for $f^{\2}$ now yields the desired estimates
  \eqref{eq:stahl_1} and \eqref{eq:stahl_3}.

  We will now consider the function $v$.
  Observe that this function is bounded at the origin
  by virtue of Lemma \ref{lem:cond_normal}, namely by estimate
  \eqref{eq:cond_nu_0}.
  Again we will use an ODE comparison argument, therefore we need
  to estimate the derivative of $v$.
  it follows that
  \begin{align*}
    \abs{Dv} & \le v^3 \bigl( \tsum_j \abs{\tsum_i \partial_i u^0
      \partial_i \partial_j u^0}^2\bigr)^{1/2} \le v^3 \abs{Du^0}
    \abs{D^2 u^0}.
  \end{align*}
  From the inequality \eqref{eq:Du_est} we derive
  together with the bound 
  \eqref{eq:stahl_3} for the second derivatives 
  \begin{align*}
    \abs{Dv} & \le v^3 \Lambda
    \quad\text{with}\quad \Lambda =
    \bigl(\abs{Du^0(0)} + B_1^4 \tfrac{r}{2\lambda\rho_1 }
    \bigr) B_1^4 \tfrac{1}{2\lambda\rho_1 }
  \end{align*}
  This gives us an integral inequality analog of the above case.
  The corresponding ODE is solved by 
  $\tilde{f}(t) = (w_1 - 2 \Lambda t)^{-1/2}$. 
  Analogously to the argument for the function $f$ we estimate $(\omega_1 -
  2 \Lambda t)^{-3/2}$ by 2, which is possible due to the value of $\lambda$.
  This yields the last estimate.
\end{proof}

\begin{rem}
  \label{rem:graph_general}
  Assume we have an uniformly spacelike submanifold with 
  an extrinsic curvature bounded by an arbitrary constant $C_0$. If we set 
  $$
  \rho_1 = (2 \lambda m^{1/2} (n+1 - m)^{1/2} C_0)^{-1},
  $$
  then Lemma \ref{lem:graph_est}  applies to this general situation.
\end{rem}

\begin{rem}
  \label{rem:metric_eigen}
  Lemma \ref{lem:graph_est} also shows that for an uniformly spacelike
  submanifold with bounded curvature the graph representation as described
  above exists at least in the ball $B_{\rho_1}(0)$.
  There are two reasons for the graph representation to fail.
  The first one is the unboundedness of $\abs{Du^a}$ for a spacelike 
  direction
  and the second one is the convergence $\abs{Du^0} \rightarrow 1$ so that
  the matrix $\ig_{ij}$ becomes degenerate.

  The previous lemma shows that these two possibilities can not occur as long
  as we stay in the Euclidean ball with radius $\rho_1 $ about 0.
  In fact, it gives us
  control over the eigenvalues of the induced metric within this ball
  \begin{subequations}
    \begin{gather}
    \label{eq:metric_compare}
        \tilde{G}_1 \delta_{ij}  \le  \ig_{ij} \le \tilde{G}_2 \delta_{ij}  
       \\
       \label{eq:compare_constant}
       \text{with}\quad
           \tilde{G}_1   = \bigl[\omega_1 + \bigl(1 + B_1^4 
    \tfrac{1}{2\lambda}
    \bigr) B_1^4 \tfrac{1}{\lambda}\bigr]^{-2} \text{ and }
    \tilde{G}_2  = B_1^2. 
  \end{gather}
  \end{subequations}
  Observe that these constants are independent of the origin $p_0$.
  In particular estimate \eqref{eq:metric_compare}
  leads to a comparison of Euclidean balls
  and balls w.r.t.\ $\ig$.
  It holds that
  \begin{gather}
    \label{eq:ball_compare}
    B_{\tilde{G}_2^{-1/2} r}^e\bigl(x(p)\bigr) \subset
    x\bigl(B_{r}^\ig(p)\bigr) \subset B_{\tilde{G}_1^{-1/2} r}^e\bigl(x(p)\bigr)
  \end{gather}
  for a radius $r > 0$ such that the middle term makes sense.
  
\end{rem}
We will need  estimates for higher derivatives of the graph functions
depending only on curvature bounds. We
begin with a representation of the partial derivatives of the
second fundamental form in terms of the covariant derivative,
the Christoffel symbols of the ambient space and the induced one 
\begin{lem}
\label{lem:high_cov_der}
  Set $A = \{ \II_{ij}^{B}\}$.
  Then the following identity holds 
  \begin{align}
    \label{eq:high_der_cov}
    \partial^{k} A = \hn^{k} A + 
    \sum \partial^{\alpha_1} \iG \ast \cdots \ast\partial^{\alpha_p} \iG
    \ast \partial^{\beta_1} \indg \ast \cdots \ast \partial^{\beta_q} \indg
      \ast \hn^{\ell} A  
  \end{align}
  where the sum ranges over a certain subset of
  all tuples 
  \begin{gather*}
    (p,\alpha_i, q, \beta_j, \ell) \quad\text{satisfying}\quad
    p + \tsum \alpha_i + q + \tsum \beta_j + \ell
    = k+1.
  \end{gather*}
\end{lem}
\begin{proof}
  The proof will be an induction, so let $k = 1$. Then the 
  definition %
  \eqref{eq:hn} and the expression \eqref{eq:2ndfform} yield
  \begin{gather*}
    \hn A = \partial A + A \ast \indg + A\ast \iG.
  \end{gather*}
  We have $p = q = 1$
  and therefore $p + q = k+1$. 
  \\
  \newline
  $k \rightarrow k+ 1$: We compute $\partial^{k+1} A = \partial \partial^k A$
  and insert the assumption on $k$ which leaves us with considering the 
  derivative of the sum. We begin with the last term involving derivatives
  of $A$: 
  \begin{gather*}
    \partial \hn^{\ell} A  
  = \hn^{\ell} A + \hn^{\ell} A\ast \indg + \hn^{\ell} A\ast \iG
  \end{gather*}
  The terms on the RHS can be described by $\ell \rightarrow \ell +1, ~
  p \rightarrow p + 1$ or $q \rightarrow q + 1$, therefore they
  fit into the pattern for the sum replacing $k$ by $k + 1$.
  The rest of the sum can be treated with the product rule
  increasing one of the indices $\alpha_i$ or $\beta_j$.
\end{proof}
The following lemma establishes estimates for higher derivatives of 
the graph functions $u^{\alpha}$ by inequalities for the terms
involved in the preceding lemma. It will be done only for the Minkowski space.
\begin{lem}
  \label{lem:u_high_der}
  Let $k$ be an integer and assume the immersion $\varphi$ to satisfy the 
  following.
  \begin{gather*}
    \text{There exist constants $C^{\varphi}_0, \dots, C^{\varphi}_k$
      such that }
    \abs{\hn^{\ell} \II}_{\ig,e} \le C_{\ell}^{\varphi} \text{ for }0 \le \ell
  \le k.
  \end{gather*}
  Then there exists a constant $C^u_{k+2}$ such that
  $\abs{D^{k + 2} u}_e 
    \le C^u_{k+2}$.
\end{lem}

\begin{proof}
  According to $\II_{ij}^B = h_{ij}^{\alpha} N_{\alpha}^B$ we have the expression 
\begin{gather}
    \label{eq:u_high_der}
    \varepsilon_{\beta} \partial^k \partial_i \partial_j u_{\beta}
    = \partial^k \II_{ij}^B \,\eta_{BC} \,N_{\beta}^C 
    + \sum_{\genfrac{}{}{0in}{}{k_1 + k_2 = k,}{k_1 < k} }      
  \partial^{k_1} \II_{ij}^B\,
  \eta_{BC} \, \partial^{k_2} N_{\beta}^C.
  \end{gather}
  Since second derivatives are involved in the representation of the second
  fundamental form and the definition of the normal $N_{\alpha}$ only
  contains first derivatives of the graph function we have to estimate
  the first term on the RHS to apply an induction.
  The case $k = 0$ was done in Lemma \ref{lem:graph_est}.
  
  We start with the expression \eqref{eq:high_der_cov} for 
  derivatives of the second fundamental form.
  Since we work in Minkowski space only the induced Christoffel symbols
  needs to be treated. From expression \eqref{eq:ind_chr_sym} we get that
  derivatives of order $\ell$ can be estimated by derivatives of the
  graph function up to order $\ell + 2$ and in \eqref{eq:high_der_cov}
  there occur only derivatives up to order $k -1$.
  The matrix $\ig^{ij}$ occurring in the Christoffel symbols can be estimated
  by Corollary \ref{cor:pos_matrix_der} using inequality 
  \eqref{eq:metric_compare}. It is a lower order term involving first
  derivatives of the graph functions.
  
  The Euclidean norm of the derivatives of the second fundamental
  form can be estimated using the comparison \eqref{eq:metric_compare}
  for the induced metric.
  Therefore, the first term on the RHS is bounded if derivatives
  of the second fundamental form up to order $k$ are bounded.

  Derivatives of order $\ell$ of the components of the normal $N_{\beta}$
  are bounded by derivatives of order $\ell + 1$ of the 
  graph functions, thus the result follows by  induction.
\end{proof}

\subsection{Solutions for fixed coordinates}
\label{sec:sol_atlas}
In this section we will consider the Cauchy problem \eqref{eq:param_ivp}
for the membrane
equation in a given family of coordinates on the initial submanifold and
on the Minkowski space. The arguments used will rely on the theory for
existence of hyperbolic equations. We will adopt the notation used in
section \ref{sec:hyp} and especially in section \ref{sec:qlinear}.

In the following definition we introduce a such a family of coordinates
in which we search for a solution to the membrane equation. 
\begin{defn}
  \label{defn:decomp}
  Let $N$ be a manifold and $\Sigma_0$ be a submanifold.
  Suppose $M$ to be a manifold and $\varphi: M \rightarrow N$ to be an 
  immersion of $\Sigma_0$.
  A set $(U_{\lambda}, x_{\lambda}, V_{\lambda}, y_{\lambda})_{\lambda \in \Lambda}$ 
  is called a \emph{decomposition of $\varphi$} if
  $(U_{\lambda}, x_{\lambda})_{\lambda \in \Lambda}$
  is an atlas for $M$,
  and  $(V_{\lambda}, y_{\lambda})$ are charts
  on $N$ satisfying $\varphi(U_{\lambda}) \subset V_{\lambda}$ for
  all $\lambda \in \Lambda$.
\end{defn}
Assume $M$ to be an $m$-dimensional manifold and $\varphi: M \rightarrow 
\rr^{n,1}$
to be an immersion (called \emph{initial immersion}) of the initial 
submanifold $\Sigma_0$. 
Since we treat the Cauchy problem \eqref{eq:param_ivp} for fixed
direction, lapse and shift we consider an initial velocity combined
in a timelike vector field 
$\chi: M \rightarrow
T\rr^{n,1}$ along $\varphi$.
We will use an integer $s > \tfrac{m}{2} + 1$ to state differentiability
properties.

Let $(U_{\lambda}, x_{\lambda}, V_{\lambda}, y_{\lambda})_{\lambda \in \Lambda}$ be
a decomposition of $\varphi$ and let 
\begin{gather}
  \label{eq:initial_expr}
  \Phi_{\lambda}(z) = y_{\lambda} \circ \varphi \circ x_{\lambda}^{-1}(z)
  \quad\text{and}\quad
  \chi_{\lambda}(z) = (dy_{\lambda})_{\varphi\circ x_{\lambda}^{-1}(z)}(\chi)
\end{gather}
be the representations of $\varphi$ and $\chi$
in these coordinates for $\lambda \in \Lambda$.
The functions $\Phi_{\lambda}$ and $\chi_{\lambda}$ will be the initial data 
for our first main existence result. 

We make
the following uniformity assumptions.
\begin{assum}
  \label{assumptions_atlas}
  \begin{enumerate}
\item \label{assum_atlas_chart} 
  The initial immersion and the 
  decomposition admit 
  constants %
  $\omega_1$ and $\rho_1$ such
  that for each $\lambda \in \Lambda$
  the image of the coordinates $x_{\lambda}(U_{\lambda})$ contain
  a Euclidean ball with radius $\rho_1$ about
  $0$. Further, the representation of the induced metric $\ig_{ij}$ on $M$
  w.r.t.\ the coordinates $x_{\lambda}$ %
  satisfies in the center $\ig_{ij}(0)
  \ge \omega_1^{-2} \delta_{ij}$. 
\item \label{assum_atlas_graph} 
  There exist constants $C_{w_0}$ and $\tilde{C}_{\ell}^{\varphi}$ for
  $2 \le \ell \le s +2$
  such that for each $\lambda \in \Lambda$ 
  \begin{gather}
    \label{eq:assum_bd_pt}
     \abs{D\Phi_{\lambda}(0)}_e \le C_{w_0},~ 
     \abs{D^2 \Phi_{\lambda}}_e \le \tilde{C}^{\varphi}_2 \tfrac{1}{\rho_1},
    ~ \abs{D^{\ell} \Phi_{\lambda}}_e  \le \tilde{C}_{\ell}^{\varphi}
    \text{ for } 3 \le \ell \le s +2.
  \end{gather}
\item \label{assum_atlas_speed}  
  Further the initial velocity $\chi_{\lambda}$
  satisfies for each $\lambda \in \Lambda$:
  \begin{align}
    \label{eq:assum_bd_vel}
    \eta(\chi_{\lambda},\chi_{\lambda}) \le - L_2 \quad\text{and}\quad
    \abs{D^{\ell} \chi_{\lambda}}_{e,e} & \le \tilde{C}_{\ell}^{\chi}
    \text{ for }0 \le \ell \le s +1
  \end{align}
  with constants $\tilde{C}_{\ell}^{\chi}, L_2$.
\end{enumerate}
\end{assum}
The following theorem is the first main result for the  initial value problem 
\eqref{eq:param_ivp} for the membrane equation
in a given family of coordinates.
\begin{thm}%
  \label{thm:ex_uni_atlas}
  Let the decomposition $(U_{\lambda}, x_{\lambda}, 
  V_{\lambda}, y_{\lambda})_{\lambda \in \Lambda}$ and the initial data
  $\varphi_{\lambda}$ and 
  $\chi_{\lambda}$ for the membrane equation satisfy the assumptions 
  \ref{assumptions_atlas}.

  Then there exist
  constants $\bar{T} > 0$, $0 < \theta < 1$ and a
  family $(F_{\lambda})$ of bounded $C^2$-immersions 
  $F_{\lambda}: [-\bar{T},\bar{T}] \times B^e_{\theta \rho_1/2}(0)
  \subset \rr \times x_{\lambda}(U_{\lambda})
  \rightarrow \rr^{n + 1}$
  solving the reduced membrane equation 
  \begin{equation}
    \label{eq:mem_mink}
    g^{\mu\nu} \partial_{\mu} \partial_{\nu} F^A - 
    g^{\mu\nu}\hat{\Gamma}_{\mu\nu}^{\lambda} 
    \partial_{\lambda} F^A  = 0
  \end{equation}
  w.r.t.\ the background metric $\hat{g}$ defined in \eqref{eq:back_metric}
  and attaining the initial values 
  \begin{gather}
    \label{eq:mink_initial}
    \restr{F_{\lambda}} = \Phi_{\lambda}\quad\text{and} \quad
    \restr{\partial_t F_{\lambda}} = \chi_{\lambda}.
  \end{gather}

  Let $F_{\lambda}$ and $\bar{F}_{\lambda}$
  be two such solutions defined on 
  the image of the coordinates $x_{\lambda}$.
  Assume $z \in \rr^m$ to be a point contained in the image
  of $x_{\lambda}$. If $F_{\lambda}$ and $\bar{F}_{\lambda}$
  attain the initial values $\Phi_{\lambda}$ and $\chi_{\lambda}$ on
  a ball $B^e_r(z)$, then they coincide
  on the double-cone with base $B_{r}^e(z)$ and slope $c_0$.
\end{thm}

\begin{rem}
  \begin{enumerate}
  \item A lower bound for the existence time $\bar{T} > 0$ will be given in
    Remark \ref{rem:geom_lower}.
  \item   A precise value for the slope $c_0$ of the uniqueness cone
    will be defined in \eqref{eq:est_c_0}. 
  \end{enumerate}
\end{rem}

\begin{rem}
  \label{rem:mem_atlas_sol_diff}
  Let $\ell_0$ be an integer.
  Assume that the initial values and the decomposition satisfy
  the assumptions \ref{assumptions_atlas} with an integer 
  $r = s + \ell_0 > \tfrac{m}{2} + 1 + \ell_0$
  instead of an $s > \tfrac{m}{2} + 1$.
  Then the family $(F_{\lambda})$ of solutions to the 
  reduced membrane
  equation are immersions of class $C^{2+ \ell_0}$.
\end{rem}
We will begin with the construction of a solution to equation 
\eqref{eq:mem_mink} in fixed charts $x_{\lambda}$ and $y_{\lambda}$.
The strategy will be to obtain a formulation of the equation to which
Theorem \ref{asym_ex} applies. The construction will  provide us with an 
estimate on the
existence time and the part of the image of the coordinates $x_{\lambda}$
on which the solution attains the initial
data in dependency on the constants occurring in the assumptions 
\ref{assumptions_atlas}. 

Let
 $\lambda \in \Lambda$ be fixed. 
Since the existence result for hyperbolic equations
obtained in section \ref{sec:hyp}
is only suitable for functions defined on all of $\rr^m$, we have to
extend the functions $\Phi_{\lambda}$ and $\chi_{\lambda}$. 
To obtain such an extension we introduce a cut-off function 
$\zeta \in C_c^{\infty}(\rr^m)$ with the property that for a constant
$0 < \theta <1$ we have 
$0 \le \zeta \le 1$, $\zeta \equiv 1$ on $B_{\theta \rho_1/2 }(0)$ 
and 
  $\zeta \equiv 0$ outside $B_{\theta \rho_1 }(0)$. The derivatives of
$\zeta$ are bounded by
\begin{gather}
  \label{eq:cut_off_est_mink}
  \abs{D^{\ell} \zeta}_e
    \le \tilde{C}_{\ell} (\theta \rho_1)^{-\ell},
\end{gather}
 where $\tilde{C}_{\ell}$
denote constants independent of $\theta$ and $\rho_1$.

Define a linear function $w(t,x)$ on $\rr^{m+1}$ by
\begin{gather}
  \label{eq:asymptotics}
  w(t,x)  = w_0(x) +  t \,w_1
  \\
  \text{with} \quad w_0(x) = x^{\ell} \partial_{\ell} \Phi_{\lambda}(0)
  \text{ for } x \in \rr^m\quad\text{and}\quad
  w_1 = \chi_0 = \chi_{\lambda}(0).
  \nonumber
\end{gather}
This function is defined as the tangent plane of the function $\Phi_{\lambda}$
 at the
origin moving
with the constant velocity $\chi_0$.
The linear function $w(t,x)$ satisfies the membrane equation \eqref{eq:membrane}
w.r.t.\ the coordinates $x_{\lambda}$ and $y_{\lambda}$.
The idea is now to apply the cut-off function $\zeta$ to 
terms  of the reduced
equation \eqref{eq:mem_mink} which are independent of the solution.
Then we are able to search for solutions within
Sobolev-space perturbations of $w(t,x)$.
The only independent term in equation \eqref{eq:mem_mink} are the 
Christoffel symbols of the background metric $\hat{g}$ defined by 
\eqref{eq:back_metric}. Since the definition only contains the initial values, 
we begin with a cut-off process for them.
Set
\begin{subequations}
  \begin{align}
  \label{eq:def_inter}
  && \init{\Phi}(x) 
   & = %
  \zeta(x)\bigl(\Phi_{\lambda}(x) - w_0(x) \bigr)
  &&
  \text{ for } x \in \rr^m
  \\
  \label{eq:vel_inter}
  \text{and} && \init{\chi}(x) &=\zeta(x) \bigl(\chi_{\lambda}(x) 
  - \chi_0\bigr) &&\text{ for } x\in \rr^m.
\end{align}
\end{subequations}
These functions constitute the interpolation of
the function $\Phi_{\lambda}$ and its tangent plane at the origin 
and the velocity $\chi_{\lambda}$ with the velocity
at the origin. 
From these functions we derive a cut-off of the background metric $\hat{g}$.
Let $\hat{a}_{\mu\nu}$ be the matrix with the components
\begin{multline}
  \label{eq:back_comp}
  \hat{a}_{00} = (\chi_0 + \init{\chi})^A \eta_{AB} (\chi_0 + \init{\chi})^B,~
  \hat{a}_{0j} = (\chi_0 + \init{\chi})^A \eta_{AB} 
  \partial_j \init{\Phi}^B,~
  \hat{a}_{ij} = \partial_i \init{\Phi}^A \eta_{AB} 
  \partial_j \init{\Phi}^B.
\end{multline}
This metric coincides with the background metric on the ball
$B_{\theta \rho_1/2}(0)$.
The Christoffel symbols of
this metric will be denoted by $\hat{\gamma}_{\mu\nu}^{\lambda}$.

We search for a function $F: {\cal V}\subset\rr^{m+1} \rightarrow 
\rr^{n+1}$ defined on a neighborhood of $\{t = 0\}$ solving the IVP
\begin{multline}
  \label{eq:geom_solve}
  g^{\mu\nu}(DF,\partial_t F) \partial_{\mu} \partial_{\nu} F^A = 
  f^A(t,DF,\partial_t F) = 
  g^{\mu\nu}(DF,\partial_t F) 
  \hat{\gamma}_{\mu\nu}^{\lambda}(t) \partial_{\lambda} F^A, 
  \\
  \restr{F} = w_0 + \init{\Phi}, ~ \restr{\partial_t F}
    = w_1 + \init{\chi}.
\end{multline}
The coefficient matrix 
$ g^{\mu\nu}(DF,\partial_t F)$ is defined as the inverse of the 
matrix 
\begin{gather}
  \label{eq:def_coeff}
  g_{\mu\nu}(DF,\partial_t F) = \partial_{\mu} F^A \eta_{AB} \partial_{\nu} F^B
\end{gather}
corresponding to the pullback metric $F^{\ast} \eta$.
For notational convenience we set $g_{\mu\nu}(F) = 
g_{\mu\nu}(DF,\partial_t F)$ and 
$f(F) = f(t,DF,\partial_t F)$, if
the exact dependency is not important for the argument.
We will show the following proposition.
\begin{prop}
  \label{prop:cutoff_ex}
    There exist a constant $T' >0$ and a unique
  $C^2$-solution $F:[-T',T']\times\rr^m
  \rightarrow \rr^{n+1}$ of the IVP \eqref{eq:geom_solve}
  satisfying
  \begin{gather}
    \label{eq:loc_lsg_diff}
    F(t) - w(t) \in C([-T',T'],H^{s+1}), ~\partial_t F(t) - w_1 \in 
    C([-T',T'],H^{s}).
  \end{gather}
\end{prop}
From this proposition it follows immediately
  the existence claim of Theorem \ref{thm:ex_uni_atlas} by setting $F_{\lambda}
  := F$. The cut-off process yields that
  a solution of equation \eqref{eq:geom_solve} solves the reduced membrane
  equation \eqref{eq:mem_mink} and attains the initial values 
  \eqref{eq:mink_initial} within 
  the ball $B_{\theta \rho_1/2}(0)$.

The differentiability properties \eqref{eq:loc_lsg_diff} suggest that
we make use of Theorem \ref{asym_ex} discussing asymptotic equations. 
To obtain the conditions of Theorem \ref{qlin_ex} for the asymptotic 
coefficients and RHS
(cf. \eqref{eq:coeff_asym_def})  we will follow the treatment
of the Cauchy problem for the Einstein equations in \cite{Kato:1976}.

Let $\Omega \subset \rr^{m(n+1)}
\times \rr^{n+1}$ be a set chosen later and define for $(Y,X) \in \Omega$
with $Y = (Y_k)$ the matrix $g^{\mathrm{a}}_{\mu\nu}$ by
\begin{gather}
  \label{eq:def_coeff_asym}
  g_{0\ell}^{\mathrm{a}}(Y,X) 
  := (\partial_{t} w + X)^A
  \eta_{AB}(\partial_{\ell} w + Y_{\ell})^B,
\end{gather}
where the other parts $g^{\mathrm{a}}_{00}$ and $g^{\mathrm{a}}_{k\ell}$ are defined analogously (see
\eqref{eq:back_comp}). The inverse of $g_{\mu\nu}^{\mathrm{a}}$ will be denoted by 
$g^{\mu\nu}_{\mathrm{a}}$. In an analogous way we define a function $f_{\mathrm{a}}$ by
\begin{gather}
  \label{eq:def_rhs_asym}
  f_{\mathrm{a}}^A(t,Y,X) := g^{\mu\nu}_{\mathrm{a}}(Y,X)\bigl( \hat{\gamma}_{\mu\nu}^{0}(t)
  (\partial_{t} w + X)^A + \hat{\gamma}_{\mu\nu}^{\ell}(t)
  (\partial_{\ell} w + Y_{\ell})^A 
  \bigr)
\end{gather}
As in \cite{Kato:1976} the set $\Omega$ will be used to ensure that
the matrix $g^{\mu\nu}_{\mathrm{a}}$ has the desired signature $(\,{-}\,{+}\,\cdots\,{+}\,)$.

These definitions give rise to the  following operators
\begin{multline*}
  g_{\mu\nu}^{\mathrm{a}}(\varphi_0, \varphi_1) = g_{\mu\nu}^{\mathrm{a}}(D\varphi_0, \varphi_1),~
  g^{\mu\nu}_{\mathrm{a}}(\varphi_0, \varphi_1) = g^{\mu\nu}_{\mathrm{a}}(D\varphi_0, \varphi_1)
  \\
  \text{and} \quad f_{\mathrm{a}}(t,\varphi_0, \varphi_1) = f_{\mathrm{a}}(t, D\varphi_0, \varphi_1)
\end{multline*}
with domain $W \subset H^{s+1} \times H^s$ chosen later in dependency 
on $\Omega$. 
Observe that the choice
of the background metric $\hat{g}$ (cf. \eqref{eq:back_metric}) gives us that
the RHS $f_{\mathrm{a}}$ in fact does not depend explicitly on the time parameter $t$.
These operators correspond to the definition 
\eqref{eq:coeff_asym_def}
of asymptotic coefficients for the IVP \eqref{eq:geom_solve}.
Therefore, if we find a solution $\psi$ of the following 
\emph{asymptotic equation}
\begin{subequations}
  \begin{gather}
  \label{eq:geom_solve_asym}
  g^{\mu\nu}_{\mathrm{a}}(D\psi,\partial_t \psi) 
  \partial_{\mu} \partial_{\nu} \psi^A  = f_{\mathrm{a}}(t,D\psi,\partial_t \psi)
  \\
  \label{eq:asym_initial}
  \text{with initial values}\quad\restr{\psi^A} = \init{\Phi}^A \text{ and }
  \restr{\partial_t \psi^A} = \init{\chi}^A
\end{gather}
\end{subequations}
then $\psi(t) + w(t)$ is a solution to IVP \eqref{eq:geom_solve}.

To obtain suitable definitions of $\Omega$ and $W$, the strategy 
will be first to choose the set 
$\Omega$ such that each $(Y,X) \in \Omega$
satisfies $g_{00}^a(Y,X) 
\le - L_2(1 - r_0)$ and $g_{k\ell}^a(Y,X) \ge 
\omega_1^{-2}(1 - R_0)
\delta_{k\ell}$
with fixed constants $0 < r_0, R_0 <1$.
Then we define
$W \subset H^{s+1} \times H^s$ in a way such that $(\varphi_0, \varphi_1)
\in W$ yields $(D\varphi_0,\varphi_1) \in \Omega$ pointwise. The exact definition of
$W$ will be given in the next section.

According to condition \eqref{cond_qlin4} we need to show that
$g_{00}^a$ can be bounded away from $0$ and the submatrix $g_{ij}^a$
is positive definite. Lemma \ref{lem:est_metric_inverse} then yields the 
desired estimates.
\begin{lem}
  \label{lem:ind_metric}
  For the matrix $g_{\mu\nu}^{\mathrm{a}}(Y,X)$ 
  the following inequalities hold
  \begin{gather*}
    g_{k\ell}^a \ge \bigl(\omega_1^{-2} -  2 
      \abs{Y}_e (\abs{Y}_e + \abs{Dw_0}_e)\bigr)\delta_{k\ell}
    \quad  \text{and} \quad 
      g_{00}^a  \le 
      \llangle \chi_0 , \chi_0 \rrangle 
      + 2 \abs{X}_e (\abs{\chi_0}_e + \abs{X}_e),
  \end{gather*}
  where we set $\abs{Y}_e :=(\tsum_k \abs{Y_k}^2)^{1/2} $.
\end{lem}

\begin{proof}
  The estimates follow from the definition of the matrix
  in \eqref{eq:def_coeff_asym} and the condition $\init{g}_{ij} = g_{ij}(0) 
  \ge \omega_1^{-2} \delta_{ij}$ for the induced metric
  at the origin.
\end{proof}
Our starting point will be the following definition
\begin{equation}
  \label{eq:omega_def}
  \Omega := B^e_{\delta_1}(0) \times B^e_{\delta_2}(0) \subset \rr^{m(n+1)}
  \times \rr^{n+1}
\end{equation}
with constants
$\delta_1, \delta_2 > 0$ to be chosen with the help of
 the estimates established in  Lemma \ref{lem:ind_metric}. 
Let $(Y, X) \in \Omega$. Then Lemma \ref{lem:ind_metric} yields
\begin{gather}
  \label{eq:pos_neg}
  \begin{split}
    g_{k\ell}^a(Y, X)  & \ge \omega_1^{-2}\bigl(1 - 2 \omega_1^2
  \delta_1 (C_{w_0} + \delta_1 )\bigr)\delta_{k\ell} 
  \\
  \text{and} \qquad
  g_{00}^a(Y, X)  & \le 
  - L_2\bigl(1 -  2 L_2^{-1}\delta_2 (\tilde{C}_0^{\chi} +  \delta_2)\bigr).
  \end{split}
\end{gather}
Here, we used the assumptions for $\chi_{\lambda}$ and $\Phi_{\lambda}$.
These estimates provide us with the following conditions for $\delta_1$ and 
$\delta_2$
\begin{gather}
  \label{eq:delta12}
  2 \delta_1 (C_{w_0} + \delta_1 ) \le \omega_1^{-2} R_0 \quad \text{and} \quad
  2 \delta_2 (\tilde{C}_0^{\chi} +  \delta_2)\le L_2 r_0.
\end{gather}
Let $\delta_1, \delta_2 > 0$ be two constants satisfying the preceding
inequalities. This choice provides us with a definition of the set $\Omega$
such that for $(Y,X) \in \Omega$, the following holds:
\begin{gather}
  \label{eq:est_part}
  g_{00}^a(Y,X) \le - \tilde{\lambda} := -L_2(1 - r_0)\text{ and } 
  g_{k\ell}^a(Y,X)
 \ge \tilde{\mu} \delta_{k\ell} := \omega_1^{-2}(1 - R_0)\delta_{k\ell}.
\end{gather}
From the definition of $\Omega$ we infer bounds for the norms  of the matrices
$(g_{\mu\nu}^{\mathrm{a}})$ and $(g^{\mu\nu}_{\mathrm{a}})$, and estimates for
$g^{00}_a$ and $g^{ij}_a$ needed by condition \eqref{cond_qlin4} of the existence
theorem.
For notational convenience we define
\begin{subequations}
\begin{align}
  \label{eq:est_space_full}
  K_0 &:= (C_{w_0}^2 + \delta_1^2)^{\2} && \Longrightarrow &
  \abs{Dw_0 + Y}_e &\le K_0
  \\
   \label{eq:est_time_full}
  K_1 &:= \bigl((\tilde{C}_0^{\chi})^2 + \delta_2^2\bigr)^{\2} &&
  \Longrightarrow &
  \abs{\chi_0 + X}_e &\le K_1
\end{align}  
\end{subequations}
for $(Y,X) \in \Omega$.
With the help of Lemma \ref{lem:est_metric_inverse} we obtain
the desired estimates for
$g^{00}_a, ~ (g^{ij}_a)$ and the norm of the inverse matrix.
Estimates for several submatrices of
the matrix $g_{\mu\nu}^{\mathrm{a}}$ will be needed. We obtain
such estimates in the 
following lemma.
\begin{lem}
  \label{lem:norm_metric}
  Assume $(Y, X) \in \Omega$. Then the following inequalities hold
  \begin{subequations}
      \begin{align}
    \label{eq:norm_metric}
    \abs{g_{00}^a(Y, X)}  & \le  K_1^2, &&&
    \babs{\bigl(g_{0\ell}^{\mathrm{a}}(Y, X)\bigr)}_e  &\le  K_0 K_1, 
    \\
    \babs{\bigl(g_{ij}^a(Y, X)\bigr)}_e  &\le K_0^2 
    &\text{and}&&
    \babs{\bigl(g_{\mu\nu}^{\mathrm{a}}(Y, X)\bigr)}_e  &\le K_1^2 + K_0^2.
  \end{align}
  \end{subequations}
\end{lem}
\begin{proof}
  We begin with the second estimate since all terms occurring in the matrix
  $g^{\mathrm{a}}_{\mu\nu}$ are involved (see definition \eqref{eq:def_coeff_asym}).
  We have
  \begin{eqnarray*}
    \tsum_{\ell} \abs{g_{0\ell}^{\mathrm{a}}}^2 & \le & 
    \bigl(\abs{X} + \abs{\chi_0}\bigr)^2
    \abs{\eta}^2 \tsum_{\ell} \bigl( \abs{Y_{\ell}} + \abs{\partial_{\ell}
      w_0}\bigr)^2 \\
    & \le & \bigl(\abs{X} + \abs{\chi_0}\bigr)^2
    \bigl( \abs{Y} + \abs{D w_0}\bigr)^2.
  \end{eqnarray*}
  Hence, the inequality for $g_{0\ell}^{\mathrm{a}}$ follows from the definition
  of the constants $K_0$ and $K_1$.
  A similar argument yields the other estimates.
\end{proof}
The preceding lemma and  Lemma \ref{lem:est_metric_inverse}
immediately yield the following result.
\begin{lem}
  \label{lem:norm_metric_inverse}
  Assume $(Y, X) \in \Omega$. Then the following inequalities hold
  for the coefficients $g^{\mu\nu}_{\mathrm{a}}(Y,X)$
  \begin{subequations}
      \begin{align}
    \label{eq:est_inverse_full}
    &&\babs{\bigl(g^{\mu\nu}_{\mathrm{a}}(Y, X)\bigr)}_e^2  
    &\le
    \tilde{\lambda}^{-2} + 
     2 \tilde{\lambda}^{-2}  \tfrac{m}{\tilde{\mu}^2} 
     K_0^2 K_1^2
     + \tfrac{m}{\tilde{\mu}^2} &&=: \Delta^{-2}, 
     \\
     \label{eq:est_lambda_mu}
    &&g^{00}_{\mathrm{a}}(Y, X)  & \le   - K_1^{-2} 
    \bigl( 1+ \tfrac{m^{1/2}}{\tilde{\mu}}
    K_0^2 \bigr)^{-1} &&=:-\lambda
    \\
    \label{eq:est_mu}
    \text{and} &&\quad
    g^{ij}_{\mathrm{a}}(Y, X)  & \ge  K_0^{-2}
    \bigl(1  + \tilde{\lambda}^{-1} 
    K_1^2 
    \bigr)^{-1}\delta^{ij} &&=: \mu\delta^{ij}
  \end{align}
  \end{subequations}
  where we used the constants $K_0, K_1, \tilde{\lambda}$ and $\tilde{\mu}$
  defined in \eqref{eq:est_space_full},
  \eqref{eq:est_time_full} and \eqref{eq:est_part}.
\end{lem}

\subsubsection{Sobolev estimates}
\label{sec:sob_est_mink}
Having in mind the estimates following from the definition of $\Omega$
we are now able to define the domain $W \subset H^{s+1}\times H^s$ of
the coefficients and the RHS of the asymptotic
equation \eqref{eq:geom_solve}. 
The construction should be done in such a way that if $(\varphi_0, \varphi_1)
\in W$, then the estimates stemming from the definition of
$\Omega$ are applicable.
We start with balls around the initial values \eqref{eq:asym_initial} 
of the asymptotic equation. 
Let $\rho>0$ be a constant chosen later and define 
\begin{gather}
  \label{eq:domain}
  W := B_{\rho}(\init{\Phi})
\times B_{\rho}(\init{\chi}).
\end{gather}
The goal is to ensure the following property 
\begin{gather}
  \label{eq:cond_W}
  (\varphi_0, \varphi_1) \in W \quad\Longrightarrow\quad
  (D\varphi_0, \varphi_1) \in \Omega \text{ pointwise.}
\end{gather}
Two steps are necessary to make the definition of $W$ reasonable.
Firstly, the initial values have to
satisfy condition \eqref{eq:cond_W}, and secondly there needs to be enough 
space left for members of $W$.
The first step will be obtained by the following lemma.
\begin{lem}
    \label{lem:initial_est}
    The following inequalities hold for the initial values
    $\init{\Phi}$ and $\init{\chi}$
    \begin{gather}
      \label{eq:est_initial}
      \abs{D\init{\Phi}}_e  \le 
      \tilde{C}^{\varphi}_2 \theta \bigl(
      1 + \tilde{C}_1
  \bigr)\qquad\text{and}\qquad
    \abs{\init{\chi}}_e  \le \tilde{C}^{\chi}_1 \rho_1 
      \theta.
    \end{gather}
    The constants are taken from the assumptions \eqref{eq:assum_bd_pt} and 
    \eqref{eq:assum_bd_vel},
    and from the bounds \eqref{eq:cut_off_est_mink} for the cut-off function.
\end{lem}
  \begin{proof}
    The definition of the interpolated function $\init{\Phi}$ yields
    \begin{eqnarray*}
      \abs{D\init{\Phi}}_e & \le & \zeta \abs{D\Phi_{\lambda} - D\Phi_{\lambda}(0)}_e 
      + \abs{D\zeta}_e 
      \abs{\Phi_{\lambda} - x^{\ell} \partial_{\ell} 
        \Phi_{\lambda}(0)}_e.
    \end{eqnarray*}
    To obtain an estimate for the two terms we use the bound for
    the second derivative of $\Phi_{\lambda}$.
    The mean value theorem and the H\"older inequality
    yield for $\abs{x} =: r < \rho_1$
    \begin{gather*}
      \abs{D\Phi_{\lambda} - D\Phi_{\lambda}(0)}_e\le 
      \tilde{C}^{\varphi}_2 \tfrac{r}{\rho_1}  \quad\text{and}\quad
      \abs{\Phi_{\lambda} - x^{\ell} \partial_{\ell} 
        \Phi_{\lambda}(0)}_e\le 
      \tilde{C}^{\varphi}_2 \tfrac{r^2}{\rho_1}.
    \end{gather*}
    The fist estimate follow by considering the bounds for cut-off function 
    $\zeta$ and $r \le \theta \rho_1$.

    The  estimate for the initial velocity can be derived via a
    similar device using the bound for $D\chi_{\alpha}$.
\end{proof}
Now we choose $\theta$ small enough 
that the RHS of the first inequality in \eqref{eq:est_initial}
 is less than $\delta_1/2$
and the RHS of the second inequality  is less than $\delta_2/2$.
Furthermore, due to the 
Sobolev embedding theorem it is possible to adapt $\rho$ so that condition 
\eqref{eq:cond_W} is 
fulfilled.

To obtain estimates meeting the conditions \eqref{cond_qlin1} to 
\eqref{cond_qlin_add2} %
 we need  bounds on $\norm{\init{\Phi}}_{s+1}$ and 
$\norm{\init{\chi}}_s$. 
These can be derived by a similar device as in the proof of Lemma
\ref{lem:initial_est} via the assumptions \ref{assumptions_atlas} on the 
representations
of $\varphi$ and $\chi$ w.r.t.\ the given decomposition and the bounds for 
derivatives of the cut-off function $\zeta$.
Therefore, constants $D_0, \tilde{D}_0, D_1$ and $\tilde{D}_1$ exist such that
\begin{gather}
  \label{eq:sob_initial}
  \norm{\init{\Phi}}_{s+1} \le D_0, ~ \norm{D^2 \init{\Phi}}_{s } \le \tilde{D}_0
  \qquad\text{and}\qquad
  \norm{\init{\chi}}_{s} \le D_1,~ \norm{D\init{\chi}}_{s} \le \tilde{D}_1.
\end{gather}
The bounds $\tilde{D}_0$ and $\tilde{D}_1$ are needed to estimate the 
Christoffel symbols of the cut-off background metric $\hat{a}_{\mu\nu}$.

We define constants  similar to $K_0$ and $K_1$ declared in 
\eqref{eq:est_space_full}
and \eqref{eq:est_time_full} we define for $(\varphi_0, \varphi_1) \in W$
\begin{subequations}
  \begin{align}
    \label{eq:Hs_est_space_full}
  K_{0,s} &:= (C_{w_0}^2 + \rho^2 + D_0^2)^{\2} && \Longrightarrow &
   \norm{Dw_0   + D\varphi_0}_{s,\mathrm{ul}}  &\le K_{0,s}
  \\
   \label{eq:Hs_est_time_full}
  K_{1,s} &:= \bigl((\tilde{C}_0^{\chi})^2 + \rho^2 + D_1^2\bigr)^{\2} &&
  \Longrightarrow &
   \norm{\chi_0 + \varphi_1}_{s,\mathrm{ul}} &\le K_{1,s}
 \end{align}
\end{subequations}
where we used the property of the special norm for the uniformly local Sobolev
spaces $H^s_{\mathrm{ul}}$
derived in Lemma \ref{lem:est_ul_asym}.

The next lemma states generic estimates for the coefficients and the RHS
of the equation in \eqref{eq:geom_solve} depending on the unknown function 
$F$.
\begin{lem}
  \label{lem:gen_est}
  The following estimates hold for two functions $F$ and $\bar{F}$
  \begin{subequations}
    \begin{align}
      \label{eq:gen_rhs_0}
      \abs{f( F)}  &\le  
      \babs{\bigl(g^{\mu\nu}( F)\bigr)}\,
      \abs{(\hat{\gamma}_{\mu\nu}^{\lambda})} 
      \bigl(\abs{\partial_t F}^2 + \abs{DF}^2\bigr)^{\2} 
      \\
      \label{eq:gen_rhs_lip}
      \begin{split}
        \abs{f( F)  - f(\bar{F})}   &\le 
      \babs{\bigl(g^{\mu\nu}( F)\bigr)
        - \bigl(g^{\mu\nu}( \bar{F})\bigr)}\abs{(\hat{\gamma}_{\mu\nu}^{\lambda})} 
      (\abs{\partial_t \bar{F}}^2 + \abs{D\bar{F}}^2)^{\2}
      \\
      & {} \qquad\qquad \qquad
      + \babs{\bigl(g^{\mu\nu}( F)\bigr)}\,
      \abs{(\hat{\gamma}_{\mu\nu}^{\lambda})}(\abs{\partial_t V}
      + \abs{DV})
      \end{split}
      \\
      \label{eq:gen_coeff_lip}
      \begin{split}
        \babs{\bigl(g_{\mu\nu}(F)\bigr)  
        - \bigl(g_{\mu\nu}( \bar{F})\bigr)}_e
       & \le  (\abs{\partial_t V}
      + \abs{DV}) 
      \bigl((\abs{\partial_t \bar{F}}^2 
      + \abs{D\bar{F}}^2)^{\2} 
      \\
      &{} \qquad\qquad\qquad
      + (\abs{\partial_t F}^2 
      + \abs{DF}^2)^{\2}\bigr),
      \end{split}  
    \end{align}
  \end{subequations}
  where we set $V := F - \bar{F}$.
\end{lem}
\begin{proof}
  The inequalities follow from a straight-forward calculation using
  the fact 
  \begin{gather*}
    u_1 \cdots u_n - v_1  \cdots  v_n = \tsum_j u_1 \cdots u_{j-1}
    (u_j - v_j) v_{j+1} \cdots v_n. \qedhere
  \end{gather*}
\end{proof}
These generic estimates will be used to derive a bound for the RHS $f_{\mathrm{a}}$
and Lipschitz estimates  for the coefficients and the RHS of the asymptotic
equation.
Recall the following notation introduced in section \ref{sec:qlinear} 
$$
E_r(u - v)
= \norm{u_0 - v_0}_{r+1} + \norm{u_1 - v_1}_r
\text{ for }(u_0, u_1), (v_0, v_1) \in H^{r+1} \times H^r.
$$
Sobolev norm estimates for the coefficients will be obtained in a 
way similar to that for pointwise estimates. We begin with estimates for the 
matrix $g_{\mu\nu}^{\mathrm{a}}$
and then turn to the inverse $g^{\mu\nu}_{\mathrm{a}}$.
\begin{lem}
  \label{lem:Hs_metric}
  Let $(\varphi_0, \varphi_1), (\psi_0, \psi_1) \in W$. 
  Then the following estimates hold
  \begin{subequations}
    \begin{align}
      \label{eq:Hs_full}
    & \bnorm{\bigl(g_{\mu\nu}^{\mathrm{a}}(\varphi_0, \varphi_1)\bigr)}_{e,s,\mathrm{ul}}  
    \le K_{0,s}^2 + K_{1,s}^2 =: \tilde{K} 
    \\
    \label{eq:Hs_lip_full}
     & \bnorm{\bigl(g_{\mu\nu}^{\mathrm{a}}(\varphi_0, \varphi_1)\bigr) 
      - \bigl(g_{\mu\nu}^{\mathrm{a}}(\psi_0, \psi_1)\bigr)}_{e,s-1, \mathrm{ul}}  
    \le \tilde{\theta}
    E_s\bigl( (\varphi_0, \varphi_1)- (\psi_0, \psi_1)\bigr)
    \\
    \label{eq:0_lip_full}
    & \bnorm{\bigl(g_{\mu\nu}^{\mathrm{a}}(\varphi_0, \varphi_1)\bigr) 
      - \bigl(g_{\mu\nu}^{\mathrm{a}}(\psi_0, \psi_1)\bigr)}_{e,0, \mathrm{ul}}  \le 
    \tilde{\theta}' E_1\bigl( (\varphi_0, \varphi_1)- (\psi_0, \psi_1)\bigr)
  \end{align}
  \end{subequations}
  where $\tilde{\theta}  = 2^{3/2}\,  (K_{0,s} + K_{1,s})$
   and 
   $\,\tilde{\theta}'  = 2^{3/2}\,  (K_0 + K_1)$. 
\end{lem}
\begin{proof}
  Replacing the pointwise norms for $Y$ and $X$ in the proof of Lemma 
  \ref{lem:norm_metric}
  by the norm $\norm{\,.\,}_{s,\mathrm{ul}}$ for $D\varphi_0$ and $\varphi_1$
  yields the first bound.

  To obtain the Lipschitz estimates we consider the generic estimate
  \eqref{eq:gen_coeff_lip} replacing 
  \begin{gather}
    \label{eq:replace}
    \partial_t F \mapsto \varphi_1,~ DF \mapsto D\varphi_0, ~
    \partial_t \bar{F} \mapsto 
    \psi_1  \quad\text{and} \quad D\bar{F}
    \mapsto D\psi_0.
  \end{gather}
  The last estimate follows then by using the $L^{\infty}$-norm and the
  second estimate follows by using the norm $\norm{\,.\,}_{s,\mathrm{ul}}$.
\end{proof}
The preceding lemma provides us with estimates for the coefficients
of the asymptotic equation \eqref{eq:geom_solve_asym} meeting the conditions 
of the existence Theorem \ref{qlin_ex}. These will be stated in the next
lemma.
\begin{lem}
  \label{lem:Hs_metric_est}
  Suppose $(\varphi_0, \varphi_1), (\psi_0, \psi_1) \in W$. Then we have
  \begin{subequations}
    \begin{align}
    & \bnorm{\bigl(g^{\mu\nu}_{\mathrm{a}}(\varphi_0, \varphi_1)\bigr)}_{e,s,\mathrm{ul}} \le 
    c \Delta^{-1}\bigl(1 + (\Delta^{-1} \tilde{K})^s\bigr) &&=:K\\
    \label{eq:coeff_Hs_lip}
    &\bnorm{\bigl(g^{\mu\nu}_{\mathrm{a}}(\varphi_0, \varphi_1)\bigr) 
      - \bigl(g^{\mu\nu}_{\mathrm{a}}(\psi_0, \psi_1)\bigr)}_{e,s-1, \mathrm{ul}}  \le \theta
    E_s\bigl( (\varphi_0, \varphi_1)- (\psi_0, \psi_1)\bigr)\\
    \label{eq:coeff_0_lip}
    & \bnorm{\bigl(g^{\mu\nu}_{\mathrm{a}}(\varphi_0, \varphi_1)\bigr) 
      - \bigl(g^{\mu\nu}_{\mathrm{a}}(\psi_0, \psi_1)\bigr)}_{e,0, \mathrm{ul}}  \le 
    \theta' E_1\bigl( (\varphi_0, \varphi_1)- (\psi_0, \psi_1)\bigr)
  \end{align}
  \end{subequations}
  where the Lipschitz constants are given by
  $
    \theta  :=  K^2 \tilde{\theta} \text{ and }
    \theta'  := \Delta^{-2} \tilde{\theta}' 
    $.
\end{lem}
\begin{proof}
  The first inequality follows directly from Lemma \ref{inverse} and the
  estimates \eqref{eq:est_inverse_full} and \eqref{eq:Hs_full}. 
  The other inequalities follow from the observation
  \begin{gather*}
    A^{-1} - B^{-1} = A^{-1}(B - A)B^{-1}
  \end{gather*}
  for matrices $A$ and $B$ and the Lipschitz estimates \eqref{eq:Hs_lip_full} 
  and \eqref{eq:0_lip_full}   for $g_{\mu\nu}^{\mathrm{a}}$.
\end{proof}
We now head to the RHS $f_{\mathrm{a}}$ of the asymptotic equation 
\eqref{eq:geom_solve_asym}. To obtain the desired estimates
we have to control the Christoffel symbols of the cut-off 
background metric $\hat{a}$.
To control derivatives of the Christoffel symbols up to order $s$ we need
bounds for derivatives of the initial values 
$\init{\Phi}$ and $\init{\chi}$ up to order $s + 2$ and $s + 1$ resp.
These were stated in \eqref{eq:sob_initial}.

A Lipschitz estimate for the RHS $f_1$ in $L^2$ will use local bounds for
the Christoffel symbols $\hat{\gamma}_{\mu\nu}^{\lambda}$.
To this end we need local estimates
for $D^2 \init{\Phi}$ and $D \init{\chi}$
 which will be derived in the next lemma. 
\begin{lem}
  The following inequalities hold
  \begin{gather}
    \label{eq:initial_d2_0}
    \abs{D^2 \init{\Phi}}_e  \le 
      \tilde{C}^{\varphi}_2 \tfrac{1}{\rho_1}\bigl(
      1 + \tilde{C}_1 + \tilde{C}_2
  \bigr) \quad \text{and}\quad
  \abs{D\init{\chi}}_e \le \tilde{C}_1^{\chi}(1 + \tilde{C}_1). 
  \end{gather}
\end{lem}
\begin{proof}
  The result follows from
  \begin{align*}
    \abs{D^2 \init{\Phi}}_e & \le \zeta \abs{D^2 \Phi_{\alpha}}_e 
    + \abs{D \zeta}_e \abs{D\Phi_{\alpha} - D\Phi_{\alpha}(0)}_e
    + \abs{D^2 \zeta}_e \abs{\Phi_{\alpha} 
      - x^{\ell} \partial_{\ell} \Phi_{\alpha}(0)}_e.
  \end{align*}
  using the bounds for the derivatives of the cut-off
  function and the bounds for $\init{\Phi}$ established in the 
  proof of Lemma \ref{lem:initial_est}.
  A similar device gives us the second inequality.
\end{proof}
Using this preparation we are able to estimate the 
Christoffel symbols
of $\hat{a}$. 
\begin{prop}
  \label{prop:chr_sym}
  The following inequalities  hold 
  \begin{gather*}
    \norm{(\hat{\gamma}_{\mu\nu}^{\lambda})}_{e,s}  \le C^{\hat{\gamma}}  
    \quad
    \text{and}
    \quad
    \abs{(\hat{\gamma}_{\mu\nu}^{\lambda})}_e  \le C^{\hat{\gamma}}_0,
  \end{gather*}
  where 
  \begin{multline*}
    C^{\hat{\gamma}}   =   2 K \bigl(
    \tilde{D}_1  \bigl(2 ((\tilde{C}_0^{\chi})^2 + 
  D_1^2)^{\2} + 4 (C_{w_0}^2 + D_0^2)^{\2} \bigr) 
  \\
  {}+
  \tilde{D}_0 \bigl(4 ((\tilde{C}_0^{\chi})^2 + 
  D_1^2)^{\2} + 2(C_{w_0}^2 + D_0^2)^{\2} \bigr)
   \bigr).
  \end{multline*}
  and $C^{\hat{\gamma}}_0$ arises from $C^{\hat{\gamma}}$ by applying the 
  replacements
  \begin{gather*}
    K \mapsto \Delta^{-1}, ~ D_0 \mapsto \delta_1/2, ~ D_1 \mapsto
    \delta_2/2
  \end{gather*}
  and $\tilde{D}_0, ~ \tilde{D}_1$ are respectively replaced by
  the bounds for $\abs{D\init{\chi}}_e$  and $\abs{D^2 \init{\Phi}}_e$ stated in
  \eqref{eq:initial_d2_0}.
\end{prop}

\begin{proof}
  Taking a generic norm for $\hat{\gamma}_{\mu\nu}^{\lambda}$ we see that
  it is necessary to estimate $\hat{a}^{\mu\nu}$ and $D\hat{a}_{\mu\nu}$.
  Observe that Lemma \ref{lem:Hs_metric_est_hyp} also
  applies to the matrix $\hat{a}^{\mu\nu}$
  so that the  $H^s$-norm of $\hat{a}^{\mu\nu}$ is bounded by $K$.
  To obtain a bound for derivatives of $\hat{a}_{\mu\nu}$
  we compute the following representative part of the derivative of $\hat{a}$
  \begin{gather*}
    \partial_k \hat{a}_{0\ell} = 
    \partial_k \init{\chi}^A \eta_{AB} (\partial_{\ell} w_0 
    + \partial_{\ell} \init{\Phi})^B
    + (\chi_0 + \init{\chi})^A \eta_{AB} \partial_k \partial_{\ell} \init{\Phi}^B.
  \end{gather*}
  We see that the highest-order terms can be estimated by the bounds 
  $\tilde{D}_0$
  and $\tilde{D}_1$ for  the $H^s$-norm of
  $D^2\init{\Phi}$ and $D\init{\chi}$, respectively.
  
  Observe that the choice of $\theta$
  ensures that $\abs{D\init{\Phi}}_e < \delta_1/2$ 
  and $\abs{\init{\chi}}_e < \delta_2/2$.
  By the same argument as for the first inequality using a pointwise norm we 
  infer the second inequality from the local bounds 
  for $g^{\mu\nu}$ stated in \eqref{eq:est_inverse_full} and the bounds
  for the initial values 
  stated in \eqref{eq:est_initial} and \eqref{eq:initial_d2_0}.
\end{proof}
The preceding considerations yield estimates for the RHS of the asymptotic
equation \eqref{eq:geom_solve_asym} meeting the conditions \eqref{cond_qlin1},
\eqref{cond_qlin_add1} and  \eqref{cond_qlin_add2}
of the existence Theorem \ref{qlin_ex}.
\begin{lem}
  \label{lem:Hs_rhs_est}
  Let $(\varphi_0, \varphi_1), (\psi_0, \psi_1) \in W$. Then 
  the following estimates hold
  \begin{align*}
    \norm{f_{\mathrm{a}}(t,\varphi_0, \varphi_1)}_{s} & \le K_f,\\
    \norm{f_{\mathrm{a}}(t,\varphi_0, \varphi_1) 
      - f_{\mathrm{a}}(t,\psi_0, \psi_1)}_{s-1} & 
    \le \theta_f' E_s\bigl( (\varphi_0, \varphi_1)- (\psi_0, \psi_1)\bigr)  \\
    \norm{f_{\mathrm{a}}(t,\varphi_0, \varphi_1) 
      - f_{\mathrm{a}}(t,\psi_0, \psi_1)}_{L^2} & 
    \le \theta_f E_1\bigl( (\varphi_0, \varphi_1)- (\psi_0, \psi_1)\bigr)
  \end{align*}
  with $K_f  =  K C^{\hat{\gamma}}(K_{0,s}^2 + K_{1,s}^2)^{\2},~
    \theta_f'  =  \theta C^{\hat{\gamma}} (K_{0,s}^2 + K_{1,s}^2)^{\2}
     + K C^{\hat{\gamma}}$
  and $\theta_f$ arises from $\theta_f'$ by applying the replacements
  \begin{gather*}
    \theta \mapsto \theta', ~ K \mapsto \Delta^{-1}, ~
    C^{\hat{\gamma}} \mapsto C_0^{\hat{\gamma}},~
    (K_{0,s}, K_{1,s}) \mapsto (K_0,  K_1).
  \end{gather*}
\end{lem}
\begin{proof}
  The first statement can be derived using the generic estimate 
  \eqref{eq:gen_rhs_0} and  proposition
  \ref{prop:chr_sym}. 

  The Lipschitz estimates can be derived in a similar way to those
  for the coefficients stated in Lemma \ref{lem:Hs_metric} 
  regarding the generic estimate 
  \eqref{eq:gen_rhs_lip}.
\end{proof}
To obtain uniqueness of solutions to the IVP \eqref{eq:geom_solve} we need to 
show
the conditions \eqref{eq:loc_cond_uni} for coefficients and RHS of the 
asymptotic equation \eqref{eq:geom_solve_asym} are such that the uniqueness 
Theorem  \ref{loc_uni} is  applicable.
\begin{lem}
  \label{lem:uni_satisfy}
  The coefficients $g^{\mu\nu}_{\mathrm{a}}$ and the RHS $f_{\mathrm{a}}$ of equation
  \eqref{eq:geom_solve_asym} satisfy the conditions \eqref{eq:loc_cond_uni}.
\end{lem}
\begin{proof}
  The claim follows from the observation that the Lipschitz estimates
  for the $L^2$-norm of the RHS and the coefficients do in fact involve
  local constants. Therefore, the desired constants are given by
  $\theta_f'$ from Lemma \ref{lem:Hs_rhs_est} and $\theta'$ from Lemma 
  \ref{lem:Hs_metric_est}.
\end{proof}
\begin{rem}
  To make this argument possible we derived a local estimate for 
  the Christoffel symbols $\hat{\gamma}_{\mu\nu}^{\lambda}$ in proposition
  \ref{prop:chr_sym}.
\end{rem}
We will now give a proof of the main result of this section.
The steps existence and uniqueness will be done separately, where the
proof of existence reduces to a proof of Proposition \ref{prop:cutoff_ex}.
\begin{proof}[\textbf{Proof of Proposition \ref{prop:cutoff_ex}}]
  The Lemmata \ref{lem:norm_metric_inverse}, \ref{lem:Hs_metric_est}
  and \ref{lem:Hs_rhs_est} show that the assumptions of the asymptotic
  existence Theorem \ref{asym_ex} are satisfied. 
  Therefore, we infer the existence of a solution $F$ to the IVP
  \eqref{eq:geom_solve} with the desired properties provided 
  by \eqref{eq:der_sol_asym}. 

  Lemma \ref{lem:uni_satisfy} provides us with the conditions needed to
  apply the local uniqueness Theorem \ref{loc_uni} to the asymptotic equation
  \eqref{eq:geom_solve_asym}. Uniqueness of solutions 
  therefore follows by considering balls of arbitrary radius
  about each point $\in \rr^m$.
\end{proof}
\begin{rem}
  \label{rem:stay_W}
  The solution is given by $F(t) = w(t) + \psi(t)$,
  where $\psi(t)$ solves equation \eqref{eq:geom_solve_asym} with initial 
  values given by \eqref{eq:asym_initial}.
  From the existence Theorem \ref{qlin_ex} applied to the asymptotic equation
  we derive that
  \begin{gather*}
    (DF(t) - Dw_0, \partial_t F(t) - w_1)
    \in W \subset H^{s+1} \times H^s.
  \end{gather*}
  Therefore, all estimates on the coefficients remain valid for $t > 0$.
\end{rem}
\begin{rem}
  \label{rem:geom_lower}
  A lower bound for the existence time $T'$ is given in
  Remark \ref{rem:special_lower} with the constants defined in remark 
  \ref{rem:gen_lower}. The constant $c_E$ plays
  an important role controlling the inverse of the existence time. 
  For convenience we state it in terms of 
  quantities occurring in our estimates.
  It is given by
  \begin{multline*}
    c_E %
    =   2\bigl(2(C_{w_0}^2 + \delta_1^2)
    \bigl(1  + L_2^{-1}(1 - r_0)^{-1}
    \bigl((\tilde{C}_0^{\chi})^2 + \delta_2^2\bigr) 
    \bigr)
    \\
    {}+ \bigl((\tilde{C}_0^{\chi})^2 + \delta_2^2\bigr) 
    \bigl( 1+ m^{1/2} \omega_1^{2}(1 - R_0)^{-1}
    (C_{w_0}^2 + \delta_1^2) \bigr)
    \bigr).
  \end{multline*}
\end{rem}

\begin{rem}
  \label{rem:extend}
  From Theorem \ref{qlin_ex} we only obtain a solution for $t > 0$, but
  hyperbolic equations are invariant under time reversal, so that
  for $t < 0$ the function
  $F(-t,z)$ is a solution to the equation in \eqref{eq:geom_solve} if the 
  ingredients are
  defined for negative values of $t$. The definition \eqref{eq:back_metric}
  of the background metric $\hat{g}$ ensures that the RHS is
  defined for all values of $t$.  This also shows the  
  uniqueness claim by considering double-cones, i.e. forward
  and backward w.r.t.\ the time parameter.
\end{rem}

\begin{rem}
  \label{rem:proof_mem_atlas_sol_diff}
  From the assumptions on the initial values we get that the 
  cut-off background metric $\hat{a}_{\mu\nu}$ defined by
  \eqref{eq:back_comp} is $C_c^{r+1}$.
  Therefore the Christoffel symbols of $\hat{a}$ 
  occurring in the RHS of the asymptotic equation defined by
  \eqref{eq:def_rhs_asym} are $C_c^{r}$ with $r = s + \ell_0$.
  This shows that derivatives of the RHS up to order $\ell_0$
  are in $C_c^s$ with $s > \tfrac{m}{2} + 1$ as it is necessary
  to obtain the estimates for the RHS in Lemma \ref{lem:Hs_rhs_est}.
  Since the coefficients and the RHS in fact do not depend
  on the time parameter, remark
  \ref{rem:mem_atlas_sol_diff} follows from corollary
  \ref{cor:sol_diff}. 
\end{rem}

\begin{proof}[\textbf{Proof of the uniqueness claim in Theorem
    \ref{thm:ex_uni_atlas}}]   The strategy will be to compare an
    arbitrary $C^2$-solution $\bar{F}_{\lambda}$, defined in a chart
    $x_{\lambda}$, with the solution constructed in proposition
    \ref{prop:cutoff_ex} denoted by $F_{\lambda}$.  
    Since the coefficients of the asymptotic equation 
    \eqref{eq:geom_solve_asym} equal the coefficients of the equation solved
    by $F_{\lambda}$ defined on a different domain, Lemma 
    \ref{lem:uni_satisfy} yields the local constants 
    desired by the uniqueness Theorem \ref{loc_uni}.
    The proof of the theorem provides us with an expression for the slope
    $c_0$ of the cone on which uniqueness holds by choosing $u_1 = F_{\lambda}$
    in the difference
    equation \eqref{eq:diff_eq}. It is given by the following consideration
    \begin{multline}
      \label{eq:est_c_0}
      1 + \bigl( 2 \babs{\bigl(g^{0j}(F_{\lambda})\bigr)}_e + 
      \babs{\bigl(g^{ij}(F_{\lambda})\bigr)}_e \bigr)
      \tfrac{1}{\lambda}
      \\
      \le  1 + m^{\2} \tilde{\mu}^{-1}K_1^2
      \bigl(2 \tilde{\lambda}^{-1} K_0 K_1
       + 1
      \bigr) (1 + m^{1/2} \tilde{\mu}^{-1} K_0^2) =: c_0,
    \end{multline}
    where we used estimate \eqref{eq:slope} and Lemma 
    \ref{lem:norm_metric_inverse}.
\end{proof}

\subsubsection{Properties of a solution}
\label{sec:sol_property}
In this section we will develop some estimates for the family of solutions
to the membrane equation
obtained in Theorem \ref{thm:ex_uni_atlas}. 
The goal is to show that for a small enough
time parameter and radius the solutions are in fact embeddings and therefore
represent a submanifold.

Let $\lambda \in \Lambda$ be fixed.
We begin with an estimate for the second fundamental form of a solution
$F := F_{\lambda}$ obtained by Theorem \ref{thm:ex_uni_atlas}.  
To measure the second fundamental form
a Riemannian metric on the pre-image of $F$ is necessary.
A natural choice is the pullback of the Riemannian metric on the ambient
manifold, so let $\hat{e} := F^{\ast} E$.
To gain advantage from the estimates following from the construction of
the solution a comparison between $\hat{e}$ and the Euclidean metric
w.r.t.\ coordinates will be established in the next lemma.
\begin{lem}
  \label{lem:hat_e_est} 
  The matrix $\hat{e}_{\mu\nu}$ satisfies
  \begin{gather*}
    \mu_e \delta_{\mu\nu} \le \hat{e}_{\mu\nu} \le M_e \delta_{\mu\nu}
    \\
    \begin{split}
    \text{with}  \quad\mu_e & = (1 - \delta_0/2)
    \min\bigl(L_2(1 - r_0), (1 + \delta_0/2)^{-1} 
  \omega_1^{-2}(1 - R_0)
    \bigr)\\
    \text{and} \quad M_e & = (K_0^2 + K_1^2) C_0^h.
    \end{split}
  \end{gather*}
\end{lem}
\begin{proof}
  Positive definiteness follows from the comparison of $E_{AB}$
  with the Euclidean metric stated in \eqref{eq:metric_E_est}
  and the property of $\partial_t F$
  and $DF$ to be timelike and spacelike, respectively, stated in
  \eqref{eq:est_part}.

  The bound for $\hat{e}_{\mu\nu}$ follows from the bounds
  for $\partial_t F$ and $DF$ stated in \eqref{eq:est_time_full} and 
  \eqref{eq:est_space_full},
  respectively, and the bound for $E_{AB}$, which coincides with
  the one for $h_{AB}$ (cf. \eqref{eq:h_norm}).
\end{proof}
This yields a bound for the second fundamental form of the mapping $F$
by using second-order estimates.
\begin{prop}
  The second fundamental form of a solution $F$ of the IVP \eqref{eq:geom_solve}
  obtained by Theorem \ref{thm:ex_uni_atlas} satisfies the following
  uniform inequality
  \begin{gather*}
    \abs{\IIfull}_{\hat{e}, E} \le C_{\IIfull}.
  \end{gather*}
\end{prop}

\begin{proof}
  We first consider $\abs{\II}_{e,e}$ with the representation \eqref{eq:2ndfform}
  for the second fundamental form of $F$. Estimates for second-order derivatives
  of the solution $F$ to the reduced membrane equation can be obtained
  via the definition of $W$ (cf. \eqref{eq:domain}) and the Sobolev embedding
  theorem. A bound for the 
  second-order
  time derivative follows by considering the equation and applying the
  bounds for the coefficients and the RHS. Bounds for first-order
  derivatives of $F$ follow from the definition of $\Omega$ in
  \eqref{eq:omega_def}. 
  Applying bounds for the induced metric $g_{\mu\nu}$,
  bounds for the metric $h_{AB}$ and its Christoffel symbols $\G_{BC}^A$
  established in \eqref{eq:est_inverse_full}, 
  \eqref{eq:h_norm} and \eqref{eq:chr_target}, respectively,
  gives an estimate for $\abs{\II}_{e,e}$.
  The desired bound follows from the comparison of $\hat{e}$
  and $E$ with the Euclidean metric stated in Lemma \ref{lem:hat_e_est} 
  and \eqref{eq:metric_E_est}, respectively.
\end{proof}

\begin{prop}
  \label{prop:embedded}
  Suppose $F$ is a solution of the IVP \eqref{eq:geom_solve}
  for the membrane equation
  obtained by Theorem \ref{thm:ex_uni_atlas}. 
  \\
  Then there exist constants $0 < \tilde{T} \le T'$ and $\theta'$
  depending on the bounds for second-order derivatives and
  the estimates ensuring the timlikeness of $\partial_t F$
  and the spacelikeness of $DF$ stated in \eqref{eq:est_part} such that 
  $F:[-\tilde{T}, \tilde{T}] \times B^e_{\theta' \rho_1/2}(0)$ is an
  embedding.
\end{prop}

\begin{proof}
  The claim follows from the inverse function theorem. An inspection of
  the proof to be found in \cite{Spivak:1965} shows that the region,
  in which the function is invertible depends on a bound for the inverse
  of the differential at the origin and a bound for the second-order
  derivatives.
\end{proof}

\subsection{Gluing local solutions}
\label{sec:prop_graph_mink}
The goal of this section is to obtain existence and uniqueness
of  solutions to the Cauchy problem \eqref{eq:param_ivp} for fixed
initial immersion, direction, lapse and shift. From the uniformity
assumptions imposed
on the initial data it will follow that the solution obtained as an immersion
has the domain $[-T,T] \times M$ with a fixed constant $T > 0$.

Recall the notations of the main problem \eqref{eq:geom_problem}.
Assume $M$ to be an $m$-dimensional manifold and $\varphi: M \rightarrow 
\rr^{n,1}$
to be an immersion of $\Sigma_0$. 
As in section \ref{sec:sol_atlas} we use a timelike vector field 
$\chi: M \rightarrow
T\rr^{n,1}$ along $\varphi$ as initial velocity for the Cauchy problem 
\eqref{eq:param_ivp}. Recall that $\tau_0$ denotes the timelike direction of 
the   Minkowski space and $\hn$ denotes the connection on the pullback bundle 
$\varphi^{\ast}   T\rr^{n,1}$.

Let $s > \tfrac{m}{2} + 1$ be an 
integer.
We make the following uniform assumptions on the initial data 
$\varphi$ and $\chi$.
  \begin{assum}
    \label{assumptions_mink}
     \begin{subequations}
    \begin{gather}
      \label{eq:assum_smf}
      \begin{split}
        &\text{There exist constants $\omega_1, C^{\varphi}_0, \dots,
          C^{\varphi}_s$
      such that}
    \\
      &\inf\{ -\eta( \gamma , \tau_0 ) : \gamma \text{ timelike future-directed
        unit normal to }\Sigma_0 \}\le \omega_1
      \\
      \text{and}   \quad&\abs{\hn^{\ell} \II}_{\ig,e} \le C_{\ell}^{\varphi}
      \text{ for }0 \le \ell \le s.
      \end{split}
    \end{gather}
      \begin{gather}
      \label{eq:assum_speed}
      \begin{split}
       & \text{There exist constants $L_1, L_2, L_3, C^{\chi}_1, \dots, C_{s+1}^{\chi}$
      such that}
    \\
        &- L_1 \le \eta(\chi, \chi ) \le - L_2,~ 
      - \eta\bigl( \tfrac{\chi}{\mabs{\chi}}, \tau_0 \bigr) \le L_3
      \text{,
      where }\mabs{\chi}^2 = - \eta(\chi,\chi)
      \\
      \text{and}   \quad&\abs{\hn^{\ell} \chi}_{\ig,e} \le C_{\ell}^{\chi}
      \text{ for }1 \le \ell \le s+1.
      \end{split}
    \end{gather}
  \end{subequations}
  \end{assum}
  
  \begin{rem}
    Suppose the initial submanifold $\Sigma_0$ to be compact. Then the 
    assumptions are 
    satisfied since we can pick a finite covering of the
    submanifold and choose the largest constants occurring in the 
    finitely many subsets.
  \end{rem}
  The next theorem states the main result of this section providing
  a solution to the IVP \eqref{eq:param_ivp} in Minkowski space.
\begin{thm}
  \label{thm:ex_uni_geom}
  Suppose the initial data  $\varphi \in C^{s+2}$ and $\chi \in C^{s+1}$
  satisfy 
  the assumptions \ref{assumptions_mink}.

  Then
  there exist a constant $T>0$ and a $C^2$-solution $F:[-T,T]\times M
  \rightarrow \rr^{n,1}$ of the membrane equation 
  \begin{gather}
    \label{eq:mem_geom_mink}
      g^{\mu\nu} \partial_{\mu} \partial_{\nu} F^A - 
  \Gamma^{\lambda} \partial_{\lambda} F^A 
  = 0
  \end{gather}
  in harmonic map gauge w.r.t.\ the background metric %
  defined by 
  \eqref{eq:back_metric} attaining the initial values
  \begin{gather}
    \label{eq:mem_initial}
    \restr{F} = \varphi\qquad \text{and}\qquad \restr{\dt F} = \chi.
  \end{gather}

  Let $F$ and $\bar{F}$ be two such solutions. Then there exists a constant
  $T_0$ such that $F$ and $\bar{F}$ coincide 
  for $- \min(T, T_0) \le t \le \min(T, T_0)$.
\end{thm}
\begin{rem}
  \begin{enumerate}
  \item The precise value of the constant $T_0$ will be given in 
  \eqref{eq:uni_cone_height}.
\item The choice of the background metric to be defined by the initial
  values \eqref{eq:mem_initial}
  is not optimal in the sense that we need bounds for 
  derivatives of one order higher for the initial values to use
  the existence Theorem \ref{qlin_ex} for hyperbolic equations.
  \end{enumerate}
\end{rem}
\begin{rem}
  \label{rem:mem_sol_diff}
  From Remark \ref{rem:mem_atlas_sol_diff} we obtain that, 
  providing $s > \tfrac{m}{2} + 1 + \ell_0$ with an integer $\ell_0$,
  the solution $F:[-T,T] \times M \rightarrow \rr^{n,1}$ 
  is in fact of class $C^{2 + \ell_0}$.
\end{rem}
\begin{rem}
  \label{rem:mink_non_uniform}
  The theorem applies to the situation, where the assumptions
  \ref{assumptions_mink} are only valid in a neighborhood of the initial 
  submanifold
  $\Sigma_0$.
\end{rem}
To obtain the existence claim we will show that the special graph 
representation
derived in section \ref{par:graph_repr} satisfies the assumptions
\ref{assumptions_atlas} and therefore Theorem \ref{thm:ex_uni_atlas}
gives a family of solutions in a decomposition consisting of special graph 
representations. The desired estimates
will be derived in the following proposition. 
\begin{prop}
  \label{prop:graph_assum}
  Let the initial data $\varphi\in C^{s+2}$ and $\chi\in C^{s+1}$ satisfy the 
  assumptions
  \ref{assumptions_mink}. 
  Let $p \in M$ and suppose $x,y$ are coordinates for $M$ and $\rr^{n,1}$,
  respectively,
  such that $y \circ \varphi \circ x^{-1}$ is the special
  graph representation \eqref{eq:graph_repr} obtained in section 
  \ref{par:graph_repr}.
  
  Then the coordinates $x,y$ and the representations $\Phi$ and $\chi_{xy}$ of 
  $\varphi$ and $\chi$ w.r.t.\ these charts satisfy 
  the conditions \ref{assumptions_atlas}.
\end{prop}
\begin{proof}
  The proof is divided in three steps. In the first step we will show
  the 
  parts \ref{assum_atlas_chart} and \ref{assum_atlas_graph} of
  the 
  assumptions concerning the chart $x$ and the representation $\Phi$.
  In the last two steps we will derive the assumptions concerning
  the representation of the initial velocity.
  \begin{enumerate}
  \item     The first condition on the chart $x$ follows from Lemma 
  \ref{lem:graph_est} yielding that a Euclidean ball with radius $\rho_1$ 
  about $0$ 
  is contained in the image
  of $x$. The positive definiteness of the induced metric
  at the origin can be obtained from Lemma \ref{lem:cond_normal}
  by using that $\varphi$ is an immersion of a uniformly spacelike
  submanifold with bounded curvature defined in \ref{defn:spacelike}.
  
  The first condition of \eqref{eq:assum_bd_pt}
  for derivatives of the representation $\Phi$
  can be derived from $\abs{Du^0}_e < 1$ and $Du^a(0) = 0$ for all $a$, which 
  gives
  us $\abs{D\Phi(0)}_e^2 < 1 + m =: C_{w_0}^2$.
  The second condition was established in Lemma \ref{lem:graph_est}
  and bounds for higher derivatives of
  $\Phi$ were derived in Lemma \ref{lem:u_high_der}.
\item The desired inequality \eqref{eq:assum_bd_vel} for 
  $\eta(\chi_{xy}, \chi_{xy})$ is satisfied by assumption.
To estimate $\abs{\chi_{xy}}_e$ we compute 
\begin{gather*}
  - (\chi_{xy}^0)^2 + \tsum_{\si{a}} (\chi_{xy}^{\si{a}})^2 \le - L_2
\quad\Longrightarrow
  \quad \tsum_{\si{a}} (\chi_{xy}^{\si{a}})^2 \le - L_2 + (\chi_{xy}^0)^2.
\end{gather*}
The assumption on  the angle yields
\begin{gather*}
  - \llangle \chi / \mabs{\chi}, \tau_0 \rrangle \le L_3 \quad\Longrightarrow
  \quad
 \chi_{xy}^0 \le \mabs{\chi} L_3 \le L_1^{1/2} L_3.
\end{gather*}
Since $\chi$ is supposed
to be future directed, we have 
$\llangle \chi_{xy}, \tau_0\rrangle
= - \chi_{xy}^0 < 0$.
Hence 
\begin{gather}
  \label{eq:est_chi0}
  \abs{\chi_{xy}}_e^2 = (\chi_{xy}^0)^2 + \tsum_{\si{a}} (\chi_{xy}^{\si{a}})^2
\le - L_2 +  2 (\chi_{xy}^0)^2 
\le - L_2 + 2 L_1 L_3^2 := (\tilde{C}^{\chi}_0)^2.
\end{gather}
\item \label{eq:high_vel}  
  Bounds for higher derivatives of the representation $\chi_{xy}$
  can be obtained by a similar device to the proof for higher derivatives of 
  the graph
  functions stated in Lemma \ref{lem:u_high_der}.
  We start  with the identity
  \begin{gather*}
    \partial^{k} \chi_{xy} = \hn^{k} \chi_{xy}
    + \tsum \partial^{\alpha_1} \iG \ast \cdots 
    \ast\partial^{\alpha_p} \iG
    \ast \hn^{\ell} \chi_{xy}
  \end{gather*}
  for the components of the initial velocity $\chi_{xy}^B$ which follows
  analogously to equation \eqref{eq:high_der_cov} for the second fundamental 
  form. The sum ranges over a 
  certain subset of
  all tuples $(p,\alpha_i, \ell)$
  such that $p + \tsum \alpha_i + \ell
  = k+1$. We infer from this expression that coordinate derivatives of
  the initial velocity are bounded, if the covariant derivatives are
  bounded. The Christoffel symbols can be estimated as in the proof
  of Lemma \ref{lem:u_high_der}. Observe that also the full range
  of bounded derivatives of the second fundamental form is needed.
  \qedhere
  \end{enumerate}
\end{proof}
\begin{rem}
  The preceding proposition shows that any decomposition consisting of special
  graph representations of $\varphi$ satisfies the conditions of
  Theorem \ref{thm:ex_uni_atlas}.
\end{rem}
To obtain a solution of the form $F:[-T,T] \times M \rightarrow \rr^{n,1}$, we
need to glue solutions obtained in different charts of a
decomposition.
This can be done if the family of solutions obtained by Theorem
\ref{thm:ex_uni_atlas} coincide on common domains of coordinates on $M$,
which will be shown in the  sequel.

Let $(U, x, V, y)$ and $(\hat{U}, \hat{x}, \hat{V}, \hat{y})$
be two decompositions of $\varphi$ such that the representations of the initial
data and the decompositions satisfy the assumptions \ref{assumptions_atlas}
possibly with different constants. Let $\Phi, ~ \chi_{xy}$ and $\hat{\Phi},~
\chi_{\hat{x}, \hat{y}}$ denote the representations of the initial data $\varphi$ and
$\chi$ w.r.t.\ the decompositions $(U, x, V, y)$ and 
$(\hat{U}, \hat{x}, \hat{V}, \hat{y})$ resp.
Assume the images of the coordinates $x$ and $\hat{x}$ to contain a Euclidean 
ball of radius $\rho_1$ and $\hat{\rho}_1$ resp.
Let $F_0$ and $\hat{F}$ be the solutions to the IVP \eqref{eq:geom_solve} 
obtained by Proposition \ref{prop:cutoff_ex}. Assume the function $F_0$
to be defined on a ball $B^e_{\theta \rho_1}(0)$ and analogously $\hat{F}$
to be defined on $B^e_{\hat{\theta}\hat{\rho}_1}(0)$. We want to compare the two 
solutions, so suppose $U \cap \hat{U} \neq \varnothing$.

To derive a quantitative description of an uniqueness result a further
condition on the coordinates $x$ and $\hat{x}$ are needed.
Let the following assumption be satisfied
\begin{gather}
  \label{cond_next}
  \parbox{0.7\textwidth}{There exist constants $G_{1}, G_{2}$ 
      such that
    the components of the induced metric w.r.t.\ $x$ satisfy 
  $G_{1} \delta_{ij} \le \ig_{ij} \le G_{2}
  \delta_{ij}$.}
\end{gather}
Suppose the components of $\ig$ w.r.t.\ the coordinates $\hat{x}$
satisfy the same assumption with the constants $G_1, G_2$ replaced
by $ \hat{G}_1, \hat{G}_2$, respectively.
Denote the change 
of coordinates
by $\hat{u}$ and $u$ such that $\hat{x} = u \circ x$
and $\hat{y} = \hat{u} \circ y$. 
Set
\begin{equation}
  \label{eq:transfer_solution}
   G(t,z) = \hat{u}^{-1} \circ \hat{F}\bigl(t, u(z)\bigr).
\end{equation}
Then $G$ is defined on the image of the chart $x$ and can by compared
with $F_0$.
We have to show that $G$ is a solution to the reduced membrane equation
\eqref{eq:mem_red}
and that the initial values of $G$ and $F_0$ coincide wherever $G$ is defined.
This will be done in the next lemma.
\begin{lem}
  \label{lem:initial_G}
  The function $G$ defined in \eqref{eq:transfer_solution} satisfies the reduced
  membrane equation
  \eqref{eq:mem_red} with the background metric defined by 
  \eqref{eq:back_metric} and the initial values of $F_0$ and $G$ coincide.
\end{lem}
\begin{proof}
 The membrane equation \eqref{eq:membrane} is coordinate independent as well as
 is
 the gauge condition \eqref{eq:harm_cond_coord}. Therefore both equations 
 remain valid after a change of coordinates. 

  We compute the initial values for $G$ arriving at
  \begin{gather*}
    G(0,z) = \hat{u}^{-1} \circ \hat{F}\bigl(0,u(z)\bigr) = \hat{u}^{-1}\circ
    \hat{y} \circ\varphi\bigl(\hat{x}^{-1}\circ u(z)\bigr) = 
    y \circ\varphi\bigl(x^{-1}(z)\bigr) = \Phi(z)
    \\
    \text{and}\quad
       d\hat{u}^{-1} \bigl(
  \partial_t  \hat{F} \bigl(0,u(z)\bigr) \bigr)
  = d(\hat{u}^{-1} \circ \hat{y})\bigl(\chi(\hat{x}^{-1}\circ u(z))\bigr)
  = \chi_{xy}(z). \qedhere
\end{gather*}
\end{proof}
The next proposition will show equality of $F_0$ and $G$ providing a 
quantitative description of the cone on which equality holds.
\begin{prop}%
  \label{prop:uni_coord}
    Let $q \in M$ such that 
  \begin{gather*}
    B_r^{\ig}(q) \subset B_{G_{1}^{1/2}  \theta\rho_1  /2}^{\ig}\bigl(x^{-1}(0)\bigr)
    \cap B_{\hat{G}_{1}^{1/2}  \hat{\theta}\hat{\rho}_1  /2}^{\ig}\bigl(\hat{x}^{-1}(0)\bigr).
  \end{gather*}
  Then $F_{0} = G$ within the cone  
  \begin{gather*}
    C = \bigl\{ (t,z) : \abs{z - x(q)}_e \le 
    - c_0 t + G_{2}^{-1/2}r\bigr\}
    \subset \rr^{1+m}.
  \end{gather*}
\end{prop}
\begin{proof}
  The additional condition \eqref{cond_next} for  both 
  decompositions provides
  us with
  a comparison between Euclidean balls and balls w.r.t.\ the metric $\ig$.
  it follows that
  \begin{gather*}
    B_{G_{2}^{-1/2}r}^e\bigl(\hat{x}(q)\bigr) \subset x\bigl(B_r^{\ig}(q)\bigr)
    \subset 
    B^e_{\theta\rho_1/2}\bigl(x^{-1}(0)\bigr),
  \end{gather*}
  where the Euclidean balls are taken w.r.t.\ the chart $x$
  and an analog result holds for the chart $\hat{x}$.
  We conclude that the transferred solution $G$ is at least defined on
  $B_{G_{2}^{-1/2}r}^e\bigl(x(q)\bigr)$ attaining the initial
  values $\Phi$ and $\chi_{xy}$ in that region.
  From the uniqueness claim of Theorem \ref{thm:ex_uni_atlas}
  it follows that $F_{0} = G$ on the desired cone with the slope
  $c_0$ defined by \eqref{eq:est_c_0}.
\end{proof}
\begin{rem}
  \label{rem:neg_t}
  We only considered positive values of $t$, but the same holds
  for negative values by applying a time reversal (cf. Remark \ref{rem:extend}).
\end{rem}
It it now possible to construct a solution to the IVP \eqref{eq:param_ivp} 
with a domain of the form $[-T,T] \times M$.
\begin{proof}[\textbf{Proof of the existence claim in \ref{thm:ex_uni_geom}}]
  The first step consists of constructing a suitable covering of the manifold
  $M$. For $p \in M$ consider charts $x_p$ on $M$ and $y_p$ on $\rr^{n,1}$ such 
  that
  $y_p \circ \varphi \circ x_p^{-1}$ is the special graph representation
  with center $p$
  constructed in section \ref{par:graph_repr}. 
  From Proposition \ref{prop:graph_assum} we derive that the representations of 
  the initial data
  w.r.t.\ $x_p$ and $y_p$ satisfy the conditions \ref{assumptions_atlas} with 
  constants independent of $p$.
  Hence, Proposition \ref{prop:cutoff_ex} yields the existence of a solution
  $F_p$ defined on the image of $x_p$ and attaining the initial data on the 
  ball $B^e_{\theta \rho_1 /2}(0)$ with
  a constant $\theta$ independent of the point $p$.

  Observe that the special graph representation about each point
  satisfies uniformly the additional assumption \eqref{cond_next} which is 
  stated
  in Remark \ref{rem:metric_eigen}. Let $\tilde{G}_1$ and $\tilde{G}_2$
  denote the constants defined in \eqref{eq:compare_constant} providing
  a comparison of Euclidean balls and balls w.r.t.\ $\ig$.

  Let $\mathcal{U}$ be the family of balls $B^{\ig}_{\tilde{G}_1^{1/2} 
    \theta\rho_1  /4}(p)$ for $p \in M$. Choose a locally finite covering
  of $M$ subordinate to $\mathcal{U}$ and denote it by 
  \begin{gather}
    \label{eq:def_atlas}
    (W_{\lambda})_{\lambda \in \Lambda} \text{ with }
    W_{\lambda} = B^{\ig}_{\tilde{G}_1^{1/2} 
      \theta\rho_1  /4}(p_{\lambda}).
  \end{gather}
  Consider the decomposition 
  $(U_{\lambda}, x_{\lambda},V_{\lambda},  y_{\lambda})_{\lambda \in \Lambda}$ of $\varphi$
  where $y_{\lambda} \circ \varphi \circ x_{\lambda}^{-1}$ is the special
  graph representation \eqref{eq:graph_repr} with center $p_{\lambda}$.
  From Theorem \ref{thm:ex_uni_atlas} we get a family $(F_{\lambda})$ of 
  solutions 
  defined on $[-T',T'] \times B^e_{\theta \rho_1 /2}(0)$, where the constant $T'$
  is independent of $\lambda$ and the ball $B^e_{\theta \rho_1 /2}(0)$ is a 
  subset of  the image of $x_{\lambda}$. 

  Define a mapping $F: U \subset \rr \times M \rightarrow \rr^{n,1}$ by
  \begin{equation}
    \label{eq:ex_def}
    F(t,p) = y_{\lambda}^{-1} \bigl( F_{\lambda}(t,x_{\lambda}(p))\bigr) \quad
    \text{if } p \in W_{\lambda}.
  \end{equation}
  Here, the set $U$ is a neighborhood of $\{0\} \times M$.
  A priori it is not clear whether  the set $U$ 
  has the form $[-T,T] \times M$ for a fixed constant $T > 0$.
  
  We need to show that $F(t,p)$ is well-defined. 
  Let $q \in W_{\lambda_1} \cap W_{\lambda_2}$ for $\lambda_1, \lambda_2 \in \Lambda$.
  By enlarging the radii of the balls
  $W_{\lambda_1} $ and $W_{\lambda_2}$ it follows that
  \begin{gather*}
    B^{\ig}_{\tilde{G}_1^{1/2} \theta
   \rho_1  /4}(q) \subset B^{\ig}_{\tilde{G}_1^{1/2} \theta
   \rho_1  /2}(p_{\lambda_1}) \cap B^{\ig}_{\tilde{G}_1^{1/2} \theta
   \rho_1  /2}(p_{\lambda_2}).
   \end{gather*}
   Denote changes of coordinates by
  $\hat{u} \circ y_{\lambda_1} = y_{\lambda_2}$ and $u \circ x_{\lambda_1}
  = x_{\lambda_2}$. 
  Proposition \ref{prop:uni_coord} now yields that
  \begin{gather*}
    y_{\lambda}^{-1} \bigl( F_{\lambda}(t,x_{\lambda}(q))\bigr)  
    = y_{\lambda}^{-1} \circ \hat{u}^{-1} \circ 
    F_{\lambda_2}\bigl(t, u\circ x_{\lambda}(q)
    \bigr) = 
    y_{\lambda_2}^{-1} \bigl( F_{\lambda_2}(t,x_{\lambda_2}(q))\bigr) 
  \end{gather*}
  provided $t$ satisfies
  \begin{gather}
    \label{eq:uni_cone_height}
    0 \le t \le T_0 := 
    \tfrac{1}{4c_0} \tilde{G}_2^{-1/2} \tilde{G}_1^{1/2} 
    \theta \rho_1.
  \end{gather}
  Applying  Remark \ref{rem:neg_t} gives us that the values
  of $F(t,p)$ coincide also for negative $t$.
  
  The preceding considerations yield that there exist a constant $T> 0$
  such that the mapping $F$ given by \eqref{eq:ex_def} is defined on 
  $[-T,T] \times M$. %
\end{proof}
\begin{rem}
  \label{rem:ex_time}
  The proof provides an estimate for the existence time $T$ of the
  constructed solution $F$. We get $T \ge \min(T',T_0)$ with the
  lower bound $T'$ from Remark \ref{rem:geom_lower} and $T_0$ introduced
  in \eqref{eq:uni_cone_height}. Further, we can give an expression for the 
  constant $c_E$ based on the assumptions \ref{assumptions_mink}.
  From estimate \eqref{eq:est_chi0} we derive the bound
  $\tilde{C}_0^{\chi}$ and the value $C_{w_0}^2 = 1 + m$ follows from the special
  graph representation.
  This yields the following expression for $c_E$ controlling the existence
  time $T'$
  \begin{multline}
    \label{eq:cE}
    c_E     =  
    4(1 + m + \delta_1^2)
    \bigl(1  + L_2^{-1}(1 - r_0)^{-1}
    \bigl(- L_2 + 2 L_1 L_3^2 + \delta_2^2\bigr) 
    \bigr)
    \\
    {}+ 2\bigl(- L_2 + 2 L_1 L_3^2 + \delta_2^2\bigr) 
    \bigl( 1+ m^{1/2} \omega_1^{2}(1 - R_0)^{-1}
    (1 + m + \delta_1^2) \bigr).
  \end{multline}
  The constants controlling the angle for the 
  initial submanifold and the initial velocity enter the constant $c_E$
  having the effect that the existence time shrinks if the initial
  submanifold or the initial velocity are closer to the light cone.
  The bound for the second fundamental form enters the
  definition of $T_0$ through the radius $\rho_1$.
\end{rem}

\begin{rem}
  \label{rem:sol_def_emb}
  By Proposition \ref{prop:embedded} it follows that local solutions
  obtained by Theorem \ref{thm:ex_uni_atlas} can be restricted such that they 
  are embeddings, provided that the initial submanifold is locally embedded.
  By using this shrinked domain we see that a solution defined by
  \eqref{eq:ex_def} is also locally an embedding.
\end{rem}

\begin{rem}
  \label{rem:ex_time_mink}
  The construction provides us with a lower bound for the supremum over
  length of
  timelike curves in $\im F$.
  Consider the curve $\gamma(t) := F(t,p)$ for $p \in M$.
  To derive a lower bound for the length of $\gamma$ we 
  introduce the function $w(r)$ by
  \begin{gather*}
    w(r) = \tint_0^r \bigl(- h(\dot{\gamma}, \dot{\gamma})\bigr)^{\2} \, d\sigma
    = \tint_0^r\bigl( - g_{00}(\sigma, p)\bigr)^{\2}\, d\sigma.
  \end{gather*}
  The construction of $F$ based on Proposition \ref{prop:cutoff_ex_hyp} yields
  that $\dot{w}(r) \ge \tilde{\lambda}^{\2}$ as can be seen from
  \eqref{eq:est_part} where the constant $\tilde{\lambda} > 0$ is defined.
   The derivative of $w$ is non-zero, therefore 
   an
  inverse function $\rho(s)$ to $w(r)$ exists. The curve $\tilde{\gamma}(s) =
  \gamma\bigl(\rho(s)\bigr)$ is parametrized by proper time.
  Since the inverse
  function is defined at least on the interval $[0,w(T)]$, the length of
  $\gamma$ is at least $T \tilde{\lambda}^{\2}$.

  The preceding consideration is valid for each $p \in \Sigma_0$
  which follows from the uniformity of the assumptions on $\Sigma_0$.
\end{rem}

\begin{proof}[\textbf{Proof of the uniqueness claim in \ref{thm:ex_uni_geom}}]
  Suppose $\bar{F}:[-T,T]\times M \rightarrow \rr^{n,1}$
  is a $C^2$-solution to the membrane equation in harmonic map gauge
  attaining the initial values \eqref{eq:mem_initial}. 
  To show uniqueness we compare this solution with the solution 
  $F:[-T,T]\times M \rightarrow \rr^{n,1}$ defined by \eqref{eq:ex_def}. 

  Consider the decomposition 
  $(U_{\lambda}, x_{\lambda},V_{\lambda},  y_{\lambda})_{\lambda \in \Lambda}$ of $\varphi$
  used in the construction of $F$.
  Proposition \ref{prop:graph_assum} yields that the assumptions 
  \ref{assumptions_atlas}
  are satisfied uniformly in $\lambda$.
  If $q \in M$, then we find a $\lambda \in \Lambda$ such that
  the image of the chart $x_{\lambda}$ contains a ball of radius
  $\tilde{G}_2^{-1/2} \tilde{G}_1^{1/2} 
  \theta \rho_1/4$ about $z = x_{\lambda}(q)$ 
  (see Remark \ref{rem:metric_eigen}).
  From the uniqueness result of Theorem \ref{thm:ex_uni_atlas} we conclude
  that the representations of $\bar{F}$ and $F$ w.r.t.\ $x_{\lambda}$
  and $y_{\lambda}$ coincide on a double-cone with height $2T_0$
  as defined in \eqref{eq:uni_cone_height}. This yields
  $\bar{F}(t,q) = F(t,q)$ for $ -\min(T,T_0) \le t \le \min(T,T_0)$.
\end{proof}

\section{Generalization to a Lorentzian manifold}
\label{sec:mem_ex_hyp}
In this section we will consider the Cauchy problem \eqref{eq:param_ivp}
for immersions in the case of a Lorentzian ambient manifold.
Restrictions on the causal structure of the ambient manifold 
will be made by 
assuming that a time function exists, i.e. a function whose gradient
is timelike everywhere. As a consequence of this assumption the ambient
manifold is globally hyperbolic (see e.g. \cite{ONeill:1983} or 
\cite{Wald:1984} for 
further
reference). Such manifolds carry a causal structure
similar to the Minkowski space providing the possibility to take
advantage of the results obtained in section \ref{sec:ex_mink}.

The layout of this section is based on section \ref{sec:ex_mink}.
To obtain a graph representation (section \ref{sec:graph_repr_hyp}) with similar
properties than derived in section \ref{par:graph_repr} it is necessary to
construct coordinates on the ambient manifold which will be done
in section \ref{sec:special_coord}.

Some notations taken from \cite{CHRKLA:1993} and 
\cite{Bart:1984} will be used.
Let $N$ be an $(n+1)$-dimensional  smooth manifold endowed with
a Lorentzian metric $h$. From the notations \ref{sec:notation} we adopt
the designations for the Christoffel symbols and the curvature of $h$.
An assumption on the causal structure
will be made by the following term.
\begin{defn}
  A real-valued function $\tau$ on $N$ is called \emph{time function}
  if its gradient $\D \tau$ is timelike everywhere. 
  Such a time function induces a \emph{time foliation} of $N$ by 
  its levelsets $\slice_{\tau}$, which are spacelike 
  hypersurfaces. The induced connection on a slice $\slice_{\tau}$ will
  be denoted by $\del$, the corresponding Christoffel symbols
  by $\gamma$ and the curvature by $\Rr$.

  We further introduce the \emph{lapse} $\psi$ of the foliation
by 
\begin{gather}
  \label{eq:def_lapse}
  \psi^{-2} = - h(\D \tau, \D \tau)
\end{gather}
and derive the unit future directed normal to the time foliation
by
\begin{gather}
  \label{eq:normal_slices}
  \hT = - \psi \D\tau.
\end{gather}
The dual vector field $\partial_{\tau}$ to 
the differential $d\tau$ of the time function is given by
\begin{gather*}
  \partial_{\tau} = - \psi^2 \D\tau \quad\Rightarrow\quad \partial_{\tau}
  = \psi \hT.
\end{gather*}
\end{defn}
The existence of such a function provides us with a possibility
of introducing a Riemannian
metric.
\begin{defn}
  \label{defn:def_E}
  Suppose $N$ admits a time function $\tau$.
  We consider the 
  decomposition into tangential and normal part w.r.t.\ the slices, namely
  \begin{gather*}
    v \in T_q N \Rightarrow v = v^{\top} + v^{\bot} = v^{\top} + \lambda \hT.
  \end{gather*}
  We introduce a Riemannian metric $E$ \emph{associated to the time foliation}
  by a flip of the sign of the unit normal $\hT$, more precisely
  \begin{gather}
    \label{eq:def_E}
    E(v,w) := h\bigl(v^{\top}, w^{\top}\bigr) + \lambda \mu, \quad\text{where }
    v^{\bot} = \lambda \hT\text{ and } w^{\bot} = \mu \hT.
  \end{gather}
  Since the slices are spacelike this metric is Riemannian.
\end{defn}
In the sequel we assume $N$ to admit a time function.

\subsection{Special coordinates}
\label{sec:special_coord}
The coordinates constructed in this section will make contact
to the results for the Minkowski space in the sense that  the representation
of the metric $h$ coincides in the center of the coordinates with the 
Minkowski metric.
These coordinates will be used to establish a special graph representation
similar to section \ref{par:graph_repr}.

Pick a point $q_0 \in N$ and set $\tau_0 = \tau(q_0), ~\psi_0 = \psi(q_0)$.
To adjust the lapse we introduce a new time function 
$\tilde{\tau} = \psi_0 \tau$ with
the corresponding lapse and dual vector field to 
$d\tilde{\tau}$ defined as follows
\begin{align*}
  \tilde{\psi}^{-2} & = - h(\D \tilde{\tau}, \D \tilde{\tau}) =  
  \psi_0^{2} \psi^{-2}
  \\
  \partial_0 & = \tilde{\psi} \widetilde{T} = \psi_0^{-1} \psi \hT = \psi_0^{-1}
  \partial_{\tau}.
\end{align*}
The following calculation shows that the unit normals w.r.t.\ both time
functions coincide
\begin{gather*}
   - \tilde{\psi} \D \tilde{\tau} 
  = - \psi_0^{-1} \psi \psi_0 \D\tau 
  = \hT.
\end{gather*}
Now we have a lapse equal to $1$ in the point
$q_0$. The slices remain the same hypersurfaces, only
the indexing has changed by means of $ \tilde{\tau}(q) = \sigma
\Leftrightarrow \tau(q) = \psi_0^{-1} \sigma$.

Let $\Phi_s$ be the flow of the vector field $\partial_0$ emanating from
a neighborhood $U$
of $q_0 \in \slice_{\tau_0}$.
If $\gamma(\sigma)$ 
is a part of the flow with $\gamma(0) = q \in U \subset \slice_{\tau_0}$
and $\dot{\gamma}(\sigma) = \partial_0\bigl(\gamma(\sigma)\bigr)$, then 
\begin{gather*}
  \tilde{\tau}\bigl(\gamma(\sigma)\bigr) 
  - \tilde{\tau}(q_0) = \tint_{\gamma} d\tilde{\tau} = 
  \tint_0^{\sigma} d\tilde{\tau}(\dot{\gamma})\, d\tau = \sigma.
\end{gather*}
Therefore, the flow $\Phi_s$ induces a diffeomorphism from $U \subset
\slice_{\tau_0}$ onto its image contained in $\slice_{t_0 + \psi_0^{-1}s}$.

We now construct coordinates on $N$ as follows. Pick normal coordinates
$z$ w.r.t.\ the slice $\slice_{\tau_0}$ with center $q_0$ and define
\begin{equation}
  \label{eq:special_coord}
  y(q) = \bigl(\tilde{\tau}(q) - \tilde{\tau}_0, z\circ 
  \Phi_{- (\tilde{\tau}(q) - \tilde{\tau}_0)}(q) \bigr)
\end{equation}
where $\tilde{\tau}_0 = \tilde{\tau}(q_0)$.
The basis vector fields of the tangent space of $N$ are given by
\begin{gather*}
  dy^{-1}_{(s,x)}(e_A) = 
  \begin{cases}
    \partial_0\bigl(y^{-1}(s,x)\bigr), & A = 0 \\
    \partial_{\si{b}}(s,x) = d\Phi_s \circ dz^{-1}_{x}(e_{\si{b}}), & A = \si{b}
    \text{ with } 1 \le \si{b} \le n.
  \end{cases}
\end{gather*}
The definition of the spatial vector fields $\partial_{\si{b}}$ yields that 
$\Lie_{\partial_0}
\partial_{\si{b}}(s,x) = 0$. 
The metric $h$ has therefore the following representation in these
coordinates
\begin{gather}
  \label{eq:special_metric}
  h_{AB}(s,x) = 
  \begin{pmatrix}
    - \psi_0^{-2} \psi^2(s,x) & 0 \\
    0 & h_{\si{a}\si{b}}(s,x)
  \end{pmatrix}
  \quad\text{with}\quad
  h_{AB}(0,0) = 
  \eta_{AB}
\end{gather}
where $h_{\si{a}\si{b}}(0,x)$ is the representation of the induced metric on the
slices w.r.t.\ the normal coordinates $z$.
The zeros come from the fact that the gradient of the time function
is normal to the slices.

We seek quantitative control over the representation of the metric and 
higher derivatives of it. This will be done by first considering
the metric components and then the Christoffel symbols $\G_{BC}^A$. 

A quantitative description makes it necessary to impose the following 
conditions on $(N,h)$.
\begin{gather}
  \label{eq:lorentz_cond}
  \begin{split}
    & \text{There exist constants $C_0^N, C_1, C_2$ and 
    $C_1^{\tau}$ 
      such that}
  \\
  & \abs{\RRiem}_E \le C_0^N, ~ C_1 \le \psi \le C_2
  \quad\text{and}\quad \abs{\D(\D\tau)}_E \le C^{\tau}_1,
  \end{split}
\end{gather}
where $\D(\D\tau)$ denotes the $(1,1)$-tensor obtained by applying the
covariant derivative to the gradient of $\tau$. Then we apply the norm
w.r.t.\ the metric $E$, introduced in definition \ref{defn:def_E}, to this 
tensor.
These conditions yield a bound for the curvature of the slices
due to the Gau\ss\  equations $\abs{\Rriem} \le c(\abs{\RRiem}_E
+ \abs{k}^2)$, where $k$ denotes the $(0,2)$-version of the second
fundamental form of the slices w.r.t.\ $\hT$.
The norm of $k$ is taken w.r.t.\ the Riemannian metric on the slices
induced by $h$.
A bound for the extrinsic curvature of the slices can be derived via the 
following computation for
$v,w \in T\slice_{\tau}$
\begin{align*}
   k(v,w) &= - h(\hT, \D_v w) = h(\D_v \hT, w)
  \\
 \text{and} \qquad \D_v \hT  &= - v(\psi) \hT - \psi\D_v (\D\tau). 
\end{align*}
If inserted the term involving a derivative of the lapse function vanishes 
since $w \bot T$ and the last term of the second line is bounded by assumption 
\eqref{eq:lorentz_cond}.

The second fundamental form $k$ 
can also be seen to be bounded via the representation
of $\D(\D\tau)$ in this coordinates, namely
\begin{align*}
  S_{\si{a}}^{\si{b}} & = h^{\si{b}\si{c}} h(\D_{\si{a}}(\D\tau), \partial_{\si{c}}) 
  = - \tilde{\psi}^{-1} k_{\si{a}}{}^{\si{b}}
  \\
  S_A^0 & = h^{00} h(\D_A(\D\tau), \partial_0 ) 
  = \tfrac{1}{2} \partial_A \tilde{\psi}^2,
  \\
  \text{and} \qquad S_0^{\si{b}} & = g^{\si{a}\si{b}} 
  h(\D_0(\D\tau), \partial_{\si{a}}) = - \tilde{\psi}^{-1} \del^{\si{b}} \tilde{\psi},
\end{align*}
where we set $S_A{}^B = \D_A \D^B \tau$.

We begin with estimates for the metric components and the Christoffel symbols
w.r.t.\ the normal coordinates $z$ on the slice $\slice_{\tau_0}$.
It is a well known fact that the components of the metric admits
a Taylor expansion where the coefficients are derived from linear
combinations of the curvature components. The error term can also
be estimated by such combinations. The Taylor expansion can be found
in \cite{Willmore:1993}. It follows for the components of
the metric that
\begin{gather*}
  h_{\si{a}\si{b}}(0,x) - \delta_{\si{a}\si{b}} = C_{\si{a}\si{c}\si{d} \si{b}} x^{\si{c}} 
  x^{\si{d}}.
\end{gather*}
If we assume $\tfrac{1}{3} \delta_{\si{a}\si{b}} \le h_{\si{a}\si{b}}(0,x)
\le 3\delta_{\si{a}\si{b}}$, then there exist a constant $c_1$ such that
\begin{gather*}
  \babs{\bigl(h_{\si{a}\si{b}}(0,x)\bigr) - \bigl(\delta_{\si{a}\si{b}}\bigr)}_e
  \le c_1 \sup_{w \in B_{\abs{x}}(0)} \abs{\Rr(w)} \, \abs{x}^2.
\end{gather*}
From the Gau\ss-equations we derive that the curvature of the slices
is bounded by the curvature of $N$ and the derivative of
the gradient of the time function $\tau$. By choosing $\rho_0 =
c_2 \bigl( C_0^N + (C_1^{\tau})^2 \bigr)^{-\2}$  with a constant
$c_2$ and the constants occurring in the assumptions \eqref{eq:lorentz_cond} 
it follows
\begin{gather}
  \label{eq:normal_coord_est}
  \bigl(1 
  - \tfrac{1}{6} \tfrac{\abs{x}^2}{\rho_0^2}\bigr)
  \delta_{\si{a}\si{b}} \le h_{\si{a}\si{b}}(0,x) \le \bigl(1 
  + \tfrac{1}{6} \tfrac{\abs{x}^2}{\rho_0^2}\bigr)\delta_{\si{a}\si{b}}.
\end{gather}
with $\abs{x} < \rho_0$.
We further need the evolution of the metric components to get a bound
on the whole domain of the coordinates. 
From the vanishing Lie derivative we derive that the covariant
derivative of the vector fields $\partial_0$ and $\partial_{\si{a}}$ commute.
Hence a similar device as stated in  \cite{CHRKLA:1993},
\cite{FriRe:2000} p. 138 (or 
\cite{HuiPold:1996} for the Riemannian case) gives us that
\begin{gather}
  \label{eq:spatial_metric_evol}
  \tfrac{d}{ds} h_{\si{a}\si{b}} = 2 \tilde{\psi} k_{\si{a}\si{b}}.
\end{gather}
Assuming that $\tfrac{1}{3} \delta_{\si{a}\si{b}} \le h_{\si{a}\si{b}}(s,x)
\le 3\delta_{\si{a}\si{b}}$ gives us a bound for the RHS of equation
\eqref{eq:spatial_metric_evol} in terms of the
bound for the lapse and the derivative of the gradient of the time
function. %

By shrinking the constant $c_2$ in the definition of $\rho_0$ 
and assuming $\abs{s} < \rho_0$ 
we arrive at
\begin{gather*}
  \babs{\bigl(h_{\si{a}\si{b}}(s,x)\bigr) - \bigl(h_{\si{a}\si{b}}(0,x)\bigr)} 
  < \tfrac{1}{6} \tfrac{\abs{s}}{\rho_0}.
\end{gather*}
For abbreviation we set $\lambda(s,x)  = \tfrac{1}{6} \tfrac{\abs{s}}{\rho_0} 
  +\tfrac{1}{6}\tfrac{ \abs{x}^2}{\rho_0^2} < 2^{-1}\tfrac{2}{3} =: 
  \tfrac{1}{2}\delta_0$
then the inequalities \eqref{eq:normal_coord_est} yield estimates for
$h_{\si{a}\si{b}}(s,x)$ as follows
\begin{multline}
  \label{eq:special_coord_est}
  (1 - \delta_0/2) \delta_{\si{a}\si{b}} \le \bigl(1 - \lambda(s,x)\bigr)
  \delta_{\si{a}\si{b}} \le h_{\si{a}\si{b}}(s,x)
  \\
  \le
  \bigl(1 + \lambda(s,x)\bigr)\delta_{\si{a}\si{b}} \le (1 + \delta_0/2)
  \delta_{\si{a}\si{b}}.
\end{multline}

With $\D(\D\tau)$ also the full gradient of 
the lapse function $\tilde{\psi}$ is bounded.
Hence
\begin{gather*}
  \abs{h_{00}(s,x) - h_{00}(0,0)} \le 2 C_1^{\tau} C_1^{-1}  \bigl(\abs{s}
  + 2 C_1^{-1}  \abs{x} \bigr).
\end{gather*}
By further shrinking the radius $\rho_0$ 
and setting 
$
\delta(s,x) = \tfrac{1}{6}\tfrac{1}{\rho_0}\bigl(\abs{s} 
  + \abs{x}\bigr) < \delta_0/2$
we obtain 
\begin{subequations}
  \begin{align}
  \label{eq:special_h00}
  - 1 - \delta_0/2 \le -1 - \delta(s,x)
  \le  h_{00} & \le 
  -1 + \delta(s,x) \le -1 + \delta_0/2,
  \\  
  \label{eq:special_lapse_est}
  (1 - \delta_0/2)^{\2} \le
  \bigl(1 - \delta(s,x)\bigr)^{\2}
  \le  \tilde{\psi} &\le
  \bigl(1 + \delta(s,x)\bigr)^{\2} \le (1 + \delta_0/2)^{\2}.
\end{align}
\end{subequations}
Therefore, 
the coordinates $y$ exist in a cylinder
$(s,x) \in [-\rho_0, \rho_0] \times B_{\rho_0}(0)
\subset \rr \times \rr^n$.
For the norms of $h_{AB}$ and its inverse we derive from the preceding
estimates
\begin{subequations}
  \begin{align}
  \label{eq:h_norm}
   \abs{(h_{AB})}_e & \le %
   2^{\2} (1  + \delta_0)
   && =: C_0^h
  \\
  \label{eq:h_norm_inverse}
  \text{and} \quad  \abs{(h^{AB})}_e & \le
  2^{\2} (1 - \delta_0)^{-1}
  && =: \Delta_h^{-1}.
\end{align}
\end{subequations}
For later reference we state an estimate for the difference
$h_{AB} - h_{AB}(0,0) $ within the domain
of the coordinates.
it holds that
\begin{gather}
  \label{eq:h_diff_est}
  \abs{(h_{AB}) - (\eta_{AB})}_e 
    \le (\delta(s,x)^2 + \lambda(s,x)^2)^{\2}
    \le \delta_0\tfrac{1}{2\rho_0}(\abs{s} + \abs{x}).
\end{gather}

From this consideration it follows  
a comparison estimate for the metric $E$ (see definition \ref{defn:def_E})
and the Euclidean metric in these coordinates.
We have
\begin{multline}
  \label{eq:metric_E_est}
 (1 - \delta_0/2) \delta_{AB}  \le \bigl(1 - \delta(s,x)\bigr)
 \delta_{AB} 
 \le E_{AB} 
 \\
 \le\bigl(1 + \delta(s,x)\bigr) \delta_{AB} 
 \le (1 + \delta_0/2)\delta_{AB}.
\end{multline}
Observe that $\abs{(E_{AB})}$ and $\abs{(h_{AB})}$ coincide and therefore,
the estimates \eqref{eq:h_norm} and \eqref{eq:h_norm_inverse} also hold 
for $E_{AB}$.

The next step is an estimate for the Christoffel symbols 
$\gamma_{\si{b}\si{c}}^{\si{a}}$ of the slice $\slice_{\tilde{\tau_0}}$. 
A similar expression as for the metric components exists
for the Christoffel symbols in 
Riemannian normal coordinates. It holds that (see e.g. \cite{Nest:1999})
\begin{gather*}
  \gamma_{\si{a}\si{b}}^{\si{c}}(0,x) - \gamma_{\si{a}\si{b}}^{\si{c}}(0,0) = 
  - \tfrac{2}{3} \Rr^{\si{c}}{}_{(\si{a}\si{b}) \si{d}}(0) x^{\si{d}}
  + \tilde{C}^{\si{c}}{}_{\si{a}\si{b}\si{d}\si{e}} x^{\si{d}} x^{\si{e}},
\end{gather*}
where we used the anti-commutator $T_{(ij)} = \tfrac{1}{2}(T_{ij} + T_{ji})$.
To control the last term of the preceding equation we have to assume 
the following.
\begin{gather}
  \label{eq:2nd_bd}
  \begin{split}
    & \text{There exist constants $C_2^{\tau}$ and $C_1^N$
      such that}
  \\
  & \abs{\D \RRiem}_E \le C^N_1
  \quad\text{and}\quad\abs{\D^2 (\D \tau)}_E \le C^{\tau}_2.
  \end{split}
\end{gather}
By shrinking the value
of the radius $\rho_0$ in the following manner
\begin{gather*}
  \rho_0 = c_3 \, \bigl((C_0^N + (C^{\tau}_1)^2)^{-\2} 
  + (C_1^N + C_1^{\tau} C_2^{\tau})^{-1/3}\bigr)
\end{gather*}
with a constant $c_3$ it follows
from the Gau\ss-equations that
there exist a constant $c_4$ such that
\begin{gather}
  \label{eq:special_chr_ind_est}
  \babs{\bigl(\gamma_{\si{a}\si{b}}^{\si{c}}(0,x)\bigr)}_e 
  \le c_4 \tfrac{\abs{x}}{\rho_0^2}
\end{gather}
assuming $\abs{x} < \rho_0$.

To control the time evolution of the Christoffel symbols we compute
the evolution of $h(\D_{\si{a}} \partial_{\si{b}}, \partial_{\si{c}})$ 
along the flow
of the vector field $\partial_0$. it holds that
\begin{eqnarray*}
  \tfrac{d}{ds} h(\D_{\si{a}} \partial_{\si{b}}, \partial_{\si{c}}) 
  & = & 
  h(\D_0 \D_{\si{a}} \partial_{\si{b}}, \partial_{\si{c}}) 
  + h(\D_{\si{a}} \partial_{\si{b}},
  \D_0 \partial_{\si{c}}).
  \\
  & = &
  \R_{0\si{a}\si{b}\si{c}} + h(\D_{\si{a}} \D_0 \partial_{\si{b}}, \partial_{\si{c}})
  + h(\D_{\si{a}} \partial_{\si{b}}, \D_0 \partial_{\si{c}})
\end{eqnarray*}
The vanishing of the Lie derivative $\Lie_{\partial_0} \partial_{\si{b}}$
yields the commutativity of the covariant derivative 
$\D_0 \partial_{\si{b}}
= \D_{\si{b}} \partial_0$.
By inserting $\partial_0 = \tilde{\psi} \hT$ we arrive at
\begin{eqnarray}
  \label{eq:evol_chr}
  h(\D_{\si{a}} \D_0 \partial_{\si{b}}, \partial_{\si{c}}) & = &
  \partial_{\si{b}} \tilde{\psi} k_{\si{a}\si{c}} + \partial_{\si{a}} \tilde{\psi}
  k_{\si{b}\si{c}} + \tilde{\psi} h(\D_{\si{a}} \D_{\si{b}} \hT, \partial_{\si{c}}).
\end{eqnarray}
These terms are controlled by the bounds for the lapse, the curvature
of $N$ and the bounds for the first and the second derivative 
of the gradient of the time function.
Since $\D_{\si{a}} \hT$ is tangential to the slices, we obtain for the other term 
\begin{gather}
  \label{eq:special_chr_part}
  h\bigl(\D_{\si{a}} \partial_{\si{b}}, \D_{\si{c}} (\tilde{\psi} T)\bigr) = -
  \D_{\si{c}} \tilde{\psi} \,
  k_{\si{a}\si{b}} + \tilde{\psi} \,\gamma_{\si{a}\si{b}}^{\si{d}} 
  \,k_{\si{d} \si{c}}.
\end{gather}
Integrating and using estimate \eqref{eq:special_chr_ind_est}
yields the existence of constants $C_3$ and $C_4$
such that it holds that
\begin{gather}
  \label{eq:chr_gronwall}
  \babs{\bigl(\gamma_{\si{a}\si{b}}^{\si{c}}(s,x)\bigr)}_e \le 
  c_4 \tfrac{\abs{x}}{\rho_0^2}
  + C_3 \tfrac{s}{\rho_0^2} + C_4  \tfrac{1}{\rho_0} \int_0^s 
  \babs{\bigl(\gamma_{\si{a}\si{b}}^{\si{c}}(\tau,x)\bigr)}_e\, d\tau.
\end{gather}
From Gronwall's lemma we derive a constant $C_0'$ such that
\begin{gather}
  \label{eq:special_chr_evol_est}
  \babs{\bigl(\gamma_{\si{a}\si{b}}^{\si{c}}(s,x)\bigr)}_e \le C'_0 \tfrac{1}{\rho_0^2}
  (\abs{s} + \abs{x}).
\end{gather}

Since $\partial_0 \bot \partial_{\si{a}}$ the Christoffel symbols of $N$ split
as follows $\G_{\si{b}\si{c}}^{\si{a}}
= \gamma_{\si{b}\si{c}}^{\si{a}}$.
The other parts of the Christoffel symbols $\G_{BC}^A$ only involve
the second fundamental form of the slices and derivatives of the lapse. 
The evolution 
along a straight line in the normal coordinates on the slices and the evolution
along the flow of the vector field $\partial_0$ can already be estimated by
assumptions \eqref{eq:lorentz_cond} and \eqref{eq:2nd_bd}. 
These considerations yield an estimate similar to 
\eqref{eq:special_chr_evol_est} for the Christoffel
symbols on the slices as follows
\begin{gather}
  \label{eq:chr_target}
  \babs{\bigl(\G_{BC}^A(s,x)\bigr)}_e \le
  C^{\G}_0 \tfrac{1}{2\rho_0^2}(\abs{s} + \abs{x}) \qquad\text{for }
  \abs{s}, \abs{x} < \rho_0
\end{gather}
with a constant $C^{\G}_0$.

Higher derivatives of the metric $h_{AB}$ can be estimated
via this considerations, if bounds for higher derivatives
of the curvature of $N$ and for derivatives of the gradient of
the lapse function are imposed. From the Taylor expansion for the components
of the metric and the Christoffel symbols we derive that
the $\ell$-th derivative of the metric
coefficients $h_{AB}$ is bounded provided that  bounds for derivatives
of the curvature
of $N$ up to order $\ell$ and bounds for derivatives of the gradient of the time
function up to one order higher exist.
We summarize the result for later reference.
Let $k_0$ be an integer and suppose 
the following assumptions hold
\begin{gather}
  \label{eq:high_der_assum}
  \begin{split}
     \text{There }& \text{exist constants $C_1, C_2, C_1^{\tau}, \dots,
      C_{k_0 + 1}^{\tau}, C_0^N, \dots, 
      C_{k_0}^N$
      such that}
    \\
    & \abs{\D^{\ell} \RRiem}_E \le C^N_{\ell} \text{ for }
    0 \le \ell \le k_0, \quad C_1 \le \psi \le C_2 
    \\
    \text{and} \quad &
    \abs{\D^{\ell}(\D\tau)}_E \le C^{\tau}_{\ell} \text{ for }
    1 \le \ell \le k_0 + 1.
  \end{split}
\end{gather}
Then there exist constants $C^h_1, \dots, C^h_{k_0}$ such that
\begin{gather}
  \label{eq:special_metric_der_est}
  \abs{(D^{\ell} h_{AB})}_e \le C_{\ell}^h = C_{\ell}^h(C_1, C_2, C_1^{\tau}, \dots,
  C_{\ell + 1}^{\tau}, C_0^N, \dots, 
  C_{\ell}^N) \quad\text{for } 1 \le \ell \le k_0.
\end{gather}

\begin{rem}
  \label{rem:loc_time}
  Suppose $N$ admits a time function $\tau$ only in a neighborhood $V$
  of a point $q$. Let $\rho > 0$ be a constant satisfying 
  $B^E_{\rho}(q)\subset V$, where $E$ denotes the flipped metric.
  According to the comparison of balls w.r.t.\ $E$ and Euclidean balls
  in $\rr^{n+1}$ in the constructed coordinates we can shrink
  $\rho_0$ depending on $\rho$ such that the coordinates remain
  to contain the cylinder $[-\rho_0, \rho_0] \times B_{\rho_0}(0)
  \subset \rr \times \rr^n$.
\end{rem}

\subsection{Special graph representation}
\label{sec:graph_repr_hyp}
In this section we develop a graph representation of a spacelike submanifold
in  analog manner to the special graph representation in Minkowski space
in section \ref{par:graph_repr}. 

Let $\Sigma_0$ be a spacelike submanifold in $N$. Suppose 
$M$ to be an $m$-dimensional manifold and
$\varphi: M \rightarrow  N$ to be an immersion of $\Sigma_0$.
The induced metric on $M$ will be denoted by $\ig$, the corresponding
Levi-civita connection and Christoffel symbols by $\inab$ and $\iG$ respectively.

Let $p_0 \in M$ and choose coordinates
on $M$ with center $p_0$. Let $y$ be the special coordinates 
constructed in the preceding section with center $q_0 = \varphi(p_0)$.
After possibly changing the order of the spatial directions of the coordinates
$y$ and applying an orthogonal transformation to the spatial directions
we get a graph representation analog to the case of the Minkowski 
space of the form
\begin{gather}
  \label{eq:graph_repr_hyp}
  \Phi(w)^A e_A = 
  y \circ \varphi \circ x^{-1}(w) = w^j e_j + u^{\alpha}(w)
  e_{\alpha}
\end{gather}
with the property $Du^a(0) = 0$ for all $a$. The index of the graph functions
$u^{\alpha}$ will be lowered with the Euclidean metric.

We state a representation of the geometry in this setting as follows
\begin{subequations}
  \begin{align}
    \label{eq:def_tangent_hyp}
    \text{tangent vectors} && 
  \partial_j \Phi^A e_A  = ~&
  e_j + \partial_j u{}^{\alpha} e_{\alpha} 
  \\
  \label{eq:def_normal_hyp}
  \text{normal vectors} && N_{\alpha}  =~&   
  (h^{Ak} \partial_k u^{\alpha} - h^{A\alpha})e_A
  \\
  \label{eq:ind_metric_hyp}
  \text{induced metric} &&
  \ig_{ij}  = ~ &   h_{ij} 
  + h_{ib} \partial_j u^b + h_{aj} \partial_i u^a
  \nonumber\\
  && & 
  {}+ \partial_i u^a h_{ab} \partial_j u^b + \partial_i u^0 h_{00} \partial_j u^0.
\end{align}
\end{subequations}
The Christoffel symbols $\iG_{ij}^k$ and the second fundamental form
$\II_{ij}^A$ are defined according to
the expressions \eqref{eq:ChrSymb} and \eqref{eq:2ndfform}.

The chosen timelike normal can be rewritten to enlighten the role
of the unit normal $\hT$ as follows. 
Let 
\begin{gather}
  \label{eq:time_normal}
  \widehat{N}_0
  := \tilde{\psi} 
  (h^{Ak} \partial_k u^0 - h^{A0}) e_A = \hT + \tilde{\psi} h^{Ak} \partial_k u_0 e_A.
\end{gather}
Then the last part cancels the
tangential part of $\hT$
\begin{gather*}
  h(\hT, \partial_j \Phi) = \tilde{\psi}^{-1}(h_{0j} + \partial_j u^{\alpha} 
  h_{0\alpha}) = 
  - \tilde{\psi} \partial_j u_0.
\end{gather*}
Observe that the Minkowski space admits a time function defined as the time 
component
and the unit normal to the corresponding time foliation is the timelike
direction $\tau_0$. 
Thus, the preceding consideration shows that $\widehat{N}_0$ is the 
correspondent to the timelike normal $N_0$ of the Minkowski case defined in 
 \eqref{eq:def_normal}.

In the following lemma we derive comparison estimates for the 
metric components to the Euclidean metric.
\begin{lem}
  \label{lem:ind_metric_control}
  Let $(N,h)$ satisfy the assumption \eqref{eq:lorentz_cond}.
  Assume $\abs{\Phi^0}, \abs{(\Phi^{\si{a}})}_e < \rho_0$. Then 
  the induced metric satisfies the estimates \eqref{eq:posdef}
  at the center $p_0$ of the chart and
  \begin{multline}
    \label{eq:ind_metric_pos_est}
    \bigl(1 - \abs{Du^0}^2 - (1 + \abs{Du}^2) %
  \delta_0\tfrac{2}{\rho_0}(\abs{\Phi^0}
  + \abs{(\Phi^{\si{a}})}_e)\bigr)
  \delta_{ij} \\
  \le \ig_{ij} \le 
  (1 + \abs{Du}^2) (1 + \delta_0)
  \delta_{ij}.
  \end{multline}
\end{lem}
\begin{proof}
  The first claim follows from the fact that the metric $h_{AB}$
  coincides with the Minkowski metric at the center and
  $Du^a(0) = 0$ for all $a$.
To obtain the positive definiteness we write the induced metric in the following
  form
  \begin{gather*}
   \ig_{ij} = \partial_i \Phi^A \eta_{AB} \partial_i \Phi^B +
  \partial_i \Phi^A(h_{AB}(\Phi) - \eta_{AB}) \partial_j \Phi^B.
  \end{gather*}
  To the first part of this equation estimate \eqref{eq:posdef} applies
  and 
  the second part can be treated with estimate \eqref{eq:h_diff_est}.
  From the same expression for $\ig_{ij}$ and again
  estimate \eqref{eq:h_diff_est} the second inequality of the claim
  follows.
\end{proof}
The next step will be to derive an estimate for the second
derivatives of the graph functions $u^{\alpha}$.
\begin{lem}
  \label{lem:2nd_der}
  Let $(N,h)$ satisfy the assumptions \eqref{eq:lorentz_cond} and
  \eqref{eq:2nd_bd}.
  Let \\
  $\abs{\Phi^0}, \abs{(\Phi^{\si{a}})}_e < \rho_0$.
  Then the following inequality holds for 
  the second derivatives of the graph functions
  \begin{gather}
  \label{eq:2nd_der_est}
  \abs{D^2 u}_e \le C_3
  \abs{\II}_{\ig,E} (1 + \abs{Du}^2)^2
  + C_0^{\G} \tfrac{1}{\rho_0}(1 + \abs{Du}^2)^{3/2},
\end{gather}
where $C_3 = 2(1 +\delta_0/2)^{\2} \Delta_h^{-1} 
m^{\2}  (1 + \delta_0)$.
\end{lem}

\begin{proof}
  Let $n_{\alpha\beta} = h(N_{\alpha}, N_{\beta})$ with the normal vector fields
  $N_{\alpha}$ defined 
  in \eqref{eq:def_normal_hyp} analog to the Minkowski case. 
  Let further $h_{ij}^{\alpha}$ be the coefficients satisfying
  $\II_{ij} = h_{ij}^{\alpha} N_{\alpha}$. Then we
  get the following expression for the second derivatives of the
  graph functions
  \begin{align}
    \label{eq:2nd_der_raw}
    - \partial_i \partial_j u_{\gamma} & = 
    n_{\gamma\alpha} %
    h^{\alpha}_{ij}
    -  \partial_i \Phi^B \partial_j \Phi^C 
    (\G_{BC}^k
    \partial_k u_{\gamma} 
    - \G_{BC}^{\gamma}).
  \end{align}
  It differs from the expression in the Minkowski case by
the last term on the RHS involving the Christoffel
symbols $\G_{BC}^A$ of the ambient manifold.
By introducing the matrix $e_{\alpha\beta} = E(N_{\alpha}, N_{\beta})$ 
with $e_{\alpha\beta} \ge C_e
\delta_{\alpha\beta}$ the first term can be estimated
by the RHS of inequality \eqref{eq:2der_est_raw}
for the Minkowski case.
From estimate \eqref{eq:h_norm_inverse} we get
\begin{gather*}
  \abs{(n_{\alpha\beta})}_e \le 2 \Delta_h^{-1} (1 + \abs{Du}^2).
\end{gather*}
To obtain an expression for the constant $C_e$, we derive with the abbreviations
$w_k = \partial_k u_a v^a$ and $w_k^0 = v^0
\partial_k u^0$ for a vector $(v^{\alpha})
\in \rr^{n+1 - m}$ that
\begin{gather*}
  v^{\alpha} e_{\alpha\beta} v^{\beta} = \tilde{\psi}^{-2} (v^0)^2  + 
  \Bigl\langle\tbinom{-v_a}{w_k+ w^0_k}, (  h^{\si{a}\si{b}})
  \tbinom{-v_b}{w_{\ell}
    + w^0_{\ell}} \Bigr\rangle
\end{gather*}
Estimate \eqref{eq:special_coord_est} for the components $(h_{\si{a}\si{b}})$
now yields that
\begin{gather}
  \label{eq:e_spatial_pos}
  \Bigl\langle\tbinom{-v_a}{w_k+ w^0_k}, (  h^{\si{a}\si{b}})
  \tbinom{-v_b}{w_{\ell}
    + w^0_{\ell}} \Bigr\rangle 
  \ge (1 + \delta_0/2)^{-1} \abs{(v^a)}^2.
\end{gather}
Using estimate \eqref{eq:special_lapse_est}
for the lapse gives us $C_e = (1 +\delta_0/2)^{-1}$.
The term on the RHS of equation \eqref{eq:2nd_der_raw}
involving the Christoffel symbols  can be estimated by means
of the bound stated in \eqref{eq:chr_target}.
\end{proof}
Replacing the Euclidean metric $e$ by the Riemannian metric $E$ 
defined in \ref{defn:def_E} gives us 
an analog definition for a uniformly spacelike submanifold with bounded 
curvature.
\begin{defn}
  \label{def:u_spacelike}
  A spacelike submanifold $\Sigma_0$ of a Lorentzian manifold $(N,h)$
  is called \emph{uniformly spacelike submanifold
with bounded curvature} if there exist constants $\omega_1, C_0$ %
  such
  that
  \begin{subequations}
    \begin{align}
  \label{eq:smf_lorentz_cond}
  \begin{split}
    \inf\{ - h( \gamma , \hT ) :   \gamma  \text{ future-directed }& \text{
        unit timelike normal to } \Sigma_0\}\le \omega_1
  \end{split}
   \\
   \text{and}\quad 
  \abs{\II}_{\ig,E}  \le C_0^{\varphi}.&
\end{align}
  \end{subequations}

\end{defn}
From now on we suppose the submanifold $\Sigma_0$
to be uniformly spacelike with bounded curvature.
Analog to Lemma \ref{lem:cond_normal} in the Minkowski case a choice of a unit 
timelike
normal to the graph is needed to gain an estimate from the condition on
the angle.
\begin{lem}
  \label{lem:cond_normal_hyp}
  Recall the definition \eqref{eq:time_normal} for the timelike
  vector field $\widehat{N}_0$.
  Set 
  \begin{gather}
    \label{eq:tnormal}
    \nu_0 : =  \mabs{\widehat{N}_0}^{-1} \widehat{N}_0
    \quad\text{with}\quad
    \mabs{\widehat{N}_0}^{-1} = (1 - \tilde{\psi}^2 \partial_k u_0 h^{k\ell} 
    \partial_{\ell} u_0)^{-\2}.
  \end{gather}
  Then it follows
  from condition \eqref{eq:smf_lorentz_cond} that
  $- h\bigl(\nu_0(0), \hT\bigr)  \le \omega_1$.
\end{lem}
\begin{proof}
  Let $\gamma = \mabs{\nu_0 + \lambda^a N_a}^{-1}
(\nu_0 + \lambda^a N_a)$ for $(\lambda^a) \neq 0$ a perturbation of 
$\nu_0$ using the other normal vectors $N_a$ defined in 
\eqref{eq:def_normal_hyp}. 
At the origin we have $N_a(0) = - h^{Aa}e_A$ and $h(N_a, \hT) = 0$.
This yields $\mabs{\nu_0(0)
  + \lambda^a N_a(0)}^2  = 1 - \lambda_a h^{ab} \lambda_b < 1$,
 where we lowered with the Euclidean metric. 
It follows that $- h\bigl(\gamma(0), \hT\bigr) > - h\bigl(\nu_0(0), \hT\bigr)$.
\end{proof}
An analog to Lemma \ref{lem:graph_est} contains an additional
estimate to ensure that the graph stays within the chosen coordinates
on the ambient manifold $N$. %
Let
\begin{align}
  \label{eq:fct_pos_def}
  && v(x) & := \bigl(1 - \abs{Du^0}^2 - (1 + \abs{Du}^2) 
  \delta_0\tfrac{1}{2\rho_0}(\abs{\Phi^0}
  + \abs{(\Phi^{\si{a}})}_e)\bigr)^{-\2}
  \\
  \label{eq:fct_radius_graph}
  \text{and} && 
  w(x) & := \bigl( 1 -  \tfrac{\abs{\Phi}^2_e}{\rho_0^2}
  \bigr)^{-\2}.
\end{align}
In the definition of the function $w$ we use the Euclidean norm on $\rr^{n+1}$
for the representation $(\Phi^A)$.
This ensures that if $\abs{(s,x)}_e < \rho_0$ then 
the pair $(s,x)$ is contained in the cylinder
$[-\rho_0, \rho_0] \times B_{\rho_0}(0)
\subset \rr \times \rr^n$ where the special coordinates constructed in
section \ref{sec:special_coord} are defined.

\begin{lem}
  \label{lem:graph_est_hyp}
  Let $(N,h)$ satisfy the assumptions \eqref{eq:lorentz_cond} and
  \eqref{eq:2nd_bd}. Let $\lambda \ge 250000$ be a fixed constant.
  Then there exist a constant $\rho_1 > 0$ such that for $z \in \rr^m,~
  \abs{z} =: r < \rho_1$ and $\{\tau z:
  0 \le \tau \le 1\}$ contained in the image of the coordinates $x$
  on $M$, which are part of the graph representation \eqref{eq:graph_repr_hyp},
  the inequalities \eqref{eq:stahl_1} and \eqref{eq:stahl_3}
  (taking the new values of $\lambda, \rho_1$ and $B_1$ into account)
  and 
  the following estimates hold
  \begin{gather}
    \label{eq:fct_radius_est}
    w(z) \le 1 + \bigl( \tfrac{r}{2\lambda\rho_1} \bigr)^2 \qquad\text{and}
    \qquad
    v(z) \le \omega_1 + \tfrac{r}{2\lambda\rho_1}.
  \end{gather}
\end{lem}

\begin{proof}
  We begin with 
  \begin{gather*}
    \rho_1 := (2\lambda)^{-1} \max(C_3 C_0^{\varphi}, \rho_0^{-1} C^{\G}_0)^{-1}.
  \end{gather*}
  This  definition  yields
  the same bound for the second derivatives of the graph functions
  as stated in \eqref{eq:2nd_der_full}.
  Therefore, from the proof of Lemma \ref{lem:graph_est} we derive
  the estimates \eqref{eq:stahl_1} and \eqref{eq:stahl_3}.

  We head to the function $w(z)$ defined in \eqref{eq:fct_radius_graph}. We 
  will apply an ODE comparison argument
  similar to the argument used in the proof of Lemma \ref{lem:graph_est}
  for the function $v$.
  The differential of $w$ can be estimated by
  \begin{gather*}
    \abs{Dw}\le w^3 
    \tfrac{1}{\rho_0^2} \abs{\Phi} \abs{D\Phi}.
  \end{gather*}
    With the Taylor expansion  for $\Phi$
  we derive
  \begin{gather*}
    \abs{Dw}\le w^3 m  \tfrac{r}{\rho_0^2} (1 + \abs{Du}^2).
  \end{gather*}
  By using estimate \eqref{eq:stahl_1} it follows that the last
  term is bounded by $B_1^2$.
  Shrinking the radius $\rho_1$ 
  such that 
  \begin{gather}
    \label{eq:radius_adjust}
    \rho_1 \le \rho_0 (4\lambda^2 m B_1^2 )^{-\2}
  \end{gather}
  and using an ODE comparison argument provides  us with the desired estimate.

  We now look at the function $v(z)$ defined in \eqref{eq:fct_radius_graph}. 
  To use an ODE 
  comparison argument
  we have to estimate the deviation of $v$ which will be done
  by considering the $\limsup$ since a norm is involved.
  In the proof of Lemma \ref{lem:graph_est} it was shown that
  \begin{gather*}
    \babs{D(\abs{Du^{\alpha}}^2)} \le 2\abs{Du^{\alpha}} \abs{D^2 u^{\alpha}}
  \end{gather*}
  and the inequalities \eqref{eq:stahl_1} and \eqref{eq:stahl_3} yield
  an estimate for the RHS.
  It remains to regard the $\limsup$ of $\abs{\Phi^0}
  + \abs{(\Phi^{\si{a}})}$ which can be controlled by
  \begin{gather}
    \label{eq:est_DPhi}
    \abs{D\Phi}_e \le m^{\2} (1 + \abs{Du}^2)^{\2}
  \end{gather}
  Using the inequalities \eqref{eq:Du_est}, \eqref{eq:stahl_1},
  \eqref{eq:stahl_3} established in the preceding step and \\
  $\abs{\Phi^0}, \abs{(\Phi^{\si{a}})}_e < \rho_0$
  yields
  \begin{gather*}
    \abs{Dv} \le v^3\bigl((1 + \delta_0)
    (1 + B_1^4\tfrac{r}{2\lambda\rho_1})B_1^4\tfrac{1}{2\lambda\rho_1}
    +  \delta_0 \tfrac{1}{2\rho_0} m^{\2} B_1^3\bigr).
  \end{gather*}
By shrinking the radius $\rho_1$ we achieve that 
$\abs{Dv} \le v^3 \tfrac{1}{2\lambda\rho_1}$.
  From the definition of $\nu_0$ in \eqref{eq:tnormal} and lemma 
  \ref{lem:cond_normal_hyp},
  we derive 
  that $v(0) \le \omega_1$.
  This gives us the desired result by an ODE comparison argument.
\end{proof}
\begin{rem}
  Analog to Remark \ref{rem:metric_eigen} the preceding lemma yields
  that the graph representation exists in the ball $B_{\rho_1}(0) \subset \rr^m$.
  The additional possibility for the graph representation to fail, 
  leaving the image of the special coordinates on $N$, was excluded by 
  estimating  the function $w(z)$.

  Set
\begin{gather}
  \label{eq:hyp_metric_compare_const}
  G_1 = \bigl( \omega_1 + \tfrac{2}{\lambda^2} \bigr)^{-2} \quad\text{and}
  \quad G_2 = B_1^2
\end{gather}
then it holds that 
\begin{gather}
  \label{eq:hyp_metric_compare}
  G_1 \delta_{ij} \le \ig_{ij} \le G_2 \delta_{ij}
  \quad\text{within the ball } B_{\rho_1}(0) \subset \rr^m.
\end{gather}
\end{rem}
To obtain bounds for higher derivatives, derived in the next lemma,
 of the graph functions  we
make use of Lemma \ref{lem:high_cov_der}.
\begin{lem}
\label{lem:hyp_u_high_der}
  Let $k$ be an integer and assume $(N,h)$ to satisfy the inequalities
  \eqref{eq:high_der_assum} for $k_0 = k$.
  Let the immersion $\varphi$ satisfy the the following assumption.
  \begin{gather*}
    \text{There exist constants $C^{\varphi}_0, \dots, C^{\varphi}_k$
      such that }\abs{\hn^{\ell} \II}_{\ig,E}
    \le C_{\ell}^{\varphi} \quad\text{for}\quad
    0 \le \ell \le k.
  \end{gather*}
  Then there exists a constant $C^u_{k+2}$ such that
  $\abs{D^{k + 2} u}_e 
    \le C^u_{k+2}$.
\end{lem}

\begin{proof}
  As in the proof of the analog statement in Lemma \ref{lem:u_high_der}
  we make
  use of expression \eqref{eq:high_der_cov}  
  derived in Lemma \ref{lem:high_cov_der} and adapt expression
  \eqref{eq:u_high_der}
  derived in
  Lemma \ref{lem:u_high_der} to the case of a Lorentzian ambient manifold.
  The proof will be done via an induction
  on the order of differentiation. The case $k = 0$ was derived in
  Lemma \ref{lem:graph_est_hyp}.

  We begin with an expression similar to \eqref{eq:u_high_der}.
  it holds that
  \begin{gather}
    \label{eq:u_high_der_hyp}
    - \partial^k \partial_i \partial_j u_{\beta} 
     = \partial^k \II_{ij}^B \,h_{BC}(\Phi) \,N_{\beta}^C 
     + \sum_{\genfrac{}{}{0in}{}{k_1 + k_2 + k_3 = k,}{k_1 < k} }      
     \partial^{k_1} \II_{ij}^B\,
     \partial^{k_2}\bigl(h_{BC}(\Phi)\bigr) \,\partial^{k_3} N_{\beta}^C.
  \end{gather}
  Only the first term on the RHS contains derivatives of the graph functions
  up to order $k + 2$, the others are lower order terms since
  the normal $N_{\beta}$ involves first derivatives and the representation
  $h_{AB}$ involves the graph functions itself. Derivatives of the components
  $h_{AB}$
  are bounded up to order $k$ which follows from inequality 
  \eqref{eq:special_metric_der_est}.
  
  According to expression \eqref{eq:high_der_cov} we need to estimate
  the induced Christoffel symbols $\iG$ computed in 
  \eqref{eq:ChrSymb} and the Christoffel symbols
  $\indg$ for the connection on the pullback bundle  $\varphi^{\ast} TN$
  defined in \eqref{eq:ind_chr}.

  Derivatives of the induced Christoffel symbols up to order $k - 1$
  need to be estimated. The first term of expression \eqref{eq:ChrSymb}
  is controlled by
  derivatives of the induced metric
  up to order $k - 1$, by  
  derivatives of the graph functions up to order $k + 1$ and by
  bounds for the metric $h_{AB}$. Thus, we get bounds using
  Corollary \ref{cor:pos_matrix_der} and  estimates 
  \eqref{eq:hyp_metric_compare} and
  \eqref{eq:special_metric_der_est}.
  The second term involves Christoffel symbols of $N$, therefore 
  we need one order of differentiation more of the metric
  $h_{AB}$ to control it.
  
  Definition \eqref{eq:ind_chr} for the Christoffel symbols
  $\indg$ yields that their derivatives of order $k - 1$ can
  be estimated by derivatives of the graph functions
  and of the metric $h_{AB}$ up to order $k$
  which is covered by the preceding consideration.
  
  It remains to treat the last term of expression 
  \eqref{eq:high_der_cov} involving
  covariant derivatives of the second fundamental form.
  Opposing to the Minkowski space it is necessary to estimate
  $\abs{\hn^{\ell} \II}_{\ig,e}$ by the term
  $\abs{\hn^{\ell} \II}_{\ig,E}$ which can be done via the estimates
  \eqref{eq:metric_E_est} for the components of the metric $E$.
\end{proof}

\subsection{Solutions for fixed coordinates}
\label{sec:setup_hyp}
In this section we will consider the Cauchy problem \eqref{eq:param_ivp}
for submanifolds $\Sigma_0$ of a globally hyperbolic ambient manifold $(N,h)$.

We use the notations of the main problem regarding the initial submanifold
and the initial direction $\nu$.
Assume $M$ to be an
$m$-dimensional manifold and $\varphi: M \rightarrow N$ to be an immersion
of $\Sigma_0$.  As initial velocity for the Cauchy problem we combine
initial direction, lapse and shift in a timelike vector field 
$\chi: M \rightarrow
TN$ along $\varphi$.
Let $s > \tfrac{m}{2} + 1$ be an integer fixing differentiability properties
in the sequel. 

Assume $(U_{\lambda}, x_{\lambda}, V_{\lambda}, y_{\lambda})_{\lambda \in \Lambda}$ 
to be a decomposition of $\varphi$ according to definition
\ref{defn:decomp}.
For $\lambda \in \Lambda$ let $\Phi_{\lambda}$ and $\chi_{\lambda}$ 
denote the representations of 
$\varphi$ and $\chi$ in these coordinates  analog to 
\eqref{eq:initial_expr}.

We make the following uniformity assumptions on the coordinates $(V_{\lambda},
y_{\lambda})_{\lambda \in \Lambda}$ for $N$ and the representation of $h$ in these
coordinates.
\begin{assum}
  \label{assumptions_atlas_hyp}
  There exist constants $\delta_0 < 1,~ \rho_0 $ such that
    for each $\lambda \in \Lambda$
    the image of $y_{\lambda}$ contains a cylinder $\bigl(y^0_{\lambda}, 
    (y^{\si{a}}_{\lambda})_{\si{a}}
    \bigr) = (s,x)
    \in [-\rho_0, \rho_0]
    \times B_{\rho_0}(0) \subset \rr \times \rr^{n}$ 
    and the  representation of the metric $h$ 
    satisfies in these
    coordinates 
    \begin{subequations}
      \begin{align}
      \label{eq:cond_pos}
      && \babs{\bigl(h_{AB}(s,x)\bigr) - (\eta_{AB})}_e 
      & \le \delta_0 \tfrac{1}{2\rho_0}(\abs{s}
      + \abs{x})
      \\
      \label{eq:cond_der}
      \text{and}&& \babs{(D^{\ell} h_{AB})}_e & \le C_{\ell}^h 
      \quad\text{for } 1 \le \ell \le s + 1
    \end{align}
    \end{subequations}
    with constants $C^h_{\ell}$ independent of $\lambda \in \Lambda$.
\end{assum}
The following theorem states the first main result of this section,
a generalization of Theorem \ref{thm:ex_uni_atlas}.
\begin{thm}
  \label{thm:ex_uni_atlas_hyp}
  Let the coordinates $(V_{\lambda}, y_{\lambda})_{\lambda \in \Lambda}$
  on $N$, which are a  part of the decomposition \\
  $(U_{\lambda}, x_{\lambda}, 
  V_{\lambda}, y_{\lambda})_{\lambda \in \Lambda}$, and the representation of the 
  metric
  $h$ in these coordinates satisfy the assumptions \ref{assumptions_atlas_hyp}.
  Let the decomposition $(U_{\lambda}, x_{\lambda}, 
  V_{\lambda}, y_{\lambda})_{\lambda \in \Lambda}$ and the initial data
  $\varphi_{\lambda}$ and 
  $\chi_{\lambda}$ for the reduced membrane equation satisfy the assumptions 
  \ref{assumptions_atlas} with $\eta_{AB}$
  replaced by $h_{AB}$.

  Then 
  there exist
  constants $0 < T \le T_1$, $0 < \theta < 1$ and a family
  $(F_{\lambda})$ of bounded
  $C^2$-immersions $F_{\lambda}: [-T,T] \times B^e_{\theta \rho_1/2}(0)
  \subset \rr \times x_{\lambda}(U_{\lambda})
  \rightarrow \rr^{n+1}$
  solving the reduced membrane equation 
  \begin{gather*}
  g^{\mu\nu} \partial_{\mu} \partial_{\nu} F^A - 
  g^{\mu\nu}\hat{\Gamma}_{\mu\nu}^{\lambda} 
  \partial_{\lambda} F^A + g^{\mu\nu} \partial_{\mu} F^B
  \partial_{\nu} F^C \G_{BC}^A(F) = 0
  \end{gather*} 
  w.r.t.\ the background metric $\hat{g}$ defined in \eqref{eq:back_metric}
  and attaining the initial values 
  \begin{gather*}
    \restr{F_{\lambda}} = \Phi_{\lambda}\quad\text{and}
    \quad\restr{\partial_t F_{\lambda}} = \chi_{\lambda}.
  \end{gather*}

  Let $F_{\lambda}$ and $\bar{F}_{\lambda}$
  be two such solutions defined on 
  the image of the same coordinates.
  Assume $z \in \rr^m$ be a point contained in the image
  of $x_{\lambda}$. If $F_{\lambda}$ and $\bar{F}_{\lambda}$
  attain the initial values $\Phi_{\lambda}$ and $\chi_{\lambda}$ on
  a ball $B^e_r(z)$, then they coincide
  on the truncated double-cone with base $B_{r}^e(z)$, slope $c_0$ and
  maximal height $2T_1$.
\end{thm}
\begin{rem}
    In contrast to the Minkowski case uniqueness only holds on a truncated 
    double-cone 
  with the maximal height given by a constant $T_1$ which will be defined in 
  \eqref{eq:choice_t}.
\end{rem}
\begin{rem}
  \label{rem:hyp_atlas_sol_diff}
  Let $\ell_0$ be an integer.
  Assume that the initial values and the decomposition satisfy
  the assumptions of Theorem \ref{thm:ex_uni_atlas_hyp} with an integer
  $r = s + \ell_0 > \tfrac{m}{2} + 1 + \ell_0$
  instead of an $s > \tfrac{m}{2} + 1$.
  Then the family $(F_{\lambda})$ of solutions to the 
  reduced membrane
  equation are immersions of class $C^{2+ \ell_0}$.
\end{rem}

Let $\lambda \in \Lambda $ be fixed.
From section \ref{sec:sol_atlas} we adopt the definition 
of the cut-off function $\zeta$ with the bounds for derivatives
of $\zeta$
in \eqref{eq:cut_off_est_mink}, the asymptotics $w(t)$ defined in
\eqref{eq:asymptotics} and the functions $\init{\Phi}$
and  $\init{\chi}$ defined in \eqref{eq:def_inter} and
\eqref{eq:vel_inter} respectively.

To handle the metric of the ambient manifold we define 
another cut-off function $\tilde{\zeta} \in C_c^{\infty}(\rr^{n+1})$
satisfying $\tilde{\zeta} \equiv 1$ on $B_{\tilde{\theta}\rho_0/2}(0)
\subset \rr^{n+1}$ and $\tilde{\zeta} \equiv 0$ outside $B_{\tilde{\theta}
\rho_0}(0) \subset \rr^{n+1}$. 
The bounds for derivatives of $\tilde{\zeta}$ are given by
\begin{gather}
  \label{eq:cutoff_est}
  \abs{D^{\ell} \tilde{\zeta}}_e
    \le \tilde{C}'_{\ell} (\tilde{\theta} \rho_0)^{-\ell}
\end{gather}
with constants $\tilde{C}_{\ell}'$ independent of $\tilde{\theta}$
and $\rho_0$.
We use this function to define
an interpolated metric as follows
\begin{align}
  \label{eq:target_inter_metric}
  \hat{h}_{AB}(z) = \eta_{AB} + 
  \init{h}_{AB}(z) 
  \quad\text{with}\quad
  \init{h}_{AB}(z)  := \tilde{\zeta}(z)\bigl(h_{AB}(z) - \eta_{AB}\bigr)
  \text{ for } z \in \rr^{n+1}.
\end{align}
The Christoffel symbols w.r.t.\ $\hat{h}_{AB}$  will be denoted by $\hg_{BC}^A$.
We define the matrix $\hat{a}_{\mu\nu}$ as in \eqref{eq:back_comp} taking
the metric $\hat{h}_{AB}$ instead of the Minkowski metric $\eta_{AB}$.
The Christoffel symbols of $\hat{a}_{\mu\nu}$ will be denoted by 
$\hat{\gamma}_{\mu\nu}^{\lambda}$.

We search for a solution to the IVP
\begin{subequations}
  \begin{align}
  \label{eq:geom_solve_hyp}
  \begin{split}
    g^{\mu\nu}(F,DF,\partial_t F) \partial_{\mu} \partial_{\nu} F^A & = 
  f^A(t,F,DF,\partial_t F) + \BF^A(t,F,DF,\partial_t F)
  \\
  \restr{F} & = w_0 + \init{\Phi}, ~ \restr{\partial_t F}
    = w_1 + \init{\chi},
  \end{split}
  \\
  \label{eq:def_rhs_first}
  \text{with } 
  f^A(t,F,DF,\partial_t F) & = 
  g^{\mu\nu}(F,DF,\partial_t F) 
  \hat{\gamma}_{\mu\nu}^{\lambda}(t) \partial_{\lambda} F^A
  \\
  \label{eq:def_rhs_BF}
  \text{and }
  \BF^A(F,DF,\partial_t F) & = 
  {}- g^{\mu\nu}(F,DF,\partial_t F) \partial_{\mu} F^B \partial_{\nu} F^C 
  \hg_{BC}^A(F).
\end{align}
\end{subequations}
The coefficients are defined as the inverse of the matrix
\begin{align}
  \label{eq:coeff}
  g_{\mu\nu}(F,DF,\partial_t F)
  = \partial_{\mu} F^A
  \hat{h}_{AB}(F) \partial_{\nu} F^B.
\end{align}
Here, the coefficients depend on $F$ itself,
not only on the derivatives
opposing to the Minkowski case.
For notational convenience we set
\begin{gather*}
  g_{\mu\nu}(F) = 
  g_{\mu\nu}(F,DF,\partial_t F),~
  f(F) = f(t,F,DF,\partial_t F)\text{ and }\BF(F)
  = \BF(F,DF,\partial_t F).
\end{gather*}
The following proposition states an analog result to proposition
\ref{prop:cutoff_ex}, from which we immediately derive
the existence claim of Theorem \ref{thm:ex_uni_atlas_hyp}.
\begin{prop}
  \label{prop:cutoff_ex_hyp}
    There exist a constant $T' >0$ and a unique
  $C^2$-solution $F:[-T',T']\times\rr^m
  \rightarrow \rr^{n+1}$ to the IVP \eqref{eq:geom_solve_hyp}
  satisfying
  \begin{gather}
    \label{eq:property_hyp}
    F(t) - w(t) \in C([-T',T'],H^{s+1}), ~\partial_t F(t) - w_1 \in 
    C([-T',T'],H^{s}).
  \end{gather}
\end{prop}
The strategy to obtain a solution to the IVP \eqref{eq:geom_solve_hyp} will be 
taken from section \ref{sec:sol_atlas}.
We begin with a definition of coefficients and RHS of an asymptotic equation
to which Theorem \ref{asym_ex} will be applied.
Since $F$ itself enters the coefficients, we have to add another slot
to the definition of $\Omega$ in \eqref{eq:omega_def}.
Let $\Omega \subset \rr^{n+1} \times \rr^{m(n+1)}
\times \rr^{n+1}$ be a set chosen later
and define for $(t,v,Y,X) \in \rr \times \Omega$
with $Y = (Y_k)$
\begin{gather}
  \label{eq:def_coeff_asym_hyp}
  g_{0\ell}^{\mathrm{a}}(t,v,Y,X) 
  := (\partial_{t} w + X)^A
  \hat{h}_{AB}\bigl(w(t) + v\bigr)(\partial_{\ell} w + Y_{\ell})^B.
\end{gather}
The other parts $g_{00}^a$ and $g_{k\ell}^a$ are defined analogously.
Observe that in this case the coefficients depend explicitly on the time
parameter $t$.
Let the functions $f_{\mathrm{a}}$ and $\BF_{\mathrm{a}}$ be defined in an analog manner to the RHS 
of the asymptotic equation for the Minkowski space in \eqref{eq:def_rhs_asym}
regarding the definitions \eqref{eq:def_rhs_first} and 
\eqref{eq:def_rhs_BF}. As in section \ref{sec:sol_atlas}
the set $\Omega$ will be used to ensure that
the matrix $g^{\mu\nu}_{\mathrm{a}}$ has the desired signature 
$(\,{-}\,{+}\,\cdots\,{+}\,)$.

These definitions give rise to the following operators 
\begin{multline}
  \label{eq:def_coeff_op_hyp}
  g_{\mu\nu}^{\mathrm{a}}(t,\varphi_0, \varphi_1) = 
  g_{\mu\nu}^{\mathrm{a}}(t,\varphi_0,D\varphi_0, \varphi_1),~
  g^{\mu\nu}_{\mathrm{a}}(t,\varphi_0, \varphi_1),~
  \\
  f_{\mathrm{a}}(t, \varphi_0, \varphi_1) \quad
  \text{and} \quad 
  \BF_{\mathrm{a}}(t, \varphi_0, \varphi_1),
\end{multline}
where the last terms are defined analogously.
The operators have a domain  $[0,T_1] \times W \subset \rr \times H^{s+1} 
\times H^s$, where $T_1$ is a constant and $W$ will be 
chosen later in dependency on $\Omega$. 
For later reference we state the asymptotic equation to be solved.
It reads
\begin{subequations}
  \begin{align}
  \label{eq:geom_solve_asym_hyp}
    g^{\mu\nu}_{\mathrm{a}}(t,\psi,\partial_t \psi) 
    \partial_{\mu} \partial_{\nu} \psi^A  = ~&
  f^A_a(t,\psi,\partial_t \psi) + \BF^A_a(t,\psi,\partial_t \psi)
  \\
  \label{eq:def_rhs_asym_first}
  \text{with } 
  f_{\mathrm{a}}^A(t,\psi,\partial_t \psi)  = ~& 
  g^{\mu\nu}_{\mathrm{a}}
  \hat{\gamma}_{\mu\nu}^{\lambda} (\partial_{\lambda} w + \partial_{\lambda} \psi)^A
  \\
  \label{eq:def_rhs_asym_BF}
  \text{and }\BF_{\mathrm{a}}^A(t,\psi,\partial_t \psi) = ~& 
  {}- g^{\mu\nu}_{\mathrm{a}}
  (\partial_{\mu} w + \partial_{\mu} \psi)^B
  (\partial_{\nu} w + \partial_{\nu} \psi)^C \hg_{BC}^A\bigl(w(t) + \psi\bigr),
\end{align}
\end{subequations}
The initial values of $\psi$ are given by \eqref{eq:asym_initial} regarding the
new  definition of these functions.

Following the strategy of section \ref{sec:sol_atlas} for the Minkowski case 
we begin
with an analog to Lemma \ref{lem:ind_metric} establishing estimates for the 
spacelike
and timelike part of the inverse of the coefficients.
\begin{lem}
  \label{lem:ind_metric_hyp}
  The following estimates hold for the %
  matrix  $g_{\mu\nu}^{\mathrm{a}}(t,v,Y,X)$
  \begin{align*}
    &&g_{k\ell}^a   
    \ge~& 
    \bigl(\omega_1^{-2} -  2 
      \abs{Y}_e (\abs{Y}_e + \abs{Dw_0}_e)
    {}- \delta_0 \tilde{\theta} ( \abs{Dw_0}_e^2 + \abs{Y}_e^2)
      \bigr)\delta_{k\ell}
    \\
     \text{and }&& 
    g_{00}^a   \le~& h(\chi_0, \chi_0) + 2 (1 + \delta_0 
    \tilde{\theta})\abs{X}_e (\abs{\chi_0}_e 
    + \abs{X}_e).
  \end{align*}
\end{lem}
\begin{proof}
  Both members of the matrix $g_{\mu\nu}^{\mathrm{a}}$  can be  divided into a part which is 
  covered by
  Lemma \ref{lem:ind_metric} and a part depending on $\init{h}_{AB}$.
  The latter part can be estimated by condition \eqref{eq:cond_pos} on the 
  representation of
  the metric $h$, where
  we use $\abs{s} < \tilde{\theta} \rho_0$ and $\abs{x} < \tilde{\theta}
  \rho_0$ due to the presence of the cut-off function $\tilde{\zeta}$.
\end{proof}
Let $\Omega$ be defined as follows
\begin{gather}
  \label{eq:Omega}
  \Omega := \rr^{n+1} \times B^e_{\delta_1}(0) \times B^e_{\delta_2}(0)\subset
  \rr^{n+1}\times \rr^{m(n+1)}
  \times \rr^{n+1},
\end{gather}
where $\delta_1$ and $\delta_2$ will be chosen in the sequel.
There is no condition for the first slot, since the cut-off process ensures 
that the matrix $\hat{h}_{AB}$ entering the definition of the coefficients
is defined on all of $\rr^{n+1}$.
We derive from lemma 
\ref{lem:ind_metric_hyp}  and the assumptions \ref{assumptions_atlas} 
the following estimates
\begin{align*}
  && g_{k\ell}^a(v, Y, X) & \ge \omega_1^{-2}\bigl(1 - 2 \omega_1^2
  \delta_1 (C_{w_0} + \delta_1 )
  - \omega_1^2
  \delta_0 \tilde{\theta}(C_{w_0}^2 + \delta_1^2 )
  \bigr) \delta_{k\ell}, 
  \\
  \text{and} && g_{00}^a(v, Y, X)  %
  &\le - L_2\bigl(1 -  2 L_2^{-1}(1 + \delta_0 
    \tilde{\theta})\delta_2 (\tilde{C}_0^{\chi} +  \delta_2)\bigr) \text{ for }
    (v, Y, X) \in \Omega.
\end{align*}
Compared to the estimates \eqref{eq:pos_neg} two additional parts occur which 
can be controlled via the parameter $\tilde{\theta}$.

Let $0 < R_0, r_0 < 1$ be fixed constants. The goal is to derive radii
$\delta_1$ and $\delta_2$ such that the matrix
$g_{\mu\nu}^{\mathrm{a}}$ satisfies similar estimates to \eqref{eq:est_part}.
First we choose $\delta_1$ and $\delta_2$ according to the inequalities
\eqref{eq:delta12} replacing the RHS by $\omega_1^{-2} R_0/2$ and
$L_2 r_0/2$ respectively. 
Now we choose $\tilde{\theta}$ so small that the terms in $g_{k\ell}^a$
and $g_{00}^a $ involving
this constant are also $\le \omega_1^{-2} R_0/2$ and
$\le L_2 r_0/2$ respectively.
This ensures that the components $g_{00}^a$ and $g_{k\ell}^a$ satisfy
the inequalities \eqref{eq:est_part}.

In the sequel we adopt the definition of the bounds $K_0$ and $K_1$ in 
\eqref{eq:est_space_full} and \eqref{eq:est_time_full} respectively.
Analog to Lemma \ref{lem:norm_metric_inverse} the estimates \eqref{eq:est_part}
and condition \eqref{eq:cond_pos} on the matrix $h_{AB}$ yield the
following result.
\begin{lem}
  \label{lem:norm_metric_inverse_hyp}
  Assume $(v,Y, X) \in \Omega$. Then it follows for the coefficients
  $g^{\mu\nu}_{\mathrm{a}}(v,Y,X)$ 
    \begin{subequations}
      \begin{align}
        \label{eq:hyp_est_inverse_full}
    \abs{(g^{\mu\nu}_{\mathrm{a}})}_e^2  &
    \le
    \tilde{\lambda}^{-2} + 
     2 \tilde{\lambda}^{-2}  \tfrac{m}{\tilde{\mu}^2} 
     K_0^2 K_1^2(1 + 
    \delta_0\tilde{\theta})^2
     + \tfrac{m}{\tilde{\mu}^2} & =:& \Delta^{-2}, 
     \\
     \label{eq:hyp_est_lambda}
    g^{00}_a  & \le   - K_1^{-2} 
    \bigl( 1+ \tfrac{m^{1/2}}{\tilde{\mu}}
    K_0^2 (1 + \delta_0 \tilde{\theta})\bigr)^{-1} (1 + 
    \delta_0\tilde{\theta})^{-1} & =:&-\lambda
    \\
    \label{eq:hyp_est_mu}
    \text{and} \quad
    g^{ij}_a  & \ge  K_0^{-2}
    \bigl(1  + \tilde{\lambda}^{-1} 
    K_1^2 (1 + \delta_0 \tilde{\theta})
    \bigr)^{-1}(1 + 
    \delta_0\tilde{\theta})^{-1}\delta^{ij} & =:& \mu\delta^{ij}.
  \end{align}
  \end{subequations}
\end{lem}

\subsubsection{Sobolev estimates}
\label{sec:sob_est_hyp}
Since the definition of $\Omega$ in \eqref{eq:Omega}
contains no constraint for the first slot, we can follow the
arguments in section
\ref{sec:sob_est_mink} for the Minkowski case, beginning with
an analog definition of the domain $W$ done in \eqref{eq:domain} involving
balls around the initial values of the asymptotic equation. 

Based on the observation that Lemma \ref{lem:initial_est} remains valid, since 
we only used 
the assumptions \eqref{eq:assum_bd_pt} and \eqref{eq:assum_bd_vel},
our first choice of the parameter $\theta$ follows the same way as in section 
\ref{sec:sob_est_mink}.

In contrast to the Minkowski case we can not be sure
that the unmodified reduced equation \eqref{eq:mem_red}
will be solved within the ball $B^e_{\theta \rho_1/2}(0)\subset \rr^m$ due to 
the bounded
image of the coordinates on the ambient manifold.
To handle this problem we consider the initial value $\init{\Phi}$.
We adjust the parameter $\theta$ to ensure that
the function $\init{\Phi}$ lies in the cylinder \\
$[-\tilde{\theta} \rho_0/4,\tilde{\theta} \rho_0/4] 
\times B_{\tilde{\theta} \rho_0/4}(0) \subset \rr \times\rr^{n}$.
This can be done by the following consideration.
From assumptions \eqref{eq:assum_bd_pt} 
for $\Phi$ we obtain via a Taylor expansion 
\begin{gather}
  \label{eq:norm_est}
  \abs{\Phi(z)}_e \le r(C_{w_0} + \tilde{C}_2^{\varphi} \tfrac{r}{\rho_1})
  \quad\text{if } z \in \rr^m \text{ with }\abs{z} =: r < \rho_1.
\end{gather}
We shrink the parameter
$\theta$ such that 
$\abs{\Phi(z)}_e < \tilde{\theta} \rho_0/4$ 
within the ball $B^e_{\theta \rho_1/2} \subset \rr^m$.
This choice of $\theta$ ensures that there is enough space
for time evolution. 

Now we choose the constant $\rho$ small enough such that $(\varphi_0, \varphi_1)
\in W$, then it follows $(\varphi_0, D\varphi_0, \varphi_1) \in \Omega$
pointwise. This gives us a definition of $W$ analog to \eqref{eq:domain}.

In the sequel we adopt the notion for the bounds $D_0, D_1, \tilde{D}_0$
and $\tilde{D}_1$ in \eqref{eq:sob_initial} and the definitions 
of $K_{0,s}$ and $K_{1,s}$ in
\eqref{eq:Hs_est_space_full} and \eqref{eq:Hs_est_time_full}.
To obtain estimates on the argument $w(t) + \varphi_0$ of the 
matrix $\hat{h}_{AB}$
and its Christoffel symbols we needed 
to introduce a  bound $T_1$ on the time parameter $t$.
For notational convenience we introduce for $t \le T_1$ and $\varphi_0
\in B_{\rho}(\init{\Phi})$
\begin{multline}
  \label{eq:arg_est}
  \norm{w(t) + \varphi_0}_{s,\mathrm{ul}} \le c \bigl(
    \abs{Dw_0}^2 + T_1^2 \abs{w_1}^2
    + \norm{\varphi_0}_s^2 \bigr)^{\2}
    \\
    \le c\bigl( C_{w_0}^2 + T_1^2 (\tilde{C}_0^{\chi})^2 + \rho^2 
    + D_0^2 \bigr)^{\2} =: \ca.
\end{multline}
The constant $c$ depends on the test function used to define the spaces
$H^s_{\mathrm{ul}}$ (see beginning of section \ref{sec:asym}).

To derive Sobolev estimates for the coefficients and the RHS,
 bounds for the Sobolev norm
of the interpolated metric $\hat{h}_{AB}$ and its Christoffel symbols
$\hg_{BC}^A$ are needed. These will be obtained by local bounds for derivatives
 of $\hat{h}_{AB}(z)$ for $z \in \rr^{n+1}$ to be shown in the next lemma.
\begin{lem}
  \label{lem:infty_h}
  The following estimates hold for the matrix $\hat{h}_{AB}$ within
  $\rr^{n+1}$
  \begin{subequations}
    \begin{align}
      \label{eq:infty_h_0}
      \abs{(\hat{h}_{AB})}  & \le n + 1 + \delta_0 \tilde{\theta}
  =: C^{\hat{h}}_0, & \abs{(\hat{h}^{AB})} &\le \Delta_{\hat{h}}^{-1}, &
  \abs{(D\hat{h}_{AB})}  &\le C^{\hat{h}}_{d,0},
  \\
  \label{eq:infty_dh_s}
  \abs{(D\hat{h}_{AB})}  & \le C^{\hat{h}}_{d,0}, &
  \abs{(\hat{h}_{AB})}_{C^s}  &\le  C^{\hat{h}}_s,  &
 \abs{(D\hat{h}_{AB})}_{C^{s}} & \le  C^{\hat{h}}_{d,s}.
\end{align}
  \end{subequations}
\end{lem}
\begin{proof}
  The first bound in \eqref{eq:infty_h_0} can be derived from the condition for 
  the difference $h_{AB} - \eta_{AB}$
  in \eqref{eq:cond_pos}. The same condition yields that
   lemma
  \ref{lem:est_metric_inverse} is applicable providing the second bound.

  Estimates for higher derivatives of $\hat{h}_{AB}$
  follow from the bounds for the cut-off function $\tilde{\zeta}$
  in \eqref{eq:cutoff_est} and from condition \eqref{eq:cond_der} for 
  derivatives of the metric 
  $h_{AB}$
  within the ball $B^e_{\rho_0}(0) \subset \rr^{n+1}$.
\end{proof}
These bounds can be used to establish Sobolev norm estimates for
the matrix $\hat{h}_{AB}$ occurring and its Christoffel symbols. 
\begin{lem}
  \label{lem:hg_Hs_est}
  Let $t \le T_1$ and  $\varphi_0, \psi_0 \in B_{\rho}(\init{\Phi}) 
  \subset H^{s+1}$.
  Then the matrix 
  $\hat{h}_{AB}$
  and its Christoffel
  symbols $\hg_{BC}^A$ satisfy
  \begin{subequations}
    \begin{align}
  \label{eq:hath_Hs_est}
  \bnorm{\bigl(\hat{h}_{AB}\bigl( w(t) + \varphi_0\bigr)\bigr)}_{e,s,\mathrm{ul}} 
  &  \le 
  c\, C_s^{\hat{h}} (1 + \ca^{s}) & =:&\tilde{K}_h,
  \\
  \label{eq:hath_inverse_est}
  \bnorm{\bigl(\hat{h}^{AB}\bigl( w(t) + \varphi_0\bigr)\bigr)}_{e,s,\mathrm{ul}} 
  &  \le c \,\Delta_{\hat{h}}^{-1}\bigl(1 
    + (\Delta_{\hat{h}}^{-1} \bigl(C_s^{\hat{h}} (1 + \ca^s)
    \bigr)^s\bigr) & =:&K_h,
  \\
  \label{eq:hg_Hs_est}
   \bnorm{\bigl(\hg_{BC}^A\bigl( w(t) + \varphi_0\bigr)\bigr)}_{e,s}  
   & \le c \,K_h
    C_{d,s}^{\hat{h}} (1 + \ca^s) & =:&C_{\hg,s},
    \\
  \label{eq:hg_0_est}
   \bnorm{\bigl(\hg_{BC}^A\bigl( w(t) + \varphi_0\bigr)\bigr)}_{e}  
   & \le 2 \Delta_{\hat{h}}^{-1} C_{d,0}^{\hat{h}} & =:&C_{\hg,0}
 \end{align}
\end{subequations}
with $\ca$ defined in \eqref{eq:arg_est}.
Further, the following Lipschitz estimates hold
\begin{subequations}
    \begin{align}
  \label{eq:hath_Hs_lip}
  \bnorm{\bigl(\hat{h}_{AB}\bigl( w(t) + \varphi_0\bigr)\bigr)  -
  \bigl(\hat{h}_{AB}\bigl( w(t) + \psi_0\bigr)\bigr)}_{e,s-1} 
  &\le  \lip_{\hat{h}, s-1} \norm{\varphi_0 - \psi_0}_{s-1} 
  \\
  \label{eq:hath_0_lip}
  \babs{\bigl(\hat{h}_{AB}\bigl( w(t) + \varphi_0\bigr)\bigr)  -
  \bigl(\hat{h}_{AB}\bigl( w(t) + \psi_0\bigr)\bigr)}_{e} 
  &\le  C^{\hat{h}}_{d,0}   \abs{\varphi_0 - \psi_0} 
  \\
  \label{eq:hath_t_lip}
  \bnorm{\bigl(\hat{h}_{AB}\bigl( w(t) + \varphi_0\bigr)\bigr)  -
  \bigl(\hat{h}_{AB}\bigl( w(t') + \varphi_0\bigr)\bigr)}_{e,s-1} 
  &\le \lip_{\hat{h}, t}
  \abs{t - t'}
  \\
  \label{eq:hg_Hs_lip}
  \bnorm{\bigl(\hg_{BC}^A\bigl( w(t) + \varphi_0\bigr)\bigr)
  - \bigl(\hg_{BC}^A\bigl( w(t) + \psi_0\bigr)\bigr)}_{e,s-1}  
  & \le \lip_{\hg, s-1}
  \norm{\varphi_0 - \psi_0}_{s-1} 
  \\
    \label{eq:hg_t_lip}
  \bnorm{\bigl(\hg_{BC}^A\bigl( w(t) + \varphi_0\bigr)\bigr)
  - \bigl(\hg_{BC}^A\bigl( w(t') + \varphi_0\bigr)\bigr)}_{e,s-1}  
  & \le \lip_{\hg, t}\abs{t - t'}
  \\
  \label{eq:hg_0_lip}
  \babs{\bigl(\hg_{BC}^A\bigl( w(t) + \varphi_0\bigr)\bigr)
  - \bigl(\hg_{BC}^A\bigl( w(t) + \psi_0\bigr)\bigr)}_{e} 
  & \le \lip_{\hg, 0} \abs{\varphi_0 - \psi_0}
\end{align}
\end{subequations}
where the occurring constants are defined by
\begin{align*}
   \lip_{\hat{h}, s-1} & = c \, C^{\hat{h}}_{d,s} \bigl( 1 + (\ca^2 
   + 3 (\rho^2 + D_0^2))^{s/2}
  \bigr)
   \\
   \lip_{\hat{h}, t} &= c\,C_{d,s}^{\hat{h}} (\tilde{C}_0^{\chi}) 
   (1 + (\ca^2 + 3T_1^2(\tilde{C}_0^{\chi})^2)^{s/2})
   \\
   \lip_{\hg, s-1} &= c\, K_h^2 \tilde{\theta}_h C_{d,s}^{\hat{h}} (1 + \ca^{s})
   + c\,K_h C_{d^2,s}^{\hat{h}} (1 + (\ca^2 + 3(\rho^2 + D_0^2))^{s/2})
   \\
   \lip_{\hg, t} & = c\, K_h^2 \lip_{\hat{h}, t} C_{d,s}^{\hat{h}} (1 + \ca^{s}) +
   c\,K_h C_{d^2,s}^{\hat{h}} (\tilde{C}_0^{\chi}) 
   (1 + (\ca^2 + 3T_1^2(\tilde{C}_0^{\chi})^2)^{s/2})
   \\
   \lip_{\hg, 0} &= 2 \bigl(\Delta_h^{-2} 
   (C_{d,0}^{\hat{h}})^2 +  \Delta_{\hat{h}}^{-1} C_{d^2,0}^{\hat{h}}
  \bigr).
\end{align*}
\end{lem}
\begin{proof}
  The bound \eqref{eq:hath_Hs_est} follows from 
  Lemma \ref{lem:comp_ul_est} using the 
  $C^s$-bounds stated in the previous lemma. The bound for the inverse
  follows from Lemma \ref{inverse}. Again, Lemma \ref{lem:comp_ul_est}
  yields the third estimate making use of the $C^s$-bounds for
  the derivative of the metric $\hat{h}_{AB}$.
  
  To obtain the Lipschitz estimates we use a similar device as
  in \cite{Kato:1975} section 5.2.
  Estimate \eqref{eq:hath_0_lip} follows directly from the 
  bound for $D\hat{h}_{AB}$ stated in
  \eqref{eq:infty_h_0}.
  Inequality \eqref{eq:hath_Hs_lip} can be derived 
  from a device similar to the proof
  of Lemma \ref{lem:comp_ul_est} regarding the $C^{s-1}$-bound
  for the derivative of the 
  metric $D\hat{h}_{AB}$ stated in \eqref{eq:infty_dh_s} and
  the estimate for the argument $w(t) + \varphi_0$ stated in
  \eqref{eq:arg_est}. To treat the argument of $D\hat{h}_{AB}$,
  namely $\lambda \varphi_0 + (1 - \lambda) \psi_0$ for $0 \le \lambda \le 1$
  we have to insert $2(\rho^2 + D_0^2)^{\2}$ instead of $(\rho^2 + D_0^2)^{\2}$.
  Therefore, an additional term appears in the estimate.
  The Lipschitz estimate \eqref{eq:hath_t_lip} w.r.t.\ the time parameter $t$ 
  follows from the same considerations. Here we have to use
  twice the radius to estimate the argument $\lambda w(t) + (1 - \lambda)
  w(t')$ which gives the additional term $3T_1^2 (\tilde{C}_0^{\chi})^2$.

  The Lipschitz estimates for the Christoffel symbols follow
  from considering  a generic norm 
  \begin{multline*}
    \babs{\bigl(\hg_{BC}^A\bigl( w(t) + \varphi_0\bigr)\bigr)
      - \bigl(\hg_{BC}^A\bigl( w(t) + \psi_0\bigr)\bigr)}_{e}
    \\
    \le 2\babs{\bigl(\hat{h}^{AB}(w(t) + \varphi_0)\bigr) - 
      \bigl(\hat{h}^{AB}(w(t) + \psi_0)\bigr)}_{e}
    \babs{\bigl(D\hat{h}_{AB}(w(t) + \psi_0)\bigr)}_{e}
    \\
    {} + 2\babs{\bigl(\hat{h}^{AB}(w(t)
      + \varphi_0)\bigr)}_{e} \babs{\bigl(D\hat{h}_{AB}(w(t) + \varphi_0)\bigr)
      - \bigl(D\hat{h}_{AB}(w(t) + \psi_0)\bigr)}_{e}
  \end{multline*}
  Inserting the bounds and Lipschitz estimates for the metric
  $\hat{h}_{AB}$, namely the inequalities
  \eqref{eq:hath_inverse_est}, \eqref{eq:hath_Hs_lip} and 
  \eqref{eq:hath_0_lip}, 
  we derive the estimates  \eqref{eq:hg_Hs_lip}
  and \eqref{eq:hg_0_lip}.
  A similar generic estimate  and Lipschitz estimate
  \eqref{eq:hath_t_lip}  yield the remaining  estimate \eqref{eq:hg_t_lip}. 
  The second term includes a Lipschitz estimate for $D\hat{h}_{AB}$
  which can be established as the Lipschitz estimate for $\hat{h}_{AB}$
  using the constant $C_{d^2, s}^{\hat{h}}$ instead of $C_{d,s}^{\hat{h}}$.
\end{proof}
The preceding lemma  provides us with Sobolev norm bounds for the matrix
$g_{\mu\nu}^{\mathrm{a}}$ and the coefficients.
\begin{lem}
  \label{lem:Hs_metric_est_hyp}
  The Lemmata \ref{lem:Hs_metric} and \ref{lem:Hs_metric_est} remain valid if 
  the constants are changed in the following way
  \begin{align}
    \tilde{K} & = \tilde{K}_h (K_{0,s}^2 + K_{1,s}^2),\nonumber
    \\
    \tilde{\theta} & = 2^{3}\, \tilde{K}_h (K_{0,s} + K_{1,s}) 
    + 2^3\, \Lip_{\hat{h},s-1} (K_{0,s}^2 + K_{1,s}^2)\nonumber
     \\
    \tilde{\theta}' & = 2^{3}\, C^{\hat{h}}_0 (K_0 + K_1)
    + 2^3\,C^{\hat{h}}_{d,0} (K_0^2 + K_1^2), \nonumber
    \\
    \label{eq:coeff_const}
  K & = c \Delta^{-1}\bigl(1 + (\Delta^{-1} \tilde{K})^s\bigr), 
  \quad     \theta  =  K^2 \tilde{\theta} \quad\text{and}
  \quad
    \theta'  = \Delta^{-2} \tilde{\theta}'.
  \end{align}
\end{lem}

\begin{proof}
  The first claim follows from the proof of Lemma \ref{lem:Hs_metric} taking 
  care
  of the additional terms involving the interpolated metric $\hat{h}_{AB}$
  and its Christoffel symbols
  which can be estimated by the Sobolev estimates obtained in lemma
  \ref{lem:hg_Hs_est}.

  In the same way as in the proof of Lemma \ref{lem:Hs_metric_est}, estimates 
  for
  the coefficients follow from estimates for the matrix $g_{\mu\nu}^{\mathrm{a}}$.
\end{proof}
The Lemmata \ref{lem:norm_metric_inverse_hyp} and 
\ref{lem:Hs_metric_est_hyp} yield  that the coefficients $g^{\mu\nu}_{\mathrm{a}}$
satisfy the conditions 
\eqref{cond_qlin1} to \eqref{cond_qlin_add3} and \eqref{cond_qlin4}
of the existence
Theorem \ref{qlin_ex}.

In Minkowski space there were no explicit dependency of the coefficients
on the time
parameter $t$. Here,
the coefficients need to satisfy condition \eqref{cond_qlin_t_lip}.
\begin{lem}
  \label{lem:coeff_lip_t}
  Let $t \le T_1$ and $(\varphi_0, \varphi_1) \in W$. 

  Then the 
  matrix $g_{\mu\nu}^{\mathrm{a}}$ %
  satisfies
  \begin{gather*}
    \norm{(g_{\mu\nu}^{\mathrm{a}}(t,\varphi_0, \varphi_1)) 
      - (g_{\mu\nu}^{\mathrm{a}}(t',\varphi_0, \varphi_1))}_{e,s-1,\mathrm{ul}}
    \le \tilde{\nu} \abs{t - t'}
  \end{gather*}
  with $\tilde{\nu} = \lip_{\hat{h}, t}(K_{0,s}^2 + K_{1,s}^2)$.
\end{lem}
\begin{proof}
  Since $\chi_0 + \varphi_1$ and $Dw_0 + D\varphi_0$ do not depend on
  the time parameter it suffices to obtain an estimate on the difference of 
  the metric
  $\hat{h}_{AB}$ which is provided by inequality \eqref{eq:hath_t_lip}.
\end{proof}
\begin{rem}
  The Lipschitz constant $\nu$ for the coefficients
  $g^{\mu\nu}_{\mathrm{a}}$ of the asymptotic equation  \eqref{eq:geom_solve_asym}
  meeting the condition \eqref{cond_qlin_t_lip}
  is given by $K^2 \tilde{\nu}$. %
\end{rem}

We will now give estimates for the RHS of equation
\eqref{eq:geom_solve_asym_hyp}.
The first part $f_{\mathrm{a}}$, defined analog to the RHS of the asymptotic 
equation \eqref{eq:geom_solve_asym} in Minkowski space, can be estimated as in 
lemma
\ref{lem:Hs_rhs_est} using the estimates for the coefficients provided that 
there exist bounds for the Christoffel 
symbols  $\hat{\gamma}_{\mu\nu}^{\lambda}$ of the metric $\hat{a}_{\mu\nu}$.
These bounds will be derived in the next lemma.
\begin{lem}
  \label{lem:back_chr_hyp}
  The Christoffel symbols of the matrix $\hat{a}_{\mu\nu}$  satisfy
  the following inequalities
    \begin{gather*}
    \norm{(\hat{\gamma}_{\mu\nu}^{\lambda})}_{e,s}  \le C^{\hat{\gamma}}  
    \quad
    \text{and}
    \quad
    \abs{(\hat{\gamma}_{\mu\nu}^{\lambda})}_e  \le C^{\hat{\gamma}}_0,
  \end{gather*}
  where 
  \begin{multline*}
    (2 K)^{-1} C^{\hat{\gamma}}  =  \tilde{D}_1 \tilde{K}_h 
    \bigl(2 ((\tilde{C}_0^{\chi})^2 + 
  D_1^2)^{\2} + 4 (C_{w_0}^2 + D_0^2)^{\2} \bigr) 
  \\
  {}+
  \tilde{D}_0 \tilde{K}_h\bigl(4 ((\tilde{C}_0^{\chi})^2 + 
  D_1^2)^{\2} + 2(C_{w_0}^2 + D_0^2)^{\2} \bigr)
  + c \, C_{d,s}^{\hat{h}}(1 + \ca^s)\bigl( ((\tilde{C}_0^{\chi})^2 + 
  D_1^2)
  \\
  {}+  (C_{w_0}^2 + D_0^2) + 4 ((\tilde{C}_0^{\chi})^2 + 
  D_1^2)^{\2} (C_{w_0}^2 + D_0^2)^{\2}\bigr)
  \end{multline*}
  and $C_0^{\hat{\gamma}}$ arises from $C^{\hat{\gamma}}$ by applying the replacements
  \begin{gather*}
    K \mapsto \Delta^{-1},~ \tilde{K}_h \mapsto \Delta_{\hat{h}}^{-1},~
    (D_0, D_1) \mapsto (\delta_1, \delta_2),~ c \, C_{d,s}^{\hat{h}}(1 + \ca^s)
    \mapsto C^{\hat{h}}_{d,0}
  \end{gather*}
  and $\tilde{D}_0, ~ \tilde{D}_1$ are replaced by
  the bounds for $\abs{D\init{\chi}}_e$  and $\abs{D^2 \init{\Phi}}_e$ stated in
  \eqref{eq:initial_d2_0}.
\end{lem}
\begin{proof}
  The proof follows by considering the additional
  terms $\norm{(\hat{h}_{AB})}_{e,s,\mathrm{ul}}$ and \\ 
  $\norm{(D\hat{h}_{AB})}_{e,s}$
  occurring in $\norm{D\hat{a}}_{e,s,\mathrm{ul}}$.
  The first term was estimated in \eqref{eq:hath_Hs_est}.
  Lemma \ref{lem:comp_ul_est} and the $C^s$-bound for
  $D\hat{h}_{AB}$ stated in \eqref{eq:infty_dh_s} yield
  the following estimate
  \begin{gather*}
    \norm{(D\hat{h}_{AB})}_{e,s} \le c \, C_{d,s}^{\hat{h}}(1 + \ca^s).
  \end{gather*}
  Further, we make use of the $H^s$-bounds $\tilde{D}_0$ and $\tilde{D}_1$
  for $D\init{\chi}$ and $D^2 \init{\Phi}$ as introduced in 
  \eqref{eq:sob_initial} and the bounds
  \begin{gather*}
    \norm{Dw_0 + Dv_0}_{s,\mathrm{ul}} \le  (C_{w_0}^2 + D_0^2)^{\2}
    \quad\text{and}\quad \norm{\chi_0 + \init{\chi}}_{s,\mathrm{ul}}
    \le ((\tilde{C}_0^{\chi})^2 + 
    D_1^2)^{\2}. 
  \end{gather*}
  The replacement is possible if we apply the local bounds 
  for $\hat{h}_{AB}$ stated
  in \eqref{eq:infty_h_0}.
\end{proof}
The Lemmata \ref{lem:Hs_rhs_est} and \ref{lem:back_chr_hyp} yield that $f_{\mathrm{a}}$ 
satisfies the assumptions
\eqref{cond_qlin1} and  \eqref{cond_qlin_add1} to \eqref{cond_qlin3} of the 
existence Theorem \ref{qlin_ex}. \\
To obtain estimates for the additional term $\BF_{\mathrm{a}}$, we first state
generic estimates to be proved similar to the ones in Lemma \ref{lem:gen_est}.
\begin{lem}
  For two functions $F$ and $\bar{F}$ the following inequalities
  hold %
  \begin{subequations}
    \begin{align}
      \label{eq:gen_BF_0}
       \abs{\BF(&F)}  \le
      \babs{\bigl(g^{\mu\nu}(F)\bigr)} 
      (\abs{\partial_t F}^2 + \abs{DF}^2) \babs{\bigl(\hg_{BC}^A(F)\bigr)}
      \\
      \label{eq:gen_BF_lip}
      \abs{\BF(F) - \BF(&\bar{F})}  \le   
      \babs{\bigl(g^{\mu\nu}(F)\bigr)
        - \bigl(g^{\mu\nu}(\bar{F})\bigr)}\,(\abs{\partial_t \bar{F}}^2 + 
      \abs{D\bar{F}}^2) \babs{\bigl(\hg_{BC}^A(\bar{F})\bigr)}
      \nonumber\\
      & {}+ \babs{\bigl(g^{\mu\nu}(F)\bigr)} (\abs{\partial_t V}
      + \abs{DV})
      (\abs{\partial_t \bar{F}}^2 + \abs{D\bar{F}}^2)^{\2} 
      \babs{\bigl(\hg_{BC}^A(\bar{F})\bigr)}
      \nonumber\\
      & {}+ \babs{\bigl(g^{\mu\nu}(F)\bigr)} 
      (\abs{\partial_t F}^2 + \abs{DF}^2)^{\2} (\abs{\partial_t V}
      + \abs{DV})
      \babs{\bigl(\hg_{BC}^A(\bar{F})\bigr)}
      \nonumber\\
      & {}+ \babs{\bigl(g^{\mu\nu}(F)\bigr)} 
      (\abs{\partial_t F}^2 + \abs{DF}^2) \babs{\bigl(\hg_{BC}^A(F)\bigr)
        - \bigl(\hg_{BC}^A(\bar{F})\bigr)},
    \end{align}
  \end{subequations}
  where we set $V = F - \bar{F}$.
\end{lem}
\begin{rem}
  As for the asymptotic equation in Minkowski space this estimates will
  be used with the replacements given in \eqref{eq:replace}.
\end{rem}
The next lemma establishes statements analog to Lemma \ref{lem:Hs_rhs_est} for
the additional term $\BF_{\mathrm{a}}$ from which the conditions of
Theorem \ref{qlin_ex} can be obtained.
\begin{lem}
  \label{lem:hyp_rhs_add_est}
  Let $(\varphi_0, \varphi_1), (\psi_0, \psi_1) \in W$.
  
  Then the term $\BF_{\mathrm{a}}$ satisfies the following estimates
\begin{align*}
    \norm{\BF_{\mathrm{a}}(t,\varphi_0, \varphi_1)}_{s} & \le K_{\BF},\\
    \norm{\BF_{\mathrm{a}}(t,\varphi_0, \varphi_1) 
      - \BF_{\mathrm{a}}(t,\psi_0, \psi_1)}_{s-1} & 
    \le \theta_{\BF}' 
    E_s\bigl( (\varphi_0, \varphi_1)- (\psi_0, \psi_1)\bigr)  \\
    \norm{\BF_{\mathrm{a}}(t,\varphi_0, \varphi_1) 
      - \BF_{\mathrm{a}}(t,\psi_0, \psi_1)}_{L^2} & 
    \le \theta_{\BF} E_1\bigl( (\varphi_0, \varphi_1)- (\psi_0, \psi_1)\bigr)
    \\
    \norm{\BF_{\mathrm{a}}(t,\varphi_0, \varphi_1) - \BF_{\mathrm{a}}(t', \varphi_0, \varphi_1)}_{s-1}
    & 
    \le \nu_{\BF} \abs{t - t'}
  \end{align*}
  with
  \begin{eqnarray*}
    K_{\BF} & = & C_{\hg,s} K (K_{1,s}^2 + K_{0,s}^2), 
    \\
    \nu_{\BF} & = & K^2 \tilde{\nu} C_{\hg,s}(K_{1,s}^2 + K_{0,s}^2) 
    + K (K_{1,s}^2 + K_{0,s}^2) \lip_{\hg, t},
    \\
    \theta_{\BF}' & = &  \theta C_{\hg,s} (K_{1,s}^2 + K_{0,s}^2)
    + 2K C_{\hg,s} (K_{1,s}^2 + K_{0,s}^2)^{\2}
    + K (K_{1,s}^2 + K_{0,s}^2) \lip_{\hg,s-1}
  \end{eqnarray*}
  and $\theta_{\BF}$ arises from $\theta_{\BF}'$ by applying the replacements
  \begin{gather*}
    \theta \mapsto \theta', ~ (K_{0,s}, K_{1,s}) \mapsto (K_0, K_1),~
    C_{\hg,s} \mapsto C_{\hg,0},~
    \lip_{\hg,s-1} \mapsto \lip_{\hg, 0}\text{ and } K \mapsto 
    \Delta^{-1}.
  \end{gather*}
\end{lem}

\begin{proof}
  The inequalities follow from the generic
  estimates \eqref{eq:gen_BF_0} and \eqref{eq:gen_BF_lip}
  using the estimates for the Christoffel symbols $\hg_{BC}^A$
  as in \eqref{eq:hg_Hs_est}, \eqref{eq:hg_Hs_lip}
    and the corresponding local estimates 
    \eqref{eq:hg_0_est} and \eqref{eq:hg_0_lip}.
    The Lipschitz constant $\nu_{\BF}$ is obtained by inserting the Lipschitz
    constants for $\hg_{BC}^A$ and $g_{\mu\nu}^{\mathrm{a}}$ stated in \eqref{eq:hg_t_lip}
    and Lemma \ref{lem:coeff_lip_t}.
\end{proof}
\begin{rem}
  \label{rem:satisfy}
  An examination of the proofs of the Lemmata \ref{lem:hg_Hs_est}, 
  \ref{lem:Hs_metric_est_hyp}
  and the preceding lemma
  reveals that the statement of Lemma \ref{lem:uni_satisfy} remains valid.
  Therefore, the local conditions \eqref{eq:loc_cond_uni} needed to apply the 
  local uniqueness
  Theorem \ref{loc_uni} to the asymptotic equation  
  \eqref{eq:geom_solve_asym_hyp} are satisfied.
\end{rem}
We will now give a proof for the main result of this section, an analog
to Theorem \ref{thm:ex_uni_atlas}. In the treatment of the Minkowski case
the main result followed from an analog to Proposition \ref{prop:cutoff_ex_hyp}.
Here, we need to apply one step more to be sure that solutions obtained
by this proposition indeed satisfy the membrane equation.
\begin{proof}[\textbf{Proof of Proposition \ref{prop:cutoff_ex_hyp}}]
  The Lemmata \ref{lem:norm_metric_inverse_hyp},
  \ref{lem:Hs_metric_est_hyp}, \ref{lem:back_chr_hyp}
  and \ref{lem:hyp_rhs_add_est} show that the conditions
  of the existence Theorem \ref{qlin_ex} are satisfied for 
  the asymptotic equation \eqref{eq:geom_solve_asym_hyp}.
  From Theorem \ref{asym_ex} we derive a solution $F$ to the IVP
  \eqref{eq:geom_solve_hyp} with the properties stated in 
  \eqref{eq:property_hyp}.
  Applying Remark \ref{rem:extend} provides us with a solution for negative
  values of $t$.
  
  Regarding Remark \ref{rem:satisfy} uniqueness of solutions can be shown
  as in the proof of Proposition \ref{prop:cutoff_ex} by applying the local
  uniqueness result \ref{loc_uni}.
\end{proof}
\begin{rem}
  \label{rem:low_bound_hyp}
  \begin{enumerate}
  \item 
    \label{stay_W}
    Remark \ref{rem:stay_W} remains valid. All estimates 
    for the coefficients and the RHS are valid for $t > 0$ and $t < 0$.
  \item   As in Remark \ref{rem:geom_lower} we
    get a lower bound for the existence
  time of the solution. The constant $c_E$ has to be changed according
  to the new definition of the constants $\lambda$ and $\mu$ in
  \eqref{eq:hyp_est_lambda} and \eqref{eq:hyp_est_mu}.
  Therefore, some factors of $(1 + 
    \delta_0\tilde{\theta})$ occur in the definition of $c_E$.
  \end{enumerate}
\end{rem}
\begin{rem}
  In contrast to the Minkowski case the coefficients and RHS of 
  the asymptotic equation depend on the time 
  parameter 
  through the components of the metric on the ambient manifold
  (cf. \eqref{eq:def_coeff_asym_hyp}).

  It suffices to consider the components of the metric $h_{AB}$
  since the coefficients and the RHS differ from the coefficients
  and the RHS in the Minkowski case with terms depending 
  on this metric.
  From the assumption \eqref{assumptions_atlas_hyp} we get that 
  derivatives of $h_{AB}$ are bounded up to order $s + 3 + \ell_0$.
  This yields that derivatives of order $0 \le \ell \le \ell_0$ admit the bounds
  stated in Lemma \ref{lem:infty_h}. From the fact that the Sobolev-norm 
  estimates derived in
  Lemma \ref{lem:hg_Hs_est} are shown by only using the bounds of
  Lemma \ref{lem:infty_h} we obtain the desired result from corollary
  \ref{cor:sol_diff}.
\end{rem}
In the sequel we will proof Theorem 
\ref{thm:ex_uni_atlas_hyp}. In contrast to the Minkowski case
existence does not follow immediately from Proposition \ref{prop:cutoff_ex_hyp},
since the solution can leave the image of the coordinates on the 
ambient manifold.
 \begin{proof}[\textbf{Proof of the existence claim in Theorem
    \ref{thm:ex_uni_atlas_hyp}}]
  In the proof of Proposition \ref{prop:cutoff_ex_hyp} the
  parameter $T_1$ was left free (cf. definition \eqref{eq:def_coeff_op_hyp}). 
  Set
  \begin{gather}
  \label{eq:choice_t}
   T_1:= K_1^{-1} \, 
 \tilde{\theta} \rho_0 /4.
 \end{gather}
 Then it follows that a solution $F(t,z)$ 
  obtained from Proposition \ref{prop:cutoff_ex_hyp}
  lies in $B^e_{\tilde{\theta} \rho_0/2}(0)
  \subset \rr^{n+1}$ for $(t,z) \in [0,T_1] \times B^e_{\theta \rho_1/2}(0)
  \subset \rr \times \rr^m$.
  Therefore it satisfies the unmodified reduced membrane equation in that 
  region.
  The same holds for negative parameter $t$ provided $- t\le T_1$.
\end{proof}

\begin{proof}[\textbf{Proof of the uniqueness claim in Theorem
    \ref{thm:ex_uni_atlas_hyp}}]
  The claim follows from similar considerations as the proof of the uniqueness 
  part of  Theorem \ref{thm:ex_uni_atlas}.
  From Remark \ref{rem:satisfy} we conclude that the local uniqueness result
  \ref{loc_uni}
  is applicable. The slope $c_0$ of the cone on which uniqueness holds
  has to be adapted according to estimate \eqref{eq:slope} and lemma 
  \ref{lem:norm_metric_inverse_hyp}.
\end{proof}
\begin{rem}
  For later reference we state the exact value of the slope.
  It is given by
  \begin{multline}
    \label{eq:hyp_est_c_0}
    c_0 = 1 + m^{\2}\tilde{\mu}^{-1} K_1^{2}\bigl(
      2 \tilde{\lambda}^{-1}K_0 K_1(1 + 
    \delta_0\tilde{\theta}) + 1\bigr) 
    \\
    {}\cdot
    \bigl( 1+ m^{\2}\tilde{\mu}^{-1}
    K_0^2 (1 + \delta_0 \tilde{\theta})\bigr) (1 + 
    \delta_0\tilde{\theta}).
  \end{multline}
\end{rem}

\subsection{Gluing local solutions}
\label{lorentz_ex}
In this section we will develop solutions to the Cauchy problem
\eqref{eq:param_ivp} analog to section \ref{sec:prop_graph_mink}.
Based on the observation that the membrane equation \eqref{eq:H0}
is invariant under scaling, i.e.
conformal transformations of the ambient manifold
with a constant, the result will be stated in a scale invariant manner.

Let $R > 0$ be the scale.
Recall the notations of the main problem \eqref{eq:geom_problem}.
Assume $M$ to be an
$m$-dimensional manifold and $\varphi: M \rightarrow N$ to be an immersion
of $\Sigma_0$.
The timelike vector field along $\varphi$ denoted by 
$\chi:M \rightarrow TN$ will serve as initial velocity combining
initial direction, lapse and shift.

Let $s > \tfrac{m}{2} + 1$ be an integer. 
We make the following uniform assumptions on the ambient manifold $(N,h)$
and the initial data $\varphi$ and $\chi$.
\begin{assum}
  \label{assumptions_hyp}
  The ambient manifold $(N,h)$ endowed with time function
  $\tau$ satisfies the following condition
    \begin{subequations}
  \begin{gather}
    \label{eq:assum_target}
    \begin{split}
      &\text{There exist constants $C_1, C_2, C_1^{\tau}, \dots,
      C_{s + 2}^{\tau}, C_0^N, \dots, 
      C_{s + 1}^N$}
    \\
    &\text{independent of $R$
      such that}
    \\
     &  C_1 \le  R^{-1} \psi \le C_2,\quad
       R^{2+\ell}\abs{\D^{\ell} \RRiem}_E \le C^N_{\ell} \quad \text{for} \quad
    0 \le \ell \le s+1 
    \\
    \text{and}\quad & 
    R^{1+\ell}\abs{\D^{\ell} (\D\tau)}_E 
    \le C_{\ell}^{\tau} \quad \text{for} \quad
    1 \le \ell \le s+2,
    \end{split}
  \intertext{where $\D (\D \tau)$ denotes the $(1,1)$-tensor
  obtained by applying the
  covariant derivative to the gradient of $\tau$.
  }
      \label{eq:assum_smf_hyp}
      \begin{split}
        &\text{There exist constants $\omega_1, C^{\varphi}_0, \dots,
          C^{\varphi}_s$ independent of $R$ such that}
    \\
       & \inf\{ - h( \gamma , \hT ) : \gamma \text{ timelike future-directed
        unit normal to }\Sigma_0\}\le \omega_1
      \\
      \text{and}   \quad &R^{\ell + 1}\abs{\hn^{\ell} \II}_{\ig,E} \le C_{\ell}^{\varphi}
      \text{ for }0 \le \ell \le s.
      \end{split}
    \end{gather}
      \begin{gather}
  \label{eq:assum_speed_hyp}
      \begin{split}
        & \text{There exist constants $L_1, L_2, L_3, C^{\chi}_1, \dots, 
          C_{s+1}^{\chi}$ independent of $R$ such that}
    \\
      &  - L_1 \le R^{-2} h(\chi, \chi ) \le - L_2,~ 
      - h\bigl( \tfrac{\chi}{\mabs{\chi}}, \hT \bigr) \le L_3
      \text{,
      where }\mabs{\chi}^2 = - h(\chi,\chi)
      \\
      &\text{and}   \quad R^{\ell}\abs{\hn^{\ell} \chi}_{\ig,E} \le C_{\ell}^{\chi}
      \text{ for }1 \le \ell \le s+1. 
      \end{split}
    \end{gather}
  \end{subequations}
\end{assum}
The following result states existence and uniqueness for solutions to
the Cauchy problem \eqref{eq:param_ivp} in analogy to Theorem 
\ref{thm:ex_uni_geom} in Minkowski space.
\begin{thm}%
  \label{thm:ex_uni_geom_hyp}
  Let the metric $h \in C^{s+3}$ and the time function $\tau \in C^{s+3}$
  of the ambient manifold $N$ satisfy the assumptions
  \eqref{eq:assum_target}. 
  Suppose the initial data  $\varphi \in C^{s+2}$ and $\chi \in C^{s+1}$
  satisfy 
  the assumptions \eqref{eq:assum_smf_hyp} and \eqref{eq:assum_speed_hyp}
  respectively.

  Then
  there exist a constant $C_0$ independent of $R$ and a $C^2$-solution 
  $F:[-RC_0,RC_0] \times M \rightarrow N$
  to the membrane equation 
  \begin{gather*}
      g^{\mu\nu} \partial_{\mu} \partial_{\nu} F^A - 
  \Gamma^{\lambda} \partial_{\lambda} F^A + g^{\mu\nu} \partial_{\mu} F^B
  \partial_{\nu} F^C \G_{BC}^A(F) = 0
  \end{gather*}
  in harmonic map gauge
  w.r.t.\ the background metric  defined in \eqref{eq:back_metric}
  attaining the initial
  values
  \begin{gather*}
    \restr{F} = \varphi \qquad \text{and}\qquad\restr{\dt F} = \chi.
  \end{gather*}

  If $F$ and $\bar{F}$ are two such solutions, then they coincide 
  for $- R\min(\bar{C}_1, \bar{C}_2)\le 
  t \le R\min(\bar{C}_1, \bar{C}_2)$. The constants
  $\bar{C}_1$ and $\bar{C}_2$ are given by the constants $T_0$
  and $T_1$ defined in respectively \eqref{eq:uni_cone_height} and 
  \eqref{eq:choice_t} at scale $R = 1$.
\end{thm}
\begin{rem}
  \label{rem:hyp_sol_diff}
  Let $s > \tfrac{m}{2} + 1 + \ell_0$. Then the solution 
  $F:[-RC_0,RC_0] \times M \rightarrow N$
  is of class $C^{2 + \ell_0}$.
\end{rem}
\begin{rem}
  \label{rem:hyp_non_uniform}
  The theorem applies to the situation, where the assumptions
  \ref{assumptions_hyp} are only valid in a neighborhood of the initial 
  submanifold
  $\Sigma_0$.
\end{rem}

The strategy will be to show the claim of the theorem for $R = 1$ first,
and then examine the behaviour of a solution to the membrane equation
under scaling.
However, we will start with the second step.
\begin{lem}
  \label{lem:scal_inv}
  The membrane equation \eqref{eq:H0} is invariant under scaling.
\end{lem}
\begin{proof}
  Let $(N,h)$ be a Lorentzian manifold and $\Sigma$ be a timelike
  submanifold of $N$ satisfying the membrane equation. Let $R > 0$ be
  a constant and define $\bar{h} := R^2 h$.
  Since the membrane equation is pointwise we pick a point $p \in \Sigma$
  and let $F: V \subset \rr^{m+1} \rightarrow N$ be an embedding of 
  $\Sigma$ in a neighborhood of $p$.
  As an embedding of a submanifold satisfying the membrane equation,
  $F^A := y^A \circ F$ satisfies equation \eqref{eq:membrane} for given 
  coordinates $y$ on $N$.
  
  We introduce the following coordinates on $\rr^{m+1}$ and $N$.
  Set $\bar{x}^{\mu} = R x^{\mu}$, where $x = (x^{\mu})$ denote the standard
  coordinates on $\rr^{m+1}$, and $\bar{y}^A = R y^A$.
  Let 
  \begin{gather*}
    \bar{F}^A(z) = \bar{y}^A \circ F \circ \bar{x}^{-1}(z)
    = R y^A \circ F \circ x(R^{-1}z) = R F^A(R^{-1}z).
  \end{gather*}
  We will show that $\bar{F}^A$ satisfies the membrane equation in 
  the coordinates $\bar{x}$ and $\bar{y}$ on 
  $(N, \bar{h})$. Then it follows that $\Sigma$ satisfies the membrane
  equation in $(N,\bar{h})$, since $\bar{F}^A$ is also an embedding of
  $\Sigma$.
  
  To this end we compute the ingredients of equation \eqref{eq:membrane}.
  The components of the metric $\bar{h}$ w.r.t.\ the coordinates
  $\bar{y}$ are
  \begin{multline*}
    \bar{h}_{AB}(w) = \bar{h}\bigl(\bar{\partial}_A(w), \bar{\partial}_B(w) \bigr)
    =R^2 h\bigl( R^{-1} dy^{-1}_{R^{-1}w}(e_A), 
    R^{-1} dy^{-1}_{R^{-1}w}(e_B)\bigr)
    = h_{AB}(R^{-1} w) 
  \end{multline*}
  where $h_{AB}$ are the components of the metric $h$ w.r.t.\ the chart $y$ on 
  $N$.
  With a similar consideration it follows that
  $\bar{\G}_{BC}^A(w) = R^{-1}\G_{BC}^A(R^{-1}w)$.
  
  These expressions yield for the induced metric $g = F^{\ast} h$
  and $\bar{g} = \bar{F}^{\ast} (\bar{h}_{AB})$ on $\rr^{m+1}$
  \begin{multline*}
    \bar{g}_{\mu\nu}(\bar{F})(z) = \partial_{\mu} \bar{F}^A(z) \bar{h}_{AB}
    \bigl(
    \bar{F}(z)
    \bigr) \partial_{\nu} \bar{F}^B(z) 
    \\
    = \partial_{\mu} F^A( R^{-1} z)
    h_{AB}\bigl(F^A( R^{-1} z) \bigr)
    \partial_{\nu} F^B(R^{-1} z) 
    = g_{\mu\nu}(F)( R^{-1} z).
  \end{multline*}
  Similar to the argument giving the representation for $\bar{\G}_{BC}^A$
  it follows for the contracted Christoffel symbols of $\bar{g}$ that
  \begin{gather*}
    \bar{\Gamma}^{\lambda}(\bar{F})(z)= R^{-1} \Gamma^{\lambda}(F)(R^{-1} z).
  \end{gather*}
  We derive that the mean curvature operator of $\bar{F}$,
  seen as a differential operator depending on the coordinates,
  computed in the coordinates $\bar{x}$ and $\bar{y}$ has the 
  form
  \begin{gather*}
    H\bigl(\bar{F}(z)\bigr) 
    = R^{-1} H\bigl(F(R^{-1}z)\bigr).
  \end{gather*}
  Hence, the desired result follows.
\end{proof}
After showing the scale invariance of the equation we will investigate
the behaviour of the assumptions of Theorem \ref{thm:ex_uni_geom_hyp} under 
scaling.
\begin{lem}
  \label{lem:scaling}
  Let $N$ be a Lorentzian manifold satisfying the assumptions
  \eqref{eq:assum_target} at scale $R = 1$. Set
  $\bar{h} = R^{2} h$. \\
  Then $(N,\bar{h})$ satisfies the assumptions \eqref{eq:assum_target} at
  scale $R > 0$.
\end{lem}

\begin{proof}
  To derive the scaling behaviour of the occurring terms we consider
  coordinates  $y^A$ and $\bar{y}^A = Ry^A$ on $N$.
  Coordinate derivatives give a factor of $R^{-1}$, therefore it follows
  from the definition of the lapse that it scales with $R^{1}$.
  The factors of $R$ cancel such that the constants $C_1$ and $C_2$
  are independent of $R$.
  The same argument shows that the curvature scales with $R^{-2}$ and
  every covariant derivative gives a factor of $R^{-1}$.
  We conclude that the terms involving the curvature and covariant derivatives
  of the time function satisfy the desired inequalities with constants
  independent of $R$.
\end{proof}
Based on the Lemmata 
\ref{lem:scal_inv} and 
\ref{lem:scaling} the next proposition will establish the claim of Theorem 
\ref{thm:ex_uni_geom_hyp} concerning
the scaling property of solutions to the Cauchy problem \eqref{eq:param_ivp}.
\begin{prop}
  \label{prop:scaling}
  Suppose Theorem \ref{thm:ex_uni_geom_hyp} holds at scale $R = 1$.
  Then it also holds for an arbitrary scaling constant $R > 0$.
\end{prop}
\begin{proof}
  Let $h,\tau, \varphi$ and $\chi$ satisfy the assumptions
  of Theorem \ref{thm:ex_uni_geom_hyp} at scale $R$.
  Let $\bar{h} = R^{-2} h$. From lemma  \ref{lem:scaling} we derive that
  $(N, \bar{h})$ satisfies the assumptions \eqref{eq:assum_target}
  at scale $R = 1$ and from the proof we infer that $\varphi$ and $\chi$
  also satisfy the assumptions \eqref{eq:assum_smf_hyp} and 
  \eqref{eq:assum_speed_hyp} at scale $R = 1$.
  
  Theorem \ref{thm:ex_uni_geom_hyp} at scale $R = 1$ yields that
  there exist a constant $C_0$ and a solution $\bar{F}:[-C_0, C_0] \times M
  \rightarrow N$ of the membrane equation \eqref{eq:membrane} in $(N, \bar{h})$ 
  attaining the initial
  values $\restrb{\bar{F}} = \varphi$ and $\restrb{\dtb \bar{F}} = \chi$.

  Set $F(R^{-1} t,p) = \bar{F}(\bar{t}, p)$ defined on $[-RC_0, RC_0] \times M$.
  Since the image remains the same we conclude from Lemma \ref{lem:scal_inv} 
  that $F$ also satisfies the membrane equation.
  Since $\restrb{\bar{F}} = \varphi$, it follows that $\restr{F} = \varphi$.
  The initial velocity of $F$ can be computed via
  \begin{gather*}
    \dt F(t,p) = \dt \bar{F}(R\bar{t}, p) = R^{-1} \dtb \bar{F}(R\bar{t}, p)
    = \chi(p) \quad\text{for } p \in M.
  \end{gather*}
  
  The considerations for the existence claim show that only a scaling
  of the time parameter is necessary which provides the uniqueness claim.
\end{proof}
The preceding result yields that it remains to investigate Theorem 
\ref{thm:ex_uni_geom_hyp} at scale $R = 1$.
Following the method of section \ref{sec:prop_graph_mink}
we will show in the next proposition 
that the special coordinates
of $N$ introduced in section \ref{sec:special_coord} 
satisfy the conditions \ref{assumptions_atlas_hyp} and the 
graph representation derived in section
\ref{sec:graph_repr_hyp} with help of the special coordinates satisfy the 
conditions \ref{assumptions_atlas}.
\begin{prop}
  \label{prop:graph_assum_hyp}
  Suppose $h, \tau \in C^{s+3}$, $\varphi \in C^{s+2}$ and $\chi\in C^{s+1}$ 
  satisfy the assumptions 
  \ref{assumptions_hyp}.
  Let $p \in M$. Suppose $x,y$ are coordinates for $M$ and special
  coordinates on $N$ introduced in section \ref{sec:special_coord}
  such that $y \circ \varphi \circ x^{-1}$ is the special
  graph representation obtained in section \ref{sec:graph_repr_hyp}
  with center $p$.
  
  Then $y$ and the representation of $h$ w.r.t.\ $y$ satisfy the conditions
  \ref{assumptions_atlas_hyp}. Further the coordinates $x,y$
  and the representations $\Phi$ and $\chi_{xy}$ w.r.t.\ $x$ and $y$ satisfy
  the conditions \ref{assumptions_atlas} with $\eta_{AB}$
  replaced by the representation of $h$.
\end{prop}
\begin{proof}
  We will first show that the special coordinates $y$ and the representation
  of $h$ in $y$ have the desired property. The image of the coordinates $y$ 
  contain a cylinder 
    $[-\rho_0, \rho_0] \times B^e_{\rho_0}(0) \subset \rr \times \rr^n$ 
    which follows from the estimates
    \eqref{eq:special_coord_est} and \eqref{eq:special_h00}.
    
    The choice of the parameters within the construction 
    estimate \eqref{eq:h_diff_est} yields for $\delta_0 = \tfrac{2}{3}$.
    Hence, condition \eqref{eq:cond_pos} is satisfied.
    Bounds for derivatives of the representation of the metric $h$ w.r.t.\ $y$
    can be obtained from \eqref{eq:special_metric_der_est}.
    We see
    that the assumptions \eqref{eq:assum_target} meet the conditions 
    \eqref{eq:high_der_assum} giving
    us the boundedness of the desired order of derivatives of $h$.

    We now head to the conditions \ref{assumptions_atlas}.
    The parts 
    \ref{assum_atlas_chart} and \ref{assum_atlas_graph}
    can be obtained by the same arguments as in the first step of the proof
    of Proposition \ref{prop:graph_assum} taking the Lemmata
    \ref{lem:graph_est_hyp}, \ref{lem:cond_normal_hyp} 
    and \ref{lem:hyp_u_high_der} into account.
    
    To derive the conditions for the initial velocity we first consider
    the second step of the proof of Proposition \ref{prop:graph_assum}
    establishing the boundedness of $\abs{\chi_{xy}}_e$.
    By using the representation of $h$ in the special coordinates 
    (cf. \eqref{eq:special_metric})
    we get from the assumptions \eqref{eq:assum_speed_hyp} in a similar way 
    the boundedness of $\abs{\chi_{xy}}_E$. Estimate \eqref{eq:metric_E_est} 
    for $E_{AB}$ now yields
    the desired boundedness of the Euclidean norm.

    Following part \ref{eq:high_vel} of the proof of proposition 
    \ref{prop:graph_assum} 
    we use Lemma \ref{lem:high_cov_der} to obtain the following identity for 
    derivatives of the initial velocity
    \begin{gather}
      \label{eq:high_der_cov_speed}
      \partial^{k} \chi_{xy} = \hn^{k} \chi_{xy}
      + \tsum \partial^{\alpha_1} \iG \ast \cdots 
      \ast\partial^{\alpha_p} \iG
      \ast \partial^{\beta_1} \indg \ast \cdots \ast \partial^{\beta_q} \indg
      \ast \hn^{\ell} \chi_{xy}.
    \end{gather}
    This expression can be estimated analog to a similar expression for
    the second fundamental form in Lemma \ref{lem:hyp_u_high_der}.
    Thus,  the desired bounds
    for derivatives of $\chi_{xy}$ up to order $s + 1$ follow.
\end{proof}
The next proposition shows that solutions obtained by Theorem
\ref{thm:ex_uni_atlas_hyp} are independent of the specific decomposition.
\begin{prop}
  \label{prop:uni_coord_hyp}
  Proposition \ref{prop:uni_coord} remains valid, if the slope $c_0$
  given by \eqref{eq:hyp_est_c_0} is used to define the cone $C$ and the cone is
  truncated to points $(t,x)$ satisfying $t \le T_1$ with the constant
  $T_1$ defined in \eqref{eq:choice_t}.
\end{prop}

\begin{proof}
  The choice of $T_1$ ensures that a solution $F(t,z)$ 
  gained from Proposition \ref{prop:cutoff_ex_hyp}
  lies in $B_{\tilde{\theta} \rho_0/2}(0)
  \subset \rr^{n+1}$ for $(t,z) \in [0,T_1] \times B^e_{\theta \rho_1/2}(0)
  \subset\rr \times \rr^m$.
  Therefore it satisfies the unmodified reduced membrane equation in that 
  region.
  Lemma \ref{lem:initial_G} yields that the proof of 
  Proposition \ref{prop:uni_coord} goes through using the uniqueness
  result from Theorem \ref{thm:ex_uni_atlas_hyp} and regarding
  Remark \ref{rem:extend}.
\end{proof}
We now head to the proof of the main result of this section providing 
existence to the Cauchy problem \eqref{eq:param_ivp}.
\begin{proof}[\textbf{Proof of the existence claim in Theorem
    \ref{thm:ex_uni_geom_hyp}}]
  We follow the strategy to show the existence claim in Theorem
  \ref{thm:ex_uni_geom}.
  The special graph representation satisfies the additional condition
  \eqref{cond_next} which is   stated in estimate 
  \eqref{eq:hyp_metric_compare}.
  The rest of the proof is analog to the proof of 
  Theorem \ref{thm:ex_uni_geom}
  using the corresponding results from Proposition \ref{prop:graph_assum_hyp},
  Theorem \ref{thm:ex_uni_atlas_hyp} and the preceding proposition.
\end{proof}

\begin{rem}
  \label{rem:constraint}
  The new constraint $T_1$ defined in \eqref{eq:choice_t} for the time 
  parameter has to be added to
  the statements of Remark \ref{rem:ex_time}.
  This yields the following lower bound for the existence time $T$.
  It holds that $T \ge \min(T', T_0, T_1)$
  with $T'$ derived from Proposition \ref{prop:cutoff_ex_hyp} and $T_0$ from 
  \eqref{eq:uni_cone_height}.
  
  The value of the constant $c_E$ (cf. Remark \ref{rem:special_lower} and  
  \eqref{eq:cE}), controlling the
  inverse of the constant $T'$, has to be changed
  according to the estimates of Lemma \ref{lem:norm_metric_inverse_hyp}.
  Not only the angles and the bound for the second fundamental form enter
  the existence time. Through $\rho_0$, also the bounds for the curvature and
  the gradient of the time function are involved.
\end{rem}

\begin{rem}
  \label{rem:ex_time_hyp}
  As in the context of the Minkowski space the construction of a solution
  to the IVP \eqref{eq:param_ivp} provides us with a lower bound
  for the proper time of timelike curves (cf. Remark \ref{rem:ex_time_mink}).
\end{rem}

\begin{proof}[\textbf{Proof of the uniqueness claim in Theorem
    \ref{thm:ex_uni_geom_hyp}}]
  In analogy to the proof of the uniqueness claim in Theorem
    \ref{thm:ex_uni_geom} the result follows by using the uniqueness result
    of Theorem \ref{thm:ex_uni_atlas_hyp}.
\end{proof}

\section{Prescribed initial lapse and shift}
\label{indpar}
\subsection{Existence and uniqueness}
\label{sec:diff_lapse}
The goal of this section is to show that solutions of the membrane equation
represented as immersions are independent of the choice of immersion
of the initial submanifold $\Sigma_0$, as well as of the initial lapse 
$\alpha$ and shift
$\beta$. To this end we will consider the Cauchy problem \eqref{eq:param_ivp}
and, after establishing existence for an arbitrary choice of initial lapse and
shift, compare solutions with different initial values.

We will use the following norm on tensors. Let $(M, g)$ and $(\widetilde{M},
\tilde{g})$ be two Riemannian manifolds. Suppose $\psi: M \rightarrow 
\widetilde{M}$ is a $C^1$-mapping. %
Then we denote by $\abs{d\psi}_{g, \tilde{g}}$ the norm induced on
$T^{\ast}M\otimes T\widetilde{M}$ given by
\begin{gather}
  \label{eq:norm_diffeo}
  \abs{d\psi}_{g, \tilde{g}}^2 = g^{ij} \tilde{g}_{k\ell} (d\psi)_i^k (d\psi)_j^{\ell} 
  \quad\text{with}\quad
  (d\psi)_i^k = \tilde{x}^k\bigl(d\psi(\partial_i)\bigr).
\end{gather}
$\tilde{x}^k$ are assumed to be coordinates on $\widetilde{M}$ and
$\partial_i$ are the basis vector fields arising from coordinates on $M$.
In a similar way we introduce norms for higher derivatives.

Recall the notations of the main problem concerning the ambient space $(N,h)$,
$\Sigma_0$ and initial direction $\nu$ (cf.\ chapter 
\ref{sec:intro_membranes}). 

Assume $\Sigma_0$ to be a regularly immersed submanifold of dimension $m$
and suppose $\varphi: M^m \rightarrow N$ is an immersion defined on an 
$m$-dimensional manifold $M$ satisfying $\im \varphi = \Sigma_0$.
Let $\nu$ be a timelike future-directed unit vector field along $\varphi$
normal to $\Sigma_0$.
Let $\alpha> 0$ be a function on
  $M$ and
  let $\beta$  be a vector field on $M$.
We make the following uniformity assumptions on initial direction,
lapse and  shift.
\begin{assum}
  \label{assum_diff_lapse}
  \begin{itemize}
  \item There exist constants $C^{\nu}_{\ell}$ and $ L_3$ such that 
    \begin{gather}
      \label{eq:direction}
      - h(\nu, \hT) \le L_3
      \quad\text{and}\quad
      \abs{\hn^{\ell}  \nu }_{\ig,E} 
      \le C^{\nu}_{\ell} \quad \text{ for }
      1 \le \ell \le s+1,
    \end{gather}
    where $\hT$ denotes the unit normal to the slices of the time foliation
    defined by \eqref{eq:normal_slices}.
  \item There exist constants $L_2, ~C^{\alpha}_{\ell}$ and  $C^{\beta}_{\ell}$
     such that
     \begin{subequations}
    \begin{gather}
      \label{eq:lapse_time}
      - \alpha^2 + \abs{\beta}_{\ig}^2 \le - L_2 
      \\
      \label{eq:shift_der}
      \begin{split}
        \abs{\inab^{\ell} \beta}_{\ig}  \le C^{\beta}_{\ell} \quad  
             & \text{ for }
      1 \le \ell \le s+1
      \\
      \text{and}\quad\abs{\inab^{\ell} \alpha}_{\ig}  \le C^{\alpha}_{\ell} \quad
       & \text{ for }
      0 \le \ell \le s+1.
      \end{split}
    \end{gather}   
     \end{subequations}
  \end{itemize}
\end{assum}
The following result provides existence for given initial direction, lapse, and
shift and uniqueness for solutions of the membrane equation with different
prescribed immersions of the initial submanifold, initial lapse, and shift.
\begin{thm}
  \label{thm:ex_uni_lapse_hyp}
  Let the data $h, \tau \in C^{s+3}$
  satisfy the assumptions \eqref{eq:assum_target},
  let $\varphi \in C^{s+2}$ satisfy the
  assumptions \eqref{eq:assum_smf_hyp} and
  let $\alpha, \beta, \nu \in C^{s+1}$ satisfy the assumptions 
  \ref{assum_diff_lapse}.

  Then there exist a constant $T>0$
  and a $C^2$-solution $F:[-T,T]\times M \rightarrow N $
  of the membrane equation 
  \begin{subequations}
    \begin{gather}
    \label{eq:mem_diffeo}
    g^{\mu\nu} \partial_{\mu} \partial_{\nu} F^A - 
    \Gamma^{\lambda} \partial_{\lambda} F^A + g^{\mu\nu} \partial_{\mu} F^B
    \partial_{\nu} F^C \G_{BC}^A(F) = 0
    \intertext{in harmonic map gauge w.r.t.\ the background metric $\hat{g}$ 
      defined by}
    \label{eq:back_diffeo}
    \hat{g} = - \alpha^2 dt^2 +
    \ig_{ij} (\beta^i dt + dx^i)(\beta^j dt + dx^j)
    \intertext{attaining the initial values}
    \label{eq:initial_sep_def1}
    \restr{F} = \varphi \quad\text{and} \quad
    \restr{\dt F} = \alpha \,\nu %
    + d\varphi(\beta).
    \end{gather}
  \end{subequations}
  
  Suppose $\bar{\varphi}: M \rightarrow N$ is another immersion
  of the initial submanifold $\Sigma_0$ 
  satisfying the assumptions \eqref{eq:assum_smf_hyp} 
  and there exists a local diffeomorphism
  $\psi_0: M \rightarrow M$ with $\varphi \circ \psi_0^{-1}
  = \bar{\varphi}$ and admitting constants $C_1^{\psi}$ and $C_2^{\psi}$ 
  such that
  \begin{gather}
    \label{eq:cond_initial_diffeo}
   \abs{d\psi_0}_{\ig,\mathring{\bar{g}}} \le C_1^{\psi}\quad
    \text{and}\quad \abs{d^2 \psi_0}_{\ig,\mathring{\bar{g}}} \le C_2^{\psi},
  \end{gather}
  where we set $\ibarg = \bar{\varphi}^{\ast} h$ and use the norm
  introduced in \eqref{eq:norm_diffeo}. Let the initial direction $\nu$
  be defined on $\Sigma_0$.
  Further let $\bar{\alpha} > 0$
  and $\bar{\beta}$
  denote another choice of initial lapse and shift satisfying the assumptions 
  \eqref{eq:lapse_time} and \eqref{eq:shift_der}.
  Assume $\bar{F}:[-T,T]\times M
  \rightarrow N$  to be the solution of the membrane equation
  \eqref{eq:mem_diffeo} in harmonic map gauge 
  w.r.t.\ the background metric $\hat{\bar{g}}$ defined
  by replacing $\ig, \alpha, \beta$ with $\ibarg, \bar{\alpha}$
  and $\bar{\beta}$ in definition \eqref{eq:back_diffeo}
  originating from the existence claim and
  attaining the initial values
  \begin{gather}
    \label{eq:initial_sep_def2}
    \bar{F}\bigr|_{\bar{t} = 0}
    = \bar{\varphi}
    \quad\text{and} \quad\dtb \bar{F}\bigr|_{\bar{t} = 0}
    = \bar{\alpha}\, \nu \circ \bar{\varphi} 
    + d\bar{\varphi}(\bar{\beta}).
  \end{gather}

    Then, for all $p \in M$,
  there exists a %
  local
  diffeomorphism $\Psi$ about $(0,p)\in [-T,T] \times M$ such that
  $F \circ \Psi^{-1}$ and $\bar{F}$ coincide.
\end{thm}
\begin{rem}
  The solution $F$ mentioned in the uniqueness statement is assumed to
  have the initial values
  \begin{gather}
    \label{eq:initial_rem_sep}
     \restr{F} = \varphi \quad\text{and} \quad
    \restr{\dt F} = \alpha\, \nu \circ \varphi
    + d\varphi(\beta).
  \end{gather}
  Since the initial submanifold is locally embedded it is possible to
  replace the vector field along $\varphi$ with the composition of the initial
  direction defined on the initial submanifold and $\varphi$.
\end{rem}

The proof has the following structure.
Firstly, the existence claim is proven similar to that of Theorem
 \ref{thm:ex_uni_geom_hyp}
Due to
the combination of initial direction, initial lapse, and shift 
the result of Theorem \ref{thm:ex_uni_geom_hyp} is not directly
applicable.
However, we make contact with the existence result of Theorem 
\ref{thm:ex_uni_atlas_hyp}.
Afterwards we will consider the uniqueness claim.

Let $x$ be a chart on $M$ with center $p \in M$ and let
$y$ be a chart on $N$ with center $\varphi(p) \in N$.
Let $\Phi$ be the expression $y \circ \varphi \circ x^{-1}$.
Suppose $\alpha_x, ~ \beta_x$ are the representations of $\alpha$ and $\beta$
w.r.t.\ the chart $x$, and $\nu_{xy}$ is the representation
of $\nu $%
w.r.t.\ the charts $x$ and
$y$.
We make the following uniformity assumptions on the representations of
direction, lapse and shift.
\begin{assum}
  \label{assum_diffeo_atlas}  
\begin{itemize}
  \item For a positive 
    constant $L_2$ the following inequality
    holds
    \begin{gather}
      \label{eq:atlas_shift_est}
      - \alpha_{x}^2
      + \beta_{x}^i \ig_{ij} \beta_{x}^j \le - L_2, 
    \end{gather}
    where $\ig_{ij}$ denotes the representation of the induced metric $\ig$
    w.r.t.\ the coordinates $x$. 
  \item There exist constants $\tilde{C}^{\alpha}_{\ell}, 
    \tilde{C}^{\beta}_{\ell},
    \tilde{C}^{\nu}_{\ell}$ such that
    \begin{gather}
      \label{eq:atlas_shift_der}
      \begin{split}
        & \abs{D^{\ell} \alpha_{x}}_e  \le \tilde{C}^{\alpha}_{\ell},
        \quad
        \abs{D^{\ell} \beta_{x}}_e \le \tilde{C}^{\beta}_{\ell}
        \\
        \text{and}\qquad  
        & \abs{D^{\ell} \nu_{xy}}_{e,e} \le \tilde{C}^{\nu}_{\ell}
        \quad\text{ for } 0 \le \ell \le s + 1.
      \end{split}
    \end{gather}
  \end{itemize}
\end{assum}
The next lemma will show that these assumptions are sufficient
to get the conditions \eqref{eq:assum_bd_vel} for the initial velocity
defined by %
\begin{gather}
  \label{eq:def_chi_atlas}
  \chi_{xy} := \alpha_{x} \nu_{xy} + \beta_{x}^j
  \partial_j \Phi,
\end{gather}
the representation of the initial velocity given in \eqref{eq:initial_sep_def1}
 w.r.t.\ the
charts $x$ and $y$.
\begin{lem}
  \label{lem:shift_atlas_speed}
  Suppose $y$ and the representation of the metric $h$ in these coordinates
  satisfy the assumptions \ref{assumptions_atlas_hyp}.
  Let the chart $x$ and the representation $\Phi$ satisfy
  part \ref{assum_atlas_chart} and \ref{assum_atlas_graph} of the assumptions 
  \ref{assumptions_atlas}. %
  Assume the representations $\alpha_x, \beta_x $ and $\nu_{xy}$ to satisfy
  the conditions \ref{assum_diffeo_atlas}.
  \\
  Then $\chi_{xy}$ defined in \eqref{eq:def_chi_atlas}
  satisfies part \ref{assum_atlas_speed} of the assumptions
  \ref{assumptions_atlas}.
\end{lem}
\begin{rem}
  \label{rem:nu_phi}
  If $\nu$ is replaced by $\nu\circ \varphi$ then the lemma remains valid.
\end{rem}
\begin{proof}
  The condition $h(\chi_{xy}, \chi_{xy}) \le - L_2$ follows immediately
  from inequality \eqref{eq:atlas_shift_est}.
  The bounds for $\chi_{xy}$ desired by assumption \eqref{eq:assum_bd_vel}
  follow from the assumptions for derivatives
  of direction, lapse, and shift using the bounds for derivatives
  of $\Phi$ stated in \eqref{eq:assum_bd_pt}. A bound for $D\Phi$ follows
  from a Taylor expansion by considering the bound for the second derivative.
\end{proof}
The preceding lemma yields that Theorem \ref{thm:ex_uni_atlas_hyp} is 
applicable.
In the next step we will show that the assumptions \ref{assum_diff_lapse} 
as they are independent of
coordinates lead to the conditions \ref{assum_diffeo_atlas}.
This step is analogous to the propositions \ref{prop:graph_assum} and
\ref{prop:graph_assum_hyp}.
\begin{lem}
    \label{lem:shift_sgr_speed}
    Let the metric $h \in C^{s+3}$ and the time function $\tau \in C^{s+3}$ 
    satisfy the assumptions
    \eqref{eq:assum_target}. Assume the initial immersion $\varphi \in C^{s+2}$
    to satisfy the assumptions \eqref{eq:assum_smf_hyp}
    and $\nu \circ \varphi, \alpha,\beta \in C^{s+1}$
    to satisfy the assumptions
    \ref{assum_diff_lapse}.
    
    Let $p \in M$ and let $x$ be a chart on $M$ with center $p$.
    Let $y$ be the special coordinates on $N$ introduced in section 
    \ref{sec:special_coord} and assume $y \circ \varphi \circ x^{-1}$ to be 
    the special
    graph representation obtained in section \ref{sec:graph_repr_hyp}.
    
    Then the representations $\Phi,~\alpha_x, \beta_x$, and $\nu_{xy}$
    of $\varphi,~ \alpha,~\beta$, and 
    $\nu$ %
    w.r.t.\ these charts satisfy the assumptions \ref{assum_diffeo_atlas}.
\end{lem}
\begin{rem}
  Here, the Remark \ref{rem:nu_phi} for the preceding lemma also applies.
  The bounds follow from the bounds for the representation of the immersion
  $\varphi$.
\end{rem}
\begin{proof}
    The first condition \eqref{eq:atlas_shift_est} follows immediately
  from \eqref{eq:lapse_time}. 
    A bound for the lapse $\alpha$
    is given in condition \eqref{eq:shift_der}. Hence,
    we can derive the inequality
    $\abs{\beta}_{\ig}^2
  \le -L_2 + (C_0^{\alpha})^2$ from condition \eqref{eq:lapse_time}.
  To obtain an estimate for the Euclidean norm in coordinates
  we use the comparison of the Euclidean metric in coordinates and
  the metric $\ig$ on M stated in \eqref{eq:hyp_metric_compare}.
  Similar to the argument used in the proof of proposition
  \ref{prop:graph_assum_hyp} (cf. Proposition \ref{prop:graph_assum})
  it follows from condition \eqref{eq:direction}
  that $\abs{\nu}_E$ is bounded. The comparison of the metric $E$ 
  with the Euclidean
  metric in the special coordinates stated in 
  \eqref{eq:metric_E_est} provides a bound for $\abs{\nu_{xy}}_e$.

  The conditions for derivatives of the representation of the initial direction
  $\nu$ %
  follow from a similar consideration
  as in the proof of Proposition \ref{prop:graph_assum_hyp}.
  An identity similar to \eqref{eq:high_der_cov_speed}
  also holds for the representations of
  lapse and shift involving the induced covariant derivative on $M$
  and its Christoffel symbols.
  Thus, the desired bounds follow from a device similar to the proof
  of Lemma \ref{lem:hyp_u_high_der}.  
\end{proof}
The preceding lemmata yield the existence claim of Theorem
\ref{thm:ex_uni_lapse_hyp} as we will now see.
\begin{proof}[\textbf{Proof of the existence claim of Theorem 
    \ref{thm:ex_uni_lapse_hyp}}]
  We will appeal to the proof of  Theorem \ref{thm:ex_uni_geom_hyp}.
  We need to verify whether the steps taken there can be paralleled here.
  From Lemma \ref{lem:shift_atlas_speed} we derive that, providing the 
  assumptions
  \ref{assum_diffeo_atlas} are satisfied, Theorem \ref{thm:ex_uni_atlas_hyp} 
  applies.
  If follows from Lemma \ref{lem:shift_sgr_speed} that, providing the 
  assumptions \ref{assum_diff_lapse}
  are satisfied, the conditions \ref{assum_diffeo_atlas} are met by 
  decompositions 
  built by special graph representations.
  Therefore, the result follows.
\end{proof}
We now head to  the uniqueness claim of
Theorem \ref{thm:ex_uni_lapse_hyp}.

Let $F$ and $\bar{F}$ be two solutions of the membrane equation 
\eqref{eq:mem_diffeo} attaining the initial values \eqref{eq:initial_rem_sep}
and \eqref{eq:initial_sep_def2}, respectively. The solutions $F$ and $\bar{F}$ 
are assumed 
to be in harmonic
map gauge w.r.t.\ the background metrics $\hat{g}$ and $\hat{\bar{g}}$
defined by the initial values as in \eqref{eq:back_diffeo}, respectively.
Let $g := F^{\ast} h$ and $\bar{g} := \bar{F}^{\ast} h$.

Our strategy
 will be to show that there exists a local diffeomorphism $\Psi$  
such that
$F \circ \Psi^{-1}$ and $\bar{F}$ satisfy the reduced membrane equation
\eqref{eq:mem_red} in harmonic map gauge w.r.t.\ the same
background metric and attaining
the same initial values. 

In the following proposition we will derive conditions for such a diffeomorphism
$\Psi$.
\begin{prop}
  \label{prop:cond_diffeo}
  Let $p \in M$, $U \subset M$ a neighborhood of $p$ and $V \subset M$
  a neighborhood of $\psi_0(p) \in M$.
  Suppose there exist constants $0 < T', \tilde{T}' \le T$ and a diffeomorphism
  $\Psi  : (-T', T') \times U \rightarrow (- \tilde{T}',\tilde{T}') \times V$
  such that
     \begin{gather}
       \label{eq:diffeo_harm}
     \Psi: \bigl((-T', T')  \times U, g \bigr)
     \rightarrow \bigl((- \tilde{T}',\tilde{T}') \times V, \hat{\bar{g}}
     \bigr) 
     \text{ is a harmonic map}.
   \end{gather}
   Assume further that the inverse satisfies the initial conditions
   \begin{gather}
     \label{eq:def_inverse_trafo_initial}
     \Psi^{-1}\bigr|_{\tilde{t} = 0} = (0,\psi_0^{-1}) \quad\text{and}\quad
     \partial_{\tilde{t}} \Psi^{-1}\bigr|_{\tilde{t} = 0}
     = \hat{\lambda} \partial_t + \hat{\chi}
     \\ 
     \text{with }
     \hat{\lambda}(p) = \tfrac{\bar{\alpha}(p)}{\alpha(\psi_0^{-1}(p))}
     \quad\text{and} \quad
     \hat{\chi}(p) = d(\psi_0^{-1})_p(\bar{\beta}) 
     - \hat{\lambda}(p) \beta(\psi_0^{-1}(p)). \nonumber
   \end{gather}
   Then $F \circ \Psi^{-1}$ satisfies the reduced membrane
   equation \eqref{eq:mem_red} w.r.t.\ the background metric
   $\hat{\bar{g}}$ attaining the initial values of $\bar{F}$.   
 \end{prop}
 \begin{proof}
   From the harmonic map equation satisfied by $\Psi$ we derive the
   condition for the solution $F \circ \Psi^{-1}$ of the membrane
   equation to be in harmonic map gauge w.r.t.\ the background metric 
   $\hat{\bar{g}}$ (cf. \eqref{eq:harm_cond}). This can be seen by computing
   the contracted Christoffel symbols of the metric $(F \circ \Psi^{-1})^{\ast}h$
   to satisfy condition \eqref{eq:harm_cond_coord} w.r.t.\ the background
   metric $\hat{\bar{g}}$. 
   
   The next step is to show that the initial values of $\Psi^{-1}$
   give us the appropriate initial values for $F \circ \Psi^{-1}$.
   We compute
   \begin{multline*}
          \partial_{\bar{t}}(F   \circ \Psi^{-1})\bigr|_{\tilde{t} = 0}(p) 
     =  \hat{\lambda}(p) \restr{\partial_t F}(\psi_0^{-1}(p))
     + d\varphi_{\psi_0^{-1}(p)}(\hat{\chi})
     \\
     = \hat{\lambda}(p) 
     \bigl(\alpha(\psi_0^{-1}(p)) \nu\circ \varphi(\psi_0^{-1}(p))
     + d\varphi_{\psi_0^{-1}(p)}(\beta) \bigr)
     + d\varphi_{\psi_0^{-1}(p)}(\hat{\chi})
     \\
     = \hat{\lambda}(p) \alpha(\psi_0^{-1}(p)) 
     \nu\circ \tilde{\varphi}(p) + \hat{\lambda}(p) 
     d\varphi_{\psi_0^{-1}(p)}(\beta) 
     + d(\varphi \circ \psi_0^{-1})_p(\bar{\beta})  
     - \hat{\lambda}(p) d\varphi_{\psi_0^{-1}(p)}(\beta)
   \end{multline*}
   Since the second and the last term cancel the result follows.
 \end{proof}
As an immediate consequence we can show uniqueness.
\begin{proof}[\textbf{Proof of the uniqueness claim of Theorem 
    \ref{thm:ex_uni_lapse_hyp}}]
  As for the existence claim Theorem \ref{thm:ex_uni_geom_hyp} is not directly 
  applicable. 
  The Lemmata  \ref{lem:shift_atlas_speed} and 
  \ref{lem:shift_sgr_speed} yield that the uniqueness result of 
  Theorem \ref{thm:ex_uni_atlas_hyp} is applicable to
  a decomposition consisting of 
  special graph representations. 
  
  Assuming the existence of a diffeomorphism $\Psi$ satisfying the 
  assumptions of Proposition \ref{prop:cond_diffeo} provides us with 
  two solutions
  satisfying the reduced membrane equation \eqref{eq:mem_red}
  w.r.t.\ the same background metric
  and attain the same initial values.
  Following the proof of the uniqueness claim of Theorem 
  \ref{thm:ex_uni_geom_hyp} gives the desired result.
\end{proof}
The next section is devoted to construct a local diffeomorphism satisfying
the condition \eqref{eq:diffeo_harm} in Proposition \ref{prop:cond_diffeo}.

\subsection{Construction of a reparametrization}
\label{sec:constr_diffeo}
From Proposition \ref{prop:cond_diffeo} we derive a harmonic map
equation to be satisfied. Henceforth, we obtain the diffeomorphism by 
constructing an immersion solving this equation.

The following designation concerning coordinates will be used throughout
this section.
Coordinates on $(\rr \times M, g)$ will be denoted by $x^{\mu}$
with indices $\mu,\nu,\lambda, \dots$ and
coordinates on $(\rr\times M, \hat{\bar{g}})$ will be denoted
by $\bar{x}^{\delta}$ with indices $\delta, \varepsilon, \kappa, \dots$.

Condition \eqref{eq:diffeo_harm} for the desired diffeomorphism
leads to the following harmonic map equation in coordinates
\begin{gather}
  \label{eq:repa}
  g^{\mu\nu} \partial_{\mu} \partial_{\nu} \Psi^{\delta} = 
  \Gamma^{\lambda} \partial_{\lambda} \Psi^{\delta}
  - g^{\mu\nu} \partial_{\mu} \Psi^{\varepsilon} \partial_{\nu} \Psi^{\kappa}
  \hbarG{}_{\varepsilon\kappa}^{\delta}(\Psi),
\end{gather}
where $\Gamma^{\lambda}$ denote the contracted Christoffel symbols of $g$ and
$\hbarG{}_{\varepsilon\kappa}^{\delta}$ denote the
Christoffel symbols of the metric $\hat{\bar{g}}$.
Observe that this is a semilinear equation, since the coefficients 
are fixed. The solution $F$ of the membrane equation
is assumed to be in harmonic map
gauge w.r.t.\ $\hat{g}$. Therefore, the Christoffel symbols 
$\Gamma^{\lambda}_{\mu\nu}$
on the RHS of the equation can be replaced by the Christoffel symbols
$\hG$ of $\hat{g}$ (cf. condition \eqref{eq:harm_cond_coord}).
The initial values are given by the initial values for the inverse
stated in \eqref{eq:def_inverse_trafo_initial}.
Hence, the initial values for $\Psi$ are as follows
 \begin{subequations}
   \begin{gather}
  \label{eq:def_trafo_initial}
       \restr{\Psi}= (0,\psi_0), ~
     \restr{\dt \Psi}
     = \hat{\alpha} \tfrac{d}{d\bar{t}}+ \hat{\beta}
     \\ 
     \label{eq:def_vel}
     \text{ with }~
     \hat{\alpha}(p) = \tfrac{\alpha(p)}{\bar{\alpha}(\psi_0(p))}
     ~\text{ and }~
     \hat{\beta}(p) = - \hat{\alpha}(p)\bar{\beta}(\psi_0(p)) + 
     d(\psi_0)_p(\beta).
   \end{gather}
 \end{subequations}
The following proposition states the main result of this section
providing a solution of the initial value problem for $\Psi$.
\begin{prop}
  \label{prop:ex_uni_diffeo}
  Let the assumptions of the uniqueness claim of Theorem 
  \ref{thm:ex_uni_lapse_hyp} be satisfied.
  
  Then there exist a constant $\bar{T}>0$ and a
  $C^2$-immersion 
  $\Psi: [- \bar{T},\bar{T}] \times M \rightarrow \rr \times M$ satisfying
  the equation
  \eqref{eq:repa}  and
  attaining the initial values \eqref{eq:def_trafo_initial}.
  
  If $\Psi$ and $\bar{\Psi}$ are two such solutions, then they
  coincide for $- \min(T_0, \tilde{T}_1)  \le t \le \min(T_0, \tilde{T}_1)$, 
where $T_0, \tilde{T}_1>0$
  are constants.
\end{prop}
\begin{rem}
  The inverse function theorem yields that for $\abs{t}$ small
  enough $\Psi(t,p)$ is a diffeomorphism around each point $p \in M$.
  It therefore has the properties described in proposition
  \ref{prop:cond_diffeo}. 

  The constant $T_0$ is defined in \eqref{eq:uni_cone_height} 
  with the constant $c_0$ being defined in \eqref{eq:hyp_est_c_0} 
  and $\tilde{T}_1$ being defined similarly to \eqref{eq:choice_t}.
\end{rem}
The method to solve the IVP will follow the same strategy as we used in
the sections \ref{sec:setup_hyp} and 
\ref{lorentz_ex} in order to solve the membrane equation.
We begin with a local result similar to Theorem \ref{thm:ex_uni_atlas_hyp}.

Let $(U_{\lambda}, x_{\lambda}, V_{\lambda}, y_{\lambda})_{\lambda \in \Lambda}$ 
be a decomposition of $\varphi$ (cf.\ Definition \ref{defn:decomp})
and \\ $(\bar{U}_{\sigma}, \bar{x}_{\sigma}, \bar{V}_{\sigma}, 
\bar{y}_{\sigma})_{\sigma
  \in \Lambda'}$ be a decomposition of $\bar{\varphi}$.
Let $\Phi_{\lambda}, \bar{\Phi}_{\sigma}$
 be the representations of $\varphi$ and $\bar{\varphi}$ w.r.t.\ the
charts $x_{\lambda}, ~y_{\lambda}$ and $\bar{x}_{\sigma}, ~\bar{y}_{\sigma}$
 as defined in \eqref{eq:initial_expr}.
Let $(\bar{\psi}_0)_{\lambda\sigma}$ denote the expression
$ \bar{x}_{\sigma} \circ \psi_0 \circ x_{\lambda}^{-1}$. %
Assume $\alpha_{\lambda}, \beta_{\lambda}$ to be the representations of
$\alpha$ and $\beta$ w.r.t.\ the charts $x_{\lambda}$ and define
$\bar{\alpha}_{\sigma}, \bar{\beta}_{\sigma}$ in an analogous manner.
Further let 
  $\nu_{\lambda}$ be the representation of $\nu \circ \varphi$ w.r.t.
the charts $x_{\lambda}$ and $y_{\lambda}$ and let $\bar{\nu}_{\lambda}$
be the representation of $\nu \circ \bar{\varphi}$ .
Further let  $\hat{\alpha}_{\lambda\sigma}$ and $\hat{\beta}_{\lambda\sigma}$ be the 
expressions
\begin{subequations}
  \begin{align}
  \label{eq:diffeo_init_atlas_fct}
  \hat{\alpha}_{\lambda\sigma}(z)
  & = \frac{\alpha_{\lambda}(z)}{
    \bar{\alpha}_{\sigma}((\bar{\psi}_0)_{\lambda \sigma}(z))}
  \\
  \label{eq:diffeo_init_atlas}
  \text{and}\quad
  \hat{\beta}_{\lambda\sigma}(z) & =  
  - \hat{\alpha}_{\lambda}(z) 
  \bar{\beta}_{\sigma}\circ(\bar{\psi}_0)_{\lambda\sigma}(z)
     + d(\bar{\psi}_0)_{\lambda\sigma}  \bigl(\beta_{\lambda}(z)\bigr).
\end{align}
\end{subequations}
The assumption on the representation of the local 
diffeomorphism $\psi_0$ is as follows.
\begin{gather}
  \label{eq:assum_atlas_init_diffeo}
  \begin{split}
      \text{
    There are constants $\tilde{C}^{\psi}_1$
    and $\tilde{C}^{\psi}_2$
    such that }
  \\
  \abs{D(\bar{\psi}_0)_{\lambda\sigma}}_e \le \tilde{C}_1^{\psi}
  \quad\text{and}\quad
  \abs{D^2(\bar{\psi}_0)_{\lambda\sigma}}_e \le \tilde{C}_2^{\psi}.
  \end{split}
\end{gather}
The following proposition includes a local existence and uniqueness
result for the harmonic map equation \eqref{eq:repa} with the initial values
\eqref{eq:def_trafo_initial}.
\begin{prop}
  \label{prop:atlas_diffeo}
  We make the following assumptions uniformly in $\lambda \in \Lambda$.
  Suppose $y_{\lambda}$ and the representation of the metric $h$ in these coordinates
  satisfy the assumptions \ref{assumptions_atlas_hyp}.
  Let the chart $x_{\lambda}$ and the representation $\Phi_{\lambda}$ satisfy
  part \ref{assum_atlas_chart} and \ref{assum_atlas_graph} of the assumptions 
  \ref{assumptions_atlas}. %
  Assume the representations $\alpha_{\lambda}, \beta_{\lambda} $, and $\nu_{\lambda}$ 
  to satisfy
  the conditions \ref{assum_diffeo_atlas}.
  \\
  Assume $\bar{y}_{\sigma}, \bar{x}_{\sigma}, \bar{\Phi}_{\sigma}, 
  \bar{\alpha}_{\sigma},
  \bar{\beta}_{\sigma}, \bar{\nu}_{\sigma}$ to satisfy the same assumptions
  uniformly in $\sigma \in \Lambda'$ possibly with different constants.
  \\
  Let the representation $(\bar{\psi}_0)_{\lambda\sigma}$ of 
  the local diffeomorphism $\psi_0$ %
  satisfy the assumptions
  \eqref{eq:assum_atlas_init_diffeo}.
  
  Then 
  there exist
  constants $\bar{T}' > 0$, $0 < \hat{\theta} < 1$ and a family $(\Psi_{\lambda})$
  of bounded $C^2$-immersions $\Psi_{\lambda}: [- \bar{T}',\bar{T}']
  \times B^e_{\hat{\theta} \rho_1/2}(0) \subset \rr \times x_{\lambda}(U_{\lambda})
  \rightarrow \rr^{m+1}$ solving equation \eqref{eq:repa} and attaining the
  initial values 
  \begin{gather}
    \label{eq:local_diffeo_initial}
    \Psi_{\lambda}(0,z) = (\bar{\psi}_0)_{\lambda\sigma}\text{ for } z \in
    B^e_{\hat{\theta} \rho_1/2}(0)\quad\text{and} \quad
    \restr{\partial_t \Psi_{\lambda}}
    = \binom{\hat{\alpha}_{\lambda\sigma}}{\hat{\beta}_{\lambda\sigma}}.
  \end{gather}

  If $\Psi_{\lambda}$ and $\bar{\Psi}_{\lambda}$ are two such solutions
  defined on the image of the same chart with the same
  initial values in $B_r^e(z)$, then they coincide on
  the double-cone with basis $B_{r}^e(z)$ and slope $c_0$ defined by
  \eqref{eq:hyp_est_c_0}.
\end{prop}
The proof of the above proposition will occupy the rest of this section.
We will follow the strategy of section \ref{sec:setup_hyp}, especially the
proof of the local 
result presented in Theorem \ref{thm:ex_uni_atlas_hyp}.
Since solutions constructed in Theorem \ref{thm:ex_uni_lapse_hyp} are $C^2$, 
a cut-off process
does not lead to the conditions of the existence Theorem \ref{qlin_ex}.
We circumvent this issue by using the result of proposition
\ref{prop:cutoff_ex}.
Assume $\lambda \in \Lambda$ and $\sigma \in \Lambda'$ to be fixed
such that the definition of $(\bar{\psi}_0)_{\lambda\sigma}$ makes sense.
Let $F_{\lambda}$ denotes the representation of the solution $F$ w.r.t.\ the charts
$x_{\lambda}$ and $y_{\lambda}$. Then we derive from Theorem 
\ref{thm:ex_uni_atlas_hyp} that $F_{\lambda}$ coincides with the solution 
$F_0$ constructed in Proposition \ref{prop:cutoff_ex} on the cone
with basis $B^e_{\theta \rho_1/2}(0)$ and slope $c_0$ as described in the
uniqueness statement of Proposition \ref{prop:atlas_diffeo}.
From the proof of Proposition \ref{prop:cutoff_ex} 
it follows that the induced metric
w.r.t.\ $F_0$ satisfies the assumptions of Theorem \ref{qlin_ex}.
Let $g_{\mu\nu}^0$ be the coefficients defined in \eqref{eq:coeff} 
w.r.t.\ the solution $F_0$. 

The cut-off function that comes along with the construction of $F_0$
will be used to define cut-off functions corresponding to the initial
values $ \Phi_{\lambda}, \alpha_{\lambda}, \nu_{\lambda}$
and $\beta_{\lambda}$ and the initial values with a bar.
These definitions give rise to a cut-off process for 
the background metrics $\hat{g}$ and 
$\hat{\bar{g}}$ providing us with
matrices 
$\hat{a}$ and $\hat{\bar{a}}$ analogous to \eqref{eq:back_comp}.
The Christoffel symbols of $\hat{a}$ and $\hat{\bar{a}}$ will be denoted
by $\hat{\gamma}_{\mu\nu}^{\lambda}$ and $\hat{\bar{\gamma}}_{\delta\kappa}^{\varepsilon}$,
respectively.

In the sequel we will develop the asymptotic equation to be solved.
For notational convenience we set $\bar{\psi}_0 
:= (\bar{\psi}_0)_{\lambda\sigma}$.
Define a linear function $w(t,x)$ by
\begin{gather}
  \label{eq:asym_harm}
  w(t,x) = w_0(x) + tw_1 
  \\
  \text{with}\quad w_0(z) 
  = x^{\ell} \partial_{\ell} \bar{\psi}_0(0)
  \text{ for } x \in \rr^m\quad\text{and}\quad 
  w_1 = \binom{\hat{\alpha}_{\lambda\sigma}(0)}{\hat{\beta}_{\lambda\sigma}(0)}.
  \nonumber
\end{gather}
Set
\begin{gather*}
  \init{\psi}_0 := \hat{\zeta}(\bar{\psi}_0 
  - w_0) \quad\text{and}\quad\init{X} := 
  \hat{\zeta}\bigl( \tbinom{\hat{\alpha}_{\lambda\sigma}}{\hat{\beta}_{\lambda\sigma}} 
  - w_1\bigr)
\end{gather*}
where $\hat{\zeta}$ is a cut-off function with the same properties as
the  cut-off function $\zeta$ used in the 
proof of Proposition \ref{prop:cutoff_ex}
with the parameter $\theta$ replaced by $\hat{\theta}$.

We will now define the operator corresponding to the RHS of equation
\eqref{eq:repa}.
Let $\Omega \subset  \rr^{m+1} \times 
\rr^{m(m+1)}
\times \rr^{m+1}$ be a set chosen later
and define 
\begin{gather*}
  \label{eq:def_rhs_diffeo}
  f_a^{\delta}(t,v,Y,X) := g^{\mu\nu}_0(t) \bigl(
  \hat{\gamma}_{\mu\nu}^{0}(t) (\partial_{t} w + X)^{\delta}
  + \hat{\gamma}_{\mu\nu}^{\ell}(t) (\partial_{\ell} w + Y_{\ell})^{\delta} \bigr)
\end{gather*}
for $(t,v,Y,X) \in \rr \times \Omega$. Define $\BF_a^{\delta}(t,v,Y,X)$  as in 
\eqref{eq:def_rhs_asym_BF} 
replacing 
the metric $g^{\mu\nu}_a$ by the fixed metric $g^{\mu\nu}_0$
and replacing the Christoffel symbols $\hg_{BC}^A$ by the Christoffel symbols
$\hat{\bar{\gamma}}_{\varepsilon\kappa}^{\delta}\bigl(w(t) + v\bigr)$.
These definitions immediately reveal the similarity to equation
\eqref{eq:geom_solve_asym_hyp}. Analogously to definition 
\eqref{eq:def_coeff_op_hyp}
we obtain from the preceding definitions of $f_a$ and $\BF_a$ operators
defined on a set $[0,\bar{T}_1] \times W \subset \rr \times H^{s+1} \times H^s$.
The operators will carry the same names and it will be clear from the 
context which notation is used.
The constant $\bar{T}_1$ and the form of $W$ will be chosen later  
depending on $\Omega$.
For later reference we state the asymptotic IVP to be solved in the sequel
\begin{subequations}
  \begin{align}
    \label{eq:repa_solve}
    \begin{split}
      g^{\mu\nu}_0(t) \partial_{\mu} \partial_{\nu} \psi^{\delta}  = ~&
      f^{\delta}_a(t,\psi,D\psi,\partial_t \psi) 
      + \BF^{\delta}_a(t,\psi, D\psi,\partial_t \psi)
      \\
      ~ &\restr{\psi} = \init{\psi}_0
      \quad\text{and}\quad \restr{\partial_t \psi}
      = \init{X}
    \end{split}
    \\
    \label{eq:def_rhs_diffeo1}
    \text{with } 
    f^{\delta}_a(t,\psi, D\psi,\partial_t \psi)  = ~& 
    g^{\mu\nu}_0(t) \hat{\gamma}_{\mu\nu}^{\lambda}(t) (\partial_{\lambda} w 
    + \partial_{\lambda}
    \psi)^{\delta}
    \\
    \label{eq:def_rhs_diffeo2}
     \text{and }
     \BF^{\delta}_a(t,\psi,D\psi,\partial_t \psi) = ~& 
      {}- g^{\mu\nu}_0(t) (\partial_{\mu} w + \partial_{\mu} \psi)^{\varepsilon} 
      (\partial_{\nu} w + \partial_{\nu} \psi)^{\kappa}
      \hat{\bar{\gamma}}_{\varepsilon\kappa}^{\delta}\bigl(w(t) + \psi\bigr).
\end{align}
\end{subequations}
In contrast to section \ref{sec:setup_hyp} the set $\Omega$ will only be used
to control the differential of the mapping $w(t) + \psi(t)$,
since the coefficients of equation \eqref{eq:repa_solve} are fixed.
We consider the matrix
\begin{gather*}
  b_{0\ell}(v,Y,X) 
  = (w_1 + X)^{\delta}  \hat{\bar{a}}_{\delta\varepsilon} (\partial_{\ell} w_0
  + Y_{\ell})^{\varepsilon} \quad\text{for } (v,Y,X) \in \Omega
\end{gather*}
where the other parts $b_{00}$ and $b_{k\ell}$ are defined analogously.
Recall  that the matrix $\hat{\bar{a}}_{\delta\lambda}$ is independent of
the time parameter.
For notational convenience we set
  \begin{gather*}
    \init{b}_{0\ell} = w_1^{\delta}
    \hat{\bar{a}}_{\delta\varepsilon}(0)
    \partial_{\ell} w_0^{\varepsilon}\quad\qquad\text{$\init{b}_{00}$ and 
      $\init{b}_{k\ell}$ being defined analogously.}
  \end{gather*}
The argument of $\hat{\bar{a}}$ corresponds to the center of the chart
$\bar{x}_{\sigma}$.

The following lemma establishes estimates for components of the
matrix $b_{\mu\nu}$ which will be used to derive a definition of 
the set $\Omega$.
\begin{lem}
  \label{lem:ind_metric_diffeo}
  The following inequalities hold
  \begin{align*}
    \init{b}_{00} &\le - L_2, &  b_{00} & \le \init{b}_{00} 
    + 2 \abs{X} \bigl(C_{w_0}^2 + (\tilde{C}_0^{\alpha} \tilde{C}_0^{\nu}
    + \tilde{C}_0^{\beta})^2\bigr)(1 + \tilde{\theta} \delta_0)
    (\abs{w_1} + \abs{X}),
    \\
    \init{b}_{ij} & \ge \omega_1^{-2} \delta_{ij}, 
    & b_{ij} & \ge \init{b}_{ij} - 2 \abs{Y}
    \bigl(C_{w_0}^2 + (\tilde{C}_0^{\alpha} \tilde{C}_0^{\nu}
    + \tilde{C}_0^{\beta})^2\bigr)
    (1 + \tilde{\theta} \delta_0)(\abs{Dw_0} + \abs{Y})
    \delta_{ij}.
  \end{align*}
\end{lem}
\begin{proof}
  From the definitions of 
  the initial values $\hat{\alpha}$ and $\hat{\beta}$
  in \eqref{eq:def_vel} it follows that
  $\init{b}_{00} = - \alpha_{\lambda}^2(0) + \beta_{\lambda}^k \ig_{k\ell}
  \beta_{\lambda}^{\ell}$ and $\init{b}_{ij} = \ig_{ij}(0)$.
  The inequalities on the LHS therefore follow from assumption
  \eqref{eq:atlas_shift_est} and 
  part \ref{assum_atlas_chart}
  of the assumptions \ref{assumptions_atlas}. 

  The other estimates can be obtained from the proof of 
  Lemma \ref{lem:ind_metric_hyp}.
\end{proof}
Define 
\begin{gather}
  \label{eq:def_O_diffeo}
  \Omega = \rr^{m+1} \times B^e_{\delta_1}(0) \times B^e_{\delta_2}(0)
  \subset \rr^{m+1} \times \rr^{m(m+1)} \times \rr^{m+1}
\end{gather}
with constants $\delta_1$ and $\delta_2$ such that if $(v,Y, X) \in \Omega$
then we have
\begin{gather*}
  b_{00}(v,Y, X) \le -L_2(1 - r_0)\quad\text{and} \quad 
  b_{ij}(v,Y, X) \ge \omega_1^{-2}(1 - R_0)
\end{gather*}
for fixed constants $0 < r_0, R_0 < 1$.
This can be done in a similar way as for the constants $\delta_1$ and
$\delta_2$ in definition \eqref{eq:Omega}.

To ensure that the initial values $\init{\psi}_0$ and
$\init{X}$ can be controlled by the
parameter $\hat{\theta}$, a
statement similar to Lemma \ref{lem:initial_est} is needed.
Let the representations $\bar{\alpha}_{\sigma},
  \bar{\beta}_{\sigma}$ and $\bar{\nu}_{\sigma}$ satisfy the assumptions
\eqref{eq:atlas_shift_der} with constants denoted by
 $\tilde{C}^{\bar{\alpha}}_{\ell}, \tilde{C}^{\bar{\beta}}_{\ell}$,
  and $   \tilde{C}^{\bar{\nu}}_{\ell}$, respectively.
\begin{lem}
  The following inequalities hold
  \begin{gather*}
    \abs{D\init{\psi}_0} \le \tilde{C}_2^{\psi} \hat{\theta}(1 + \tilde{C}_1)
    \qquad\text{and}\qquad\abs{\init{X}} \le C^X_1 \rho_1 \hat{\theta}
    \\
    \text{with}\quad C_1^X = \tilde{C}_2^{\psi} \tilde{C}_0^{\beta} 
    + \tilde{C}_1^{\psi} \tilde{C}_1^{\beta} +
    L_2^{-2} \bigl(\tilde{C}_1^{\alpha}
    \tilde{C}_0^{\bar{\alpha}} + \tilde{C}_0^{\alpha} 
    \tilde{C}_1^{\bar{\alpha}}\bigr)
    (1 + \tilde{C}_0^{\bar{\beta}}) 
    + L_2^{-\2} \tilde{C}_0^{\alpha} \tilde{C}_1^{\bar{\beta}}
    \tilde{C}_1^{\psi}.
  \end{gather*}
\end{lem}
\begin{proof}
  The proof follows from a device similar to that used in the proof
  of Lemma \ref{lem:initial_est} taking into account the conditions 
  \eqref{eq:atlas_shift_der} on $\alpha_{\lambda}, \bar{\alpha}_{\sigma}, 
  \beta_{\lambda}, \bar{\beta}_{\sigma}$, and $(\bar{\psi}_0)_{\lambda\sigma}$. 
\end{proof}

We choose the parameter $\hat{\theta}$ small enough such that
$\abs{\init{X}} < \delta_2/2$, 
$\abs{D\init{\psi}_0} < \delta_1/2$ and $\abs{\bar{\psi}_0(z)}_e \le 
\theta \rho_1/4$
for $\abs{z} < \hat{\theta} \rho_1/2$ which can be done via Taylor expansion.
This yields %
that the asymptotic IVP \eqref{eq:repa_solve} coincides with the IVP consisting
of 
equation \eqref{eq:repa} and the initial values defined by
\eqref{eq:local_diffeo_initial}, in the region
$B^e_{\hat{\theta} \rho_1/2}(0) \subset \rr^m$.

Let $\rho$ be chosen in a way such that
for $(\varphi_0, \varphi_1) \in W := B_{\rho}\bigl(\init{\psi}_0\bigr)\times 
B_{\rho}\bigl(\init{X}\bigr)$ 
it follows that $\abs{D\varphi_0} < \delta_1$
and $\abs{\varphi_1} < \delta_2$.

From part \ref{stay_W} of remark 
\ref{rem:low_bound_hyp}
we get that $g^{\mu\nu}_0$ satisfies the conditions 
\eqref{cond_lin1} to 
\eqref{cond_lin4} of the linear existence Theorem \ref{lin_ex}.
The construction of the solution $F_0$ in Proposition \ref{prop:cutoff_ex_hyp}
further provides bounds for the Christoffel symbols of the matrix
$\hat{a}_{\mu\nu}$ specified in Lemma \ref{lem:back_chr_hyp}.
The Lemmata \ref{lem:norm_metric_inverse_hyp}, \ref{lem:coeff_lip_t}
and estimate 
\eqref{eq:coeff_const} 
yield that 
the conditions
\eqref{cond_qlin1} and \eqref{cond_qlin_add1} through %
\eqref{cond_qlin3} for the first part $f_a$ of equation
\eqref{eq:repa_solve} given by \eqref{eq:def_rhs_diffeo1} can by deduced 
by the same device as used in section \ref{sec:sob_est_hyp}.

It therefore remains to derive the preceding conditions for the second part
$\BF_a$ of the RHS defined by \eqref{eq:def_rhs_diffeo2}.
A careful examination of  the proof of
Lemma \ref{lem:hyp_rhs_add_est} shows that the estimates are derived from
the Sobolev estimates for the metric $\hat{h}_{AB}$ and its Christoffel
symbols. From the proof of Lemma \ref{lem:hg_Hs_est} we deduce that only the 
local bounds
for $h_{AB}$ stated in Lemma \ref{lem:infty_h} were used.
Hence, taking the structure of $\BF_a$ %
into account,  similar local estimates for the metric $\hat{\bar{a}}$
lead %
to the conditions for $\BF_a$ needed by
Theorem \ref{qlin_ex}. 
The local bound will be derived in the next lemma.
\begin{lem}
  \label{lem:est_back2}
  The metric $\hat{\bar{a}}$ satisfies the inequalities
  \begin{align*}
    \abs{(\hat{\bar{a}}_{\mu\nu})}_e & \le C_{\hat{a}_2,0}, &
    \abs{(D\hat{\bar{a}}_{\mu\nu})}_e & \le C_{D\hat{a}_2,0},
    &
    \abs{(D^2\hat{\bar{a}}_{\mu\nu})}_e & \le C_{D^2\hat{a}_2,0}, 
    \\
    \abs{(\hat{\bar{a}}_{\mu\nu})}_{e,C^s} & \le C_{\hat{a}_2,s}, &
    \text{and }\quad
    \abs{(D\hat{\bar{a}}_{\mu\nu})}_{e,C^s} & \le C_{D\hat{a}_2,s}.
  \end{align*}
\end{lem}
\begin{proof}
  Let $\bar{\chi}$ denote the expression arising from definition  
  \eqref{eq:def_chi_atlas} 
  by replacing $\Phi, \alpha_x, \beta_x$, and $\nu_{xy}$ with 
  $\bar{\Phi}_{\sigma},
  \bar{\alpha}_{\sigma}, \bar{\beta}_{\sigma}$, and $\bar{\nu}_{\sigma}$, 
  respectively.
  Suppose
  $\init{\bar{\Phi}}$ and $\init{\bar{\chi}}$ are defined as in 
  \eqref{eq:def_inter} and \eqref{eq:vel_inter} by replacing $\Phi$ and $ \chi$ 
  with $\bar{\Phi}$  and $\bar{\chi}$. 
  The matrix $\hat{\bar{a}}_{\mu\nu}$ is defined as in \eqref{eq:back_comp} by
  replacing $\init{\Phi}$ and $ \init{\chi}$ with $\init{\bar{\Phi}}$ and 
  $\init{\bar{\chi}}$.
  The proof of Proposition \ref{prop:chr_sym} implies that we need 
  bounds for the representations
  \begin{gather*}
    \abs{D\init{\bar{\chi}}}_{\infty}, ~\abs{D^2 \init{\bar{\chi}}}_{\infty},~
    \abs{D\init{\bar{\chi}}}_{C^s},~
    \abs{D^2 \init{\bar{\Phi}}}_{\infty},~ \abs{D^3 \init{\bar{\Phi}}}_{\infty} 
    \quad\text{and} \quad \abs{D^2 \init{\bar{\Phi}}}_{C^s}.
  \end{gather*}
  These can be derived from the bounds for the cut-off function $\zeta$
  in \eqref{eq:cutoff_est} and from the assumptions on $\bar{\Phi}_{\sigma}$
  and $\bar{\alpha}_{\sigma}, \bar{\beta}_{\sigma}, \bar{\nu}_{\sigma}$
  stated in \eqref{eq:assum_bd_pt} and \eqref{eq:atlas_shift_der}.
\end{proof}
\begin{proof}[\textbf{Proof of Proposition \ref{prop:atlas_diffeo}}]
  The preceding considerations show that the IVP 
  \eqref{eq:repa_solve} satisfies the assumptions
  of Theorem \ref{qlin_ex}. 
  From the asymptotic existence Theorem \ref{asym_ex} we derive 
  a solution $\Psi(t) = w(t) + \psi(t)$ of equation \eqref{eq:repa}
  attaining the  initial values \eqref{eq:def_trafo_initial} within the ball
  $B^e_{\hat{\theta} \rho_1/2}(0)$. 
  From the choice of the domain $W$
  of the RHS based on the set $\Omega$ defined in \eqref{eq:def_O_diffeo} 
  we obtain
  that the differential of $\Psi(t)$ is invertible for all $(t,x)$ in the
  domain of $\Psi$. 

  In analogy to the proof of Theorem
    \ref{thm:ex_uni_atlas_hyp} we intend to 
  impose a bound on the time parameter by the constant
  $\tilde{T}_1$ defined by
  $\tilde{T}_1 = K_1^{-1} \, 
  \theta \rho_1 /4$. This choice guarantees that $\Psi(t,z)$ solves
  the original equation \eqref{eq:repa} for $\abs{z} < \hat{\theta}\rho_1/2$
  and $0\le t \le \tilde{T}_1$ if it exists up to this time. 
  Remark \ref{rem:extend}
  applies so that the solution %
  can be extended to negative time parameters.

  The proof of the uniqueness claim is analogous to the 
  proof of the uniqueness result in Theorem
  \ref{thm:ex_uni_atlas_hyp}. Since the coefficients are taken from
  the solution obtained by Proposition \ref{prop:cutoff_ex_hyp}, the value
  of the slope $c_0$
  of the uniqueness cone is described in the proof of the uniqueness 
  result of Theorem \ref{thm:ex_uni_atlas_hyp}.
\end{proof}
In Lemma \ref{lem:shift_sgr_speed} we obtained the transition from the 
assumptions 
\ref{assum_diff_lapse} to the local assumptions \ref{assum_diffeo_atlas} via 
the special graph  representation
introduced in section \ref{sec:graph_repr_hyp}. The next step is to show that 
the 
assumptions \eqref{eq:cond_initial_diffeo}
on the diffeomorphism $\psi_0$ localize in a way that the assumptions 
\eqref{eq:assum_atlas_init_diffeo} are satisfied.
\begin{lem}
  Suppose $\psi_0 \in C^{s+2}$ satisfies the assumptions 
  \eqref{eq:cond_initial_diffeo}.
  Let $p \in M$. Suppose $x$ are coordinates on $M$ with center $p$ and
  $\bar{x}$ are coordinates with center $\psi_0(p)$. 
  Let %
  $y$ and $\bar{y}$ be the special coordinates on $N$ 
  introduced in section \ref{sec:special_coord}
  such that $y \circ \varphi \circ x^{-1}$ and $\bar{y} \circ \bar{\varphi}
  \circ \bar{x}^{-1}$ are the special
  graph representations of $\varphi$ and $\bar{\varphi}$
  obtained in section \ref{sec:graph_repr_hyp}.
  Set
  $\bar{\psi}_0 = \bar{x} \circ \psi_0 \circ x^{-1}$.
  Then $\bar{\psi}_0$ satisfies the local assumptions
  \eqref{eq:assum_atlas_init_diffeo}.
\end{lem}
\begin{proof}
  The claim follows by considering the comparison
  of the eigenvalues of the representations of the metrics $\ig$ and 
  $\mathring{\bar{g}}$ 
  in the special graph representation with the Euclidean
  metric  stated in
  \eqref{eq:hyp_metric_compare}.
\end{proof}
Using considerations elaborated for the membrane equation in section 
\ref{lorentz_ex} we
are now able to give a proof of the main result of this section.
\begin{proof}[\textbf{Proof of Proposition \ref{prop:ex_uni_diffeo}}]
  \begin{itemize}
  \item \underline{Existence} \\ %
    We can 
  show  analogously to the proof of Proposition \ref{prop:uni_coord_hyp} that
    solutions $\Psi_{\lambda}$ derived in proposition
    \ref{prop:atlas_diffeo} are independent of the specific decomposition.
    A solution can be transferred to another coordinate system since
    equation \eqref{eq:repa} is invariant under change of coordinates.
    The  parameter $\tilde{T}_1$ controlling the region where the constructed
    solution solves the unmodified equation was set to $K_1^{-1} \theta/4$
    in Proposition \ref{prop:atlas_diffeo}.
    The constant $T_0$ controlling the height of the uniqueness
    cone depends on the slope $c_0$. Since we used
    the coefficients of a solution derived from proposition 
    \ref{prop:cutoff_ex_hyp}, 
    the definition of $c_0$ made in \eqref{eq:hyp_est_c_0} can be used.
    Therefore, it is possible to construct a solution $\Psi$ defined on
    an interval $\times M$ just as we have proceeded in the proof of
    Theorem \ref{thm:ex_uni_geom_hyp} in the case of the membrane equation. 
  \item \underline{Uniqueness} \\
    Uniqueness follows as in the proof of the uniqueness result of Theorem
    \ref{thm:ex_uni_geom_hyp}, i.\,e.\  using the local uniqueness
    result of Proposition \ref{prop:atlas_diffeo}. \qedhere
  \end{itemize}
\end{proof}

\section{Main results}
\label{main}
In this section we will consider existence and uniqueness for the membrane
equation \eqref{eq:H0} in purely geometric terms.
An existence result will be obtained including a geometric notion of 
time of existence of a solution. The uniqueness result will show that
two submanifolds solving the IVP coincide on a neighborhood
of the initial submanifold.
\subsection{Existence}
In this section we present the main existence theorem for the Cauchy
problem \eqref{eq:geom_problem} in Theorem \ref{thm:main_ex}. 
The statement is given in a scale invariant
way providing the scaling behaviour of the time of existence of a solution.
Furthermore, it is shown that a smooth solution exists for smooth ambient
manifolds and smooth initial data.

Throughout this section we use the following assumptions.
Let $N^{n+1}$ be an $(n+1)$-dimensional manifold endowed with
a Lorentzian metric $h$. Suppose $\Sigma_0$ is an
$m$-dimensional spacelike regularly immersed submanifold of $N$ with
$m \le n-1$. Suppose $\varphi: M^m \rightarrow N$ is an immersion
with $\im \varphi = \Sigma_0$.
Let $\nu$ be a timelike future-directed unit vector field along $\varphi$
normal 
to $\Sigma_0$.
Recall that we use $\D$ to denote the Levi-Civita connection on $N$ 
corresponding to $h$.

Let $s_0 \ge 2$ be an integer and let $R > 0$ be the scale.
\begin{assum}[on the ambient space]
  \label{assum_ambient_main}
  Let $N$ admit a time function $\tau$. Let $\psi$ denote
  the lapse of the time foliation  induced by $\tau$ 
  as defined
  by \eqref{eq:def_lapse} and let $E$ denote 
  the flipped Riemannian metric defined in \ref{defn:def_E}.
  
  Suppose there are constants $C_1, C_2,
  C^N_{\ell}$ and $C^{\tau}_{\ell}$ independent of $R$
  such that the following inequalities
  are satisfied 
  \begin{gather*}
    \begin{split}
     &  C_1 \le  R^{-1}\psi \le C_2,~
       R^{2+\ell}\abs{\D^{\ell} \RRiem}_E \le C^N_{\ell} \quad\text{for } 
    0 \le \ell \le s_0+1 
    \\
    \text{and}\quad & 
    R^{1+\ell}\abs{\D^{\ell} (\D\tau)}_E 
    \le C_{\ell}^{\tau} \quad\text{for } 
    1 \le \ell \le s_0+2
    \end{split}
  \end{gather*}
  where $\D(\D \tau)$ denotes the $(1,1)$-tensor
    obtained by applying the covariant derivative to the gradient of the time 
    function $\tau$.
\end{assum}
\begin{assum}[on the initial submanifold]
  \label{assum_init_smf_main}
  
  Suppose there exist constants $\omega_1, C^{\varphi}_{\ell}$ 
  independent of $R$ such that
    \begin{gather*}
      \begin{split}
         \inf\{ -h( \gamma , \hT ) : \gamma \text{ timelike future-directed 
           unit normal}& \text{ to }\Sigma_0
      \}\le \omega_1,
      \\
      R^{\ell + 1}\abs{\hn^{\ell} \II}_{\ig,E}  \le C_{\ell}^{\varphi}
      \quad\text{for }0 \le \ell \le s_0&
      \end{split}
    \end{gather*}
    where $\hT$ denotes the future-directed timelike unit normal to
    the time foliation on $N$ and
    $\II$ denotes the second fundamental form of $\Sigma_0$.
\end{assum}

\begin{assum}[on the initial direction]
  \label{assum_direction_main}
  Suppose there exist constants $L_3, C^{\nu}_{\ell}$ independent of $R$ such that 
    \begin{gather*}
      - h(\nu, \hT) \le L_3
      \quad\text{and}\quad
      R^{\ell}\babs{\hn^{\ell} \nu}_{\ig,E} 
      \le C^{\nu}_{\ell} \quad \text{ for }
      1 \le \ell \le s_0+1
    \end{gather*}
\end{assum}
The next definition gives a notion of ``time of existence'' which respects the 
geometric
behaviour of a solution in contrast to the time parameter
obtained in the existence theorems 
\ref{thm:ex_uni_geom}, \ref{thm:ex_uni_geom_hyp} and 
\ref{thm:ex_uni_lapse_hyp}.
\begin{defn}[Time of existence]
  \label{def:proper_time}
  Let $\Sigma$ be a solution of the IVP \eqref{eq:geom_problem}.
  The \emph{time of existence} $\tau_{\Sigma}$ of $\Sigma$ is given by
  \begin{gather*}
    \begin{split}
      \tau_{\Sigma} :=
      \inf_{p \in \Sigma_0}\sup\{\text{length } 
      \text{of all timelike } & \text{future-directed}
        \\
        &
      \text{curves in $\Sigma$ emanating from } p \}.
    \end{split}
  \end{gather*}
\end{defn}
Our aim is an existence result for the IVP \eqref{eq:geom_problem} which
includes a lower bound on the time of existence according to the preceding
definition. The following assumptions on ambient manifold,
initial submanifold, and initial direction will provide us with
such a lower bound. 
\begin{thm}
  \label{thm:main_ex}
 Let $s > \tfrac{m}{2} + 1$ be an integer and let $\rho > 0$ be a constant. 
  Assume that for each point 
  $q \in \Sigma_0$
  there exists a neighborhood $V \subset N$ of $p$ such that
  $N$ admits a time function $\tau$ in $V$. Let $E$ denote the flipped
  metric on $V$ defined in \ref{defn:def_E}. Let $B^E_{R\rho}(q) \subset V$
  and suppose $h, \tau \in C^{s+3}$
  satisfy the assumptions \ref{assum_ambient_main} 
  in $V$
  with $s_0 = s$ and constants independent of $q$.
  Let $\Sigma_0$ be of class $C^{s+2}$ and let the immersion $\varphi$ with 
  image $\Sigma_0$ 
  satisfy the assumptions \ref{assum_init_smf_main} with $s_0 = s$ and
  constants independent of $q$. Let  the initial direction
  $\nu \in C^{s+1}$ %
  satisfy  the assumptions \ref{assum_direction_main} with $s_0 = s$
  and  constants independent of $q$.
  
  Then there exists an open $(m + 1)$-dimensional regularly immersed
  Lorentzian submanifold
  $\Sigma$ of class $C^2$ satisfying the IVP
  \begin{gather*}
    H(\Sigma) \equiv 0,~ \Sigma_0 \subset \Sigma,
    \text{ and }\nu \text{ is tangential to }\Sigma.
  \end{gather*}
  Furthermore, there exists a constant $\delta>0$ such that
  \begin{gather*}
    \tau_{\Sigma} \ge R\,\delta %
  \end{gather*}
\end{thm}

\begin{proof} %
  Then $\varphi$ satisfies the assumptions \eqref{eq:assum_smf_hyp} and 
  $\nu \circ \varphi$
  satisfies the assumptions \eqref{eq:direction}.
  From the assumptions on the time function of the ambient manifold
  we get that the special coordinates introduced in section 
  \ref{sec:special_coord} are defined in a ball
  depending %
  the radius $\rho$. %

  Let $F$ denote the solution of the IVP \eqref{eq:param_ivp} obtained
  from Theorem \ref{thm:ex_uni_lapse_hyp} with initial lapse 
  equal to 1 and initial shift equal to 0.
  A lower bound for the time of existence of the solution $\Sigma:= \im F$
  is obtained by the uniformity of the assumptions by taking remark
  \ref{rem:ex_time_hyp} into account.
\end{proof}

\begin{rem}
  \label{rem:ex_main_non_uniform}
  The proof shows that the theorem applies to the situation where
  the assumptions are valid only locally with constants depending
  on the specific point.
\end{rem}
In the sequel we will specialize the result of Theorem \ref{thm:main_ex}
to various types of initial submanifolds which will be subject to the next
definitions.
\begin{defn}
  \label{defn:loc_emb}
  Let $\Sigma_0$ be a regularly immersed submanifold of dimension $m$
  being the image of an immersion
  $\varphi: M^m \rightarrow N$.
  $\Sigma_0$ is called a \emph{locally embedded} submanifold, if
  for every point $q \in \Sigma_0$ there exist open sets
  $q \in V \subset N$ and $U \subset M$ such that
  \begin{gather}
    \label{eq:defn_loc_emb}
    \varphi: U \rightarrow \varphi(U) \text{ is a diffeomorphism and} \quad
    \varphi^{-1}(V \cap \Sigma_0) = U.
  \end{gather}
\end{defn}

\begin{defn}
  \label{defn:fin_loc_emb}
  Let $\Sigma_0$ be a regularly immersed submanifold of dimension $m$
  being the image of an immersion %
  $\varphi: M^m \rightarrow N$.
  $\Sigma_0$ is called \emph{regularly immersed with locally finite intersections}, if
  for every point $q \in \Sigma_0$ there exist a neighborhood
  $V \subset N$ of $q$ and finitely many open pairwise disjoint sets 
  $U_{\ell} \subset M$ 
  such that
  \begin{gather}
    \label{eq:defn_fin_loc_emb}
    \varphi: U_{\ell} \rightarrow \varphi(U_{\ell}) 
    \text{ is a diffeomorphism for every $\ell$ and} \quad
    \varphi^{-1}(V \cap \Sigma_0) = \bigcup U_{\ell}.
  \end{gather}
\end{defn}
The following corollaries will show that a solution to the IVP 
\eqref{eq:geom_problem}
for the membrane equation can be constructed in such a way that 
these properties of initial submanifolds are preserved.
\begin{cor}
  \label{cor:main_ex_loc_emb}
  Let $s > \tfrac{m}{2} + 1$ be an integer. Let $\Sigma_0$ be locally 
  embedded and for a point $q \in \Sigma_0$ let $U\subset M$ and $V \subset N$ 
  denote
  sets satisfying the conditions \eqref{eq:defn_loc_emb} of definition 
  \ref{defn:loc_emb}.
  Suppose that for every point $q \in \Sigma_0$ the immersion
  $\varphi \in C^{s+2}$ satisfies the assumptions \ref{assum_init_smf_main} in 
  $U$, and $N$ with metric $h\in C^{s+3}$ and time function $\tau \in C^{s+3}$
  defined in $V$ satisfies the assumptions \ref{assum_ambient_main}
  in $V$ with $s_0 = s$. Assume further
  the initial direction $\nu \in C^{s+1}$ to satisfy
  the assumptions \ref{assum_direction_main} in $V$ with $s_0 = s$.
  
  Then there exists a locally embedded timelike $(m+1)$-dimensional
  submanifold $\Sigma$
  of class $C^2$ solving the IVP 
  \begin{gather*}
    H(\Sigma) \equiv 0,~ \Sigma_0 \subset \Sigma,
    \text{ and }\nu \text{ is tangential to }\Sigma.
  \end{gather*}  
\end{cor}

\begin{proof}
  We will appeal to the proof of Theorem \ref{thm:ex_uni_geom_hyp}.
  From the assumptions on $N$ we obtain existence of the special
  coordinates constructed in section \ref{sec:special_coord}, at least 
  locally in $V$.
  Therefore, the special graph representation of $\varphi$ introduced
  in section \ref{sec:graph_repr_hyp} can be constructed in a neighborhood of 
  $\varphi^{-1}(q) = p \in U$. This is possible by the assumptions 
  \ref{assum_init_smf_main}  on $\Sigma_0$.
  
  Proposition \ref{prop:graph_assum_hyp} yields that the local existence
  Theorem \ref{thm:ex_uni_atlas_hyp} for the reduced membrane equation can be 
  applied. From Proposition \ref{prop:uni_coord_hyp} we get that
  the family of solutions is independent of the chosen coordinates.
  It follows that for every point $q \in \Sigma_0$ there exist neighborhoods
  $q \in V \subset N$ and $\varphi^{-1}(q) \in U \subset M$ and 
  a solution $F_q$ of the membrane equation defined on a neighborhood
  $W$ of $\{0\} \times U$ in $\rr \times M$
  with values in $V$. By considering proposition
  \ref{prop:embedded} we may shrink the domain of the solution $F_q$ such that
  $F_q$ is a diffeomorphism onto its image.

  Consider the family 
  $\mathcal{U} = \bigl(\tilde{U}(\varphi^{-1}(q)) \bigr)_{q\in \Sigma_0}$, where
  $\tilde{U}(\varphi^{-1}(q))$ is the neighborhood of $\varphi^{-1}(q)$
  on which the local solution associated to $\varphi^{-1}(q)$
  is defined.
  Choose a locally finite covering subordinate to $\mathcal{U}$.
  Use this covering to define a mapping 
  $F: W \subset \rr \times M \rightarrow N$ by \eqref{eq:ex_def}.
  This construction is well-defined by virtue of proposition 
  \ref{prop:uni_coord_hyp}.
  
  By construction it follows that $\Sigma := \im F$
  is a locally embedded submanifold of class $C^2$.
\end{proof}
To obtain a solution for regularly immersed initial submanifolds with
finite intersection it is necessary to change the way in which the local 
solutions are pieced together.
\begin{cor}
  \label{cor:main_ex_finite}
  Let $s > \tfrac{m}{2} + 1$ be an integer. 
   Let $\Sigma_0$ be regularly immersed with locally finite
  intersections and for a point $q \in \Sigma_0$ let $U_{\ell}\subset M$ and 
  $V \subset N$ denote
  sets satisfying the conditions \eqref{eq:defn_fin_loc_emb} of definition 
  \ref{defn:fin_loc_emb}.
  Suppose that for every point $q \in \Sigma_0$
  the immersion $\varphi \in C^{s+2}$ satisfies the assumptions 
  \ref{assum_init_smf_main} in 
  $U_{\ell}$ for every $\ell$,  and $N$ with metric $h\in C^{s+3}$ and time 
  function $\tau \in C^{s+3}$
  defined in $V$ satisfies the assumptions 
  \ref{assum_ambient_main}
  in $V$ with $s_0 = s$. Assume %
  the initial direction $\nu%
  \in C^{s+1}$ 
  to satisfy
  the assumptions \ref{assum_direction_main} in each $U_{\ell}$ with $s_0 = s$.

  Then there exists a timelike $(m+1)$-dimensional
  regularly immersed submanifold $\Sigma$ with locally finite intersections 
  of class $C^2$ solving the IVP 
  \begin{gather*}
    H(\Sigma) \equiv 0,~ \Sigma_0 \subset \Sigma,
    \text{ and }\nu \text{ is tangential to }\Sigma.
  \end{gather*}  
\end{cor}

\begin{proof}
  We mimic the proof of Corollary \ref{cor:main_ex_loc_emb}.
  For each $q \in \Sigma_0$ we pick the finitely many $U_{q,\ell}$ and
  solve the membrane equation in $U_{q,\ell}$ with values in $V$.
  We shrink the domain of the solutions $F_{q,\ell}$ to obtain embeddings.
  By shrinking the 
  set $V \subset N$ to a subset $\tilde{V}_q \subset N$ we achieve
  that 
  \begin{gather*}
    \bigl(F_{q,\ell}\bigr)^{-1}\bigl(\tilde{V}_q \cap \im F_{q,\ell}\bigr)
    = W_{q,\ell} \subset \dom(F_{q,\ell}) \subset \rr \times U_{q,\ell}
    \quad\text{for all } \ell.
  \end{gather*}
  Consider the family $\mathcal{U} = 
  \bigl(\tilde{V}_q \cap \Sigma_0 \bigr)_{q \in \Sigma_0}$.
  Choose a locally finite covering $\bigl(\tilde{V}_{q_{\lambda}} \cap 
  \Sigma_0 \bigr)_{\lambda \in \Lambda}$
  of $\Sigma_0$
  subordinate to $\mathcal{U}$.
  Let $\tilde{U}_{q,\ell}$ denote the part of $W_{q,\ell}$ which belongs
  to $\{0\} \times M$.
  Consider the family $\bigl(\tilde{U}_{q,\ell} \bigr)_{q\in \Sigma_0}$.
  Then the family $(\tilde{U}_{q_{\lambda}, \ell})_{\lambda}$ is a locally finite
  covering of $M$ subordinate to $\bigl(U_{q,\ell} \bigr)_{q \in \Sigma_0}$
  due to the finiteness of the sets $U_{q, \ell}$ for fixed
  $q$.
  
  We use this covering to define a mapping 
  $F: W \subset \rr \times M \rightarrow N$ by \eqref{eq:ex_def};
  proposition 
  \ref{prop:uni_coord_hyp} shows that it is well-defined.
  
  By construction it follows that $\Sigma := \im F$
  is  regularly immersed with locally finite intersections 
  and of class $C^2$.
\end{proof}
We show that smooth data lead to a smooth solution of the IVP
\eqref{eq:geom_problem} for the membrane equation again respecting the type of 
the initial submanifold.
\begin{cor}
  \label{cor:main_ex}
  Assume $(N,h)$ to be  smooth and suppose
  $\Sigma_0$ is  smooth %
  \begin{enumerate}
  \item regularly immersed 
  \item locally embedded 
  \item regularly immersed  with locally finite intersections.
  \end{enumerate}
  Suppose $N$ admits a smooth time function $\tau$ in a neighborhood 
  of the initial submanifold $\Sigma_0$. Assume  the initial direction
  $\nu$ to be smooth.
  
  Then there exists an open smooth $(m + 1)$-dimensional timelike
  \begin{enumerate}
  \item regularly immersed submanifold $\Sigma$ 
  \item locally embedded submanifold $\Sigma$ 
  \item regularly immersed submanifold  $\Sigma$ with locally finite 
    intersections, respectively,
  \end{enumerate}
  satisfying the IVP
  \begin{gather*}
    H(\Sigma) \equiv 0,~ \Sigma_0 \subset \Sigma,
    \text{ and }\nu \text{ is tangential to }\Sigma.
  \end{gather*}
\end{cor}

\begin{proof}
  For each integer $\ell_0$ it follows from the 
  smoothness
  of $h, \tau$, the immersion $\varphi$, and the initial direction that
  the assumptions \ref{assum_ambient_main} to \ref{assum_direction_main} 
  are satisfied for $s_0 = s > \tfrac{m}{2} + 1 + \ell_0$ in
  a neighborhood about each point $q \in \Sigma_0$.
  Depending on $\Sigma_0$ to be regularly immersed, locally embedded, or
  regularly immersed with locally finite intersections we apply theorem 
  \ref{thm:main_ex}
  and Remark \ref{rem:ex_main_non_uniform}
  or the corollaries
  \ref{cor:main_ex_loc_emb} and \ref{cor:main_ex_finite}, respectively, and
  obtain a
  solution $\Sigma$ of class $C^{2 + \ell_0}$ by taking remark 
  \ref{rem:hyp_sol_diff} into account. By considering the constructions made 
  in the proof of corollary 
  \ref{cor:main_ex_finite}
  and Theorem \ref{thm:main_ex} the other cases follow.
\end{proof}

\subsection{Uniqueness}
In this section we will consider the uniqueness claim \eqref{eq:geom_uni}
of the main problem.
In the preceding section, we  showed that the construction of a solution 
accomplished %
in Theorem \ref{thm:ex_uni_lapse_hyp} are independent of the choice of an 
immersion of the
initial submanifold, as well as of the initial lapse and shift. It therefore 
remains to
construct an immersion of a solution to the membrane equation 
which is in harmonic map gauge
w.r.t.\ the background metric defined by its initial values as in
\eqref{eq:back_metric}.

Throughout this section we use the following assumptions.
Let $N^{n+1}$ be an $(n+1)$-dimensional manifold endowed with
a Lorentzian metric $h$. Suppose $\Sigma_0$ is a 
$m$-dimensional spacelike locally embedded submanifold of $N$. 
Suppose $\varphi: M^m \rightarrow N$ is an immersion
with $\im \varphi = \Sigma_0$ satisfying the conditions \eqref{eq:defn_loc_emb}
 of
Definition \ref{defn:loc_emb}.
Let $\nu$ be a timelike future-directed unit vector field on $\Sigma_0$
normal 
to $\Sigma_0$.
\begin{thm}
  \label{thm:main_uni}
  Assume $(N,h)$ to be smooth  and suppose
  $\Sigma_0$ is smooth.
  Let $N$ admit a smooth time function $\tau$ in a neighborhood 
  of the initial submanifold $\Sigma_0$. Suppose the initial direction
  $\nu$ is  smooth.

  Let $\Sigma_1$ and $\Sigma_2$ be two open smooth $(m + 1)$-dimensional
  locally embedded
  Lorentzian submanifolds of $N$ solving the IVP
   \begin{gather*}
    H(\Sigma) \equiv 0,~ \Sigma_0 \subset \Sigma,
    \text{ and }\nu \text{ is tangential to }\Sigma.
  \end{gather*}

  Then there exists a neighborhood $\Sigma_0 \subset V \subset N$ of $\Sigma_0$
  such that
  \begin{gather*}
    V \cap \Sigma_1 = V \cap \Sigma_2.
  \end{gather*}
\end{thm}
Our strategy to prove this theorem will be to compare an arbitrary solution
with the solution constructed in the previous section.
To apply the uniqueness result of section \ref{lorentz_ex} we need to 
construct
an immersion satisfying the IVP \eqref{eq:param_ivp}.
\begin{prop}
  \label{prop:ex_embed}
  Let $(N,h),~\Sigma_0$ and $\nu$ satisfy the assumptions of Theorem
  \ref{thm:main_uni}.
  Let $\Sigma$ be a smooth locally embedded 
  solution
  to the IVP \eqref{eq:geom_problem}.

  Then there exists an immersion  $F:W \subset
  \rr \times M \rightarrow N$ with $\im F \subset \Sigma$ is a locally 
  embedded submanifold.
  Furthermore, $F$ has the properties 
  that $\dt F$ is timelike and that $F(t): M \rightarrow N$ has a spacelike
  image. The initial values of $F$ are given by
  $\restr{F} = \varphi$ and $\restr{\dt F}
  = \nu \circ \varphi$. 
\end{prop}

\begin{proof}
  Let $p \in M$ and let $\gamma_{\varphi(p)}(t)$ be a geodesic in $\Sigma$
  attaining
  the initial values $\gamma_{\varphi(p)}(0) = \varphi(p)$ and
  $\dot{\gamma}_{\varphi(p)}(0) = \nu \circ \varphi(p)$.
  Set $F(t,p) = \gamma_{\varphi(p)}(t)$, then
  $F$ is an immersion, since $\varphi$ is assumed to be an immersion
  and $\nu$ is assumed to be unit timelike.
  The claim about the initial values of $F$ follows from the properties of the 
  geodesic. %
  
  This construction is similar to that of Gaussian coordinates (see e.g.
  \cite{Wald:1984}). In an analogous way as it is shown that those are 
  coordinates
  it follows that $\dt F$ is timelike and that $F(t): M \rightarrow N$ has a 
  spacelike image. From the same argument we derive that
  the geodesics do not cross in a neighborhood of any point
  in $\Sigma_0$ as long as $\varphi$ is an embedding and
  $\Sigma$ is an embedded submanifold around that point.
\end{proof}
The following proposition makes contact with the reduction in section 
\ref{sec:reduction}.
If $\hat{g}$ denotes the special background metric on $\rr \times M$ 
defined in \eqref{eq:back_metric}, we will construct a reparametrization
of $F$ 
such that it satisfies the membrane equation in harmonic map gauge w.r.t.
$\hat{g}$.
\begin{prop}
  \label{prop:diffeo_ex}
    Let $(N,h),~\Sigma_0$ and $\nu$ satisfy the assumptions of Theorem
  \ref{thm:main_uni}.
  Let $\Sigma$ be a smooth locally embedded 
  solution
  to the IVP \eqref{eq:geom_problem}.
  Suppose $F:W \subset
  \rr \times M \rightarrow N$ of $\Sigma$ is the immersion with locally
  embedded image constructed in Proposition \ref{prop:ex_embed}.

  Then, for all $p \in M$,
  there exists a local diffeomorphism $\Psi$ defined on a neighborhood
   of $(0,p) \in \rr \times M$
  such that
  $F \circ \Psi^{-1}$ is a solution of the membrane equation
  in harmonic map gauge w.r.t.\ the background
  metric $\hat{g}$ defined by \eqref{eq:back_metric} using the  initial values
  of $F$. Further, the initial values of $F \circ \Psi^{-1}$ coincide
  with the initial values of $F$.
\end{prop}

\begin{proof}
  The condition \eqref{eq:harm_cond_coord} for a solution $F$ to be in 
  harmonic map gauge
  computed for $F \circ \Psi^{-1}$ leads to a harmonic map equation 
  analogous to \eqref{eq:repa} with $\hat{\bar{g}}$ replaced by $\hat{g}$.
  We use the coordinates $x$ on $M$ and $y$ on $N$ belonging to the
  special graph representation of section \ref{sec:graph_repr_hyp}.
  We have to show that the metric $g = F^{\ast} h$ has the 
  property that the existence Theorem \ref{qlin_ex}
  is applicable. To this end we will show that
  it satisfies estimates analog to \eqref{eq:cond_pos} and
  \eqref{eq:cond_der} in the coordinates $x$.
  
    Since the geodesic used to define the immersion $F$ in the proof of the 
    previous proposition remains orthogonal to the slices of constant parameter
  $t$ the metric has the following form
 \begin{gather*}
    g_{\mu\nu}(t,z) = 
    \begin{pmatrix}
      -1 & 0 \\
      0 & g_{ij}(t,z)
    \end{pmatrix}
    .
  \end{gather*}
  The components $g_{ij}(0,z)$ are the expression of the metric
  induced on $M$ by $\varphi$ w.r.t.\ the coordinates $x$ which belong
  to the special graph representation. Therefore, from
  estimate \eqref{eq:ind_metric_pos_est} and lemma
  \ref{lem:graph_est_hyp} we infer control over the metric
  at $t = 0$.
  The derivatives of the matrix $\bigl(g_{\mu\nu}(t,z)\bigr)$
  can be bounded up to order $s+1$ if the domain of the coordinates
  $x$ and the parameter $t$ are bounded appropriately.
  From this fact the remaining estimates follow.

  A cut-off process analogous to \eqref{eq:target_inter_metric} shows that
  the metric components of $g$ and the corresponding Christoffel
  symbols can be estimated to meet the conditions of Theorem \ref{qlin_ex}.
  Existence of a solution $\Psi$ to the harmonic map equation
  \eqref{eq:repa} follows from the consideration made in section 
  \ref{sec:constr_diffeo},
  where we used the initial values 
  \begin{gather*}
    \restr{\Psi}(z)= (0,z)
    \quad\text{and}\quad
    \restr{\partial_{t} \Psi}(z)
    = \tbinom{1}{0}. \qedhere
  \end{gather*}
\end{proof}
We are now in the position to give a proof of the main uniqueness result.
\begin{proof}[\textbf{Proof of 
    Theorem \ref{thm:main_uni}}]
  Let $F_0$ be the solution to  the IVP \eqref{eq:param_ivp}
  with initial values $\restr{F_0} = \varphi$ and $\restr{\dt F_0}
  = \nu \circ \varphi$ constructed in Corollary \ref{cor:main_ex_loc_emb}.
  Let $\widehat{\Sigma} := \im F_0$ denote the locally embedded image
  of the solution $F_0$.
  Our strategy will be to compare a smooth solution $\Sigma$ to the
  IVP \eqref{eq:geom_problem}
  with the solution $\widehat{\Sigma}$.

  From Proposition \ref{prop:ex_embed} we obtain that $\Sigma$ admits an 
  immersion $F:W \subset \rr \times M \rightarrow N$ with locally embedded
  image and initial values
  $\restr{F} = \varphi$ and $\restr{\dt F} = \nu \circ \varphi$
  in a neighborhood of the initial submanifold $\Sigma_0$.

  From Proposition \ref{prop:diffeo_ex} we
  obtain a local diffeomorphism $\Psi$ of $\rr \times M$
  defined in a neighborhood of $\{0 \} \times M$ such that
  $F \circ \Psi^{-1}$ is in harmonic map gauge w.r.t.\ the background metric
  defined by the initial values of $F$. We shrink the domain of this
  mapping to ensure that $F$ is a diffeomorphism, defined on
  the image of $\Psi$, onto
  its image.

  By applying the uniqueness result of Theorem \ref{thm:ex_uni_geom_hyp} and 
  taking
  Remark \ref{rem:hyp_non_uniform} into account we obtain that locally 
  $F \circ \Psi^{-1}$
  and $F_0$ coincide. Since $F$ is an embedding of a portion
  of $\Sigma$, also $F\circ \Psi^{-1}$ is a local embedding of $\Sigma$.
  Hence, the desired result follows.
\end{proof}

\begin{rem}
  The existence and uniqueness results for the membrane equation
  \ref{eq:H0} are kept in  purely geometric terms.
  This indicates that a notion of  maximal developments as introduced
  by Y.\ Choquet-Bruhat (cf.\ \cite{CBG:1969}) 
  in the context of the Einstein equations should also be possible for the
  membrane equation.
\end{rem}

\newpage

\renewcommand{\appendixtocname}{Appendix}
\renewcommand{\appendixpagename}{Appendix}

\begin{appendices}
\section{Proofs of the statements in section \ref{sec:qlinear}}
\label{sec:proofs}
In this section we state the proofs of the lemmas and propositions used
to the existence Theorem \ref{qlin_ex}
for the hyperbolic IVP \eqref{second}.

\begin{proof}[\textbf{Proof of Lemma \ref{u00}}]
We know from the assumptions on the coefficients that 
\begin{gather*}
  C \ge 
  4 c_E^{1/2} \tau_{s}^{1/2}(\lambda + (m+1) \mu).
\end{gather*}
We have to distinguish between two cases:
\begin{case}
$4 c_E^{1/2} \tau_{s}^{1/2}(\lambda + (m+1) \mu)\ge 1$:  Choose $r> 0$ as the largest radius satisfying $
(v_0, v_1) \in B_r(u_0) \times B_r(u_1)$, then $(v_0, v_1) \in W$.
Set $\delta = r/2$. $C$ is bounded by the assumption 
$\norm{(g^{\mu\nu})}_{e,s,\mathrm{ul}}
 \le K$ on $W$.
Therefore we can find a $\rho >0$ such that $\rho C^{1/2} \le 
\delta/3$. In this case we get directly that $\rho \le \delta/3$. Since 
$H^{s+2}\times H^{s+1}$ is dense in $H^{s+1}\times H^{s}$ we find $u_{00}$ such
that $E_{s+1}(\init{u} - u_{00}) \le \rho$. The first condition comes
from
$$ E_{s+1}(v - \init{u}) \le E_{s+1}(v - u_{00})+ E_{s+1}(\init{u}
- u_{00}) \le \delta + \rho \le 4 \delta /3 < r
$$
\end{case}
\begin{case}
$4 c_E^{1/2} \tau_{s}^{1/2}(\lambda + (m+1) \mu)\le 1$: 
Choose $\delta$ as above. We can find a $\rho$
such that $\rho C^{1/2} \le 4 c_E^{1/2} \tau_{s}^{1/2}
(\lambda + (m+1) \mu) \delta /3 \le \delta/3$.
$C \ge 4 c_E^{1/2} \tau_{s}^{1/2}(\lambda + (m+1) \mu)$ 
gives us now that $\rho \le \delta/3$. The rest
is done as in the above case.
\end{case}
\end{proof}

\begin{proof}[\textbf{Proof of Lemma \ref{e_est}}]
  Let $(u_{\ell})$ be a Cauchy-sequence in $(Z_{\delta, L'}, d)$. Then
  $(u_{\ell})$ , $(\partial_t u_{\ell})$ are Cauchy-sequences in $C([0,T], H^1)$
  and $C([0,T], L^2)$ resp.

  We conclude $u_{\ell} \rightarrow u$ in $C([0,T], H^1)$ and $\partial_t 
  u_{\ell} \rightarrow w$ in $C([0,T], L^2)$.
  By the fundamental theorem of calculus we get 
  $$ u_{\ell}(t+\tau) = u_{\ell}(\tau) + \tint_{\tau}^{t+\tau} \partial_t u_{\ell}
  (s)\, ds\qquad\text{in } L^2
  $$
  Passing to the limit yields
  $$ u(t+\tau) = u(\tau) + \tint_{\tau}^{t+\tau} w(s)\, ds \qquad\text{in } L^2
  $$
  Therefore $\partial_t u = w$.

  It follows from $E_{s+1}(u_{\ell} - u_{00}) \le \delta$ that $(u_{\ell})$ is a 
  bounded 
  sequence in the Hilbert-space $H^{s+1}$ and $(\partial_t u_{\ell})$ analogously 
  in $H^{s}$. We have therefore weak convergence of subsequences. %
  We consider the sequence $(u_{\ell})$, the other is treated analogously.
  It holds that
  $u_{\ell_n} \rightharpoonup v$ in $H^{s+1}$ pointwise and $\partial_t u_{\ell_n}
  \rightharpoonup h$ in $H^{s}$. Again, the fundamental theorem of calculus
  gives us after passing to the weak limit in $H^{s}$ that $\partial_t v = h$.
  Since we have a limit in $L^2$-norm of $u_{\ell_n}$ and $\partial_t u_{\ell_n}$
  we conclude that $u=v$. 
  
  Consider the normed space $H^{s+1} \times H^{s}$ with norm 
  $\norm{\, . \,} = 
  (\norm{\, . \,}_{s+1}^2 + \norm{\, . \,}_{s}^2)^{1/2}$. 
  This is a Hilbert-space with equivalent
  norm $\norm{\, . \,}_{s+1} + \norm{\, . \,}_{s}$.
  We get from Riesz' theorem that weak convergence is equivalent to  weak 
  convergence by components, so that $(u_{\ell_n}
  , \partial_t u_{\ell_n}) \rightharpoonup (u, \partial_t u)$ here.
  It 
  follows
  \begin{align*}
    \delta \ge \liminf E_{s+1}(u_{\ell_n} - u_{00}) & 
    \ge \tfrac{1}{\sqrt{2}} \liminf 
    \norm{(u_{\ell_n} - y_0, \partial_t u_{\ell_n})- y_1} \\
    &\ge \tfrac{1}{\sqrt{2}} 
    \norm{(u - y_0,\partial_t u - y_1)}
    \ge E_{s+1}(u - u_{00}) \text{ pointwise in } t.
  \end{align*}

  It was shown that $\partial_t u_{\ell}
  \rightharpoonup \partial_t u$ in $H^{s}$. By virtue of Rellich's theorem
  we get $\partial_t u_{\ell} \rightarrow \partial_t u$ in $H^{s-1}$ and therefore
  \begin{multline*}
    \norm{\partial_t u(t) - \partial_t u(t')}_{s-1} \le \norm{\partial_t u(t) -
  \partial_t u_{\ell}(t)}_{s-1} + \norm{\partial_t u_{\ell}(t) - 
  \partial_t u_{\ell}(t')}_{s-1} 
  \\
  {}+ \norm{\partial_t u_{\ell}(t') - 
  \partial_t u(t')}_{s-1} \le 2\varepsilon + L'\abs{t - t'}
  \end{multline*}
  if $\ell$ is chosen large enough.
\end{proof}

\begin{proof}[\textbf{Proof of Proposition \ref{prop:e_estimate}}]
To derive estimate \eqref{eq:e_estimate} we introduce the energy
\begin{multline}
E(u - u_{00}) =  \mu \norm{u - y_0}^2_{L^2} - \langle g^{00}(v, Dv) (
\partial_t u - y_1) , \partial_t u - y_1
\rangle \\
{}+ \langle g^{ij}(v, Dv) \partial_i (u - y_0), \partial_j (u
- y_0) \rangle
\end{multline}
Here $\langle \, , \, \rangle$ denotes the $L^2$ scalar product.
This energy is equivalent to $E_1(u - u_{00})$, namely
\begin{gather*}
  E_1^2(u - u_{00}) \le c_E 
E(u - u_{00})\text{ and }
\\
E(u - u_{00}) \le ( \mu + \norm{g^{00}}_{\infty}
+  \norm{(g^{ij})}_{e,\infty}) E_1^2(u - u_{00})\le \hat{C}, 
E_1^2(u - u_{00})
\\
\text{where } c_E \text{ is defined prior to Lemma \ref{u00}  and $\hat{C}$
  is defined in Lemma \ref{u00}.}
\end{gather*}
We want to estimate the change of $E$ in time. 
$E$ is not differentiable because of the coefficients, so we have to use the 
$\limsup$.
By considering the modified equation
\begin{equation}
  \label{eq:modified}
  g^{00} \partial_t^2 u + 2 g^{0j} \partial_j (\partial_t u-y_1) 
  + g^{ij} \partial_i \partial_j (u-y_0) = f - 2g^{0j} \partial_j y_1
  - g^{ij} \partial_i \partial_j y_0
\end{equation}
and omitting the argument $u - u_{00}$ of $E$ we get
\begin{equation*}
\begin{split}
 \limsup_{\tau}&  \tfrac{1}{\tau}\bigl( E(t+\tau) - E(t)
\bigr) = \\
& 2\langle - f + 2g^{0j} \partial_j (\partial_t u - y_1)
+ g^{ij} \partial_i \partial_j(u-y_0) +  2 g^{0j} \partial_j y_1
+ g^{ij} \partial_i \partial_j y_0,  \partial_t u - y_1 \rangle \\
 {}+ {}&{}2 \langle \partial_i (\partial_t u - y_1), g^{ij} \partial_j (u - y_0)
 \rangle  +  2 \langle \partial_i y_1, g^{ij} \partial_j (u - y_0)
 \rangle \\
{} + {}& 2 \mu
 \langle u - y_0, \partial_t u - y_1
 \rangle  + 2 \mu \langle u - y_0, y_1
 \rangle\\
{}- {}&\langle \partial_t u - y_1 ,\limsup \tfrac{1}{\tau}\bigl( 
g^{00}(t+\tau) - 
g^{00}(t)
\bigr) \partial_t u  - y_1\rangle \\
{}+ {}& \langle \partial_i (u - y_0)
, \limsup \tfrac{1}{\tau}\bigl( g^{ij}(t+\tau) - 
g^{ij}(t)
\bigr)
\partial_j (u - y_0) \rangle 
\end{split}
\end{equation*}
Terms with second-order derivatives need a closer examination. It holds that
$$ \langle \partial_i (\partial_t u - y_1), g^{ij} \partial_j (u - y_0)
 \rangle = - \langle \partial_t u - y_1 , g^{ij} \partial_i \partial_j
(u - y_0)\rangle - \langle \partial_t u - y_1, \partial_i g^{ij} 
\partial_j (u - y_0) \rangle 
$$ and
$$ \langle g^{0j} \partial_j (\partial_t u - y_1) , \partial_t u - y_1 \rangle
 = - \langle  \partial_t u - y_1, g^{0j} \partial_j (\partial_t u - y_1)\rangle
- \langle \partial_j g^{0j}  (\partial_t u - y_1) , \partial_t u - y_1 \rangle.
$$ 
Therefore, the terms consisting $g^{0j} \partial_j$ and 
$g^{ij} \partial_i \partial_j$
cancel. The coefficients are supposed to be Lipschitz-continuous w.r.t.\ 
$H^{s-1}_{
\mathrm{ul}} 
\hookrightarrow C_b$. Hence 
$$\bnorm{\bigl(g^{\mu\nu}(t,v(t))\bigr) 
  - \bigl(g^{\mu\nu}(t',v(t'))\bigr)}_{e,s-1,\mathrm{ul}} \le 
\nu \abs{t - t'} + \theta 
E_{s}(v(t) - v(t'))
$$
which can be estimated by means of 
inequality \eqref{eq:Es_lip_t}.
This estimate and the Sobolev embedding theorem applied to 
$\norm{(g^{\mu\nu})}_{\infty}$
and $\norm{(Dg^{\mu\nu})}_{\infty}$ yield
\begin{eqnarray*}
 \limsup_{\tau} \tfrac{1}{\tau}\bigl( E(t+\tau) - E(t)
\bigr) & \le  &  c\bigl( \norm{f}_{L^2} + ( \mu 
+  K )
E_1(u_{00}) 
\bigr) E_1(u-u_{00}) \\
&&{}+ c \bigl( \mu +  K +  \nu + \theta (\delta + E_{s+1}(u_{00}) + L')
\bigr) E_1^2(u-u_{00}).
\end{eqnarray*}
By using the equivalence of $E$ and $E_1$ it follows that
\begin{align*}
  && E^{1/2}(t) & \le e^{\hat{C}_2 t} \bigl( 
E^{1/2}(0) + \hat{C}_1 t
\bigr) \nonumber
\\
  \text{with} &&\hat{C}_1 & = c~ c_E^{1/2} \bigl( \norm{f}_{L^{\infty}L^2} + ( \mu 
+  K )
E_1(u_{00}) 
\bigr)
\\
\text{and} &&\hat{C}_2 & = c ~ c_E\bigl( \mu +  K +  \nu + \theta (\delta 
+ E_{s+1}(u_{00})
+ L')
\bigr).
\end{align*}
For the energy $E_1$ we have the following estimate.
\begin{equation}
\label{energy1}
 E_1\bigl(u(t) - u_{00}\bigr)
 \le c_E^{1/2} E^{1/2}(t) \le  c_E^{1/2} e^{\hat{C}_2 t} \bigl( 
\tilde{C} E_1(\init{u} - u_{00}) + \hat{C}_1 t
\bigr).
\end{equation}

To obtain estimates for higher derivatives of the solution $u$ we have to 
take the 
differentiability of $u$ into account. 
Define a new energy $\tilde{E}$ by
\begin{equation}
  \label{eq:energy_diff}
  \tilde{E} = \sum_{\abs{\beta} \le s} E^{\beta}
  = \sum_{\abs{\beta} \le s}E(\partial^{\beta} u).
\end{equation}
From the equivalence of the energies $E$ and $E_1$ it follows that
\begin{gather*}
  E_{s+1} \le \sum_{\abs{\beta} \le s} E_1(\partial^{\beta} u ) \le 2 E_{s+1}
\end{gather*}
and therefore
\begin{equation}
\label{equiv}
E_{s+1}^2 \le \tau_{s} \sum_{\abs{\beta} \le s} E_1^2(\partial^{\beta}
u) \le c_E \tau_{s} \tilde{E} \quad\text{and}\quad
\tilde{E} \le \sum_{\abs{\beta} \le
s} \tilde{C} E_1^2(\partial^{\beta} u) \le 4 \tilde{C} E_{s+1}^2,
\end{equation}
where $\tau_{s} $ is defined prior to Lemma \ref{u00}.

We want to estimate $\tilde{E}$ but we cannot handle the occurring term
$\partial^{\beta} \partial_t^2 u$, so we have to use a mollified version of the
energy. Suppose $J_{\varepsilon}$ is a standard Friedrich's mollifier
for $\rr^m$. 
Define
\begin{gather*}
  \tilde{E}_{\varepsilon} = \sum_{\abs{\beta} \le s} 
E^{\beta}_{\varepsilon}
  = \sum_{\abs{\beta} \le s}E(\partial^{\beta} J_{\varepsilon} u).
\end{gather*}
This mollified energy satisfies $\tilde{E}_{\beta} \rightarrow \tilde{E}$
since $J_{\varepsilon} h \rightarrow h$ in $H^{\ell}$ if $h \in H^{\ell}$.
It is not admitted to differentiate the equation and then apply the
mollifier because the equation holds only in $H^{s-1}$. 
We have to find an equation for the term $g^{00}\partial^{\beta} J_{\varepsilon} \partial_t^2 u$ occurring in the $\limsup$ of the energy $\tilde{E}_{\varepsilon}$
which holds within $L^2$.
By virtue of Lemma \ref{inv_ul} it follows that $\tfrac{1}{g^{00}} \in H^{s}$,
so we can use the equation
\begin{equation}
\label{diffequation}
 g^{00} \partial^{\beta} J_{\varepsilon} \partial^2_t u = 
g^{00} \partial^{\beta} J_{\varepsilon} \bigl(-\tfrac{2}{g^{00}} g^{0j} 
\partial_j \partial_t u - \tfrac{g^{ij}}{g^{00}} \partial_i \partial_j u
+ \tfrac{1}{g^{00}} f \bigr)
\end{equation}
It gives us for example
\begin{equation}
\label{komeq}
g^{00} \partial^{\beta} J_{\varepsilon} \tfrac{g^{ij}}{g^{00}} \partial_i 
\partial_j u = g^{ij}\partial^{\beta} J_{\varepsilon} \partial_i 
\partial_j u {}+ g^{00} \bigl[ \partial^{\beta}, \tfrac{g^{ij}}{g^{00}}
\bigr] J_{\varepsilon} \partial_i \partial_j  u  
{} + g^{00} \partial^{\beta}
\bigl[  J_{\varepsilon}, \tfrac{g^{ij}}{g^{00}}
\bigr] \partial_i \partial_j u. 
\end{equation}
The commutators can be estimated by inequality \eqref{diff_komm} and we arrive
at
\begin{multline*}
  \bnorm{g^{00} \bigl[ \partial^{\beta}, \tfrac{g^{ij}}{g^{00}}
    \bigr] J_{\varepsilon} \partial_i \partial_j  u}_{L^2} \le 
  \norm{g^{00}}_{\infty} \bnorm{[\partial^{\beta}, \tfrac{g^{ij}}{g^{00}}] 
    J_{\epsilon} \partial_i \partial_j u}_{L^2} 
  \\
  \le c \norm{g^{00}}_{\infty}
  \norm{(g^{ij})}_{e,s,\mathrm{ul}} \bnorm{\tfrac{1}{g^{00}}}_{s,\mathrm{ul}} \norm{u}_{s+1}
\end{multline*}
and by the properties of the mollifier it follows that
\begin{gather}
  \label{eq:comm_didju}
  \bnorm{g^{00} \partial^{\beta}
\bigl[  J_{\varepsilon}, \tfrac{g^{ij}}{g^{00}}
\bigr] \partial_i \partial_j u }_{L^2} \le \norm{g^{00}}_{\infty} 
\bnorm{\bigl[  J_{\varepsilon}, \tfrac{g^{ij}}{g^{00}}
\bigr] \partial_i \partial_j u }_{s} \le c 
\norm{g^{00}}_{\infty} \bnorm{\tfrac{(g^{ij})}{g^{00}}}_{e,C^1} \norm{u}_{s+1}.
\end{gather}
Our goal is to estimate $E_{s+1}(u - u_{00})$, so we have to modify
the RHS of equation \eqref{komeq} in the same way as equation \eqref{second}
was modified to obtain equation \eqref{diffequation}.
No differentiability issues occur since we required $y_0$
to be $H^{s+2}$ and $y_1$ to be $H^{s+1}$.
The commutators, where $u$ is replaced by $u - u_{00}$ and $u_{00}$ can be
estimated by inequality \eqref{eq:comm_didju} providing us with the bounds
$$  c K^2 \bnorm{\tfrac{1}{g^{00}}}_{s,\mathrm{ul}} E_{s+1}(u - u_{00}) \quad\text{and}
\quad cK^2 \bnorm{\tfrac{1}{g^{00}}}_{s,\mathrm{ul}} E_{s+1}(u_{00})
$$
The term in equation \eqref{diffequation} containing the RHS $f$ can be 
estimated by
$$ \bnorm{g^{00} \partial^{\beta} \tfrac{1}{g^{00}} f}_{L^2}
\le c \norm{g^{00}}_{\infty} \bnorm{\tfrac{1}{g^{00}}}_{s,\mathrm{ul}} \norm{f}_{s}
\le c K \bnorm{\tfrac{1}{g^{00}}}_{s,\mathrm{ul}} \norm{f}_{\infty,s}
$$
The terms without a commutator can be treated as in the above case.
We conclude the following inequality for the mollified energy 
\begin{equation*}
\begin{split}
\limsup_{\tau}&  \tfrac{1}{\tau}\bigl( E^{\beta}_{\varepsilon}(t+\tau) - 
E^{\beta}_{\varepsilon}(t)
\bigr) \le  C_1 E_{s+1}(u - u_{00}) + C_2 E_{s+1}^2(u - u_{00}).
\end{split}
\end{equation*}
In the definition of $C_1$ and $C_2$, 
$c$ denotes a constant depending on $m,k$ and on
the mollifier.
Using the equivalence stated in \eqref{equiv} we derive
\begin{multline}
   \label{eq:est_diffE}
\limsup_{\tau} \tfrac{1}{\tau}\bigl( \tilde{E}_{\varepsilon}(t+\tau) - 
\tilde{E}_{\varepsilon}(t)
\bigr) \le \sum_{\abs{\beta} \le k-1}
\limsup_{\tau} \tfrac{1}{\tau}\bigl( E^{\beta}_{\varepsilon}(t+\tau) 
- E^{\beta}_{\varepsilon}(t)
\bigr) \\
\le 
c_E^{1/2} \tau^{3/2}_{k-1} 
C_1 \tilde{E}^{1/2} + c_E \tau_{k-1}^2 C_2 \tilde{E}.
\end{multline}
To perform the limit $\varepsilon \rightarrow 0$ we integrate this inequality
$$ \tilde{E}_{\varepsilon}(t + \tau) \le \tilde{E}_{\varepsilon}(t) + \int_t^{t+\tau} 
\bigl( c_E^{1/2} \tau^{3/2}_{k-1} 
C_1 \tilde{E}^{1/2} + c_E \tau_{k-1}^2 C_2 \tilde{E}\bigr)\, ds
$$
Taking the limit and again considering the $\limsup$ we obtain
$$ \limsup_{\tau} \tfrac{1}{\tau}\bigl( \tilde{E}(t+\tau) - 
\tilde{E}(t)
\bigr) \le 
c_E^{1/2} \tau^{3/2}_{k-1} 
 C_1 \tilde{E}^{1/2}(t) + c_E \tau_{k-1}^2 C_2 \tilde{E}(t).
$$
With the help of considerations to be found in \cite{Sogge:1995}, we
get
\begin{gather*}
\limsup_{\tau} \tfrac{1}{\tau}\bigl( \tilde{E}^{1/2}(t+\tau) - 
\tilde{E}^{1/2}(t)
\bigr) \le c_E^{1/2} \tau^{3/2}_{k-1} 
C_1  + c_E \tau_{k-1}^2 C_2 \tilde{E}^{1/2}
\end{gather*}
and further
\begin{gather*}
\tilde{E}^{1/2}(t) \le e^{c_E \tau_{k-1}^2 C_2 t}(\tilde{E}^{1/2}(0) + 
c_E^{1/2} \tau^{3/2}_{k-1} C_1 t).
\end{gather*}
Using again the equivalence stated in \eqref{equiv} it follows
\begin{equation*}
\begin{split}
E_{s+1}(u - u_{00}) \le c_E^{1/2} \tau_{s}^{1/2} \tilde{E}^{1/2}(t) \le 
e^{c_E 
\tau_{s}^2 C_2 t}\bigl( C^{1/2} E_{s+1}(\init{u} - u_{00}) + 
c_E \tau^2_{s} C_1 t\bigr),
\end{split}
\end{equation*}
where the constant $C$ is taken from Lemma \ref{u00}.
\end{proof}

\begin{lem}
  \label{conti}
  A solution $u$ to the quasilinear second-order equation \eqref{second}
  obtained by Theorem \ref{qlin_ex}
  satisfies $u \in C([0,T'], H^{s+1}) \cap C^1([0,T'], H^s)$.
\end{lem}
\begin{proof}%
  Firstly we will show that we have weak continuity.
  Only $u$ itself will be treated, the case for $\partial_t u$ is analog.
  If $t_n \rightarrow t$ 
  is a sequence, then we know $\bigl(u(t_n)\bigr)$ is bounded in $H^{s+1}$.
  There is a weakly convergent subsequence $u(t_{n_k}) \rightharpoonup
  v \in H^{s+1}$. By Rellich's theorem it follows that $u(t_{n_k}) \rightarrow v$
  in $H^s$. Since $u \in C([0,T], H^s)$ it holds $v = u(t)$. 
  We conclude that every 
  convergent subsequence of $\bigl(u(t_n)\bigr)$ has the limit $u(t)$. 
  If we can find
  a subsequence not converging to $u(t)$, then we can apply the argument
  since this subsequence is also bounded in $H^{s+1}$.

  The next step is to show that $\norm{u(t)}_{s+1}$ is continuous.
  This will be done with the help of the continuity of the energy $\tilde{E}$
  (cf. \ref{eq:energy_diff}). Integrating estimate \eqref{eq:est_diffE}
  for the $\limsup$
  of $\tilde{E}_{\varepsilon}$ gives us for $\tau > 0$
  \begin{gather*}
    \tilde{E}_{\varepsilon}(t + \tau) \le \tilde{E}_{\varepsilon}(t) + 
    \int_t^{t+\tau} \bigl(
    c_E^{1/2} \tau^{3/2}_{k-1} 
    C_1 \tilde{E}_{\varepsilon}^{1/2} + c_E \tau_{k-1}^2 C_2 \tilde{E}_{\varepsilon}
    \bigr)\, ds.
  \end{gather*}
  Passing to the limit $\varepsilon \rightarrow 0$ we derive 
  $ \lim_{\tau \rightarrow 0} \tilde{E}(t + \tau) \le \tilde{E}(t)$. The same 
  argument gives us $\lim_{\tau \rightarrow 0} \tilde{E}(t - \tau) \ge \tilde{E}(t)$.
  The other parts of the semi-continuity follow from the invariance of the
  equation under time reversal. It follows 
  \begin{gather*}
    \lim_{\tau \rightarrow 0} \tilde{E}(T - (r + \tau)) \le \tilde{E}(T - r)
    \text{ and }
    \lim_{\tau \rightarrow 0} \tilde{E}(T - (r - \tau)) \ge \tilde{E}(T - r).
  \end{gather*}
  By setting $T - r = t$, the desired continuity of the energy
  $\tilde{E}$ follows.

  We have now that $\tilde{E}$ is continuous and $(u,\partial_t u)
  \in C^w([0,T], H^s) \times
  C^w([0,T], H^{s-1})$.
  Only the highest derivatives are interesting, so we assume $\abs{\beta} = s$.
  We will show that 
  \begin{gather*}
    \mu \norm{\partial^{\beta} (u(t_n)
    - u(t))}^2_{L^2} - 
    \langle g^{00} 
    \partial^{\beta}  (\partial_t u(t_n) - \partial_t u(t)) ,  \partial^{\beta} 
    (\partial_t u(t_n) - \partial_t u(t))
    \rangle \\
    {}+ \langle g^{ij}\partial_i \partial^{\beta} ( u(t_n) - u(t)), \partial_j 
    \partial^{\beta} (u(t_n) - u(t)) \rangle
  \end{gather*}
  converges to $0$, if $t_n \rightarrow t$.
  Expanding this term gives us
  \begin{gather*}
    \label{tedious}
    \mu \norm{\partial^{\beta} (u(t_n)- y_0)}_{L^2}^2 
    + \mu \norm{\partial^{\beta} (u(t)- y_0)}_{L^2}^2
    - 2 \mu \langle \partial^{\beta} (u(t_n) - y_0), \partial^{\beta} (u(t) - y_0)
    \rangle \\
    - \langle g^{00}(t) \partial^{\beta}  (\partial_t u(t_n) - y_1), \partial^{\beta} 
    (\partial_t u(t_n) - y_1)  
    \rangle - \langle g^{00}(t) \partial^{\beta}  (\partial_t u(t) - y_1), 
    \partial^{\beta} 
    (\partial_t u(t) - y_1)
    \rangle \\
    + 2\langle g^{00}(t) \partial^{\beta}  (\partial_t u(t_n) - y_1),
    \partial^{\beta}  (\partial_t u(t) - y_1)\rangle 
    + \langle g^{ij}(t) \partial_i \partial^{\beta}  (u(t_n) - y_0),
    \partial_j \partial^{\beta}  (u(t_n) - y_0) \rangle \\
    + \langle g^{ij}(t) \partial_i \partial^{\beta}  (u(t) - y_0),
    \partial_j \partial^{\beta}  (u(t) - y_0) \rangle - 2\langle g^{ij}(t )
    \partial_i \partial^{\beta}  (u(t_n) - y_0),
    \partial_j \partial^{\beta}  (u(t) - y_0) \rangle
  \end{gather*}
  The mixed terms converge to $- 2 E(u(t))$ due to the weak continuity. 
  Exemplarily we consider the following term to examine the other terms
  without $t$ and $t_n$ mixed. It holds that
  \begin{multline*}
    - \langle g^{00}(t) \partial^{\beta}  (\partial_t u(t_n) - y_1), 
    \partial^{\beta} 
    (\partial_t u(t_n) - y_1)  
    \rangle 
    = \\
    - \langle g^{00}(t_n) \partial^{\beta}  (\partial_t u(t_n) - y_1), 
    \partial^{\beta} 
    (\partial_t u(t_n) - y_1)  
    \rangle\\
    {}-  \langle (g^{00}(t) - g^{00}(t_n))
    \partial^{\beta}  (\partial_t u(t_n) - y_1), \partial^{\beta} 
    (\partial_t u(t_n) - y_1)  
    \rangle.
  \end{multline*}
  The last term has to vanish in the limit which follows from
  \begin{multline*}
    \abs{\langle (g^{00}(t) - g^{00}(t_n))
    \partial^{\beta}  (\partial_t u(t_n) - y_1), \partial^{\beta} 
    (\partial_t u(t_n) - y_1)  
    \rangle} 
  \\
  \le \norm{g^{00}(t) - g^{00}(t_n)}_{\infty} (\norm{\partial_t 
    u(t_n)}_s + 
    \norm{y_1}_s)^2.
  \end{multline*}
  The assumptions for the linear equation yield that the coefficients are
  Lipschitz w.r.t.\ $H^{s-1}_{\mathrm{ul}}
  \hookrightarrow C_b$.
  The structure of the expanded term is 
  \begin{gather*}
    E(u(t_n)) + E(u(t)) + \text{mixed
  terms} + \text{terms involving } g^{\mu\nu}(t) - g^{\mu\nu}(t_n).
  \end{gather*}
  This converges 
  to $0$.
\end{proof}

\section{Matrix computations}

First, we state some
computations for the inverse of the metric.
\begin{lem}
  \label{lem:coeff_inverse}
  Suppose $(a_{\mu\nu})$ is a symmetric $(m+1)$$\times$$(m+1)$-matrix satisfying
  $a_{00} < 0$ and $(a_{ij}) > 0$ where $(a_{ij})$ is the submatrix of 
  $(a_{\mu\nu})$. Let $(\bar{a}^{ij})$ denote the inverse of the submatrix
  $(a_{ij})$. Then we have for the components of the inverse $(a^{\mu\nu})$
  \begin{gather*}
    a^{00} = \tfrac{1}{a_{00}}\bigl(1+ \tfrac{1}{a_{00}} a_{0i} a_{0j} a^{ij}\bigr),
    \qquad a^{0j} = - \tfrac{1}{a_{00}} a^{jk} a_{0k}, \\
    a^{ik}\bigl( a_{kj} - \tfrac{1}{a_{00}} a_{0k} a_{0j}
    \bigr) = \delta_j^i
  \end{gather*}
  or 
  \begin{gather*}
    a^{00} = (a_{00} - a_{0i} a_{0j} \bar{a}^{ij})^{-1}, \qquad 
    a^{0j} = - a^{00} a_{0k}
    \bar{a}^{jk}, \\
    a^{ij} = a^{00} \bar{a}^{i\ell} a_{0\ell} a_{0k} \bar{a}^{kj} + \bar{a}^{ij}
  \end{gather*}
\end{lem}
\begin{proof}
  The results follow from the definition of the inverse after tedious
  computations.
\end{proof}
\begin{lem}
  \label{lem:est_metric_inverse}
  Assume $A = (a_{\mu\nu})$ is an $m+1$$\times$$m+1$ matrix with 
  $a_{00} \le - C_1 < 0$ and $a_{ij} \ge C_2 \delta_{ij}$ where $C_1$ and
  $C_2$ are positive constants. Let $\bar{a}^{ij}$ denote the inverse to the
  submatrix $a_{ij}$. Then it follows
  \begin{align*}
    a^{00} & 
    \le -\bigl( \abs{a_{00}} + m^{1/2} C_2^{-1} \abs{(a_{0\ell})}^2_e 
    \bigr)^{-1}, &
    a^{ij} & \ge \bigl( \abs{(a_{ij})}_e + C_1^{-1} \abs{(a_{0\ell})}^2_e
    \bigr)^{-1}
    \delta^{ij}
  \end{align*}
  and for the norm of the inverse it holds
  \begin{gather*}
    \abs{A^{-1}}^2_e \le C_1^{-2} + 
    2 m C_1^{-2} C_2^{-2} \abs{(a_{0\ell})}_e^2
    + m C_2^{-2}.
  \end{gather*}
\end{lem}

\begin{proof}
  Consider first the inverse $(\bar{a}^{ij})$ of the submatrix.
  The following estimates hold
  \begin{gather*}
    \abs{(a_{ij})}_e^{-1} \delta^{ij} \le \bar{a}^{ij} \le C_2^{-1} \delta^{ij}.
  \end{gather*}
  Therefore we have an estimate for the norm $\abs{(\bar{a}^{ij})}_e \le
  m^{1/2} C_2^{-1}$. We are now heading to the first claim.
  From Lemma \ref{lem:coeff_inverse} we get a description of $a^{00}$ so
  we have to estimate $a_{00} - a_{0i} a_{0j} \bar{a}^{ij}$.
  Estimating the absolute value gives us
  \begin{align*}
    \abs{a_{00} - a_{0i} a_{0j} \bar{a}^{ij}} & = 
    - a_{00} + a_{0i} a_{0j} \bar{a}^{ij} 
    \le \abs{a_{00}} + \abs{(a_{0\ell})}^2_e \abs{(\bar{a}^{ij})}_e 
  \end{align*}
  and $- a_{00} + a_{0i} a_{0j} \bar{a}^{ij}  \ge  C_1 $.
  Therefore the estimates on $a^{00}$ follow.
  
  For the positive definiteness of $a^{ij}$ consider 
  the norm of the inverse $a_{ij} - \tfrac{1}{a_{00}} a_{0i} a_{0j}$.
  The inequality $c^{ij}\ge \abs{(c_{ij})}_e^{-1} \delta^{ij}$ for a matrix 
  $(c_{ij})$ gives the desired
  estimate. For the norm of $(a^{ij})$ consider the positive definiteness
  of the inverse. We have
  $(a_{0i} a_{0j}) \ge 0$ and therefore 
  $a_{ij} - \tfrac{1}{a_{00}} a_{0i} a_{0j} \ge C_2 \delta_{ij}$.
  The norm of $a^{0\ell}$ can be estimated via
  $\abs{(a^{0\ell})}_e \le C_1^{-1} \abs{(a_{0\ell})}_e \abs{(\bar{a}^{ij})}_e$.
  To derive an estimate for the norm of the full matrix $(a^{\mu\nu})$
  we assemble the results. 
  It follows
  \begin{multline*}
    \tsum_{\mu\nu} \abs{a^{\mu\nu}}^2  =  \abs{a^{00}}^2 + 2\abs{(a^{0\ell})}^2_e
      + \abs{(a^{ij})}^2_e  
      \le C_1^{-2} + 2 C_1^{-2} \abs{(a_{0\ell})}_e^2
      \abs{(\bar{a}^{ij})}_e^2 + mC_2^{-2}
  \end{multline*}
  Together with $\abs{(\bar{a}^{ij})}_e \le m^{1/2} C_2^{-1}$ 
  this yields the claim.
\end{proof}
The next lemma establishes a result similar to Lemma \ref{inverse} for 
differentiable matrices.
\begin{lem}
  \label{lem:inverse_der_high}
  Let $A = (a_{\mu\nu})$  be an $(m+1)$$\times$$(m+1)$ 
  matrix-valued function defined on $\rr^m$ such that
  $a_{00} \le - \lambda$ and $a_{ij} \ge \mu \delta_{ij}$. Then
  it holds 
  \begin{gather*}
    \abs{D^k A^{-1}}_{e} \le c \, \delta^{-1}\bigl(1 + (\delta^{-1}
    \norm{A}_{e,C^k})^k\bigr)
  \end{gather*}
  where $\delta^{-1}$ denotes the bound for $A^{-1}$ as constructed in
  Lemma \ref{lem:est_metric_inverse}.
\end{lem}

\begin{proof}
  The proof is due to \cite{Kato:1975}. 
  From the generic term 
  \begin{gather*}
    \partial^{\alpha} A^{-1} = \tsum_{\alpha_1 + \cdots + \alpha_k = \alpha}
    A^{-1}\partial^{\alpha_1} A \,A^{-1} \cdots
    \partial^{\alpha_k} A \,A^{-1}
  \end{gather*}
  for a multi-index $\alpha \in N_0^m$ we derive with $u_j = \partial^{\alpha_j}
  A\, A^{-1}$ and
  $\abs{A^{-1}} \le \delta^{-1}$ that
  \begin{gather*}
    \abs{\partial^{\alpha} A^{-1}}_e \le 
    \delta^{-1} \tsum \abs{u_1 \cdot\cdots\cdot
      u_k}_e \le \delta^{-1}\tsum \abs{D^{\abs{\alpha_1}} A}_e \delta^{-1}
    \cdots 
    \abs{D^{\abs{\alpha_k}} A}_e\delta^{-1}.
  \end{gather*}
  Assuming $\abs{D^{\ell} A}_e \le C_{\ell} $ yields
  \begin{gather*}
    \abs{\partial^{\alpha} A^{-1}}_e \le \delta^{-1}\tsum C_{\abs{\alpha_1}} 
    \delta^{-1}
    \cdots C_{\abs{\alpha_k}} \delta^{-1} = 
     \delta^{-1} \tsum C_{\abs{\alpha_1}}\delta^{-1}
    \cdots C_{\abs{\alpha_k}}\delta^{-1}.
  \end{gather*}
  With the norm $\norm{\,.\,}_{C^s}$ defined in
  \eqref{eq:norm_Cs} we get $\norm{A}_{e,C^k}^2 \le \tsum_{\ell \le k}
  (C_{\ell})^2$. This yields
  \begin{gather*}
    \abs{\partial^{\alpha} A^{-1}}_e \le c \,  \delta^{-1}
    (1 + B + \cdots + B^{\abs{\alpha}})
  \end{gather*}
  where $c$ is a constant dependent on the dimension and $k$ and 
  $B = \delta^{-1} \norm{A}_{e,C^k}$. Summing over $\abs{\alpha} = k$
  and using the inequality $(a + b)^p \le 2^p(a^p + b^p)$
  gives the desired estimate.
\end{proof}
The proof can also be applied to a positive definite matrix-valued function.
We state the result for reference.
\begin{cor}
  \label{cor:pos_matrix_der}
    Let $A = (a_{ij})$  be an $n$$\times$$n$ 
  matrix-valued function defined on $\rr^m$ such that
  $a_{ij} \ge \mu \delta_{ij}$ for a constant $\mu > 0$. Then
  it holds 
  \begin{gather*}
    \abs{D^k A^{-1}}_{e} \le c \, \delta^{-1}\bigl(1 + (\delta^{-1}
    \norm{A}_{e,C^k})^k\bigr)
  \end{gather*}
  with $\delta^{-2} = \tfrac{n}{\mu}$.
\end{cor}

\end{appendices}

\renewcommand{\refname}{Bibliography}
\bibliographystyle{amsalpha}
\addcontentsline{toc}{section}{Bibliography}
\bibliography{notes} 

\end{document}